\documentclass[conference,review,anonymous]{IEEEtran}
\IEEEoverridecommandlockouts

\usepackage[noadjust]{cite}

\input{env.sty}

\usepackage{titletoc}

\allowdisplaybreaks

\makeatletter
\newcommand{\linebreakand}{%
  \end{@IEEEauthorhalign}
  \hfill\mbox{}\par
  \mbox{}\hfill\begin{@IEEEauthorhalign}
}
\makeatother

\begin{document}

\title{A Complete Equational Theory for Quantum~Circuits}
\thispagestyle{plain}
\pagestyle{plain}

\author{\IEEEauthorblockN{Alexandre Clément\IEEEauthorrefmark{1}}
\IEEEauthorblockA{\href{mailto:alexandre.clement@loria.fr}{alexandre.clement@loria.fr}\\
\url{https://orcid.org/0000-0002-7958-5712}\\
\url{https://members.loria.fr/AClement}}
\and
\IEEEauthorblockN{Nicolas Heurtel\IEEEauthorrefmark{2}\IEEEauthorrefmark{3}}
\IEEEauthorblockA{\href{mailto:nicolas.heurtel@quandela.com}{nicolas.heurtel@quandela.com}\\
\url{https://orcid.org/0000-0002-9380-8396}}
\and
\IEEEauthorblockN{Shane Mansfield\IEEEauthorrefmark{2}}
\IEEEauthorblockA{\href{mailto:shane.mansfield@quandela.com}{shane.mansfield@quandela.com}}
\linebreakand
\IEEEauthorblockN{Simon Perdrix\IEEEauthorrefmark{1}}
\IEEEauthorblockA{\href{mailto:simon.perdrix@loria.fr}{simon.perdrix@loria.fr}\\
\url{https://orcid.org/0000-0002-1808-2409}\\
\url{https://members.loria.fr/SPerdrix}}
\and
\IEEEauthorblockN{Beno\^\i t Valiron\IEEEauthorrefmark{3}}
\IEEEauthorblockA{\href{mailto:benoit.valiron@universite-paris-saclay.fr}{benoit.valiron@universite-paris-saclay.fr}\\
\url{https://orcid.org/0000-0002-1008-5605}\\
\url{https://www.monoidal.net}}
\linebreakand
\IEEEauthorblockA{\IEEEauthorrefmark{1}Universit\'e de Lorraine, CNRS,\\
Inria, LORIA\\
F-54000 Nancy, France}
\and
\IEEEauthorblockA{\IEEEauthorrefmark{2}Quandela\\
7 Rue Léonard de Vinci\\
91300 Massy, France}
\and
\IEEEauthorblockA{\IEEEauthorrefmark{3}Université Paris-Saclay, CentraleSupélec,\\
Inria, CNRS, ENS Paris-Saclay,\\
Laboratoire Méthodes Formelles\\
91190, Gif-sur-Yvette, France}
}

\maketitle

\begin{abstract}
We introduce the first complete equational theory for
  quantum circuits. More precisely, we introduce a set of circuit
  equations that we prove to be sound and complete: two circuits
  represent the same unitary map if and only if they can be
  transformed one into the other using the equations. The proof is
  based on the properties of multi-controlled gates -- that are
  defined using elementary gates -- together with an encoding of
  quantum circuits into linear optical circuits, which have been
  proved to have a complete axiomatisation.
\end{abstract}

\section{Introduction}

Quantum computation is the art of manipulating the states of objects
governed by the laws of quantum physics in order to perform
computation. The standard model for quantum computation is the
\emph{quantum co-processor model}: an auxiliary device, hosting a
quantum memory. This coprocessor is then interfaced with a classical
computer: the classical computer sends the co-processor a series of
instructions to update the state of the memory. The standard formalism
for these instructions is the \emph{circuit model}
\cite{deutsch-circuit}. Akin to boolean
circuits, in quantum circuits wires represent \emph{quantum bits} and
boxes elementary operations -- \emph{quantum gates}. The mathematical
model is however very different: quantum bits (qubits) correspond to vectors in
a $2$-dimensional
Hilbert space, gates to unitary maps and parallel composition to the
tensor product -- the Kronecker product.

Quantum circuits currently form the \emph{de facto} standard for
representing low-level, logical operations on a quantum memory. They
are used for everything: resource estimation
\cite{green2013quipper}, optimization
\cite{amy2014polynomial,duncan2020graph,PhysRevA.102.022406,maslov2005quantum,maslov2008quantum,nam2018automated}, satisfaction of hardware
constraints \cite{kissinger2019cnot,nash2020quantum}, \textit{etc}.

However, as ubiquitous to quantum computing as they are, the graphical
language of quantum circuits has never been fully
formalized.
In particular, a \emph{complete
  equational theory} has been a longstanding open problem for 30 years
\cite{aaronson-slide}. It would make it possible to directly prove
properties such as circuit equivalence without having to rely on
ad-hoc set of equations. So far, complete equational theories were
only known for non-universal fragments, such as circuits acting on at
most two qubits \cite{bian2022generators,coecke2018zx}, the stabilizer
fragment \cite{makary2021generators,ranchin2014complete}, the
CNot-dihedral fragment \cite{Amy_2018}, or fragments of reversible
circuits~\cite{iwama2002transformation,cockett2018categorycnot,cockett2018categorytof}. 

Interestingly enough, other diagrammatic languages for quantum
computation have been developed on sound foundations: it is
reasonable to think
that this could help in developing a complete equational theory for circuits.
Arguably the strongest candidate has been the ZX-calculus
\cite{coecke2008interacting,coecke2011interacting},\footnote{or its
  variants like ZH \cite{Backens_2019} and ZW \cite{DBLP:conf/lics/HadzihasanovicN18}, sharing several similar
properties.} equipped with complete equational theories
\cite{jeandel2018complete,hadzihasanovic2018two,jeandel2018diagrammatic,lmcs:6532,vilmart2019near}. The ZX-calculus shares the same
underlying mathematical representation for states: wires corresponds
to Hilbert spaces and parallel composition to the tensor
operation. Nonetheless, the completeness of the
ZX-calculus does not lead \textit{a priori} to a complete equational theory for
quantum circuits. The reason lies in the expressiveness of the
ZX-calculus and the \emph{non-unitarity} of some of its generators. Any
quantum circuit can be straightforwardly seen as a ZX-diagram. On the
other hand, a ZX-diagram does not necessarily represent a unitary map, and
even when it does, extracting a corresponding quantum circuit
is known to be a hard task in general \cite{duncan2020graph,de2022circuit}.

Another example of a quantum language with a complete equational theory
is the LOv-calculus, a language for linear optical quantum circuits
for which a simple complete equational theory has recently been
introduced \cite{clement2022LOv}. While both linear optical and
regular quantum circuits are universal for unitary transformations,
they do not share the same structure. In particular, if the parallel
composition of quantum circuits corresponds to the tensor product, for
linear optical circuits it stands for the \emph{direct sum}.

In this paper, we introduce the first complete equational theory for quantum
circuits, by first closing the gap between regular and linear optical quantum
circuits.
Despite the seemingly incompatible approaches to parallel
composition, our completeness result derives from the completeness result
for linear optical circuits. Indeed, unlike ZX-generators, linear
optical components are unitary, making it possible write a translation
in both directions.

The complete equational theory for quantum circuits is derived from
the completeness of the LOv-calculus as follows: equipped with maps
for encoding (from quantum circuits to linear optical circuits) and decoding
(from linear optical circuits to quantum circuits), one can roughly speaking
prove completeness for quantum circuits as long as its equational
theory is powerful enough to derive a finite number of equations,
those corresponding to the decoding of the equations of the complete
equational theory for linear optical circuits.

 Due to the difference in its interpretation in both kinds of circuits, the parallel composition is not preserved by the encoding nor the decoding maps. The translations are actually based on a sequentialisation of circuits, since the translation of a local gate (acting on at most two wires) is translated as a piece of circuit acting potentially on all wires. Technically, it forces to work with a raw version of circuits, as a circuit may lead to a priori distinct translations depending on the choice of the sequentialisation. Moreover,  a single linear optical gate like a phase shifter (which consists in applying a phase on a particular basis state) is decoded as a piece of circuits that can be interpreted as a multi-controlled gate acting on all qubits. As we choose to stick with the usual generators of quantum circuits acting on at most two qubits, multi-controlled gates are inductively defined and we introduce an equational theory powerful enough to prove the basic algebra of multi-controlled gates, necessary to finalise the proof of completeness.

The paper is structured as follows. We first introduce a set of
structural relations for quantum circuits generated by the standard
elementary gates: Hadamard, Phase-rotations, and CNot. We define
multi-controlled gates using these elementary gates, and show that the
basic algebra of multi-controlled gates can be derived from the
structural relations.  In addition to the structural equations, we
introduce Euler-angle-based equations. We then proceed to the proof of
completeness, based on a back-and-forth translation from quantum
circuits to linear optical circuits.

\section{Quantum Circuits}
\label{sec:QC}

In quantum computation, circuits ---such as quantum circuits or
optical quantum circuits--- are graphical descriptions of quantum
processes. Akin to (conventional) boolean circuits, circuits in
quantum computations are built from wires (oriented from left to
right), representing the flow of information, and gates, representing
operations to update the state of the system.
Every circuit comes with a set of input wires (incoming the circuit
from the left) and a set of output wires (exiting the circuit on the
right).

\subsection{Graphical languages}\label{sec:raw}

To provide a formal definition of circuits, we first use the notion of
\emph{raw circuits}.\footnote{\emph{Raw terms} are for instance
  similarly used \cite{Pawel} as an intermediate step in the defintion
  of prop.}  Given a set of generators, one can generate a \emph{raw
  circuit} by means of iterative sequential ($\circ$) and parallel
($\otimes$) compositions.  For instance, given the elementary gates
$\gh$ and $\gppifour$ (with one input and one output) and $\gcnot$ (with
two inputs and two outputs), one can construct the raw circuit
$\gcnot \circ ((\gh\otimes \gppifour)\circ \gcnot)$. Notice that a
sequential composition $C'\circ C$ requires that the number of outputs
of $C$ matches the number of inputs of $C'$. This raw circuit can be
depicted by gluing the generators together and using boxes to witness
how the generators have been composed:
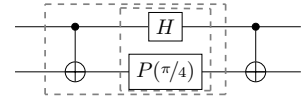
\[{\scalebox{.8}{\input{bare-circuit-ex3.tikz}}}\]
To avoid the use of boxes and recover the intuitive notion of
circuits, we formally define circuits as a prop \cite{prop}, which
consists in considering the raw circuits up to the rules given in
Figure \ref{fig:axiom}. More precisely, a prop generated by a set $\mathcal G$ of
elementary gates is the collection of raw circuits generated by
$\mathcal G\cup \{\gid, \gswap,\begin{tikzpicture}[scale=0.2]
	\begin{pgfonlayer}{nodelayer}
		\node [style=sdiagrammevide] (0) at (0, 0) {};
	\end{pgfonlayer}
\end{tikzpicture}
\}$\footnote{
  $\gid$ denotes the identity, $\gswap$ the swap and
  $\begin{tikzpicture}[scale=0.2]
	\begin{pgfonlayer}{nodelayer}
		\node [style=sdiagrammevide] (0) at (0, 0) {};
	\end{pgfonlayer}
\end{tikzpicture}
$ the empty circuit.} quotiented by the
equations of Figure \ref{fig:axiom}.

\begin{figure*}[htb]
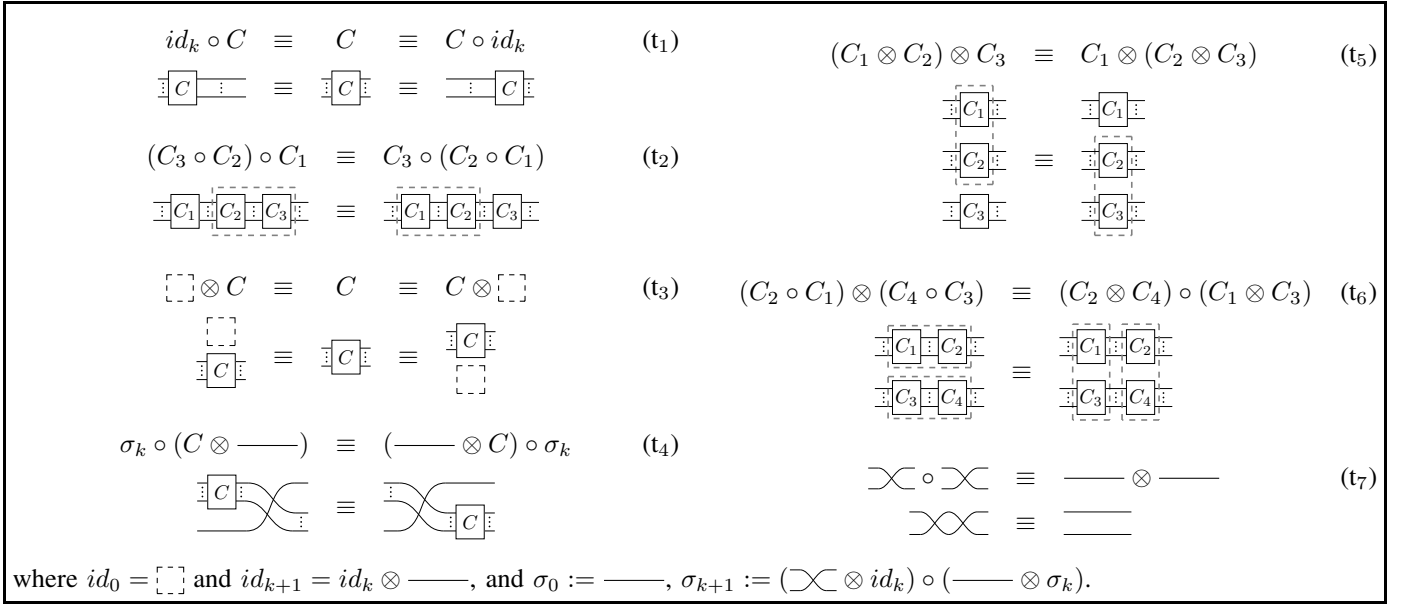

\scalebox{1}{\fbox{\begin{minipage}{1\textwidth}
\newcommand{\nspazer}{0.5em}
\vspace{-1em}
\begin{multicols}{2}%
\begin{eqnabc}\label{idneutre}\begin{array}[t]{rcccl}id_k \circ C& \equiv & C& \equiv &C\circ id_k\\[0.5em]
\minitikzfig[0.75]{Cidmultifils}&\equiv&\minitikzfig[0.75]{Cmultifils}&\equiv&\minitikzfig[0.75]{idCmultifils}
\end{array}\end{eqnabc}
\vspace{\nspazer}
\begin{eqnabc}\label{assoccomp}\begin{array}[t]{rcl}(C_3\circ C_2)\circ C_1 &\equiv & C_3\circ (C_2\circ C_1)\\[0.5em]
\minitikzfig[0.75]{C1puisC2C3}&\equiv&\minitikzfig[0.75]{C1C2puisC3}
\end{array}\end{eqnabc}
\vspace{\nspazer}
\begin{eqnabc}\label{videneutre}\begin{array}[t]{rcccl}\input{diagrammevide-s.tikz} \otimes C &\equiv &C &\equiv &C\otimes  \input{diagrammevide-s.tikz}\\[0.5em]
\minitikzfig[0.75]{diagrammevidesurC}&\equiv&\minitikzfig[0.75]{Cmultifils}&\equiv&\minitikzfig[0.75]{Csurdiagrammevide}
\end{array}\end{eqnabc}
\vspace{\nspazer}
\begin{eqnabc}\label{naturaliteswap}\begin{array}[t]{rcl}\sigma_{k}\circ (C\otimes \input{filcourt-s.tikz}) &\equiv  &(\input{filcourt-s.tikz}\otimes C) \circ     \sigma_{k}\\[0.5em]
\minitikzfig[0.75]{Cswapn1}&\equiv&\minitikzfig[0.75]{swapn1C}
\end{array}\end{eqnabc}
\vspace{\nspazer}
\begin{eqnabc}\label{assoctens}\begin{array}[t]{rcl}(C_1\otimes C_2)\otimes C_3&\equiv &C_1\otimes(C_2\otimes C_3)\\[0.5em]
\minitikzfig[0.75]{C1surC2puissurC3}&\equiv&\minitikzfig[0.75]{C1puissurC2surC3}
\end{array}\end{eqnabc}
\vspace{\nspazer}
\begin{eqnabc}\label{mixedprod}\begin{array}[t]{rcl}(C_2\circ C_1)\otimes (C_4\circ C_3) &\equiv &(C_2\otimes C_4)\circ (C_1\otimes C_3)\\[0.5em]
\minitikzfig[0.75]{C1C2surC3C4}&\equiv&\minitikzfig[0.75]{C1surC3puisC2surC4}
\end{array}\end{eqnabc}
\vspace{\nspazer}
\begin{eqnabc}\label{doubleswap}\begin{array}[t]{rcl}\gswap \circ \gswap &\equiv &\input{filcourt-s.tikz}\otimes \input{filcourt-s.tikz}\\[0.5em]
\minitikzfig[0.75]{swapswap-xxs}&\equiv&\minitikzfig[0.75]{filsparalleleslongs-xxs}
\end{array}\end{eqnabc}
\end{multicols}
where $id_0 = \input{diagrammevide-s.tikz}$ and
$id_{k+1} = id_k\otimes \input{filcourt-s.tikz}$, and
$\sigma_{0} := \input{filcourt-s.tikz}$,
$\sigma_{k+1} := (\gswap\otimes id_k )\circ
(\input{filcourt-s.tikz}\otimes \sigma_k)$.
\end{minipage}}}
\caption{Definition of $\equiv$ for raw circuits (either raw quantum
  circuits or raw optical circuits).\label{fig:axiom}}
\end{figure*}

The use of the prop formalism guarantees that circuits can be depicted
graphically without ambiguity. Circuits are thus defined up to
deformations, as for instance:
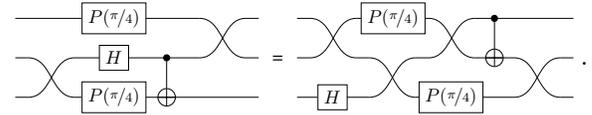
\[
  \scalebox{.7}{\input{ex1Lb.tikz}\ =\ \input{ex1Rb.tikz}}.
\]

\subsection{Quantum circuits: Syntax and semantics}

We consider quantum circuits defined on the following standard set of
generators: Hadamard, Control-Not, and Phase-gates together with
global phases.

\begin{definition}
  \label{def:QC}
  Let $\textup{\textbf {QC}}$ be the prop
  generated 
  by $\gh$, $\gcnot$, and for any $\varphi \in \mathbb R$, $\gp$ and
  $\gs$.
\end{definition}

\begin{figure*}
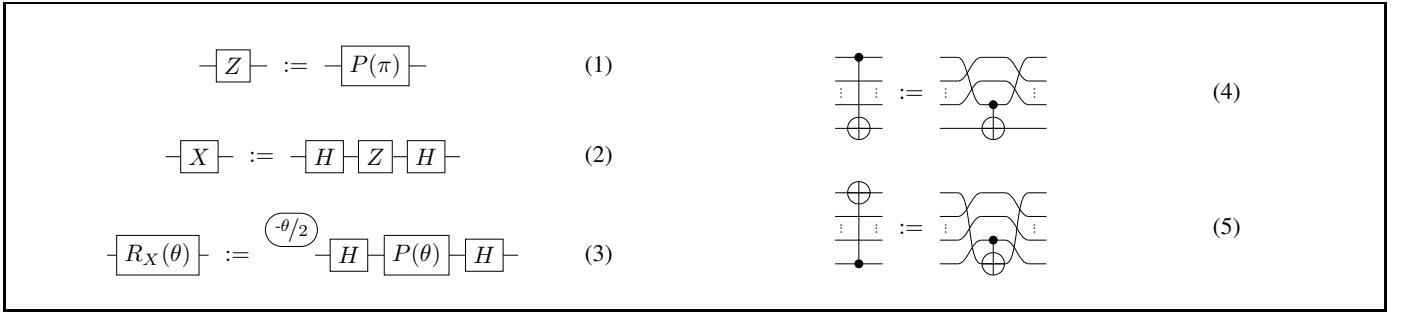

\centering
\fbox{\begin{minipage}{1\textwidth}\scalebox{.9}{\begin{minipage}{1\textwidth}
\newcommand{\nspazer}{0.5em}
\begin{multicols}{2}
\begin{equation}\label{zgate}\input{QgateZ.tikz}\ :=\ \input{Ppi.tikz}\end{equation}
\vspace{\nspazer}
\begin{equation}\label{xgate}\input{QgateX.tikz}\ :=\ \input{HZH.tikz}\end{equation}
\vspace{0em}
\begin{equation}\label{RXgate}\input{QgateRXtheta.tikz}\ :=\ \input{HPH.tikz}\end{equation}
\vspace{\nspazer}
\begin{equation}\label{CNot-def}\input{CNot-non-adj.tikz}\ :=\ \input{swap-CNot-swap.tikz}\end{equation}
\vspace{\nspazer}\vspace{\nspazer}
\begin{equation}\label{NotC-def}\input{NotC-non-adj.tikz} \ := \ \input{swap-NotC-swap.tikz}\end{equation}
\end{multicols}
\vspace{\nspazer}
\end{minipage}}
\end{minipage}}
\caption{Usual abbreviations of quantum circuits.\label{fig:shortcut}} 
\end{figure*}

\medskip The gates $\gh$ and $\gp$ have one input and one output,
while $\gcnot$ has two and $\gs$ zero. A quantum circuit $C$ with $n$ inputs
and $n$ outputs is called a $n$-qubit circuit. Given an $n$-qubit
circuit $C$, the corresponding unitary map $\interp C$ is acting on the
Hilbert space
$\mathbb C^{\{0,1\}^n}=\textup{span}(\ket{x}, x\in
\{0,1\}^n)$:\footnote{We use the standard Dirac
  notations.}

\begin{definition}[Semantics]
  For any $n$-qubit quantum circuit $C$, let
  $\interp C: \mathbb C^{\{0,1\}^n} \to \mathbb C^{\{0,1\}^n}$ be the
  linear map inductively defined as follows:
  $\interp{C_2\circ C_1} = \interp{C_2}\circ\interp{ C_1}$,
  $\interp{C_1\otimes C_3} = \interp{C_1}\otimes\interp{ C_3}$, and
  $\forall x,y\in \{0,1\}$, $\forall \varphi \in \mathbb R$, 
  \begin{align*}
    \interp{\gh} &=\ket x \mapsto {\textstyle\frac{1}{\sqrt2}}(\ket 0  +(-1)^x\ket 1),
  \\
    \interp{\gp}  &= \ket x\mapsto e^{ix\varphi}\ket x,
  \\
    \interp\gid &= \ket x\mapsto \ket x,
  \\
    \interp{\gcnot}  &=\ket{x,y}\mapsto  \ket{x,x\oplus y},
  \\
    \interp{\gswap} &= \ket{x,y}\mapsto \ket{y,x},
  \\
    \interp{\gs} &= 1\mapsto e^{i\varphi},
  \\
    \interp{\emptyc} &= 1\mapsto 1.
  \end{align*}
\end{definition}

\begin{remark}
  Although the definition of $\interp .$ relies on the inductive
  structure of raw quantum circuits, it is well-defined on quantum
  circuits as for any raw quantum circuits $C,C'$, whenever $C\equiv C'$
  we have $\interp C = \interp{C'}$.
\end{remark}

\begin{proposition}[Universality \cite{Barenco1995gates}]
 For any $n$-qubit
 unitary map $U$ acting on $\mathbb C^{\{0,1\}^n}$,
 there exists an $n$-qubit circuit $C$ such that $\interp C=U$.\qed
\end{proposition}

We use standard shortcuts in the description of quantum circuits,
given in \cref{fig:shortcut}. 
In textual description, we sometimes use CNot, $s(\varphi)$, $X$, $P(\varphi)$, \emph{etc} to denote respectively $\gcnot$, $\gs$, $\gx$, $\gp$, \emph{etc}.
Moreover, when the parameters (e.g. $\varphi$) are not specific values they can take arbitrary ones. 
We write $R_X(\theta)$ for the so-called
$X$-rotation \cite{NielsenChuang}, whereas the standard phase gate $P(\varphi)$ is a $Z$-rotation
only up to a global phase. As a consequence, they have a slightly
different behaviour: $P$ is $2\pi$-periodic: $\interp{P(2\pi)}=I$,
whereas $R_X$ is $4\pi$-periodic, and we instead have
$\interp{R_X(2\pi)} = -I$.

\subsection{Structural equations}

We introduce a set $\qczero$ of \emph{structural equations} on quantum
circuits in \cref{eq:qc0}. These equations are structural in the sense
that the transformations on the parameters are only based on the fact
that $\mathbb R$ is an additive group. In particular, these equations
are valid for any reasonable\footnote{I.e. which forms an additive
  group and contains $\pi/2$.} restriction on the angles.

We write $\qczero\vdash C_1 = C_2$ when $C_1$ can be transformed into
$C_2$ using the equations of \cref{eq:qc0}.\footnote{More formally,
  $\qczero\vdash \cdot = \cdot$ is defined as the smallest congruence
  which satisfies equations of Figures \ref{fig:axiom} and
  \ref{eq:qc0}.}

\begin{proposition}
  The structural equations of \cref{eq:qc0} are sound, i.e. if
  $\qczero\vdash C_1 = C_2$ then $\interp{C_1} = \interp{C_2}$.
\end{proposition}

\begin{figure*}
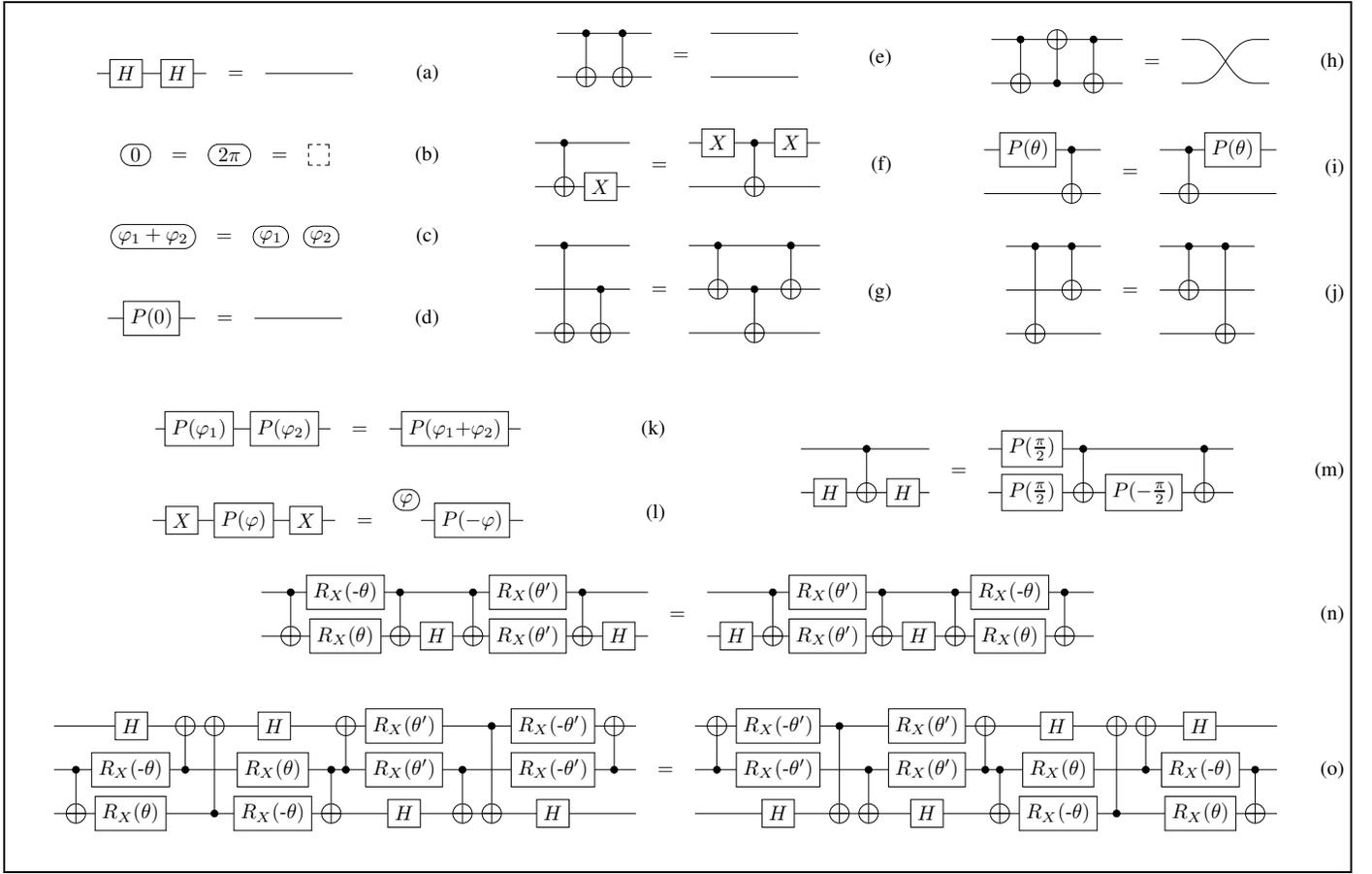

\centering
\scalebox{.8}{\fbox{\begin{minipage}{1.2603\textwidth}
\vspace{1em}
\begin{multicols}{3}
\newcommand{\nspazer}{1em}
\begin{equation}\label{HH}\tag{a}\begin{array}{rcl}\input{HH.tikz}&=&\input{filcourt.tikz}\end{array}\end{equation}
\vspace{\nspazer}
\begin{equation}\label{S0}\tag{b}\begin{array}{rrccl}\input{QScal0.tikz}&=&\input{QScal2pi.tikz}&=&\input{diagrammevide-s.tikz}\end{array}\end{equation}
\vspace{\nspazer}
\begin{equation}\label{SS}\tag{c}\begin{array}{rcl}\input{QScalsum.tikz}&=&\input{QScalphi1.tikz}\ \ \input{QScalphi2.tikz} \end{array}\end{equation}
\vspace{\nspazer}
\begin{equation}\label{P0}\tag{d}\begin{array}{rcl}\input{P0.tikz}&=&\input{filcourt.tikz}\end{array}\end{equation}
\vspace{\nspazer}
\begin{equation}\label{CNotCNot}\tag{e}\begin{array}{rcl}\input{CNotCNot.tikz}&=&\input{Qfilsparallelescourts.tikz}\end{array}\end{equation}
\vspace{\nspazer}
\begin{equation}\label{CNotX}\tag{f}\begin{array}{rcl}\input{CNotXb.tikz}&=&\input{XhCNotXh.tikz}\end{array}\end{equation}
\vspace{\nspazer}
\begin{equation}\label{CNotlift}\tag{g}\begin{array}{rcl}\input{CNotgrandCNotbas.tikz}&=&\input{CNotCNotCNot.tikz}\end{array}\end{equation}
\vspace{\nspazer}
\begin{equation}\label{tripleCNotswap}\tag{h}\begin{array}{rcl}\input{tripleCNotv1.tikz}&=&\input{Qswap.tikz}\end{array}\end{equation}
\vspace{\nspazer}
\begin{equation}\label{commutationPctrl}\tag{i}\begin{array}{rcl}\input{PCNot.tikz}&=&\input{CNotP.tikz}\end{array}\end{equation}
\vspace{\nspazer}
\begin{equation}\label{commutationCNothaut}\tag{j}\begin{array}{rcl}\input{CNotgrandCNothaut.tikz}&=&\input{CNothautCNotgrand.tikz}\end{array}\end{equation}
\vspace{\nspazer}
\end{multicols}
\begin{multicols}{2}
\newcommand{\nspazer}{1em}
\begin{equation}\label{PP}\tag{k}\begin{array}{rcl}\input{PP.tikz}&=&\input{Pphiplusphiprime.tikz}\end{array}\end{equation}
\vspace{\nspazer}
\begin{equation}\label{XPX}\tag{l}\begin{array}{rcl}\input{XPX.tikz}&=&\input{Pmoinsphieiphi.tikz}\end{array}\end{equation}
\vspace{\nspazer}
\begin{equation}\label{CZ}\tag{m}\begin{array}{rcl}\input{CZ.tikz}&=&\input{CPpi.tikz}\end{array}\end{equation}
\vspace{\nspazer}
\end{multicols}
\newcommand{\nspazer}{1em}
\begin{equation}\label{commutationdecompctrlPRY}\tag{n}\begin{array}{rcl}\input{comcontrlRXHRXH.tikz}&=&\input{comcontrlHRXHRX.tikz}\end{array}\end{equation}
\vspace{\nspazer}\vspace{\nspazer}
\begin{equation}\label{commutationdecompctrlRXpasenface}\tag{o}\begin{array}{rcl}\input{nnRXRXbncompactsansH-.tikz}&=&\input{RXbnnnRXcompactsansH-.tikz}\end{array}\end{equation}
\vspace{\nspazer}
\end{minipage}}}
\caption{\label{eq:qc0}\textbf{Axioms of $\qczero$:} Structural
  equations on quantum circuits. The equations are defined for any
  $\varphi, \varphi_1,\varphi_2, \theta, \theta'\in \mathbb R$.}
\end{figure*}

\begin{proof}
By inspection of the equations of
  \cref{eq:qc0}.
\end{proof}

Equations (\ref{HH}) to (\ref{XPX}) are fairly standard in quantum
computing.  \cref{CZ}, which is used for instance in
\cite{abdessaied2014quantum}, describes two equivalent ways to define
a controlled-Z gate. Notice that this equation cannot be derived from
the other axioms as it is the only equation on 2 qubits which does not
preserve the parity of the number of CNots plus the number of
swaps. Equations (\ref{commutationdecompctrlPRY}) and
(\ref{commutationdecompctrlRXpasenface}) are more involved and account
for some specific commutation properties of controlled gates (see
\cref{commctrl} and \cref{commctrlpasenface}).

The axioms of $\textup{QC}_0$, i.e.~the equations given in
\cref{eq:qc0}, are sufficient to derive standard elementary circuit
identities like those given in \cref{usefull_eq}.

One can also prove that some particular circuits, called phase-gadgets
\cite{Cowtan_2020}, can be flipped vertically:

\begin{equation}\label{CNotPCNotreversible}\qczero\vdash ~\scalebox{0.8}{\input{CNotPbCNot.tikz}}~=~\scalebox{0.8}{\input{CNotPbCNotalenvers.tikz}}\end{equation}

\begin{equation}\label{NotCRXNotCreversible}\qczero\vdash~\scalebox{0.8}{\input{RXhconjCNot.tikz}}~=~\scalebox{0.8}{\input{RXbconjCNot.tikz}}\end{equation}

The derivations are given in \cref{proof:usefuleq}. Combining
\cref{CNotPCNotreversible} and \cref{commutationPctrl}, one can easily
prove the following equation, used for instance in
\cite{nam2018automated} in the context of circuit optimisation:
\[
  \qczero\vdash \scalebox{.8}{\input{CPL.tikz}}
  {=}\!\scalebox{.8}{\input{CPR.tikz}}
\]

\begin{figure*}
\centering
\fbox{\scalebox{.8}{\begin{minipage}{1.2603\textwidth}
\newcommand{\nspazer}{0.5em}
\begin{multicols}{3}
\begin{equation}\label{commutationCNotsbas}\begin{array}{rcl}\input{CNotgrandCNotbas.tikz}&=&\input{CNotbasCNotgrand.tikz}\end{array}\end{equation}
\vspace{\nspazer}
\vspace{\nspazer}
\begin{equation}\label{CNotHH}\begin{array}{rcl}\input{CNotHH.tikz}&=&\input{HHNotC.tikz}\end{array}\end{equation}
\vspace{\nspazer}
\vspace{\nspazer}
\begin{equation}\label{XX}\begin{array}{rcl}\input{XX.tikz}&=&\input{filcourt-s.tikz}\end{array}\end{equation}
\vspace{\nspazer}
\begin{equation}\label{CNotliftvar}\begin{array}{rcl}\input{CNotgrandCNotbas.tikz}&=&\input{CNotgrandconjNotChauts.tikz}\end{array}\end{equation}
\vspace{\nspazer}\vspace{\nspazer}\begin{equation}\label{commutationXCNot}\begin{array}{rcl}\input{XbCNot.tikz}&=&\input{CNotXb.tikz}\end{array}\end{equation}
\vspace{\nspazer}\vspace{\nspazer}\begin{equation}\label{ZZ}\begin{array}{rcl}\input{ZZ.tikz}&=&\input{filcourt-s.tikz}\end{array}\end{equation}
\vspace{\nspazer}\begin{equation}\label{NotClift}\begin{array}{rcl}\input{NotCgrandNotCbas.tikz}&=&\input{NotCNotCNotC.tikz}\end{array}\end{equation}
\vspace{\nspazer}\begin{equation}\label{ZCNot}\begin{array}{rcl}\input{ZbCNot.tikz}&=&\input{CNotZZ.tikz}\end{array}\end{equation}
\vspace{\nspazer}\begin{equation}\label{commutationRXCNot}\begin{array}{rcl}\input{XCNot.tikz}&=&\input{CNotX.tikz}\end{array}\end{equation}
\end{multicols}
\vspace{-2.8em}
\begin{multicols}{2}
\vspace{\nspazer}\begin{equation}\label{RX0}\begin{array}{rcl}\input{RX0.tikz}&=&\input{filcourt-s.tikz}\end{array}\end{equation}
\vspace{\nspazer}\begin{equation}\label{RXRX}\begin{array}{rcl}\input{RXRX.tikz}&=&\input{RXthetaplusthetaprime.tikz}\end{array}\end{equation}
\vspace{\nspazer}\begin{equation}\label{CNotHCNot}\begin{array}{rcl}\input{CNotHhCNot.tikz}&=&\input{XbHhCNotHhCNotHh.tikz}\end{array}\end{equation}
\end{multicols}
~
\end{minipage}}}
\caption{Standard circuit identities that can be derived from the
  axioms of $\textup{QC}_0$, given in
  \cref{eq:qc0} \label{usefull_eq}. The proofs are given in
  \cref{proof:usefuleq}.}
\end{figure*}

Notice that when $\varphi = -\varphi' = \alpha/2$  the above circuits are two
equivalent standard implementations of a controlled-phase gate of angle
$\alpha$. We show in the next section how the basic algebra of (multi-) controlled gates can be derived. 

\subsection{Controlled gates}

Multi-controlled gates are useful to describe more elaborate quantum
circuits. We use the notations ``$\lambda$'' and ``$\Lambda$'' for
controls.
Given a 1-qubit gate $G$, $\lambda^1 G$ is a 2-qubit positively
controlled gate: if the control qubit (the top one) is in state
$\ket 1$ (resp.\ $\ket 0$) then $G$ (resp.\ the identity) is applied
on the target qubit (the bottom one). $\lambda^2 G$ is a 3-qubit
positively controlled gate, where the two upper qubits are controls:
they both need to be in state $\ket{1}$ for the gate $G$ to fire on
the bottom qubit.
We also consider more general multi-controlled gates
$\Lambda^{x_1\ldots x_k}G$ with positive (when $x_i=1$) and negative
(when $x_i=0$) controls: if the first qubit is in the state
$\ket{x_1}$ (resp. $\ket {\bar x_1}$) then $\Lambda^{x_2\ldots x_k}G$
(resp. the identity) is applied on the remaining qubits. Finally,
$\Lambda ^x_y G$ denotes a multi-controlled gate with control qubits
on both sides -- above and below -- of the target qubit.

We will follow a standard construction for multi-controls
using a decomposition into elementary $1$- and $2$-qubit gates (see
for instance \cite{Barenco1995gates}).  Note that we do not aim here at defining
\emph{all} controlled operators: as this construction is the main
apparatus for the completeness result, we only focus on the operations
$s(\varphi)$, $X$, $R_X(\theta)$ and $P(\varphi)$. Other controlled
operations can then be derived if needed.

We first define in \cref{def:multicontrolled-pos} circuits
implementing regular, all-positive multi-controlled gates
$\lambda^nG$. We then present in \cref{def:multicontrolled-oriented}
how to handle positive and negative controls. In
\cref{def:multicontrolled} we finally introduce controlled gates with
controls both above and below the gate $G$.

\begin{definition}[Positively multi-controlled gates]
  \label{def:multicontrolled-pos}
  For all $n\in\mathbb{N}$ and
  $G\in \{s(\varphi),X,R_X(\theta),P(\varphi)\}$, we define a quantum
  circuit $\lambda^n G$.\footnote{Note that $G$ spans non-elementary
    gates, the constructor $\lambda$ is not considered as a gate
    operator, and the fact that the circuit $\lambda^nG$ happens to be
    related to $G$ is a corollary of its definition, as 
    discussed further in the article.}  This circuit acts on
  $n$ wires when $G=s(\varphi)$ and $n+1$ otherwise.  We define each
  circuit $\lambda^nG$ as follows.
  \begin{itemize}
  \item $\lambda^{n} R_X(\theta)$ is defined by induction:%
  \[
    \lambda^{0} R_X(\theta)\coloneqq R_X(\theta),
  \]
  \[
    \lambda^{n+1} R_X(\theta)\coloneqq\scalebox{0.8}{\input{mctrlXthetaOnlyBlack.tikz}}.
  \]
\item $\lambda^{n} P(\varphi)$ is defined by induction using
  $\lambda^nR_X(\varphi)$:
  \[
    \lambda^{0} P(\varphi)\coloneqq P(\varphi),
  \]
  \[
    \lambda^{n+1}
    P(\varphi)\coloneqq\scalebox{0.8}{\input{mctrlPphi2OnlyBlack.tikz}}.
  \]
  
\item $\lambda^{n}X$ is a simple macro:
  \[ \lambda^{n}X\coloneqq
    \scalebox{0.8}{\input{mctrlPHOnlyBlack.tikz}} \]
  
\item Finally, $\lambda^{0}s(\psi)\coloneqq s(\psi)$ and
  $\lambda^{n+1}s(\psi)\coloneqq \lambda^nP(\psi)$. 
\end{itemize}
\end{definition}

\begin{definition}[Multi-controlled gates]
  \label{def:multicontrolled-oriented}
  For any $k$-length list of booleans $x = x_1,\dots, x_k$ ($x_i\in\{0,1\}$), for any
  $G\in \{s(\varphi),X,R_X(\theta),P(\varphi)\}$ we define the quantum circuit $\Lambda^x
  G$ as
  \[
    \Lambda^{x}
    G\coloneqq\scalebox{0.8}{\input{generalControledGate.tikz}}    
  \]
  when $G\in \{X,R_X(\theta),P(\varphi)\}$, and
  \[\Lambda^{x}
    s(\varphi)\coloneqq\scalebox{0.8}{\input{generalControledGateS.tikz}}.
  \]
  where $\overline x = 1-x$, $\!\scalebox{0.8}{\begin{tikzpicture}
	\begin{pgfonlayer}{nodelayer}
		\node [style=none] (86) at (-1, 0) {};
		\node [style=none] (87) at (1, 0) {};
		\node [style=boite2] (103) at (0, 0) {$X^{1}$};
	\end{pgfonlayer}
	\begin{pgfonlayer}{edgelayer}
		\draw (86.center) to (87.center);
	\end{pgfonlayer}
\end{tikzpicture}
}\!=\!\scalebox{0.8}{\begin{tikzpicture}
	\begin{pgfonlayer}{nodelayer}
		\node [style=none] (86) at (-1, 0) {};
		\node [style=none] (87) at (1, 0) {};
		\node [style=boite2] (103) at (0, 0) {$X$};
	\end{pgfonlayer}
	\begin{pgfonlayer}{edgelayer}
		\draw (86.center) to (87.center);
	\end{pgfonlayer}
\end{tikzpicture}
}$, and $\!\scalebox{0.8}{\begin{tikzpicture}
	\begin{pgfonlayer}{nodelayer}
		\node [style=none] (86) at (-1, 0) {};
		\node [style=none] (87) at (1, 0) {};
		\node [style=boite2] (103) at (0, 0) {$X^{0}$};
	\end{pgfonlayer}
	\begin{pgfonlayer}{edgelayer}
		\draw (86.center) to (87.center);
	\end{pgfonlayer}
\end{tikzpicture}
}\!=\!\scalebox{0.8}{\begin{tikzpicture}
	\begin{pgfonlayer}{nodelayer}
		\node [style=none] (0) at (-0.8, 0) {};
		\node [style=none] (1) at (0.8, 0) {};
	\end{pgfonlayer}
	\begin{pgfonlayer}{edgelayer}
		\draw [style=noire] (0.center) to (1.center);
	\end{pgfonlayer}
\end{tikzpicture}}$.  
\end{definition}

\begin{definition}[General multi-controlled gates]
  \label{def:multicontrolled}
  Given two lists of booleans $x\in \{0,1\}^k$ and $y\in
  \{0,1\}^\ell$, if $xy$ is the concatenation of $x$ and $y$ we define the
  two quantum circuits
  \begin{itemize}
  \item for any $G\in \{X,R_X(\theta),P(\varphi)\}$
    \[\Lambda^x_y G\coloneqq \scalebox{0.8}{\input{multixy.tikz}}\]
      \item $\Lambda^x_y s(\varphi)\coloneqq \Lambda^{xy} s(\varphi)$.
  \end{itemize}
\end{definition}

One can double check using the semantics that $\Lambda^x_y G$ is
actually a multi-controlled gate:

\begin{proposition}
  For any $x,u\in \{0,1\}^k$, $y,v\in \{0,1\}^\ell$, $a\in\{0,1\}$ and
  $G\in
  \{X,R_X(\theta),P(\varphi)\}$,
  \begin{align*}
    \interp{\Lambda^x_yG}\ket{u,a,v} &= \begin{cases}\ket u\otimes
      (\interp G \ket a)\otimes \ket{v}&\text{if $uv=xy$,
      }\\\ket{u,a,v} &\text{otherwise,}\end{cases}
  \intertext{and}
    \interp{\Lambda^x_ys(\varphi)}\ket{u,v}
    &= \begin{cases}e^{i\varphi}\ket {u,v}&\text{if $uv=xy$,
      }\\\ket{u,v} &\text{otherwise.}\end{cases}
  \end{align*}
\end{proposition}

We use the standard bullet-based graphical  notation for multi-controlled gates:
the $i^\text{th}$ control is black (resp.\ white) when $x_i=1$ (resp.\ $x_i=0$), and
the $j^\text{th}$ from the end control is black (resp.\ white) when
$y_{\ell-j+1}=1$ (resp.\ $=0$), e.g.:
\begin{align*}
  \Lambda^{11}_1 X&: \hspace{1.4ex}\scalebox{.9}{\input{C11X1.tikz}},
  &
  \Lambda^{0}_{10} R_X(\theta)&:  \scalebox{.9}{\input{C0X10.tikz}},
  \\[0.3cm]
  \Lambda^{10} P(\varphi)&: \scalebox{.9}{\input{C10P.tikz}},
  &
  \Lambda^{1\ldots1} R_X(\theta)&: \scalebox{.9}{\input{CCRX.tikz}}.
\end{align*}
To avoid ambiguity with CNot we will not use this notation in the particular case of $\Lambda^1 X$ and $\Lambda_1X$. Notice however that  
$\Lambda^1X$ is provably  equivalent to  $\cnot$:

\begin{proposition}\label{ctrlXCNot}$\qczero \vdash \Lambda^1 X = \gcnot$.
\end{proposition}

\begin{proof}
The proof is given in \cref{preuvectrlXCNot}.
\end{proof}

\subsection{Properties of multi-controlled gates}\label{propertiesmultictrl}
 
In a multi-qubit controlled gate, all control qubits play a similar
role. This can be expressed as the following commuting property:
\[
  \scalebox{0.8}{\input{swapCCRX3.tikz}} =
  \scalebox{0.8}{\input{CCRXswap3.tikz}}
\]
This property is provable in $\qczero$, considering three cases
depending whether the exchanged control qubits are either above or
below the target qubit:

\begin{proposition}\label{cor:swap}
  For any $x\in \{0,1\}^k, y\in \{0,1\}^\ell, z\in \{0,1\}^m$,
  $a,b\in \{0,1\}$ and any $G\in \{s(\psi),X,R_X(\theta),P(\varphi)\}$,
  \begin{alignat}{100}
    \qczero&\vdash& \scalebox{0.9}{\input{swapMctrl.tikz}}
    &{=} \scalebox{0.9}{\input{Mctrlswap.tikz}}
    \label{swapCCZ}
    \\
    \qczero&\vdash& \scalebox{0.9}{\input{swapMctrl-down.tikz}}
    &{=} \scalebox{0.9}{\input{Mctrlswap-down.tikz}}
    \label{swapCCZ2}
    \\
    \qczero&\vdash& \scalebox{0.9}{\input{swapMctrl-middle.tikz}}
    &{=} \scalebox{0.9}{\input{Mctrlswap-middle.tikz}}
    \label{swapCCZ3}
  \end{alignat}
\end{proposition}

A peculiar property of controlled phase gates (and hence controlled
scalars) is that the target qubit is actually equivalent to the
control qubits, e.g.:
\[ \scalebox{0.8}{\input{CCP3.tikz}} = \scalebox{0.8}{\input{CCP3-.tikz}}\]

This property is also provable in $\qczero$:

\begin{proposition}\label{prop:CP}
  For any $x\in \{0,1\}^k, y\in \{0,1\}^\ell$,
  \begin{equation}
    \qczero\vdash\Lambda^x_{y1} P(\varphi)=\Lambda^{x1y}P(\varphi)
    \label{phasemobile}
  \end{equation}
\end{proposition}

\begin{proof}[Proof of \cref{cor:swap} and \cref{prop:CP}] 
  The two properties are proved at once.  The proof relies on the
  following commutation property which can be proved by induction (see
  Appendix \ref{sec:proofcomRX}).
  \begin{multline}
    \label{eq:comRX}
    \qczero\vdash \scalebox{0.8}{\input{com-L.tikz}} \\=  \scalebox{0.8}{\input{com-R.tikz}}
  \end{multline}
  The proof of Equations~\eqref{swapCCZ}-\eqref{swapCCZ3} for $G=R_X(\theta)$ follows by
  induction. We then prove \cref{phasemobile} which requires a few
  technical developments.  The proof of Eq.~\eqref{swapCCZ}-\eqref{swapCCZ3} for the other
  gates then follows from the $R_X(\theta)$ case and
  \cref{phasemobile} (see \cref{preuvesswapsmultictrl}).
\end{proof}

The gates $P(\varphi)$ form a monoid, i.e.
$P(\varphi+\varphi') = P(\varphi)\circ P(\varphi')$ (\cref{PP}) and
$P(0)=\gidspace$ (\cref{P0}). Notice that $R_X(\theta)$ and $s(\varphi)$ also form
monoids. It is provable in $\qczero$ that their multi-controlled
versions enjoy the same property:

\begin{proposition}\label{prop:sum}
For any $x\in \{0,1\}^k$, $y\in \{0,1\}^\ell$,
\begin{align*}
  \qczero &\vdash \Lambda^{x}_y R_X(\theta')\circ \Lambda^{x}_y
  R_X(\theta) =  \Lambda^{x}_y R_X(\theta+\theta'),
  \\
  \qczero &\vdash \Lambda^{x}_y P(\varphi')\circ \Lambda^{x}_y
  P(\varphi)=  \Lambda^{x}_y P(\varphi+\varphi'),
  \\
  \qczero &\vdash \Lambda^{x}_y s(\varphi')\circ \Lambda^{x}_y
  s(\varphi)=  \Lambda^{x}_y s(\varphi+\varphi'),
  \\
  \qczero &\vdash \Lambda^{x}_y R_X(0)=id_{k+\ell+1},
  \\
  \qczero &\vdash \Lambda^{x}_y P(0)= id_{k+\ell+1},
  \\
  \qczero &\vdash \Lambda^{x}_y s(0)= id_{k+\ell},
\end{align*}
where $id_{k}$ is defined as in \cref{fig:axiom}. 
\end{proposition}

\begin{proof}
First, proving that multi-controlled gates with angle $0$ are equivalent to the identity is straightforward by induction.

To prove the rest of the proposition, we first prove that
  $\qczero \vdash \Lambda^{1..1}R_X(\theta')\circ
  \Lambda^{1..1}R_X(\theta)= \Lambda^{1..1} R_X(\theta+\theta')$. The proof is by induction:
  we unfold the two multi-controlled gates, use \cref{eq:comRX} to put
  the multi-controlled gates with angles $\theta/2$ and $\theta'/2$
  side by side, and merge them using the induction hypothesis. We use
  again \cref{eq:comRX} to allow the combination of the
  multi-controlled gates with angle $-\theta/2$ and $-\theta'/2$,
  closing the case.
  
  The cases with more general controls are derived from this one using \cref{def:multicontrolled-oriented,def:multicontrolled}. The cases of $P$ and $s$ are derived from the $R_X$ case using \cref{def:multicontrolled-pos} and an ancillary lemma stating that a multi-controlled phase commutes with the controls of another multi-controlled gate. The details of the proof are given in \cref{proofpropsum}.
\end{proof}
 
\begin{remark}
  Note that \cref{prop:sum} does not imply the periodicity of
  controlled gates. The latter is proven in \cref{prop:period} with
  the help of the rules of \cref{fig:euler}.
\end{remark}
Combining a control and anti-control on the same qubit makes the
evolution independent of this qubit, as in the following example in
which the evolution is independent of the second
qubit:\footnote{Notice that in the above example we implicitly use
  \cref{cor:swap} to swap the first two qubits and apply
  \cref{prop:comb}. As a consequence, the resulting multi-controlled gate
  acts on non-adjacent qubits. Similarly to the \cnot{} case (see
  \cref{CNot-def,NotC-def}), we use some syntactic sugar to represent such
  multi-controlled gates acting on non-adjacent qubits.}
\[
  \scalebox{0.8}{\input{com-fu-L.tikz}} =
  \scalebox{0.8}{\input{com-fu-R.tikz}}
\]
Such simplifications can be derived in $\qczero$.
\begin{proposition}\label{prop:comb}
  For all bitstrings $x\in \{0,1\}^k$, $y\in \{0,1\}^\ell$, and for all
  $G\in\{s(\varphi),X,R_X(\theta),P(\varphi)\}$,
  \[
    \qczero \vdash \Lambda^{0x}_y G\circ \Lambda^{1x}_{y} {G} =
    \gidspace\otimes\Lambda^{x}_{y} {G}.
  \]
\end{proposition}

\begin{proof}
  Without loss of generality, we assume $y$ as the empty string $\epsilon$ and
  $G = R_X(\theta)$, as it can derive the other cases. The proof is by induction: we unfold the multi-controlled and
  multi-anti-controlled gates. We can then move the $X$ gate through $H$
  and $\cnot$ gates due to the anti-control, changing the sign of an
  $R_X$ rotation from $-\theta/2$ to $\theta/2$.  The rest of the
  proof is similar to the one of \cref{prop:sum}, except that two
  $R_X$ gates cancel out, leading to the identity on the first qubit
  and the desired multi-controlled gate on the second
  one. The details of the proof are given in \cref{preuvepropcomb}.
\end{proof}

\cref{prop:comb} shows how control and anti-control can be combined on the
first qubit of a multi-controlled gate. Note, however, that it can be
generalised to any control qubit thanks to \cref{cor:swap}.

Another useful  property of multi-controlled gates is that they commute when
there is a control and anti-control on the same qubit, as in the following
example in which their controls differ on the third (and last) qubit:
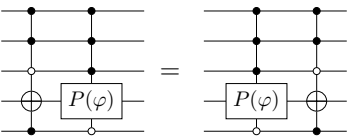
\[
  \scalebox{0.8}{\input{com-ex-L.tikz}} =\scalebox{0.8}{
    \input{com-ex-R.tikz}}
\]

When the target qubit is the same, such a commutation property can be
derived in $\qczero$, using in particular
\cref{commutationdecompctrlPRY}.

\begin{proposition}\label{commctrl}
  For any $x,x'\in \{0,1\}^k$, $y,y'\in \{0,1\}^\ell$, and
  $G,G'\in \{X,R_X(\theta),P(\varphi)\}$, if
  $xy\neq x'y'$\footnote{$xy\neq x'y'$ iff
    $\exists i, x_i\neq x'_i \vee y_i\neq y'_i$.} then
  $$\qczero \vdash \Lambda^x_y G\circ \Lambda^{x'}_{y'} {G'} =
  \Lambda^{x'}_{y'} {G'}\circ \Lambda^x_y G \, .$$
\end{proposition}

\begin{proof}
  The proof relies on a generalisation of \cref{eq:comRX}, and follows
  by an induction argument whose base case can be derived thanks to
  \cref{commutationdecompctrlPRY}. The details of the proof are given in \cref{preuvecommctrl}.
\end{proof}

Controlled and anti-controlled gates also commute when the target
qubits are not the same in both gates, as in:
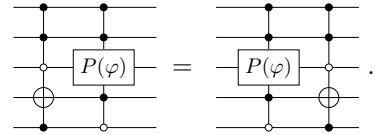
\[
  \scalebox{0.8}{\input{com-ex-L-diff.tikz}} =\scalebox{0.8}{
    \input{com-ex-R-diff.tikz}}.
\]
This property can also be derived in $\qczero$, using in particular
\cref{commutationdecompctrlRXpasenface}:

\begin{proposition}\label{commctrlpasenface}
  For any $a,b\in \{0,1\}$, $x,x'\in \{0,1\}^k$,
  $y,y'\in \{0,1\}^\ell$, $z,z'\in \{0,1\}^m$ and
  $G,G'\in \{X,R_X(\theta),P(\varphi)\}$, if $xyz\neq x'y'z'$ then
   $$\qczero \vdash \Lambda^x_{yaz} G\circ \Lambda^{x'by'}_{z'} {G'} =
   \Lambda^{x'by'}_{z'} {G'}\circ \Lambda^x_{yaz} G$$
\end{proposition}

\begin{proof}
  The proof is also based on the generalisation of \cref{eq:comRX},
  using an inductive argument whose base case can be derived thanks to
  \cref{commutationdecompctrlRXpasenface}. The details of the proof are given in \cref{preuvecommctrlpasenface}.
\end{proof}

\subsection{Euler angles and Periodicity}\label{Eulerandperiod}

$\qczero$ is not complete. In particular equations based on Euler
angles, which require non-trivial calculations on the angles, cannot
be derived. As a consequence we add to the equational theory the three rules shown in \cref{fig:euler},
leading to the equational theory $\textup{QC}$. We write $\textup{QC}\vdash C_1=C_2$ when $C_1$ can be
rewritten into $C_2$ using equations of \cref{eq:qc0} and
\cref{fig:euler} (together with the deformation rules).

\begin{figure*}[htb]
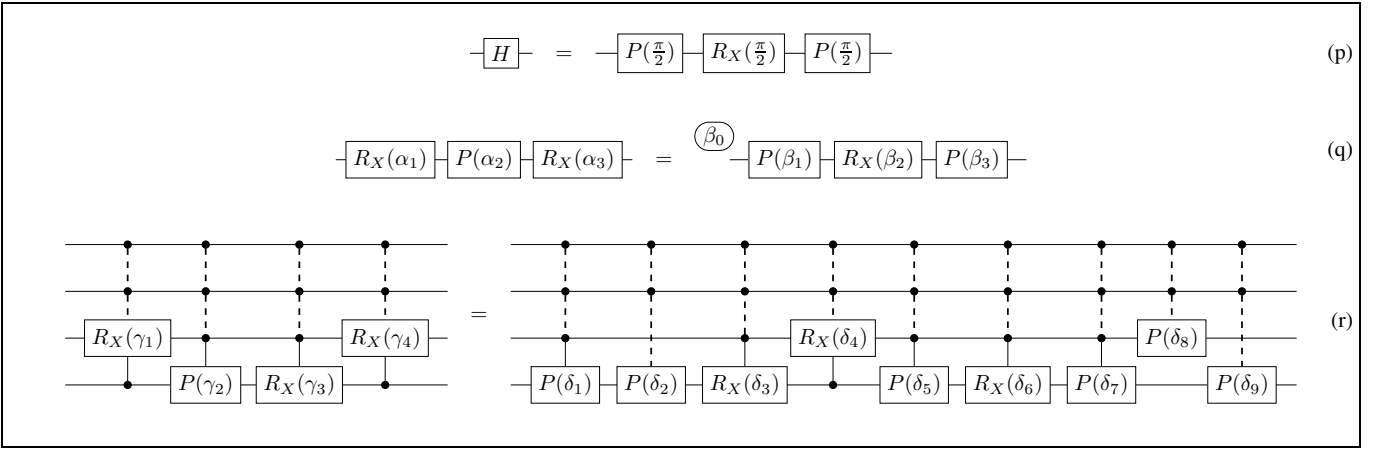

\centering
\scalebox{.83}{\fbox{\begin{minipage}{1.187\textwidth}
\newcommand{\nspazer}{1em}
\vspace{\nspazer}
\begin{equation}\label{EulerH}\tag{p}\begin{array}{rcl}\input{H.tikz}&=&\input{EulerHv2.tikz}\end{array}\end{equation}
\vspace{\nspazer}
\begin{equation}\label{Euler2d}\tag{q}\begin{array}{rcl}\input{RXPRXalphas.tikz}&=&\input{RXPRXbetas.tikz}\end{array}\end{equation}
\vspace{\nspazer}
\vspace{\nspazer}
\begin{equation}\label{Euler3dmulticontrolled}\tag{r}\begin{array}{rcl}\input{Euler3dleft-multicontrolled-simp-gammas.tikz}&=&\input{Euler3dright-multicontrolled-simp-deltas.tikz}\end{array}\end{equation}
\vspace{\nspazer}
\end{minipage}}}
\caption{\textbf{Non-structural equations.} In Equations
  (\ref{Euler2d}) and (\ref{Euler3dmulticontrolled}) the LHS circuit
  has arbitrary parameters which uniquely determine the parameters of
  the RHS circuit.  \cref{Euler2d} is nothing but the well-known
  Euler-decomposition rule which states that any unitary can be
  decomposed, up to a global phase, into basic $X$- and
  $Z$-rotations. Thus for any $\alpha_i\in \mathbb R$, there exist
  $\beta_j\in [0,2\pi)$ such that \cref{Euler2d} is sound. We make the
  angles $\beta_j$ unique by assuming that $\beta_1 \in [0,\pi)$,
  $\beta_2\in[0,2\pi)$ and if $\beta_2\in\{0,\pi\}$ then $\beta_1=0$.
  \cref{EulerH} is the particular Euler decomposition of $H$.
  \cref{Euler3dmulticontrolled} reads as follows: the equation is
  defined for any $n\ge 2$ input qubits, in such a way that all gates
  are controlled by the first $n-2$ qubits.
  \cref{Euler3dmulticontrolled} can be seen as a generalisation of the
  Euler rule, using multi-controlled gates.  Similarly to
  \cref{Euler2d}, for any $\gamma_i\in\mathbb R$, there exist
  $\delta_j\in[0,2\pi)$ such that \cref{Euler3dmulticontrolled} is
  sound. We can ensure that the angles $\delta_j$ are uniquely
  determined by assuming that $\delta_1,\delta_2,\delta_5,\in[0,\pi)$,
  $\delta_3,\delta_4,\delta_6\in[0,2\pi)$, if $\delta_3=0$
  then $\delta_2=0$, if $\delta_3=\pi$ then $\delta_1=0$, if $\delta_4=0$ then $\delta_1=\delta_3\mathrel{(=}\delta_2)=0$, if
  $\delta_4=\pi$ then $\delta_2=0$, if $\delta_4=\pi$ and $\delta_3=0$ then $\delta_1=0$, and if $\delta_6\in\{0, \pi\}$ then
  $\delta_5=0$.\label{fig:euler}}
\end{figure*}

The Euler decomposition of $H$ (\cref{EulerH}) is not unique: 

\begin{proposition}\label{EulerHmoins}
  $\textup{QC}\vdash \!\!\!\scalebox{0.8}{\begin{tikzpicture}
	\begin{pgfonlayer}{nodelayer}
		\node [style=none] (0) at (-1, 0) {};
		\node [style=none] (1) at (1, 0) {};
		\node [style=boite2] (2) at (0, 0) {$H$};
	\end{pgfonlayer}
	\begin{pgfonlayer}{edgelayer}
		\draw (0.center) to (1.center);
	\end{pgfonlayer}
\end{tikzpicture}
} {=}
  \scalebox{0.75}{\input{EulerHmoins.tikz}}$
\end{proposition}

\begin{proof}The proof is given in \cref{preuveEulerHmoins}.
\end{proof}

More generally the Euler angles are not unique, but can be made unique by adding some constraints on the angles, like choosing them in the appropriate intervals (see \cref{fig:euler}).

\begin{proposition}\label{soundnessEuleraxioms}
  Equations (\ref{Euler2d}) and (\ref{Euler3dmulticontrolled}) are
  sound.  Moreover, the choice of parameters in the RHS-circuits
  to make the equations sound is unique (under the constraints given
  in \cref{fig:euler}).
\end{proposition}

\begin{proof}
  The soundness and uniqueness of \cref{Euler2d} are well-known
  properties. Regarding \cref{Euler3dmulticontrolled}, we first notice
  that the semantics of both circuits is of the form
  $\left(\begin{array}{c|c}I&0\\\hline 0&U\end{array}\right)$ where
  $U$ is a $3\times 3$ matrix. We then use the fact that this matrix
  can be decomposed into basic rotations that can be proved to be
  unique \cite{clement2022LOv}.  The details of the proof are given in
  \cref{appendix:euler3d}.
\end{proof}

Notice that \cref{Euler2d} subsumes Equations (\ref{PP}) and
(\ref{XPX}), which can now be derived using the other axioms of
$\textup{QC}$.

\begin{proposition}\label{klfollowfromEuler}The following two equations of $\textup{QC}$,
\[\begin{array}{rclr}
 \scalebox{0.9}{\begin{tikzpicture}
	\begin{pgfonlayer}{nodelayer}
		\node [style=none] (6) at (-3, 0) {};
		\node [style=none] (7) at (3, 0) {};
		\node [style=boite2] (15) at (-1.5, 0) {$P(\varphi_1)$};
		\node [style=boite2] (18) at (1.425, 0) {$P(\varphi_2)$};
	\end{pgfonlayer}
	\begin{pgfonlayer}{edgelayer}
		\draw (6.center) to (7.center);
	\end{pgfonlayer}
\end{tikzpicture}
}&{=}&\scalebox{0.9}{\begin{tikzpicture}
	\begin{pgfonlayer}{nodelayer}
		\node [style=none] (4) at (-2.25, 0) {};
		\node [style=none] (5) at (2.25, 0) {};
		\node [style=boite2] (11) at (0, 0) {$P(\varphi_1{+}\varphi_2)$};
	\end{pgfonlayer}
	\begin{pgfonlayer}{edgelayer}
		\draw (4.center) to (5.center);
	\end{pgfonlayer}
\end{tikzpicture}
} &\text{~~~~(\ref{PP})} \\[0.4cm]
   \scalebox{0.9}{\input{XPX.tikz}}&{=}&\scalebox{0.9}{\begin{tikzpicture}
	\begin{pgfonlayer}{nodelayer}
		\node [style=none] (4) at (-1.75, 0) {};
		\node [style=none] (5) at (1.75, 0) {};
		\node [style=boite2] (11) at (0, 0) {$P(-\varphi)$};
		\node [style=cartouche] (12) at (-2.25, 0.7) {$\varphi$};
	\end{pgfonlayer}
	\begin{pgfonlayer}{edgelayer}
		\draw (4.center) to (5.center);
	\end{pgfonlayer}
\end{tikzpicture}
}&\text{~~~~(\ref{XPX})}
  \end{array}\]
 can be
  derived from the other axioms of $\textup{QC}$.
\end{proposition}

\begin{proof}
The proofs are given in \cref{preuveklfollowfromEuler}.
\end{proof}

The introduction of the additional equations of \cref{fig:euler} allows us to
prove some extra properties about multi-controlled gates, like periodicity (for
those with a parameter) in \cref{prop:period} and the fact that a
multi-controlled X gate is self-inverse.

\begin{proposition}\label{prop:CCX}
  For any $x\in \{0,1\}^k$, $y\in \{0,1\}^\ell$,
  \[
    \textup{QC} \vdash \Lambda^{x}_y X\circ \Lambda^{x}_y X=
    id_{k+\ell+1}
  \]
\end{proposition}
\begin{proof}
  The case $x=y=\epsilon$ is a direct consequence of \cref{XX}.
  For the other cases, by \cref{def:multicontrolled-pos,def:multicontrolled-oriented,def:multicontrolled,HH,XX,prop:sum}, it is equivalent to
  show that, for any $x\in \{0,1\}^k$,
  \[
    \qc \vdash \Lambda^x P(2\pi) = id_{k+1}.
  \]
  Without loss of generality, we can consider $x\in \{1\}^k$.  
  Then the result is a consequence of \cref{prop:sum} and
  \cref{Euler3dmulticontrolled}.
  Indeed, by taking
  $\gamma_1=\gamma_3=\gamma_4=0$ and $\gamma_2=2\pi$ in the 
  LHS of \cref{Euler3dmulticontrolled}, the unique angles on the
  right are all zeros:
  $\delta_1=\delta_2=\delta_3=\delta_4=\delta_5=\delta_6=\delta_7=\delta_8=\delta_9=0$.
  By \cref{prop:sum}, any multi-controlled gate with zero angle is the identity, which gives us the desired equality.
  Further
  details can be found in \cref{proof:CCX}.
\end{proof}

\begin{proposition}\label{prop:period}
  For any $x\in \{0,1\}^k$, $y\in \{0,1\}^\ell$,
  $\theta\in \mathbb R$,
  \begin{align*}
    \qc &{}\vdash \Lambda^x_y R_X(\theta+4\pi) = \Lambda^x_y R_X(\theta)
    \\
    \qc &{}\vdash \Lambda^x_y P(\theta+2\pi) = \Lambda^x_y
    P(\theta)
    \\
    \qc &{}\vdash \Lambda^x_y s(\theta+2\pi) =
    \Lambda^x_y s(\theta)
  \end{align*}
\end{proposition}

\begin{proof}
  Following the additivity of \cref{prop:sum}, it is sufficient to
  show that for any $x\in \{0,1\}^k$, $y\in \{0,1\}^\ell$,
  \begin{align*}
    \qc &{}\vdash \Lambda^x_y R_X(4\pi) = id_{k+\ell+1},
    \\
    \qc &{}\vdash \Lambda^x_y P(2\pi) = id_{k+\ell+1},
    \\
    \qc &{}\vdash \Lambda^x_y s(2\pi) = id_{k+\ell}.
  \end{align*}
  Also, with Equations \eqref{XX} and 
  \cref{def:multicontrolled-oriented,def:multicontrolled}, it is sufficient to
  show that for any $x\in \{1\}^k$,
  \begin{align*}
    \qc &{}\vdash \Lambda^x R_X(4\pi) = id_{k+1},
    \\
    \qc &{}\vdash \Lambda^x P(2\pi) = id_{k+1},
    \\
    \qc &{}\vdash \Lambda^x s(2\pi) = id_{k}.
  \end{align*}
  First, we prove the three cases with $x=\epsilon$. Then,
  we use $\qc \vdash \Lambda^x P(2\pi) = id_{k+1}$, proven in
  \cref{prop:CCX} 
  using \cref{Euler3dmulticontrolled}. We obtain the
  other 
  statements as direct consequences of the $2\pi$-periodicity of
  $P$. Further details are provided in \cref{proof:propperiod}.
\end{proof}

\section{Completeness}\label{sec:completeness}

In this section we prove the main result of the paper, namely the completeness
of $\textup{QC}$. To this end, a back and forth encoding of quantum circuits
into linear optical quantum circuits is introduced. We use the graphical
language for linear optical circuits introduced in \cite{clement2022LOv}.

\subsection{Optical circuits}\label{sec:loppcircuits}

A \emph{linear optical polarisation-preserving} (LOPP for short)
circuit is an optical circuit made of beam splitters
($\input{bs-xs.tikz}$) and phase shifters
($\begin{tikzpicture}[scale=0.3]
	\begin{pgfonlayer}{nodelayer}
		\node [style=none] (5) at (-1.6, 0) {};
		\node [style=none] (9) at (1.6, 0) {};
		\node [style=PolRot,scale=0.75] (11) at (0, 0) {\small $\varphi$};
	\end{pgfonlayer}
	\begin{pgfonlayer}{edgelayer}
		\draw (5.center) to (11);
		\draw (11) to (9.center);
	\end{pgfonlayer}
\end{tikzpicture}
$):

\begin{definition}
  \label{def:LOPP}
  Let $\textup{\textbf {LOPP}}$ be the prop generated by
 $\begin{tikzpicture}[scale=0.3]
	\begin{pgfonlayer}{nodelayer}
		\node [style=none] (5) at (-1.6, 0) {};
		\node [style=none] (9) at (1.6, 0) {};
		\node [style=PolRot,scale=0.75] (11) at (0, 0) {\small $\varphi$};
	\end{pgfonlayer}
	\begin{pgfonlayer}{edgelayer}
		\draw (5.center) to (11);
		\draw (11) to (9.center);
	\end{pgfonlayer}
\end{tikzpicture}
$, $\input{bs-xs.tikz}$ with $\varphi, \theta \in \mathbb R$.
\end{definition}

Like quantum circuits, LOPP-circuits are defined as a prop: one can see them as raw circuits quotiented by the $\equiv$-equivalence given in \cref{fig:axiom}.

In the following, we consider the single photon case, hence each input mode (or
wire) represents a possible input position for the photon. The photon moves
from left to right in the circuit. The state of the photon is entirely defined
by its position, and as a consequence the state space is of the form $\mathbb
C^n$ when there are $n$ possible modes. We consider the standard orthonormal
basis $\{\ket{p}\}_{p\in [0,n)}$ of $\mathbb C^n$. The semantics is defined as
follows.

\begin{definition}[Semantics]
  \label{semLOPP}
  For any $n$-mode $\textup{LOPP}$-circuit $C$, let
  $\interp C: \mathbb C^{n} \to \mathbb C^{n}$ be a linear map
  inductively defined as follows:
  $\interp{C_2\circ C_1} \coloneqq \interp{C_2}\circ\interp{ C_1}$,
  $\interp{C_1\otimes C_3} \coloneqq \interp{C_1}\oplus\interp{
    C_3}=\left(\begin{array}{c|c}\interp{C_1}&0\\\hline0&\interp{C_3}\end{array}\right)$,
  \begin{align*}
    \interp{\input{bs-xs.tikz}}
    &\coloneqq\ket p
      \mapsto\cos(\theta)\ket{p}+i\sin(\theta)\ket{{1-p}}
    \\
    &=
      \begin{pmatrix}
        \cos(\theta)&i\sin(\theta)\\
        i\sin(\theta)&\cos(\theta)
      \end{pmatrix}
    \\
    \interp{\gswap}
    &\coloneqq \ket{p}\mapsto
      \ket{1-p}
      =
      \begin{pmatrix}
        0&1\\1&0
      \end{pmatrix}
  \end{align*}
  \[
    \interp{\begin{tikzpicture}[scale=0.3]
	\begin{pgfonlayer}{nodelayer}
		\node [style=none] (5) at (-1.6, 0) {};
		\node [style=none] (9) at (1.6, 0) {};
		\node [style=PolRot,scale=0.75] (11) at (0, 0) {\small $\varphi$};
	\end{pgfonlayer}
	\begin{pgfonlayer}{edgelayer}
		\draw (5.center) to (11);
		\draw (11) to (9.center);
	\end{pgfonlayer}
\end{tikzpicture}
}
    \coloneqq e^{i\varphi}
    \qquad \interp\gid \coloneqq1
    \qquad
    \interp{\emptyc} \coloneqq0
  \]
\end{definition}

\begin{figure*}[htb]
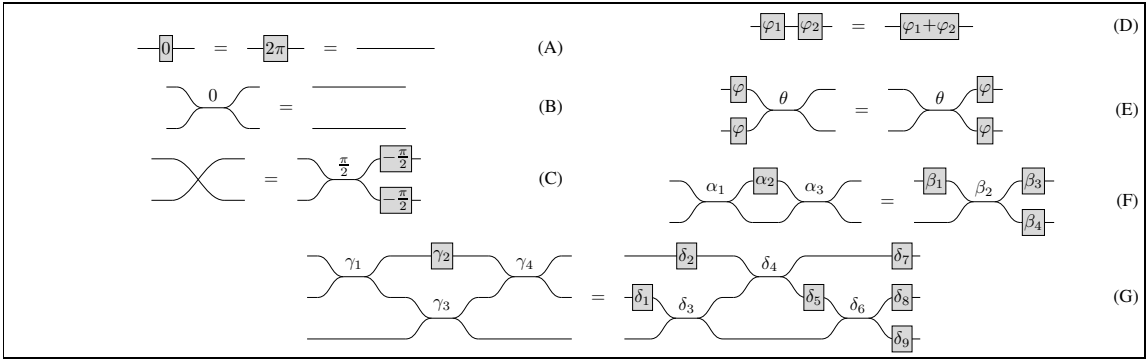

\centering
\scalebox{.69}{\fbox{\begin{minipage}{1.2\textwidth}
\begin{multicols}{2}
\newcommand{\nspazer}{0em}
\begin{equation}\tag{A}\label{phase0}\begin{array}{rcccl}\input{phase-shift0.tikz}&=&\input{phase-shift2pi.tikz}&=&\input{filcourt.tikz}\end{array}\end{equation}
\vspace{\nspazer}
\begin{equation}\tag{B} \label{bs0}\begin{array}{rcl}\input{bs0.tikz}&=&\input{filsparalleleslongbs-m.tikz}\end{array}\end{equation}
\vspace{\nspazer}
\begin{equation}\tag{C} \label{swapbspisur2}\begin{array}{rcl}\input{swap.tikz}&=&\input{bspissur2.tikz}\end{array}\end{equation}\vspace{\nspazer}
\begin{equation}\tag{D} \label{phaseaddition}\begin{array}{rcl}\input{convtp-phase-shifts-12.tikz}&=&\input{convtp-phase-shift-1plus2.tikz}\end{array}\end{equation}
\vspace{\nspazer}
\begin{equation}\tag{E} \label{globalphasepropagationbs}\begin{array}{rcl}\input{convtp-thetathetabs.tikz}&=&\input{convtp-bsthetatheta.tikz}\end{array}\end{equation}
\vspace{\nspazer}
\begin{equation}\tag{F} \label{Eulerbsphasebs}\begin{array}{rcl}\input{bsphasebsalpha.tikz}&=&\input{phasebsphasebeta.tikz}\end{array}\end{equation}
\end{multicols}
\vspace{-0.5cm}
\begin{equation}\tag{G}\label{Eulerscalaires}\begin{array}{rcl}~\qquad\qquad\input{bsyangbaxterpointeenbas-simp-gammas.tikz}&=&\input{bsyangbaxterpointeenhaut-deltas.tikz}\end{array}\end{equation}
\end{minipage}}}
\caption{Axioms of the LOPP-calculus.  In \cref{Eulerbsphasebs} and
  \cref{Eulerscalaires}, the LHS circuit has arbitrary parameters
  which uniquely determine the parameters of the RHS circuit. For any
  $\alpha_i\in\mathbb R$, there exist $\beta_j\in[0,2\pi)$ such that
  \cref{Eulerbsphasebs} is sound, and for any $\gamma_i\in\mathbb R$,
  there exist $\delta_j\in[0,2\pi)$ such that \cref{Eulerscalaires} is
  sound. We can ensure that the angles $\beta_j$ are unique by
  assuming that $\beta_1,\beta_2\in [0,\pi)$ and if
  $\beta_2\in\{0,\frac\pi2\}$ then $\beta_1=0$, and we can ensure that
  the angles $\delta_j$ are unique by assuming that
  $\delta_1,\delta_2,\delta_3,\delta_4,\delta_5,\delta_6\in[0,\pi)$. If
  $\delta_3\in\{0,\frac\pi2\}$ then $\delta_1=0$, if
  $\delta_4\in\{0,\frac\pi2\}$ then $\delta_2=0$, if $\delta_4=0$ then
  $\delta_3=0$, and if $\delta_6\in\{0,\frac\pi2\}$ then $\delta_5=0$.
  The existence and uniqueness of such $\beta_j$ and $\delta_j$ are
  given by Lemmas 10 and 11 of \cite{clement2022LOv}. 
  \label{axiomsLOPP}}
\end{figure*}

\begin{remark}
The definition of $\interp .$ relies on the inductive structure of raw \LOPP-circuits, it is however well-defined on \LOPP-circuits as for any raw \LOPP-circuits $C,C'$, $C\equiv C'$ implies $\interp C = \interp{C'}$. 
\end{remark}

We consider a simple equational theory for \LOPP-circuits
(\cref{axiomsLOPP}), which is derived from the rewriting system
introduced in \cite{clement2022LOv}. Contrary to the rewriting system
of \cite{clement2022LOv}, the swap is part of \LOPP-circuits.
Moreover, the most elaborate equation -- \cref{Eulerscalaires} -- is
slightly simplified in the present paper to have one parameter less.

We use the notation $\textup{LOPP}\vdash C_1=C_2$ whenever $C_1$ can
be transformed into $C_2$ using the equations of \cref{axiomsLOPP} (and
circuit deformations of \cref{fig:axiom}). 

\begin{theorem} \label{thm:LOPPcompleteness} The equational theory
  given by \cref{axiomsLOPP} is sound and complete: for any
  $\textup{LOPP}$-circuits $C_1, C_2$, $\textup{LOPP}\vdash C_1=C_2$
  iff $\interp {C_1} = \interp{C_2}$.
\end{theorem} 

\begin{proof}
  The soundness can be shown with the semantics given in
  \cref{semLOPP}.  Regarding completeness, we show that we can derive
  from \cref{axiomsLOPP} the rules of the strongly normalising
  rewriting system of \cite{clement2022LOv}.  The full proof is given
  in \cref{proof:thmLOPPcompleteness}.
\end{proof}

\subsection{Forgetting the monoidal structure}\label{sec:bare}

The proof of completeness for quantum circuits is based on a back and
forth translation from linear optical circuits. While both kinds of
circuits form a prop, so both have a monoidal structure, these
monoidal structures do not coincide. The monoidal structure of quantum
circuits corresponds to the tensor product, whereas that of linear
optical circuits is a direct sum. Hence the translations do not
preserve the monoidal structure.

As a consequence there is a technical issue around defining the
translation directly on circuits. We instead define the transformations
on \emph{raw} circuits (cf. \cref{sec:raw}). The collection of raw quantum (resp. LOPP) circuits is denoted $\QCbarebf$ (resp. $\LOPPbarebf$).
Notice that we
recover the standard circuits by considering the raw circuits up to
the equivalence relation $\equiv$ given in \cref{fig:axiom}:
$ \QCbf ={\QCbarebf}{/{\equiv}}$ and
$\LOPPbf = {\LOPPbarebf}{/{\equiv}} $.

To avoid ambiguity in the graphical representation of raw circuits
one can use boxes like $\scalebox{0.7}{\input{XXXbox.tikz}}$ for
$(\gx \otimes \gx)\otimes \gx$. We also use box-free graphical
representation that we interpret as a layer-by-layer description of a
raw circuit, more precisely we associate with any box-free graphical
representation, a raw-circuit of the form
$C=(\ldots((L_1\circ L_2)\circ L_3)\circ \ldots )\circ L_k$ where
$L_i=(\ldots((g_{i,1}\otimes g_{i,2})\otimes g_{i,3})\otimes \ldots
)\otimes g_{i,\ell_i}$.

For instance, $((id_1\otimes id_1)\otimes X)\circ (CNot
  \otimes H)$ is
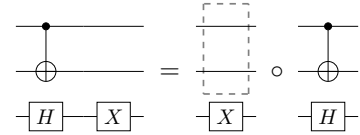
\[
  \scalebox{.8}{\input{bare-circuit-ex.tikz}} = \scalebox{.8}{\input{bare-circuit-ex-2.tikz}} \circ
  \scalebox{.8}{\input{bare-circuit-ex-1.tikz}}
\]

We extend the notation $\textup{QC}\vdash \cdot = \cdot$ and
$\textup{LOPP}\vdash \cdot = \cdot$ to raw circuits. For any raw
quantum circuits (resp.\ raw optical circuits) $C_1,C_2$, we write
$\textup{QC}\vdash C_1 = C_2$ (resp. $\textup{LOPP}\vdash C_1=C_2$) if
$C_1$ and $C_2$ are equivalent by the congruence defined in
\cref{eq:qc0}, \cref{fig:euler} and \cref{fig:axiom} (resp.\
\cref{axiomsLOPP} and \cref{fig:axiom}).\footnote{In this context, the circuits depicted in Figures \ref{eq:qc0}, \ref{fig:euler} and \ref{axiomsLOPP} are interpreted as box-free graphical representations of raw circuits.\label{rawinfigures}}

Notice that there exists a derivation between two circuits if and only
if there exists a derivation between two of their representative raw
circuits. Indeed, intuitively the only difference is that the
derivation on raw circuits is more fine-grained as the equivalence
relation $\equiv$ is made explicit.

\subsection{Encoding quantum circuits into optical ones}

We are now ready to define the encoding of (raw) quantum circuits
into (raw) linear optical circuits. For dimension reasons, an  $n$-qubit system is encoded into $2^n$ modes. 
One can
naturally choose to encode $\ket x$, with $x\in \{0,1\}^n$, into the
mode $\ket{\underline x}$ where $\underline x =\sum_{i=1}^{n}x_i2^{n-i}$ is the
usual binary encoding.  Alternatively, we use Gray codes to produce
circuits with a simpler connectivity, in particular two adjacent modes
encode basis qubit states which differ on exactly one qubit.

\begin{definition}[Gray code]\label{defgraycode}
  Let
  $\mathfrak G_n: \mathbb C^{2^n} \to \mathbb C^{\{0,1\}^n}$ be the map $\ket k
  \mapsto \ket{G_n(k)}$ where $G_n(k)$ is the Gray code of $k$,
  inductively defined by $G_0(0)=\epsilon$ and
  \[
    G_n(k)=\begin{cases}
      0G_{n-1}(k)&\text{if $k<2^{n-1}$,}
      \\
      1G_{n-1}(2^n-1-k)
      &
      \text{if $k\geq 2^{n-1}$.}
    \end{cases}
  \]
\end{definition}

For instance $G_3$ is defined as follows:
\[\begin{array}{rclrcl}
0&\mapsto& 000&\qquad 4&\mapsto 110\\
1&\mapsto& 001&\qquad 5&\mapsto 111\\
2&\mapsto& 011&\qquad 6&\mapsto 101\\
3&\mapsto& 010&\qquad 7&\mapsto 100\\
\end{array}\]
In order to get around the fact
that the encoding an $n$-qubit circuit into a $2^n$-mode optical circuit cannot preserve the parallel composition, we
proceed by `sequentialising' the circuit:
roughly speaking, an $n$-qubit circuit is seen as a sequential composition of layers, each layer being an $n$-qubit circuit made of an elementary gate $g$ acting on at most two qubits in parallel with the identity on all other qubits, \emph{e.g.} $id_k\otimes g
\otimes id_{l}$. The encoding of such a layer, denoted $E_{k,l}(g)$, is a $2^n$-mode optical circuit acting non-trivially on potentially all the modes.  

For instance,  consider a $3$-qubit layer which consists in applying $P(\varphi)$ on the second qubit. Its semantics is $\ket{x,y,z}\mapsto e^{i\varphi y}\ket{x,y,z}$. Such a circuit is encoded into an $8$-mode  optical circuit $E_{1,1}(P(\varphi))$
made of $4$ phase shifters acting on the modes $p\in[2,5]$ (those s.t. $G_3(p)=x1z$). 
Indeed, the semantics of $E_{1,1}(P(\varphi))$ is $\ket p\mapsto \begin{cases}e^{i\varphi }\ket{p} & \text{if~} p\in[2,5]\\ \ket{p} &\text{otherwise}\end{cases}$.

The encoding map is formally defined as follows: 

\begin{definition}[Encoding]\label{def:encoding}
  Let $E:\QCbarebf \to \LOPPbarebf$ be defined as follows: for any
  $n$-qubit circuit $C$, $E(C)=E_{0,0}(C)$ where $E_{k,\ell}$ is
  inductively defined as:
  \begin{itemize}
  \item $E_{k,\ell}(C_1\otimes C_2) = E_{k+n_1,\ell}(C_2)\circ
    E_{k,\ell+n_2}(C_1)$, where $C_1$ (resp.\ $C_2$) is acting on $n_1$ (resp.\ $n_2$) qubits;
  \item $E_{k,\ell}(C_2\circ C_1) = E_{k,\ell}(C_2)\circ  E_{k,\ell}(C_1)$;  
  \end{itemize}
Let us define $\sigma_{k,n,\ell}$ as a $2^{k+n+\ell}$-mode linear optical
circuit made only of swaps (that is, without any
$\begin{tikzpicture}[scale=0.3]
	\begin{pgfonlayer}{nodelayer}
		\node [style=none] (5) at (-1.6, 0) {};
		\node [style=none] (9) at (1.6, 0) {};
		\node [style=PolRot,scale=0.75] (11) at (0, 0) {\small $\varphi$};
	\end{pgfonlayer}
	\begin{pgfonlayer}{edgelayer}
		\draw (5.center) to (11);
		\draw (11) to (9.center);
	\end{pgfonlayer}
\end{tikzpicture}
$ or $\input{bs-xs.tikz}$) such that
$\mathfrak G_n\circ\interp{\sigma_{k,n,\ell}}\circ\mathfrak
G_n^{-1}(\ket{x,y,z})=\ket{x,z,y}$ for any $x\in \{0,1\}^k$,
$y\in \{0,1\}^n$ and $z\in \{0,1\}^\ell$. We then define
    \begin{align*}
      E_{k,\ell}(\gswap)
      &= \sigma_{k,\ell,2}\circ\sigma_{k+\ell,1,1}\circ   \sigma_{k,2,\ell},
      \\
      E_{k,\ell}(\begin{tikzpicture}[scale=0.2]
	\begin{pgfonlayer}{nodelayer}
		\node [style=sdiagrammevide] (0) at (0, 0) {};
	\end{pgfonlayer}
\end{tikzpicture}
)
      &=
      (\begin{tikzpicture}
	\begin{pgfonlayer}{nodelayer}
		\node [style=none] (0) at (-0.8, 0) {};
		\node [style=none] (1) at (0.8, 0) {};
	\end{pgfonlayer}
	\begin{pgfonlayer}{edgelayer}
		\draw [style=noire] (0.center) to (1.center);
	\end{pgfonlayer}
\end{tikzpicture})^{\otimes
        {2^{k+\ell}}},
      \\
      E_{k,\ell}(\gid)
      &= (\begin{tikzpicture}
	\begin{pgfonlayer}{nodelayer}
		\node [style=none] (0) at (-0.8, 0) {};
		\node [style=none] (1) at (0.8, 0) {};
	\end{pgfonlayer}
	\begin{pgfonlayer}{edgelayer}
		\draw [style=noire] (0.center) to (1.center);
	\end{pgfonlayer}
\end{tikzpicture})^{\otimes {2^{k+\ell+1}}},
      \\
      E_{k,\ell}(s(\varphi))
      &=
        \left(\begin{tikzpicture}[scale=0.3]
	\begin{pgfonlayer}{nodelayer}
		\node [style=none] (5) at (-1.6, 0) {};
		\node [style=none] (9) at (1.6, 0) {};
		\node [style=PolRot,scale=0.75] (11) at (0, 0) {\small $\varphi$};
	\end{pgfonlayer}
	\begin{pgfonlayer}{edgelayer}
		\draw (5.center) to (11);
		\draw (11) to (9.center);
	\end{pgfonlayer}
\end{tikzpicture}
\right)^{\otimes
      {2^{k+\ell}}}.
    \end{align*}
 where $C^{\otimes n}$ means $C$ $n$ times in parallel: $C^{\otimes 0} = \begin{tikzpicture}[scale=0.2]
	\begin{pgfonlayer}{nodelayer}
		\node [style=sdiagrammevide] (0) at (0, 0) {};
	\end{pgfonlayer}
\end{tikzpicture}
 $ and $C^{\otimes n+1} = C\otimes C^{\otimes n}$.

  For the remaining generators, we have:
  \begin{align*}
    E_{0,0}(\gh) &= \input{H-LOPP-xs.tikz}, 
    \\
    E_{0,0}(\gp) &= \input{Z-PHOL.tikz}, 
    \\
    E_{0,0}(\gcnot) &= \input{CNot-PHOL.tikz},
  \intertext{and whenever $(k,\ell)\neq(0,0)$:}
   E_{k,\ell}(\gh) &=
    \sigma_{k,\ell,1}\circ  \left(\input{H-LOPP-doublesym-xs.tikz}\right)^{\otimes {2^{k+\ell-1}}}\circ   \sigma_{k,1,\ell},
    \\
    E_{k,\ell}(\gp) &=
    \sigma_{k,\ell,1}\circ
                      \left(\input{Z-PHOL-doublesym.tikz}\right)^{\otimes
                      {2^{k+\ell-1}}}\circ   \sigma_{k,1,\ell},
    \\
    E_{k,\ell}(\gcnot) &=
    \sigma_{k,\ell,2}\circ \left(\input{CNot-PHOL-doublesym.tikz}\right)^{\otimes {2^{k+\ell-1}}}\circ   \sigma_{k,2,\ell}.
  \end{align*}

\end{definition}

\begin{remark}
Note that for any $n$-qubit circuit $C$, $E_{k,\ell}(C)$ is a $2^{k+n+\ell}$-mode optical circuit. Also note that $\sigma_{k,n,\ell}$ is nothing but a permutation of wires. By
\cref{decodingtoporules} -- which is independent from the definition
of $E$ -- any actual circuit satisfying the above property ($\mathfrak G_n\circ\interp{\sigma_{k,n,\ell}}\circ\mathfrak
G_n^{-1}(\ket{x,y,z})=\ket{x,z,y}$) is convenient for our purposes.
A formal definition of $\sigma_{k,n,\ell}$ is however given in \cref{defsigma}.
\end{remark}

\begin{figure*}
  \fbox{\begin{minipage}{1\textwidth}\centering
  $\begin{array}{rcl}
     E(C_0)
     &=&E_{0,0}(\gcnot \otimes \gh)\\[0.2cm]
     &=&E_{2,0}(\gh)\circ E_{0,1}(\gcnot)\\[0.2cm]
     &=&\sigma_{2,0,1} \circ
         \left(\input{H-LOPP-doublesym-xs.tikz}\right)^{\otimes 2}
         \circ \sigma_{2,1,0} \circ \sigma_{0,1,2} \circ
         \left(\input{CNot-PHOL-doublesym.tikz}\right) \circ
         \sigma_{0,2,1}
     \\[0.2cm]
     &=& id_8 \circ
         \left(\input{H-LOPP-doublesym-xs.tikz}\right)^{\otimes 2}
         \circ id_8 \circ \left(\gsigmazot\right)
         \circ \left(\input{CNot-PHOL-doublesym.tikz}\right)
         \circ \left(\gsigmazto\right)
     \\[0.2cm]
   \end{array}$
 \end{minipage}}
 \caption{Encoding of the circuit discussed in \cref{ex:encod}.}
   \label{fig:ex:encod}
\end{figure*}

\begin{example}\label{ex:encod}
  Consider the simple circuit $C_0=\gcnotexqc$. The
  encoding is as shown in \cref{fig:ex:encod}.
Using the topological rules (\cref{fig:axiom}), one can simplify $E (C_0)$ into the circuit $C_1$:
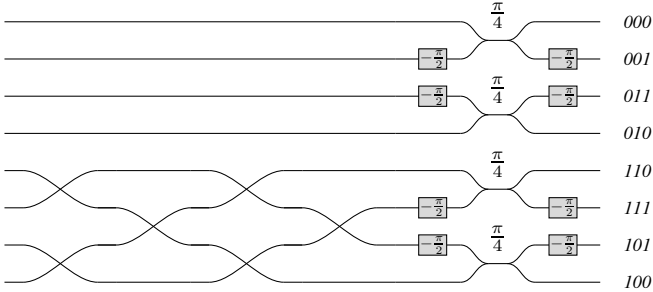
\[
     \scalebox{1.1}{\gcnotexloppgray} \\[0.2cm]
\]
\end{example}

\bigskip
The encoding of quantum circuits into linear optical circuits preserves the
semantics, up to Gray codes.

\begin{proposition} 
\label{prop:sem-preserving} For any $n$-qubit
  quantum circuit $C$,
  \[
    \mathfrak G_n \circ \interp{E(C)} = \interp{C}\circ \mathfrak
    G_n
  \]
\end{proposition}

\begin{proof}
  By induction.
\end{proof}

\subsection{Decoding}

Regarding the decoding, i.e.\ the translation back from linear optical
circuits to quantum circuits, we use the same sequentialisation
approach. Note that such a decoding is defined only for optical
circuits with a power of two number of modes. 

The decoding of a $2^n$-mode layer $id_{k}\otimes g \otimes id_{l}$ is a $n$-qubit circuit denoted $D_{k,n}(g)$.  For instance consider a $16$-mode layer which consists in applying $\begin{tikzpicture}[scale=0.3]
	\begin{pgfonlayer}{nodelayer}
		\node [style=none] (5) at (-1.6, 0) {};
		\node [style=none] (9) at (1.6, 0) {};
		\node [style=PolRot,scale=0.75] (11) at (0, 0) {\small $\varphi$};
	\end{pgfonlayer}
	\begin{pgfonlayer}{edgelayer}
		\draw (5.center) to (11);
		\draw (11) to (9.center);
	\end{pgfonlayer}
\end{tikzpicture}
$ on the fourth mode.
Its semantics is $\ket{p}\mapsto\begin{cases}e^{i\varphi }\ket{p} & \text{if~} p=3\\ \ket{p} &\text{otherwise}\end{cases}$. Such a circuit is decoded into a $4$-qubit circuit $D_{3,4}(\begin{tikzpicture}[scale=0.3]
	\begin{pgfonlayer}{nodelayer}
		\node [style=none] (5) at (-1.6, 0) {};
		\node [style=none] (9) at (1.6, 0) {};
		\node [style=PolRot,scale=0.75] (11) at (0, 0) {\small $\varphi$};
	\end{pgfonlayer}
	\begin{pgfonlayer}{edgelayer}
		\draw (5.center) to (11);
		\draw (11) to (9.center);
	\end{pgfonlayer}
\end{tikzpicture}
)$
implementing the multi-controlled phase $\Lambda^{G_4(3)}s(\varphi)$, whose semantics is $\ket{x,y,z,t}\mapsto  \begin{cases}e^{i\varphi }\ket{x,y,z,t} & \text{if~} xyzt = G_4(3)\\ \ket{x,y,z,t} &\text{otherwise}\end{cases}$.

The decoding map is formally defined as follows:

\begin{definition}[Decoding]\label{defdecoding}
  Let $D: \LOPPbarebf\to \QCbarebf$ be defined as follows: for any
  $2^n$-mode circuit $C$, $D(C)=D_{0,n}(C)$ where for any $n,k,\ell$
  with $k+\ell\leq2^n$ and $C:\ell\to\ell$, $D_{k,n}(C)$ is
  inductively defined as follows.
  \begin{itemize}
  \item $D_{k,n}(C_1\otimes C_2) = D_{k+\ell_1,n}(C_2)\circ  D_{k,n}(C_1)$, where $C_1$ is acting on $\ell_1$ modes;
  \item $D_{k,n}(C_2\circ C_1) = D_{k,n}(C_2)\circ  D_{k,n}(C_1)$;
  \item $D_{k,n}(\begin{tikzpicture}
	\begin{pgfonlayer}{nodelayer}
		\node [style=none] (0) at (-0.8, 0) {};
		\node [style=none] (1) at (0.8, 0) {};
	\end{pgfonlayer}
	\begin{pgfonlayer}{edgelayer}
		\draw [style=noire] (0.center) to (1.center);
	\end{pgfonlayer}
\end{tikzpicture}) = id_{n}$.
  \end{itemize}
  The remaining generators are treated as follows.
  \begin{align*}
    D_{k,n}(\begin{tikzpicture}[scale=0.2]
	\begin{pgfonlayer}{nodelayer}
		\node [style=sdiagrammevide] (0) at (0, 0) {};
	\end{pgfonlayer}
\end{tikzpicture}
) &= id_{n},
    &
    D_{k,n}(\begin{tikzpicture}[scale=0.3]
	\begin{pgfonlayer}{nodelayer}
		\node [style=none] (5) at (-1.6, 0) {};
		\node [style=none] (9) at (1.6, 0) {};
		\node [style=PolRot,scale=0.75] (11) at (0, 0) {\small $\varphi$};
	\end{pgfonlayer}
	\begin{pgfonlayer}{edgelayer}
		\draw (5.center) to (11);
		\draw (11) to (9.center);
	\end{pgfonlayer}
\end{tikzpicture}
) &= \Lambda^{G_n(k)}
                                               s(\varphi),
    \\
    D_{k,n}(\input{swap-s.tikz}) &= \Lambda^{x_{k,n}}_{y_{k,n}} X,
    &
      D_{k,n}(\input{bs-xs.tikz}) &= \Lambda^{x_{k,n}}_{y_{k,n}} R_X(-2\theta),
  \end{align*}
  where $x_{2k,n}\coloneqq G_{n-1}(k)$, $y_{2k,n}\coloneqq\epsilon$, $x_{2k+1,n}\coloneqq w$ and $y_{2k+1,n}\coloneqq 1.0^q$, where $q\in\{0,...,n-2\}$ and $w\in\{0,1\}^{n-q-2}$ are such that $G_{n}(2k+1)=wa1.0^q$ for some $a\in\{0,1\}$.
\end{definition}

\begin{example}\label{ex:decod}
  We consider the optical circuit $C_1$ obtained in \cref{ex:encod}. 
  With all of the gates $P$ and $R_X$ parametrized with
  $\frac{-\pi}{2}$, we can show that $D(C_1)\equiv$
  \[
      \scalebox{1}{\gcnotdec}
  \]
\end{example}
Similarly to the encoding function, the decoding function preserves the
semantics up to Gray codes.

\begin{proposition}
  For any $2^n$-mode optical circuit $C$,
  $$\interp{D(C)}\circ \mathfrak G_n = \mathfrak G_n \circ
  \interp{C}.$$
\end{proposition}

\begin{proof}
  The proof is by induction. 
\end{proof}

\subsection{Quantum circuit completeness}\label{sec:QCcompleteness}

The proof of completeness is based on the encoding/decoding of quantum
circuits into optical circuits. Intuitively, given two quantum
circuits representing the same unitary map, one can encode them as
linear optical circuits. Since the encoding preserves the semantics
and LOPP is complete, there exists a derivation proving the
equivalence of the encoded circuits. In order to lift this proof to
quantum circuits, it remains to prove that the decoding of an
encoded quantum circuit is provably equivalent to the original quantum
circuit, and that each axiom of LOPP can be mimicked in
\textup{QC}. Notice that since the encoding/decoding is defined on
raw circuits, an extra step in the proof consists in showing that the
axioms of $\equiv$ can also be mimicked in \textup{QC}.

Examples (\ref{ex:encod}) and (\ref{ex:decod}) point out that
composing encoding and decoding does not lead, in general, to the
original circuit, the decoded circuit being made of multi-controlled
gates. However, we show that the equivalence with the initial circuit can always be
derived in $\textup{QC}$:

\begin{lemma}\label{DE}
  For any $n$-qubit raw quantum circuit $C$,
  \[\textup{QC}\vdash D(E(C))=C.\]
\end{lemma}

\begin{proof}
  We prove by structural induction on $C$ that
  \[\forall k,\ell,\ \textup{QC}\vdash D(E_{k,\ell}(C))=id_k\otimes C\otimes id_\ell.\]
  For any two $n$-qubit raw circuits $C_1,C_2$, one has
  \[D(E_{k,\ell}(C_2\circ C_1))= D(E_{k,\ell}(C_2))\circ D(E_{k,\ell}(C_1))\]
  and for any $m$-qubit raw circuit $C_3$,
  \[D(E_{k,\ell}(C_1\otimes C_3))=D(E_{k+n,\ell}(C_3))\circ D(E_{k,\ell+m}(C_1)).\]
  Hence, it remains the basis cases which are proved as \cref{basecaseDE} in \cref{lemmasforbasecaseDE}.
\end{proof}

Note that in general, the decoding function does not preserve the
topological equivalence. For instance, with the raw circuits
$C_1=\scalebox{0.56}{$\input{exnaturalitebsL.tikz}$}$ and
$C_2=\scalebox{0.56}{$\input{exnaturalitebsR.tikz}$}$, we have
$C_1\equiv C_2$ but
$D(C_1)=\scalebox{0.8}{$\input{exdecodagenaturalitebsL.tikz}$}$ and
$D(C_2)=\scalebox{0.8}{$\input{exdecodagenaturalitebsR.tikz}$}$\smallskip. Thus,
the topological rules also have to be mimicked in \textup{QC}:

\begin{lemma}\label{decodingtoporules}
  For any $2^n$-mode raw optical circuits $C_1, C_2$, if
  ${C_1} \equiv C_2$ then $\textup{QC}\vdash D(C_1) = D(C_2)$.
\end{lemma}

\begin{proof}
  The proof consists intuitively in verifying that the decoding of
  every equation of \cref{fig:axiom} is provable in $\textup{QC}$. The
  proof is given in
  \cref{preuvedecodingtoporules}.
\end{proof}

\begin{lemma}\label{decodingLOPPrules}
  For any $2^n$-mode raw optical circuits $C_1, C_2$, if
  $\textup{LOPP}\vdash C_1 = C_2$ then
  $\textup{QC} \vdash D(C_1) = D(C_2)$.
\end{lemma}
\begin{proof}
  The proof consists intuitively in verifying that the decoding of
  every equation of \cref{axiomsLOPP} is provable in
  $\textup{QC}$. The proof is given in \cref{preuvedecodingLOPPrules}.
\end{proof}

We are now ready to prove the main result of the paper.

\begin{theorem}[Quantum circuit
  completeness] \label{thm:QCcompleteness} \textup{QC} is a complete
  equational theory for quantum circuits: for any quantum circuits
  $C_1$, $C_2$, if $\interp{C_1}=\interp {C_2}$ then
  $\textup{QC}\vdash C_1=C_2$.
\end{theorem} 

\begin{proof}
  Given two quantum circuits $C_1$, $C_2$ s.t.
  $\interp{C_1}=\interp {C_2}$, let $C_1'$ (resp.\ $C_2'$) be a raw
  quantum circuit, representative of $C_1$ (resp.\ $C_2$).  Thanks to
  \cref{prop:sem-preserving} we have
  $\interp{E(C'_1)}=\interp {E(C'_2)}$. The completeness of
  $\textup{LOPP}$ implies $\textup{LOPP}\vdash E(C'_1)=E(C'_2)$. By
  \cref{decodingLOPPrules}, we have
  $\textup{QC}\vdash D(E(C'_1)) = D(E(C'_2))$. Moreover \cref{DE}
  implies $\textup{QC}\vdash C'_1 = C'_2$. From this derivation we
  obtain a derivation of $\textup{QC}\vdash C_1 = C_2$, where the
  steps corresponding to the equivalence relation $\equiv$ are
  trivialised.
\end{proof}

\section{Discussions}

We have introduced the first complete equational theory for quantum
circuits. Although this equational theory is fairly simple,
\cref{Euler3dmulticontrolled} is an unbounded family of equations
---one for each possible number of control qubits. Such a family of
equations is a natural byproduct of our proof technique: The decoding
of each axiom of LOPP produces an equation made of multi-controlled
gates that has to be derived using QC.
It is actually quite surprising that \cref{Euler3dmulticontrolled} is
the only remaining equation with multi-controlled
gates. 

Notice that one can get rid of these multi-controlled gates by
extending the context rule as described below. Indeed,
\cref{Euler3dmulticontrolled} can be derived from its 2-qubit case
\begin{multline}\label{Euler3d}
  \tag{r'}
  \scalebox{.6}{\input{Euler3D-2qubit-L.tikz}}=\\
  \scalebox{.6}{\input{Euler3D-2qubit-R.tikz}}
\end{multline}

 \noindent if one allows the following control context rule  $\vdash \Lambda C_1 =
 \Lambda C_2$ when $\vdash C_1=C_2$. 
 Notice that it requires extending the $\Lambda$-construction to any circuit -- which can be done in an inductive way like $\Lambda(C_2\circ  C_1)  = \Lambda C_2 \circ \Lambda C_1$ and $\Lambda (C_1\otimes C_2) = (\Lambda C_1\otimes id_m)\circ (id_1\otimes \sigma_{m,n})\circ (\Lambda C_2\otimes id_n)\circ (id_1\otimes \sigma_{n,m})$. 
 
\vspace{0.2cm}

A natural application of
the completeness result is to design procedures for quantum circuit
optimisation based on this equational theory. One can take advantage of the terminating and confluent rewriting system for optical circuits \cite{clement2022LOv} by mimicking the applications of the rewrite rules on quantum circuits. However, the exponential blowup of the encoding map makes this approach probably inefficient as it is and requires some improvements.

Another future work is to prove (upper or lower) bounds on the size of
a derivation between two given equivalent circuits, as well as a bound
on the size of the intermediate quantum circuits. This might be
useful for providing a verifiable quantum advantage, in particular if
there exist polysize quantum circuits requiring exponentially
many rewrites \cite{aaronson-slide}.

\subsection*{Acknowledgements}

The authors wish to thank Emmanuel Jeandel for his comments on an
early version of this paper.  This work is supported by the PEPR
integrated project EPiQ ANR-22-PETQ-0007 part of Plan France 2030,
the CIFRE 2022/0081, the French National
Research Agency (ANR) under the research projects SoftQPro
ANR-17-CE25-0009-02 and VanQuTe ANR-17-CE24-0035, BPI France under the Concours Innovation PIA3 projects DOS0148634/00 and DOS0148633/00, by the STIC-AmSud
project Qapla’ 21-STIC-10, and by the European projects NEASQC and
HPCQS.

\bibliographystyle{IEEEtran}
\bibliography{ref}

\begin{thebibliography}{10}
\providecommand{\url}[1]{#1}
\csname url@samestyle\endcsname
\providecommand{\newblock}{\relax}
\providecommand{\bibinfo}[2]{#2}
\providecommand{\BIBentrySTDinterwordspacing}{\spaceskip=0pt\relax}
\providecommand{\BIBentryALTinterwordstretchfactor}{4}
\providecommand{\BIBentryALTinterwordspacing}{\spaceskip=\fontdimen2\font plus
\BIBentryALTinterwordstretchfactor\fontdimen3\font minus
  \fontdimen4\font\relax}
\providecommand{\BIBforeignlanguage}[2]{{%
\expandafter\ifx\csname l@#1\endcsname\relax
\typeout{** WARNING: IEEEtran.bst: No hyphenation pattern has been}%
\typeout{** loaded for the language `#1'. Using the pattern for}%
\typeout{** the default language instead.}%
\else
\language=\csname l@#1\endcsname
\fi
#2}}
\providecommand{\BIBdecl}{\relax}
\BIBdecl

\bibitem{deutsch-circuit}
D.~Deutsch, ``Quantum computational networks,'' \emph{Proceedings of the Royal
  Society of London. Series A, Mathematical and Physical Sciences}, vol. 425,
  no. 1868, pp. 73--90, 1989.

\bibitem{green2013quipper}
A.~S. Green, P.~L. Lumsdaine, N.~J. Ross, P.~Selinger, and B.~Valiron,
  ``Quipper: A scalable quantum programming language,'' in \emph{Proceedings of
  the {ACM} {SIGPLAN} Conference on Programming Language Design and
  Implementation, {PLDI}'13}, H.-J. Boehm and C.~Flanagan, Eds.\hskip 1em plus
  0.5em minus 0.4em\relax {ACM}, 2013, pp. 333--342.

\bibitem{amy2014polynomial}
M.~Amy, D.~Maslov, and M.~Mosca, ``Polynomial-time {T}-depth optimization of
  {C}lifford+{T} circuits via matroid partitioning,'' \emph{IEEE Transactions
  on Computer-Aided Design of Integrated Circuits and Systems}, vol.~33,
  no.~10, pp. 1476--1489, 2014.

\bibitem{duncan2020graph}
R.~Duncan, A.~Kissinger, S.~Perdrix, and J.~Van De~Wetering, ``Graph-theoretic
  simplification of quantum circuits with the {ZX}-calculus,'' \emph{Quantum},
  vol.~4, p. 279, 2020.

\bibitem{PhysRevA.102.022406}
\BIBentryALTinterwordspacing
A.~Kissinger and J.~van~de Wetering, ``Reducing the number of non-{C}lifford
  gates in quantum circuits,'' \emph{Phys. Rev. A}, vol. 102, p. 022406, Aug.
  2020. [Online]. Available:
  \url{https://link.aps.org/doi/10.1103/PhysRevA.102.022406}
\BIBentrySTDinterwordspacing

\bibitem{maslov2005quantum}
D.~Maslov, C.~Young, D.~M. Miller, and G.~W. Dueck, ``Quantum circuit
  simplification using templates,'' in \emph{Design, Automation and Test in
  Europe}.\hskip 1em plus 0.5em minus 0.4em\relax IEEE, 2005, pp. 1208--1213.

\bibitem{maslov2008quantum}
D.~Maslov, G.~W. Dueck, D.~M. Miller, and C.~Negrevergne, ``Quantum circuit
  simplification and level compaction,'' \emph{IEEE Transactions on
  Computer-Aided Design of Integrated Circuits and Systems}, vol.~27, no.~3,
  pp. 436--444, 2008.

\bibitem{nam2018automated}
Y.~Nam, N.~J. Ross, Y.~Su, A.~M. Childs, and D.~Maslov, ``Automated
  optimization of large quantum circuits with continuous parameters,''
  \emph{npj Quantum Information}, vol.~4, no.~1, pp. 1--12, 2018.

\bibitem{kissinger2019cnot}
A.~Kissinger and A.~M.-v. de~Griend, ``{CNOT} circuit extraction for
  topologically-constrained quantum memories,'' \emph{arXiv preprint
  arXiv:1904.00633}, 2019.

\bibitem{nash2020quantum}
B.~Nash, V.~Gheorghiu, and M.~Mosca, ``Quantum circuit optimizations for {NISQ}
  architectures,'' \emph{Quantum Science and Technology}, vol.~5, no.~2, p.
  025010, 2020.

\bibitem{aaronson-slide}
S.~Aaronson, ``Verifiable quantum advantage: What {I} hope will be done,'' Set
  of slides, presented at \textit{Quantum Advantage Workshop}, Chicago, IL,
  August 1, 2022. Slide 10. Online at
  \href{https://www.scottaaronson.com/talks/whatihope.ppt}{https://www.scottaaronson.com/talks/whatihope.ppt}.

\bibitem{bian2022generators}
X.~Bian and P.~Selinger, ``Generators and relations for 2-qubit {C}lifford+{T}
  operators,'' \emph{arXiv preprint arXiv:2204.02217}, 2022.

\bibitem{coecke2018zx}
B.~Coecke and Q.~Wang, ``{ZX}-rules for 2-qubit {C}lifford+{T} quantum
  circuits,'' in \emph{International Conference on Reversible
  Computation}.\hskip 1em plus 0.5em minus 0.4em\relax Springer, 2018, pp.
  144--161.

\bibitem{makary2021generators}
J.~Makary, N.~J. Ross, and P.~Selinger, ``Generators and relations for real
  stabilizer operators,'' in \emph{Proceedings of the 18th International
  Conference on Quantum Physics and Logic, {QPL} 2021}, ser. {EPTCS}, C.~Heunen
  and M.~Backens, Eds., vol. 343, 2021, pp. 14--36.

\bibitem{ranchin2014complete}
A.~Ranchin and B.~Coecke, ``Complete set of circuit equations for stabilizer
  quantum mechanics,'' \emph{Physical Review A}, vol.~90, no.~1, p. 012109,
  2014.

\bibitem{Amy_2018}
\BIBentryALTinterwordspacing
M.~Amy, J.~Chen, and N.~J. Ross, ``A finite presentation of {CNOT}-dihedral
  operators,'' \emph{Electronic Proceedings in Theoretical Computer Science},
  vol. 266, pp. 84--97, Feb. 2018. [Online]. Available:
  \url{https://doi.org/10.4204%2Feptcs.266.5}
\BIBentrySTDinterwordspacing

\bibitem{iwama2002transformation}
K.~Iwama, Y.~Kambayashi, and S.~Yamashita, ``Transformation rules for designing
  {CNOT}-based quantum circuits,'' in \emph{Proceedings of the 39th annual
  Design Automation Conference}, 2002, pp. 419--424.

\bibitem{cockett2018categorycnot}
R.~Cockett, C.~Comfort, and P.~Srinivasan, ``The category {CNOT},'' in
  \emph{Proceedings 15th International Conference on Quantum Physics and Logic,
  {QPL} 2018}, ser. {EPTCS}, P.~Selinger and G.~Chiribella, Eds., vol. 287,
  2019, pp. 258--293.

\bibitem{cockett2018categorytof}
R.~Cockett and C.~Comfort, ``The category {TOF},'' in \emph{Proceedings 15th
  International Conference on Quantum Physics and Logic, {QPL} 2018}, ser.
  {EPTCS}, P.~Selinger and G.~Chiribella, Eds., vol. 287, 2019, pp. 67--84.

\bibitem{coecke2008interacting}
B.~Coecke and R.~Duncan, ``Interacting quantum observables,'' in
  \emph{International Colloquium on Automata, Languages, and
  Programming}.\hskip 1em plus 0.5em minus 0.4em\relax Springer, 2008, pp.
  298--310.

\bibitem{coecke2011interacting}
------, ``Interacting quantum observables: categorical algebra and
  diagrammatics,'' \emph{New Journal of Physics}, vol.~13, no.~4, p. 043016,
  2011.

\bibitem{Backens_2019}
\BIBentryALTinterwordspacing
M.~Backens and A.~Kissinger, ``{ZH}: A complete graphical calculus for quantum
  computations involving classical non-linearity,'' \emph{Electronic
  Proceedings in Theoretical Computer Science}, vol. 287, pp. 23--42, Jan.
  2019. [Online]. Available: \url{https://doi.org/10.4204%2Feptcs.287.2}
\BIBentrySTDinterwordspacing

\bibitem{DBLP:conf/lics/HadzihasanovicN18}
\BIBentryALTinterwordspacing
A.~Hadzihasanovic, K.~F. Ng, and Q.~Wang, ``Two complete axiomatisations of
  pure-state qubit quantum computing,'' in \emph{Proceedings of the 33rd Annual
  {ACM/IEEE} Symposium on Logic in Computer Science, {LICS} 2018, Oxford, UK,
  July 09-12, 2018}, A.~Dawar and E.~Gr{\"{a}}del, Eds.\hskip 1em plus 0.5em
  minus 0.4em\relax {ACM}, 2018, pp. 502--511. [Online]. Available:
  \url{https://doi.org/10.1145/3209108.3209128}
\BIBentrySTDinterwordspacing

\bibitem{jeandel2018complete}
E.~Jeandel, S.~Perdrix, and R.~Vilmart, ``A complete axiomatisation of the
  {ZX}-calculus for {C}lifford+{T} quantum mechanics,'' in \emph{Proceedings of
  the 33rd Annual ACM/IEEE Symposium on Logic in Computer Science}, 2018, pp.
  559--568.

\bibitem{hadzihasanovic2018two}
A.~Hadzihasanovic, K.~F. Ng, and Q.~Wang, ``Two complete axiomatisations of
  pure-state qubit quantum computing,'' in \emph{Proceedings of the 33rd Annual
  ACM/IEEE Symposium on Logic in Computer Science}, 2018, pp. 502--511.

\bibitem{jeandel2018diagrammatic}
E.~Jeandel, S.~Perdrix, and R.~Vilmart, ``Diagrammatic reasoning beyond
  {C}lifford+{T} quantum mechanics,'' in \emph{Proceedings of the 33rd Annual
  ACM/IEEE Symposium on Logic in Computer Science}, 2018, pp. 569--578.

\bibitem{lmcs:6532}
\BIBentryALTinterwordspacing
------, ``{Completeness of the {ZX}-Calculus},'' \emph{{Logical Methods in
  Computer Science}}, vol. {Volume 16, Issue 2}, Jun. 2020. [Online].
  Available: \url{https://lmcs.episciences.org/6532}
\BIBentrySTDinterwordspacing

\bibitem{vilmart2019near}
R.~Vilmart, ``A near-minimal axiomatisation of {ZX}-calculus for pure qubit
  quantum mechanics,'' in \emph{2019 34th Annual ACM/IEEE Symposium on Logic in
  Computer Science (LICS)}.\hskip 1em plus 0.5em minus 0.4em\relax IEEE, 2019,
  pp. 1--10.

\bibitem{de2022circuit}
N.~de~Beaudrap, A.~Kissinger, and J.~van~de Wetering, ``Circuit extraction for
  {ZX}-diagrams can be \#{P}-hard,'' \emph{arXiv preprint arXiv:2202.09194},
  2022.

\bibitem{clement2022LOv}
\BIBentryALTinterwordspacing
A.~Cl{\'e}ment, N.~Heurtel, S.~Mansfield, S.~Perdrix, and B.~Valiron,
  ``{LOv}-calculus: A graphical language for linear optical quantum circuits,''
  2022. [Online]. Available: \url{https://arxiv.org/abs/2204.11787}
\BIBentrySTDinterwordspacing

\bibitem{Pawel}
J.~Paix{\~a}o and P.~Soboci{\'{n}}ski, ``Calculational proofs in relational
  graphical linear algebra,'' in \emph{Formal Methods: Foundations and
  Applications}, G.~Carvalho and V.~Stolz, Eds.\hskip 1em plus 0.5em minus
  0.4em\relax Springer, 2020, pp. 83--100.

\bibitem{prop}
S.~MacLane, ``Categorical algebra,'' \emph{Bulletin of the American
  Mathematical Society}, vol.~71, no.~1, pp. 40 -- 106, 1965.

\bibitem{Barenco1995gates}
\BIBentryALTinterwordspacing
A.~Barenco, C.~H. Bennett, R.~Cleve, D.~P. DiVincenzo, N.~Margolus, P.~Shor,
  T.~Sleator, J.~A. Smolin, and H.~Weinfurter, ``Elementary gates for quantum
  computation,'' \emph{Physical Review A}, vol.~52, no.~5, pp. 3457--3467, Nov.
  1995. [Online]. Available: \url{http://dx.doi.org/10.1103/PhysRevA.52.3457}
\BIBentrySTDinterwordspacing

\bibitem{NielsenChuang}
M.~A. Nielsen and I.~L. Chuang, \emph{Quantum Computation and Quantum
  Information}.\hskip 1em plus 0.5em minus 0.4em\relax Cambridge University
  Press, 2002.

\bibitem{abdessaied2014quantum}
N.~Abdessaied, M.~Soeken, and R.~Drechsler, ``Quantum circuit optimization by
  hadamard gate reduction,'' in \emph{International Conference on Reversible
  Computation}.\hskip 1em plus 0.5em minus 0.4em\relax Springer, 2014, pp.
  149--162.

\bibitem{Cowtan_2020}
\BIBentryALTinterwordspacing
A.~Cowtan, S.~Dilkes, R.~Duncan, W.~Simmons, and S.~Sivarajah, ``Phase gadget
  synthesis for shallow circuits,'' \emph{Electronic Proceedings in Theoretical
  Computer Science}, vol. 318, pp. 213--228, May 2020. [Online]. Available:
  \url{https://doi.org/10.4204%2Feptcs.318.13}
\BIBentrySTDinterwordspacing

\end{thebibliography}

\onecolumn
\appendices

\crefalias{section}{appendix}
\crefalias{subsection}{appendix}
\crefalias{subsubsection}{appendix}
\crefalias{paragraph}{appendix}
\crefalias{subparagraph}{appendix}

\section{Contents}

All the equations from Equation \eqref{CNotPCNotreversible} to Equation \eqref{CNotHCNot} are proven either directly from the axioms of $\textup{QC}_0$, given in
\cref{eq:qc0}, or from the equations already proven. Those proofs are in \cref{proof:usefuleq}.

Proposition \ref{ctrlXCNot} is proven in \cref{preuvectrlXCNot}.

In Appendix \ref{proofinductiveprop}, we highlight the inductive properties of multicontrolled gates which will be used in the inductive proofs of the following appendices, in the form of Lemmas \ref{lem:Lambda0} to \ref{lem:Lambda-Rx}.

Lemmas \ref{commctrlciblesdistinctes} to \ref{decompctrlblancRX} are introduced and proven by induction in \cref{proofmultictrl}. Alongside with Equation \eqref{eq:comRX} proven in \cref{sec:proofcomRX}, those properties are used to prove Propositions \ref{cor:swap} and \ref{prop:CP} in \cref{preuvesswapsmultictrl}.

To prove \cref{prop:sum}, we introduce Lemma \ref{commctrlphaseenhaut}. 
We do a proof by induction with both hypotheses, to prove at the same time \cref{prop:sum} and Lemma \ref{commctrlphaseenhaut}, as detailed in \cref{proofpropsum}.

Appendices \ref{usefulmulctrl}, \ref{proof:commctrlPNotCvar} and \ref{commutationsXtarget} introduce and prove  
Equations \eqref{mctrlX} and \eqref{Euler2dmulticontrolled} and Lemmas \ref{commctrlPNotCvar} to \ref{commctrlphaseenhautP}. Those properties on multi-controlled gates are to be used in other later proofs.

Propositions \ref{prop:comb}, \ref{commctrl}, \ref{commctrlpasenface}, \ref{EulerHmoins}, \ref{soundnessEuleraxioms}, \ref{klfollowfromEuler}, \ref{prop:CCX} and \ref{prop:period} are respectively proven in Appendix \ref{preuvepropcomb}, \ref{preuvecommctrl}, \ref{preuvecommctrlpasenface}, \ref{preuveEulerHmoins}, \ref{appendix:euler3d}, \ref{preuveklfollowfromEuler}, \ref{proof:CCX} and \ref{proof:propperiod}.

Theorem \ref{thm:LOPPcompleteness} is proven in \cref{proof:thmLOPPcompleteness}.

Appendices \ref{usefuldef} and \ref{proof:usefullemmas} introduce convenient notations and Lemmas \ref{antisymmetriecontrolee} to \ref{decodagefilsdisjoints}, useful for proving the main result.

Finally, Lemmas \ref{DE}, \ref{decodingtoporules} and \ref{decodingLOPPrules} are proven in Appendices \ref{lemmasforbasecaseDE}, \ref{preuvedecodingtoporules} and \ref{preuvedecodingLOPPrules}.

 The $\sigma_{k,n,\ell}$ are defined in \cref{defsigma}.

\section{Useful Quantum Circuits Equations}
\subsection{Proofs of Equations \eqref{CNotPCNotreversible} to \eqref{CNotHCNot}}\label{proof:usefuleq}

\noindent Proof of \cref{CNotPCNotreversible}:
\begin{eqnarray*}\!\!\!\!\!\!\!\!\!\!\!\!\!\!\!\!\input{CNotPbCNot.tikz}&=&\input{CNotPbCNot1.tikz}\\[0.5cm]
&\eqeqref{tripleCNotswap}&\input{CNotPbCNot2.tikz}\\[0.5cm]
&\eqeqref{CNotCNot}&\input{CNotPbCNot3.tikz}\\[0.5cm]
&\eqeqref{commutationPctrl}&\input{CNotPbCNot4.tikz}\\[0.5cm]
&\eqeqref{CNotCNot}&\input{CNotPbCNotalenvers.tikz}\\[0.5cm]
\end{eqnarray*}

\noindent Proof of \cref{commutationCNotsbas}:
\begin{eqnarray*} \label{proofcommutationCNotsbas}\!\!\!\!\!\!\!\!\!\!\!\!\!\!\!\!\input{CNotgrandCNotbas.tikz}&\eqeqref{CNotCNot}&\input{CNotgrandCNotbas1.tikz}\\[0.5cm]
&\eqeqref{CNotlift}&\input{CNotgrandCNotbas2.tikz}\\[0.5cm]
&\eqeqref{CNotCNot}&\input{CNotbasCNotgrand.tikz}\\[0.5cm]
\end{eqnarray*}

\noindent Proof of \cref{CNotHH}:
\begin{eqnarray*}\!\!\!\!\!\!\!\!\!\!\!\!\!\!\!\!\input{CNotHH.tikz}&\eqeqref{HH}&\input{CNotHH1.tikz}\\[0.5cm]
&\eqeqref{CZ}&\input{CNotHH2.tikz}\\[0.5cm]
&\eqeqref{CNotPCNotreversible}&\input{CNotHH3.tikz}\\[0.5cm]
&\eqeqref{CZ}&\input{CNotHH4.tikz}\\[0.5cm]
&\eqeqref{HH}&\input{HHNotC.tikz}
\end{eqnarray*}
Note that the second use of \cref{CZ} relies on the fact that \gcnot{} is defined as \minitikzfig{NotCdef-s}, and %
uses a few topological rules\bigskip.

\noindent Proof of \cref{NotCRXNotCreversible}:
\begin{longtable}{RCL}
\input{RXhconjCNot.tikz}&\eqeqref{HH}&\input{RXhconjCNot1.tikz}\\\\
&\eqeqref{CNotHH}&\input{RXhconjCNot2.tikz}\\\\
&\eqeqref{CNotPCNotreversible}&\input{RXhconjCNot3.tikz}\\\\
&\eqeqref{CNotHH}&\input{RXhconjCNot4.tikz}\\\\
&\eqeqref{HH}&\input{RXbconjCNot.tikz}.
\end{longtable}

\noindent Proof of \cref{XX}:
\begin{eqnarray*}\!\!\!\!\!\!\!\!\!\!\!\!\!\!\!\!\begin{tikzpicture}
	\begin{pgfonlayer}{nodelayer}
		\node [style=none] (4) at (-2, 0) {};
		\node [style=none] (5) at (1.75, 0) {};
		\node [style=boite2] (12) at (-1, 0) {$X$};
		\node [style=boite2] (13) at (0.75, 0) {$X$};
	\end{pgfonlayer}
	\begin{pgfonlayer}{edgelayer}
		\draw (4.center) to (5.center);
	\end{pgfonlayer}
\end{tikzpicture}
&\eqeqref{P0}&\begin{tikzpicture}
	\begin{pgfonlayer}{nodelayer}
		\node [style=none] (0) at (-2, 0) {};
		\node [style=boite2] (3) at (0.75, 0) {$P(0)$};
		\node [style=boite2] (4) at (-1.25, 0) {$X$};
		\node [style=none] (6) at (3.5, 0) {};
		\node [style=boite2] (8) at (2.75, 0) {$X$};
	\end{pgfonlayer}
	\begin{pgfonlayer}{edgelayer}
		\draw (0.center) to (6.center);
	\end{pgfonlayer}
\end{tikzpicture}
\\[0.5cm]
&\eqeqref{XPX}&\begin{tikzpicture}
	\begin{pgfonlayer}{nodelayer}
		\node [style=none] (0) at (-1.75, 0) {};
		\node [style=boite2] (1) at (0, 0) {$P(0)$};
		\node [style=none] (3) at (1.75, 0) {};
		\node [style=cartouche] (5) at (-1.75, 0.75) {$0$};
	\end{pgfonlayer}
	\begin{pgfonlayer}{edgelayer}
		\draw (0.center) to (3.center);
	\end{pgfonlayer}
\end{tikzpicture}
\\[0.5cm]
&\eqdeuxeqref{S0}{P0}&\begin{tikzpicture}
	\begin{pgfonlayer}{nodelayer}
		\node [style=none] (0) at (-0.8, 0) {};
		\node [style=none] (1) at (0.8, 0) {};
	\end{pgfonlayer}
	\begin{pgfonlayer}{edgelayer}
		\draw [style=noire] (0.center) to (1.center);
	\end{pgfonlayer}
\end{tikzpicture}\\[0.5cm]
\end{eqnarray*}

\noindent Proof of \cref{CNotliftvar}:
\begin{eqnarray*}\!\!\!\!\!\!\!\!\!\!\!\!\!\!\!\!\input{CNotgrandCNotbas.tikz}&\eqeqref{commutationCNotsbas}&\input{CNotgrandconjNotChauts1.tikz}\\[0.5cm]
&\eqeqref{CNotlift}&\input{CNotgrandconjNotChauts2.tikz}\\[0.5cm]
&=&\input{CNotgrandconjNotChauts3.tikz}\\[0.5cm]
&=&\input{CNotgrandconjNotChauts.tikz}\\[0.5cm]
\end{eqnarray*}

\noindent Proof of \cref{commutationXCNot}:
\begin{eqnarray*}\!\!\!\!\!\!\!\!\!\!\!\!\!\!\!\!\input{XbCNot.tikz}&\eqdeuxeqref{HH}{xgate}&\input{XbCNot1.tikz}\\[0.5cm]
&\eqdeuxeqref{CNotHH}{zgate}&\input{XbCNot2.tikz}\\[0.5cm]
&\eqdeuxeqref{commutationPctrl}{CNotHH}&\input{XbCNot3.tikz}\\[0.5cm]
&\eqdeuxeqref{HH}{xgate}&\input{CNotXb.tikz}\\[0.5cm]
\end{eqnarray*}

\noindent Proof of \cref{ZZ}:
\begin{eqnarray*}\!\!\!\!\!\!\!\!\!\!\!\!\!\!\!\!\begin{tikzpicture}
	\begin{pgfonlayer}{nodelayer}
		\node [style=none] (4) at (-2, 0) {};
		\node [style=none] (5) at (1.75, 0) {};
		\node [style=boite2] (12) at (-1, 0) {$Z$};
		\node [style=boite2] (13) at (0.75, 0) {$Z$};
	\end{pgfonlayer}
	\begin{pgfonlayer}{edgelayer}
		\draw (4.center) to (5.center);
	\end{pgfonlayer}
\end{tikzpicture}
&\eqdeuxeqref{xgate}{HH}&\input{ZZ1.tikz}\\[0.5cm]
&\eqeqref{HH}&\input{ZZ2.tikz}\\[0.5cm]
&\eqeqref{XX}&\begin{tikzpicture}
	\begin{pgfonlayer}{nodelayer}
		\node [style=none] (4) at (-2, 0) {};
		\node [style=none] (5) at (1.75, 0) {};
		\node [style=boite2] (12) at (-1, 0) {$H$};
		\node [style=boite2] (13) at (0.75, 0) {$H$};
	\end{pgfonlayer}
	\begin{pgfonlayer}{edgelayer}
		\draw (4.center) to (5.center);
	\end{pgfonlayer}
\end{tikzpicture}
\\[0.5cm]
&\eqeqref{HH}&\begin{tikzpicture}
	\begin{pgfonlayer}{nodelayer}
		\node [style=none] (0) at (-0.8, 0) {};
		\node [style=none] (1) at (0.8, 0) {};
	\end{pgfonlayer}
	\begin{pgfonlayer}{edgelayer}
		\draw [style=noire] (0.center) to (1.center);
	\end{pgfonlayer}
\end{tikzpicture}\\[0.5cm]
\end{eqnarray*}

\noindent Proof of \cref{NotClift}:
\begin{eqnarray*}\!\!\!\!\!\!\!\!\!\!\!\!\!\!\!\!\input{NotCgrandNotCbas.tikz}&=&\input{NotCNotCNotC1.tikz}\\[0.5cm]
&\eqdeuxeqref{commutationCNotsbas}{CNotlift}&\input{NotCNotCNotC2.tikz}\\[0.5cm]
&\eqeqref{CNotCNot}&\input{NotCNotCNotC3.tikz}\\[0.5cm]
&=&\input{NotCNotCNotC.tikz}\\[0.5cm]
\end{eqnarray*}

\noindent Proof of \cref{ZCNot}:
\begin{eqnarray*}\!\!\!\!\!\!\!\!\!\!\!\!\!\!\!\!\input{ZbCNot.tikz}&\eqdeuxeqref{xgate}{HH}&\input{ZbCNot1.tikz}\\[0.5cm]
&\eqeqref{CNotHH}&\input{ZbCNot2.tikz}\\[0.5cm]
&\eqeqref{CNotX}&\input{ZbCNot3.tikz}\\[0.5cm]
&\eqeqref{CNotHH}&\input{ZbCNot4.tikz}\\[0.5cm]
&\eqdeuxeqref{xgate}{HH}&\input{CNotZZ.tikz}\\[0.5cm]
\end{eqnarray*}

\noindent Proof of \cref{commutationRXCNot}:
\begin{eqnarray*}\!\!\!\!\!\!\!\!\!\!\!\!\!\!\!\!\input{XCNot.tikz}&\eqdeuxeqref{HH}{RXgate}&\input{XCNot1.tikz}\\[0.5cm]
&\eqeqref{CNotHH}&\input{XCNot2.tikz}\\[0.5cm]
&\eqeqref{commutationPctrl}&\input{XCNot3.tikz}\\[0.5cm]
&\eqeqref{CNotHH}&\input{XCNot4.tikz}\\[0.5cm]
&\eqdeuxeqref{HH}{RXgate}&\input{CNotX.tikz}\\[0.5cm]
\end{eqnarray*}

\noindent Proof of \cref{RX0}:
\begin{longtable}{RCL}
\begin{tikzpicture}
	\begin{pgfonlayer}{nodelayer}
		\node [style=none] (4) at (-1.5, 0) {};
		\node [style=none] (5) at (1.5, 0) {};
		\node [style=boite2] (11) at (0, 0) {$R_X(0)$};
	\end{pgfonlayer}
	\begin{pgfonlayer}{edgelayer}
		\draw (4.center) to (5.center);
	\end{pgfonlayer}
\end{tikzpicture}
&\eqeqref{RXgate}&\input{HP0H.tikz}\\\\
&\eqdeuxeqref{S0}{P0}&\begin{tikzpicture}
	\begin{pgfonlayer}{nodelayer}
		\node [style=none] (4) at (-2, 0) {};
		\node [style=none] (5) at (1.75, 0) {};
		\node [style=boite2] (12) at (-1, 0) {$H$};
		\node [style=boite2] (13) at (0.75, 0) {$H$};
	\end{pgfonlayer}
	\begin{pgfonlayer}{edgelayer}
		\draw (4.center) to (5.center);
	\end{pgfonlayer}
\end{tikzpicture}
\\\\
&\eqeqref{HH}&\begin{tikzpicture}
	\begin{pgfonlayer}{nodelayer}
		\node [style=none] (0) at (-0.8, 0) {};
		\node [style=none] (1) at (0.8, 0) {};
	\end{pgfonlayer}
	\begin{pgfonlayer}{edgelayer}
		\draw [style=noire] (0.center) to (1.center);
	\end{pgfonlayer}
\end{tikzpicture}\\\\
\end{longtable}

\noindent Proof of \cref{RXRX}:
\begin{eqnarray*}\!\!\!\!\!\!\!\!\!\!\!\!\!\!\!\!\begin{tikzpicture}
	\begin{pgfonlayer}{nodelayer}
		\node [style=none] (6) at (-3.125, 0) {};
		\node [style=none] (7) at (3.275, 0) {};
		\node [style=boite2] (15) at (-1.5, 0) {$R_X(\theta)$};
		\node [style=boite2] (18) at (1.575, 0) {$R_X(\theta')$};
	\end{pgfonlayer}
	\begin{pgfonlayer}{edgelayer}
		\draw [in=-180, out=0] (6.center) to (7.center);
	\end{pgfonlayer}
\end{tikzpicture}
&\eqeqref{RXgate}&\input{RXRX1.tikz}\\[0.5cm]
&\eqeqref{HH}&\input{RXRX2.tikz}\\[0.5cm]
&\eqdeuxeqref{SS}{PP}&\input{RXRX3.tikz}\\[0.5cm]
&\eqeqref{RXgate}&\begin{tikzpicture}
	\begin{pgfonlayer}{nodelayer}
		\node [style=none] (4) at (-2.25, 0) {};
		\node [style=none] (5) at (2.25, 0) {};
		\node [style=boite2] (11) at (0, 0) {$R_X(\theta{+}\theta')$};
	\end{pgfonlayer}
	\begin{pgfonlayer}{edgelayer}
		\draw (4.center) to (5.center);
	\end{pgfonlayer}
\end{tikzpicture}
\\[0.5cm]
\end{eqnarray*}

\noindent Proof of \cref{CNotHCNot}:
\begin{eqnarray*}\!\!\!\!\!\!\!\!\!\!\!\!\!\!\!\! \input{XbHhCNotHhCNotHhCNOTHhCNOT.tikz}&\eqeqref{HH}&\input{XbHhCNotHhCNotHhCNOTHhCNOT1.tikz}\\[0.5cm]
&\eqdeuxeqref{CNotHH}{HH}&\input{XbHhCNotHhCNotHhCNOTHhCNOT2.tikz}\\[0.5cm]
&\eqeqref{CZ}&\input{XbHhCNotHhCNotHhCNOTHhCNOT3.tikz}\\[0.5cm]
&\eqdeuxeqref{CNotPCNotreversible}{CNotCNot}&\input{XbHhCNotHhCNotHhCNOTHhCNOT4.tikz}\\[0.5cm]
&\eqeqref{commutationPctrl}&\input{XbHhCNotHhCNotHhCNOTHhCNOT4bis.tikz}\\[0.5cm]
&\eqdeuxeqref{ZZ}{ZCNot}&\input{XbHhCNotHhCNotHhCNOTHhCNOT5.tikz}\\[0.5cm]
&\eqtroiseqref{zgate}{PP}{ZZ}&\input{XbHhCNotHhCNotHhCNOTHhCNOT6.tikz}\\[0.5cm]
&\eqeqref{commutationPctrl}&\input{XbHhCNotHhCNotHhCNOTHhCNOT6bis.tikz}\\[0.5cm]
&\eqdeuxeqref{CNotPCNotreversible}{commutationPctrl}&\input{XbHhCNotHhCNotHhCNOTHhCNOT6ter.tikz}\\[0.5cm]
&\eqeqref{CZ}&\input{XbHhCNotHhCNotHhCNOTHhCNOT7.tikz}\\[0.5cm]
&\eqdeuxeqref{zgate}{commutationPctrl}&\input{XbHhCNotHhCNotHhCNOTHhCNOT8.tikz}\\[0.5cm]
&\eqdeuxeqref{xgate}{HH}&\input{XbHhCNotHhCNotHhCNOTHhCNOT9.tikz}\\[0.5cm]
&\eqdeuxeqref{XX}{CNotHH}&\input{XbHhCNotHhCNotHhCNOTHhCNOT10.tikz}\\[0.5cm]
&\eqdeuxeqref{HH}{CNotCNot}&\begin{tikzpicture}
	\begin{pgfonlayer}{nodelayer}
		\node [style=none] (0) at (-3, 0.75) {};
		\node [style=none] (1) at (4.25, 0.75) {};
		\node [style=none] (2) at (-3, -0.75) {};
		\node [style=none] (3) at (4.25, -0.75) {};
	\end{pgfonlayer}
	\begin{pgfonlayer}{edgelayer}
		\draw (0.center) to (1.center);
		\draw (2.center) to (3.center);
	\end{pgfonlayer}
\end{tikzpicture}
\\[0.5cm]
\end{eqnarray*}

It follows that

\begin{eqnarray*}\!\!\!\!\!\!\!\!\!\!\!\!\!\!\!\! \input{CNotHhCNot.tikz}&\eqdeuxeqref{CNotCNot}{HH}&\input{XbHhCNotHhCNotHh.tikz}\\[0.5cm]
\end{eqnarray*}

\subsection{Proof of Proposition \ref{ctrlXCNot}}\label{preuvectrlXCNot}
     First, we can notice that 
     \begin{longtable}{RCL}
\Lambda^{1} P(\pi)&\eqdef&\input{Ppidec1.tikz} \\\\ 
          &\overset{\text{def}}{=}& \input{Ppidec2.tikz} \\\\
          &\eqeqref{RXgate}& \input{Ppidec2bis1.tikz} \\\\
          &\eqdeuxeqref{HH}{SS}& \input{Ppidec2bis2.tikz} \\\\
          &\eqeqref{S0}& \input{Ppidec3.tikz} \\\\
          &\eqeqref{CNotHH}& \input{Ppidec4.tikz} \\\\
          &\eqdeuxeqref{HH}{CNotHH}& \input{Ppidec5.tikz} \\\\
          &\eqeqref{HH}& \input{Ppidec6.tikz} \\\\
          &\eqeqref{CZ}& \input{CZ.tikz}
\end{longtable}
It follows that 
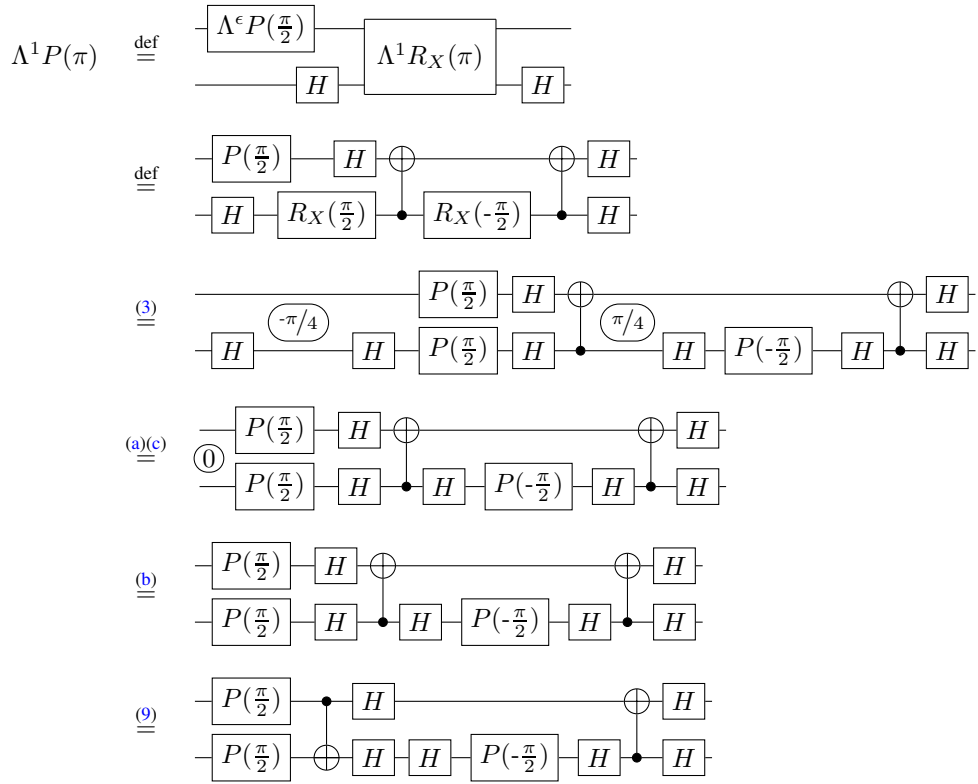
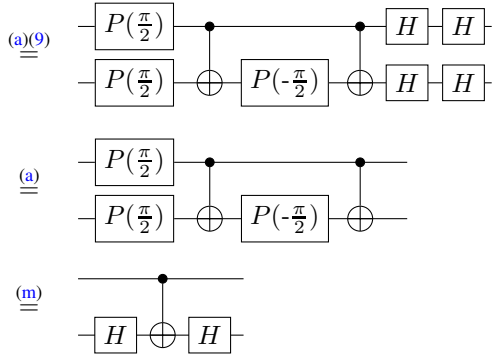
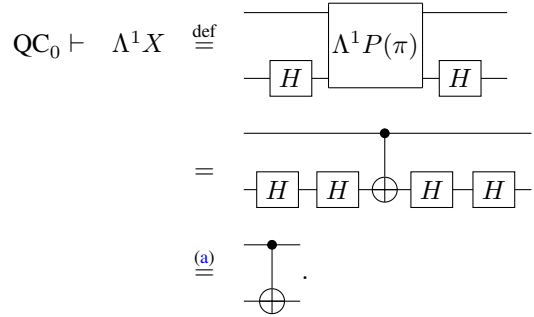
\begin{longtable}{RRCL}
  \qczero\vdash&\Lambda^1 X &\eqdef& \input{Lambda1Xdec1.tikz} \\\\
               &&=& \input{Lambda1Xdec2.tikz} \\\\
               &&\eqeqref{HH}& \input{CNot-m.tikz}.
\end{longtable}

\subsection{Inductive properties for multi-controls: Lemmas \ref{lem:Lambda0} to \ref{lem:Lambda-Rx}}
\label{proofinductiveprop}
The following technical lemmas highlight the inductive properties of
the circuits $\Lambda^xG$. They are at the heart of the proof of the
completeness result.

\begin{lemma}[Base case for the inductive properties]
  \label{lem:Lambda0}
  For all $G\in\{s(\psi),X,R_X(\theta),P(\varphi)\}$, if
  $\epsilon$ is the empty list, $\Lambda^{\epsilon} G = G$.
\end{lemma}
\begin{proof}
  In the case of an empty list, in
  \cref{def:multicontrolled-oriented} there are no gates $X^{\overline{x_i}}$, and
  $\Lambda^\epsilon G = \lambda^0G$. We can then check in
  \cref{def:multicontrolled-pos} that each $\lambda^0G$ is
  $G$: by definition this is true for $R_X(\theta)$, $s(\psi)$ and
  $P(\theta)$. For $X$ we fall back on the definition of $X$ as
  $HP(\pi)H=HZH$.
\end{proof}

\begin{lemma}[Inductive properties for $\Lambda^xG$]
  \label{lem:Lambda0x}
  For all $x\in \{0,1\}^k$,
  and $G\in \{s(\varphi),X,R_X(\theta),P(\varphi)\}$,
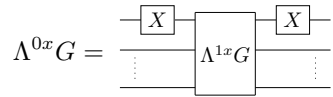
  \[\Lambda^{0x} G=\scalebox{0.8}{\input{neg-control.tikz}}\]
\end{lemma}

\begin{proof}
  This is directly derived from the definition of $\Lambda^xG$: the
  $X^{\overline{x_1}}$'s on the top wire are $X$ for $\Lambda^{0x} G$ and the identity
  for $\Lambda^{1x} G$, while the $X^{\overline{x_i}}$'s on the lower wires are the
  same.
\end{proof}

\begin{lemma}[Inductive properties for $\Lambda^xs(\varphi)$]
  \label{lem:Lambda-s}
  Suppose that $x$ is a $k$-length list of booleans. We then have
  $\Lambda^1 s(\varphi) = P(\varphi)$, $\Lambda^{1x1}s(\varphi) =
  \Lambda^{1x} P(\varphi)$, and
  \[
    \Lambda^{1x0}s(\varphi)= \scalebox{0.8}{\input{mctrlsphi0.tikz}}
  \]
\end{lemma}

\begin{proof}
  By definition, $\Lambda^1 s(\varphi)$  is $\lambda^1 s(\varphi)$:
  there are no $X^{\overline{x_i}}$ since the list only contains a single $1$. By
  definition, $\lambda^1 s(\varphi)$ is $\lambda^0 P(\varphi)$, which
  is $P(\varphi)$.
  
  Suppose now that $x$ is a $k$-length list of booleans, and $b$ is a
  single boolean. Consider $\Lambda^{1xb}s(\varphi)$: by definition
  it is
  \[
    \scalebox{0.8}{\input{generalControledGateS1xb.tikz}}.
  \]
  By definition, $\lambda^{k+2}s(\varphi) =
  \lambda^{k+1}P(\varphi)$. Now, $\Lambda^{1x}P(\varphi)$ is
  \[
    \scalebox{0.8}{\input{generalControledGateP1x.tikz}}.
  \]
  We directly recover $\Lambda^{1x1}s(\varphi)$, i.e. when $b=1$, and
  the case $b=0$ since this just amounts to add the two gates $X^{\overline{0}}=X^1 =
  X$ on the bottom wire.
\end{proof}

\begin{lemma}[Inductive properties of $\Lambda^xX$]
  \label{lem:Lambda-X}
  Suppose that $x$ is a $k$-length list of boolean. Then
  \[\Lambda^{1x}X = \scalebox{0.8}{\input{mctrlPH.tikz}}.\]
\end{lemma}

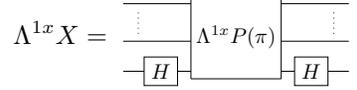
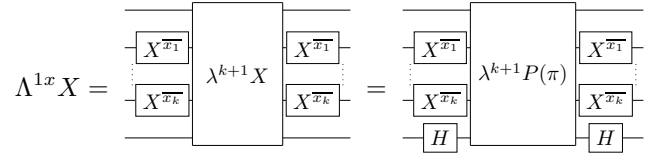
\begin{proof}
  By definition,
  \[
    \Lambda^{1x}X = \scalebox{0.8}{\input{mctrlPHdef.tikz}}
    = \scalebox{0.8}{\input{mctrlPHdef2.tikz}},
  \]
  which is exactly the right-hand-side of the desired equation.
\end{proof}

\begin{lemma}[Inductive properties of $\Lambda^xP(\varphi)$]
  \label{lem:Lambda-P}
  Suppose that $x$ is a $k$-length list of boolean. Then
  \[\qczero\vdash \Lambda^{1x} P(\varphi) = \scalebox{0.8}{\input{mctrlPphi2.tikz}} \]
\end{lemma}
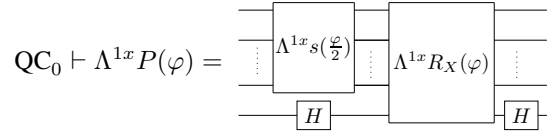
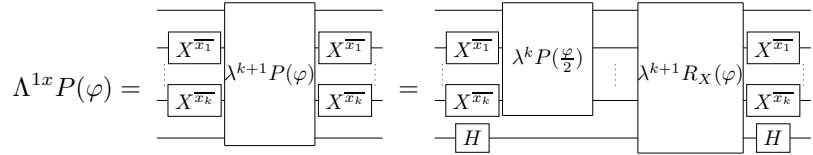
\begin{proof}
  By definition,
  \[ \Lambda^{1x} P(\varphi)
    = \scalebox{0.8}{\input{mctrlP1xdef.tikz}}
    = \scalebox{0.8}{\input{mctrlP1xdef2.tikz}}
  \]
  Since $XX$ is the identity according to \cref{XX}, this is equal to
  \[
    \scalebox{0.8}{\input{mctrlP1xdef3.tikz}}.
  \]
  We can conclude by noting that
  \[
    \Lambda^{1x} s(\frac\varphi2) =
    \scalebox{0.8}{\input{mctrlP1xdef3-2.tikz}}
    \text{ and }
    \Lambda^{1x} R_X(\varphi) =
    \scalebox{0.8}{\input{mctrlP1xdef3-1.tikz}}.
  \]
\end{proof}

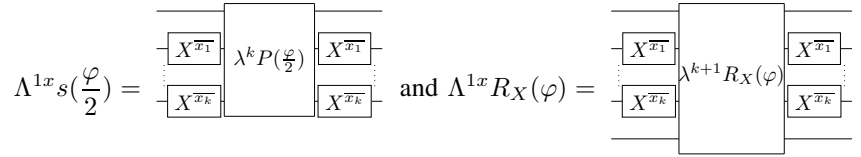
\begin{lemma}[Inductive properties of $\Lambda^xR_X(\varphi)$]
  \label{lem:Lambda-Rx}
  Suppose that $x$ is a $k$-length list of boolean. Then
  \[\qczero\vdash \Lambda^{1x} R_X(\theta)=\scalebox{0.8}{\input{mctrlXtheta.tikz}}.
  \]
\end{lemma}

\begin{proof}
  By definition of $\Lambda^{1x} R_X(\theta)$ and $\lambda^{k+1}
  R_X(\theta)$, we have:
  \[ \Lambda^{1x} R_X(\theta) =
    \scalebox{0.8}{\input{mctrlXtheta2.tikz}}.
  \]
  Using \cref{XX}, we infer that
  \[ \Lambda^{1x} R_X(\theta) =
    \scalebox{0.8}{\input{mctrlXtheta3.tikz}}.
  \]
  We can then conclude by using the definition of $\Lambda^{x}
  R_X(\frac\theta2)$ and $\Lambda^{x} R_X(\text{-}\frac\theta2)$ (and
  the deformation of circuits coming from the prop structure).
\end{proof}

Since these lemmas are essentially consequences of the definitions (except for the use of \cref{XX} in \cref{lem:Lambda-P,lem:Lambda-Rx}), in the following we will mostly keep their uses implicit.

\subsection{Ancillary lemmas: Lemmas \ref{commctrlciblesdistinctes} to \ref{decompctrlblancRX}}
\label{proofmultictrl}
\newcommand{\tripleindice}[3]{\begingroup\arraycolsep=0pt\def\arraystretch{0.5}\begin{scriptarray}{c}#1\\#2\\#3\end{scriptarray}\kern.08333em\endgroup}

For the following lemmas, it is convenient to introduce a graphical notation of multi-controlled gate which allows for more flexibility in the position of the target qubit, relatively to the control qubits:

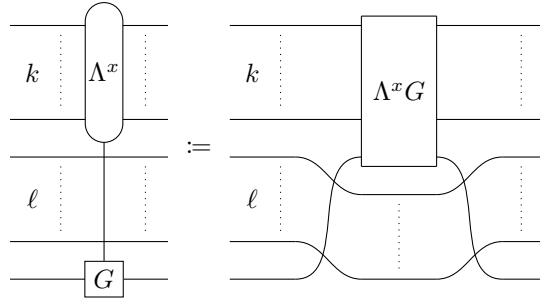
\[ \scalebox{1}{\input{LambdaGxynew.tikz}} \coloneqq \scalebox{1}{\input{LambdaGxynewdef.tikz}}\]

\begin{lemma}\label{commctrlciblesdistinctes}
For any $x\in\{0,1\}^k$,
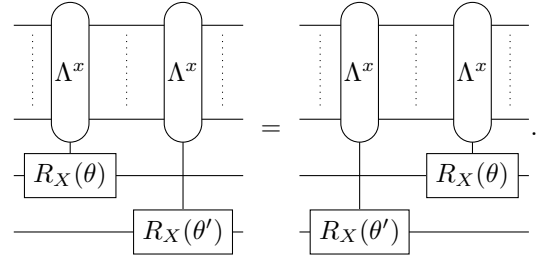
\[\input{LambdaRXhLambdaRXb.tikz}\ =\ \input{LambdaRXbLambdaRXh.tikz}.\]
\end{lemma}
\begin{proof}
We proceed by induction on $k$. If $k=0$, then the equality is a consequence of the topological rules. If $k\geq1$, by \cref{XX} we can assume without loss of generality that $x=1z$ with $z\in\{0,1\}^{k-1}$. One has
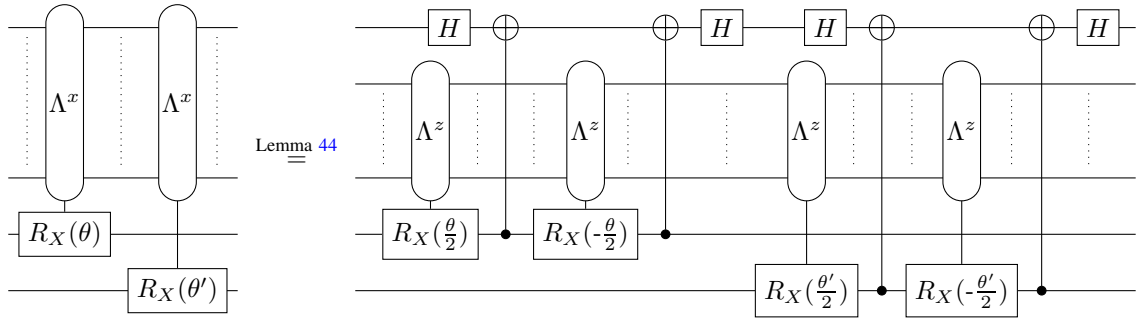
\[\input{LambdaRXhLambdaRXbgrand.tikz}\ \overset{\text{\cref{lem:Lambda-Rx}}}{=}\ \input{LambdaRXhLambdaRXbdecomp.tikz}\]
then it is easy to see that the two parts commute by induction hypothesis and \cref{HH,commutationCNotsbas}, together with topological rules.
\end{proof}

\begin{lemma}\label{CZRXCZreversiblecontrole}
For any $x\in\{0,1\}^k$,
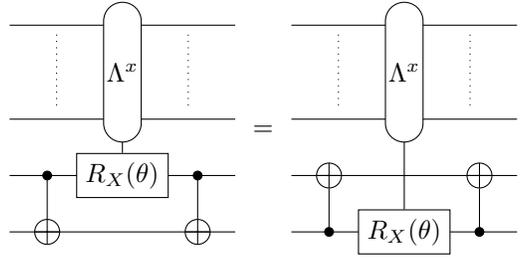
\[\input{LambdaRXhconjCNot.tikz}\ =\ \input{LambdaRXbconjCNot.tikz}.\]
\end{lemma}
\begin{proof}
We proceed by induction on $k$. If $k=0$, then 
the result is just \cref{NotCRXNotCreversible}.
If $k\geq1$, then we can assume without loss of generality that $x=1z$ with $z\in\{0,1\}^{k-1}$. One has
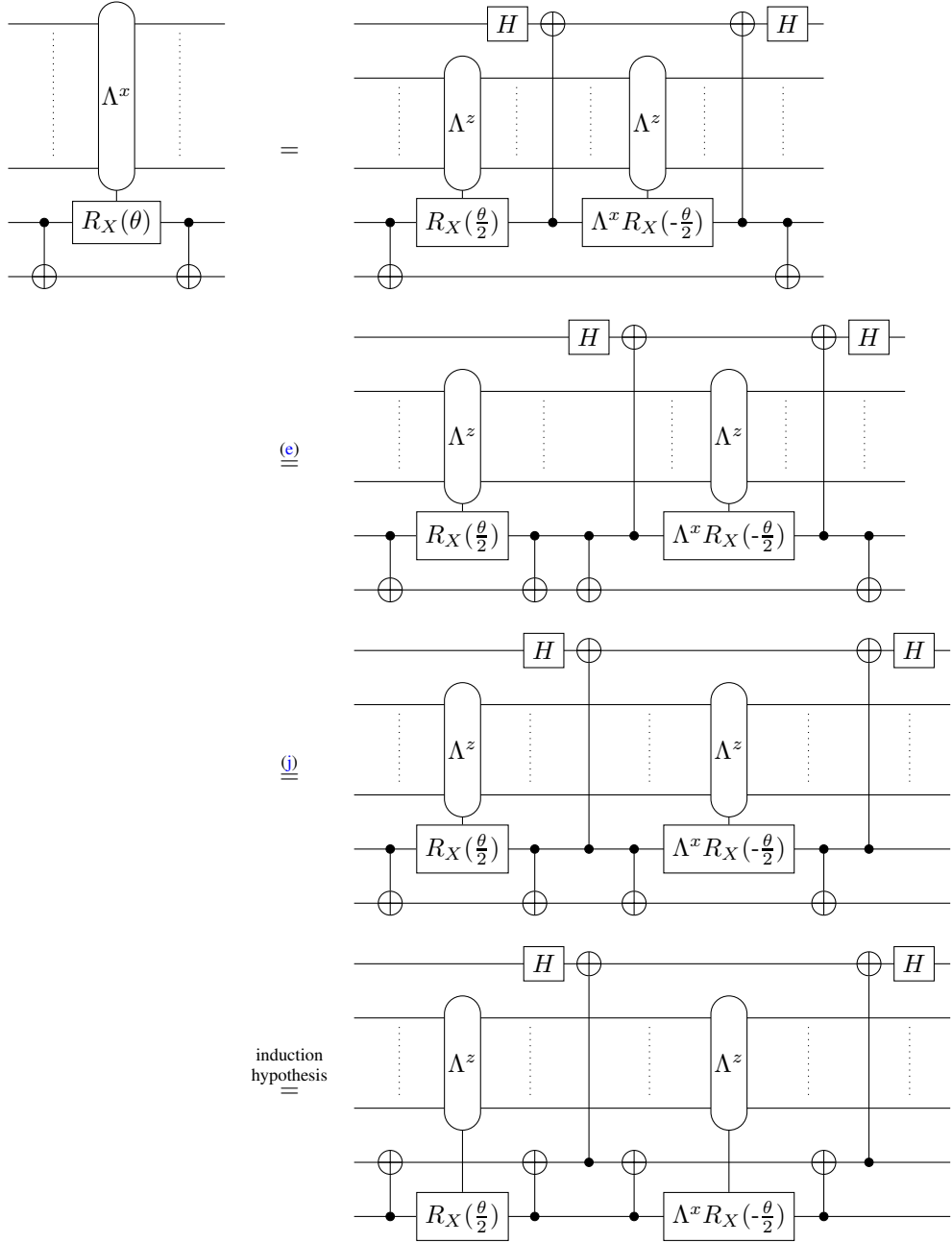
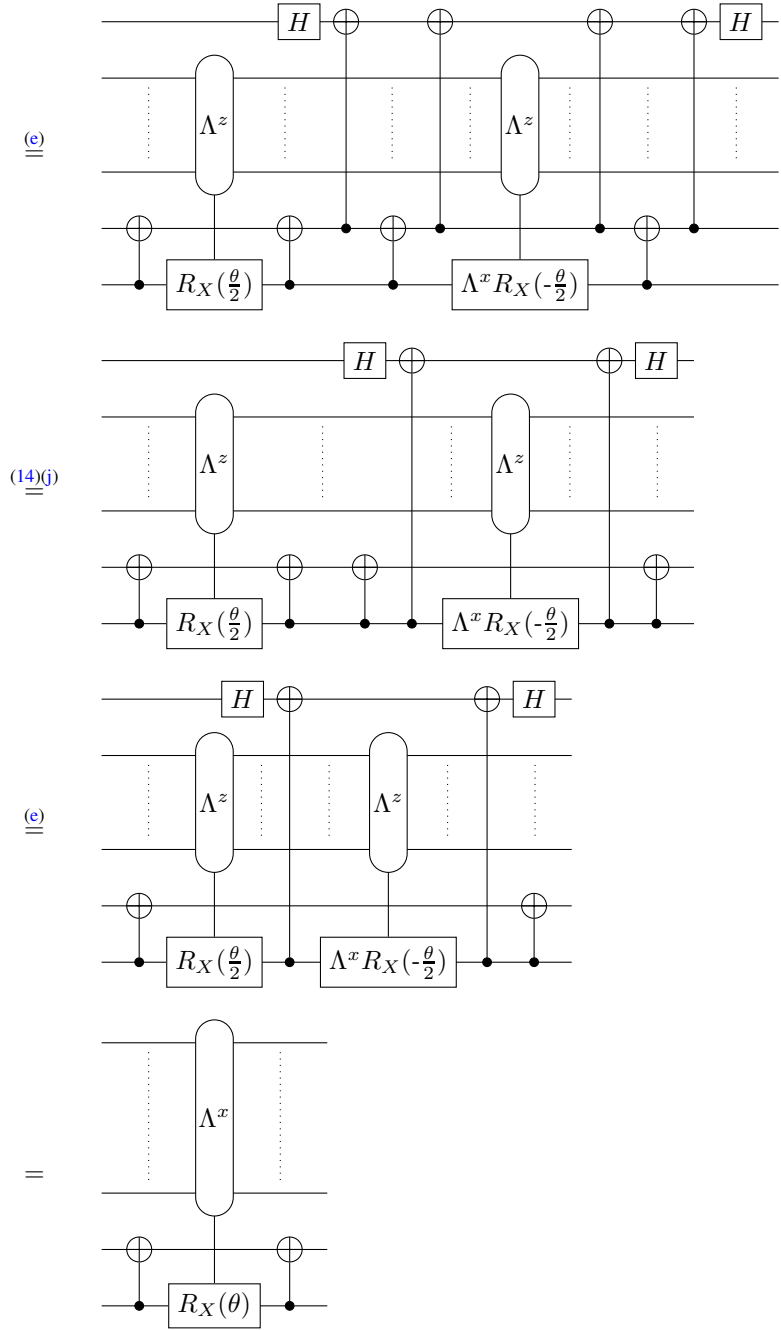
\begin{longtable}{RCL}
\input{LambdaRXhconjCNotgrand.tikz}&=&\input{LambdaRXhconjCNot1.tikz}\\\\
&\eqeqref{CNotCNot}&\input{LambdaRXhconjCNot2.tikz}\\\\
&\eqeqref{commutationCNothaut}&\input{LambdaRXhconjCNot3.tikz}\\\\
&\overset{\text{induction}}{\overset{\text{hypothesis}}{=}}&\input{LambdaRXhconjCNot4.tikz}\\\\
&\eqeqref{CNotCNot}&\input{LambdaRXhconjCNot5.tikz}\\\\
&\eqdeuxeqref{NotClift}{commutationCNothaut}&\input{LambdaRXhconjCNot6.tikz}\\\\
&\eqeqref{CNotCNot}&\input{LambdaRXhconjCNot7.tikz}\\\\
&=&\input{LambdaRXbconjCNotgrand.tikz}
\end{longtable}
\end{proof}

\begin{lemma}\label{decompctrlblancRX}
For any $x\in\{0,1\}^k$,
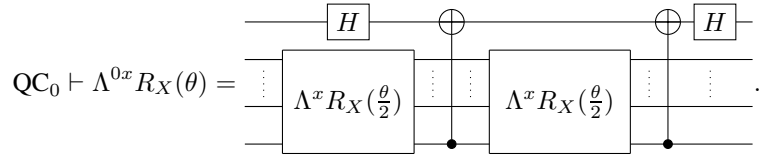
\[\qczero\vdash\Lambda^{0x}R_X(\theta)=\input{mctrlXthetablanc.tikz}.\]
\end{lemma}
\begin{proof}
The proof relies on the following property:
\begin{equation}\label{ZctrlRX}
\qczero\vdash\input{ZmctrlRX.tikz}\ =\ \input{mctrlRXZmoins.tikz}
\end{equation}
that we prove by induction on the length of $x$ as follows:

If $x=\epsilon$, then
\begin{longtable}{RCL}
\begin{tikzpicture}
	\begin{pgfonlayer}{nodelayer}
		\node [style=none] (6) at (2.525, 0) {};
		\node [style=none] (7) at (-2.375, 0) {};
		\node [style=boite2] (15) at (0.9, 0) {$R_X(\theta)$};
		\node [style=boite2] (18) at (-1.425, 0) {$Z$};
	\end{pgfonlayer}
	\begin{pgfonlayer}{edgelayer}
		\draw [in=0, out=180] (6.center) to (7.center);
	\end{pgfonlayer}
\end{tikzpicture}
&\overset{\eqref{RXgate}}{=}&\input{ZHPH.tikz}\\\\
&\eqdeuxeqref{HH}{xgate}&\input{HXPH.tikz}\\\\
&\eqtroiseqref{XX}{XPX}{SS}&\input{HPXH.tikz}\\\\
&\eqtroiseqref{xgate}{RXgate}{HH}&\begin{tikzpicture}
	\begin{pgfonlayer}{nodelayer}
		\node [style=none] (6) at (-2.625, 0) {};
		\node [style=none] (7) at (2.525, 0) {};
		\node [style=boite2] (15) at (-0.875, 0) {$R_X(\text{-}\theta)$};
		\node [style=boite2] (18) at (1.575, 0) {$Z$};
	\end{pgfonlayer}
	\begin{pgfonlayer}{edgelayer}
		\draw [in=180, out=0] (6.center) to (7.center);
	\end{pgfonlayer}
\end{tikzpicture}

\end{longtable}
If $x\neq\epsilon$, then the commutation is a direct consequence of the induction hypothesis and \cref{commutationPctrl}.\bigskip

Given this property, the result can be deduced as follows:
\begin{longtable}{RCL}
\Lambda^{0x}R_X(\theta)&=&\input{mctrlRX1.tikz}\\\\
&\eqdeuxeqref{xgate}{HH}&\input{mctrlRX2.tikz}\\\\
&\eqeqref{ZCNot}&\input{mctrlRX3.tikz}\\\\
&\eqeqref{ZctrlRX}&\input{mctrlRX4.tikz}\\\\
&\eqquatreeqref{ZCNot}{zgate}{commutationPctrl}{ZZ}&\input{mctrlRX5.tikz}\\\\
&\eqtroiseqref{xgate}{HH}{ZZ}&\input{mctrlXthetablanc.tikz}
\end{longtable}
\end{proof}

\subsection{Proof of Equation \eqref{eq:comRX}}
\label{sec:proofcomRX}

We actually prove a slightly more general result: for any $x,x'\in\{0,1\}^k$,
\begin{equation}\label{commctrlRXconjCNots}\qczero\vdash\input{com-L-prime.tikz}\ =\ \input{com-R-prime.tikz}\end{equation}

\Cref{eq:comRX} corresponds to the case where $x=x'$.

\begin{proof}[Proof of \cref{commctrlRXconjCNots}]

 The proof is by induction on $x$. 
 \\\noindent If $x=\epsilon$ (i.e. $k=0$), 
\begin{eqnarray*}\input{RXNotCRXNotC.tikz}&\eqeqref{RXgate}&\input{RXNotCRXNotC-1.tikz}\\[0.5cm]
&\eqdeuxeqref{HH}{CNotHH}&\input{RXNotCRXNotC-2.tikz}\\[0.5cm]
&\eqquatreeqref{CNotPCNotreversible}{RXgate}{HH}{SS}&\input{RXNotCRXNotC-3.tikz}\\[0.5cm]
&\eqeqref{commutationPctrl}&\input{RXNotCRXNotC-4.tikz}\\[0.5cm]
&\eqquatreeqref{CNotPCNotreversible}{HH}{SS}{RXgate}&\input{RXNotCRXNotC-5.tikz}\\[0.5cm]
&\eqtroiseqref{CNotHH}{HH}{RXgate}&\input{NotCRXNotCRX.tikz}\end{eqnarray*}

If $k\geq 1$, then we can write $x=az$ and $x'=a'z'$ with $a,a'\in\{0,1\}$. One has (where the $\pm$ signs correspond respectively to $(-1)^a$ and $(-1)^{a'}$):
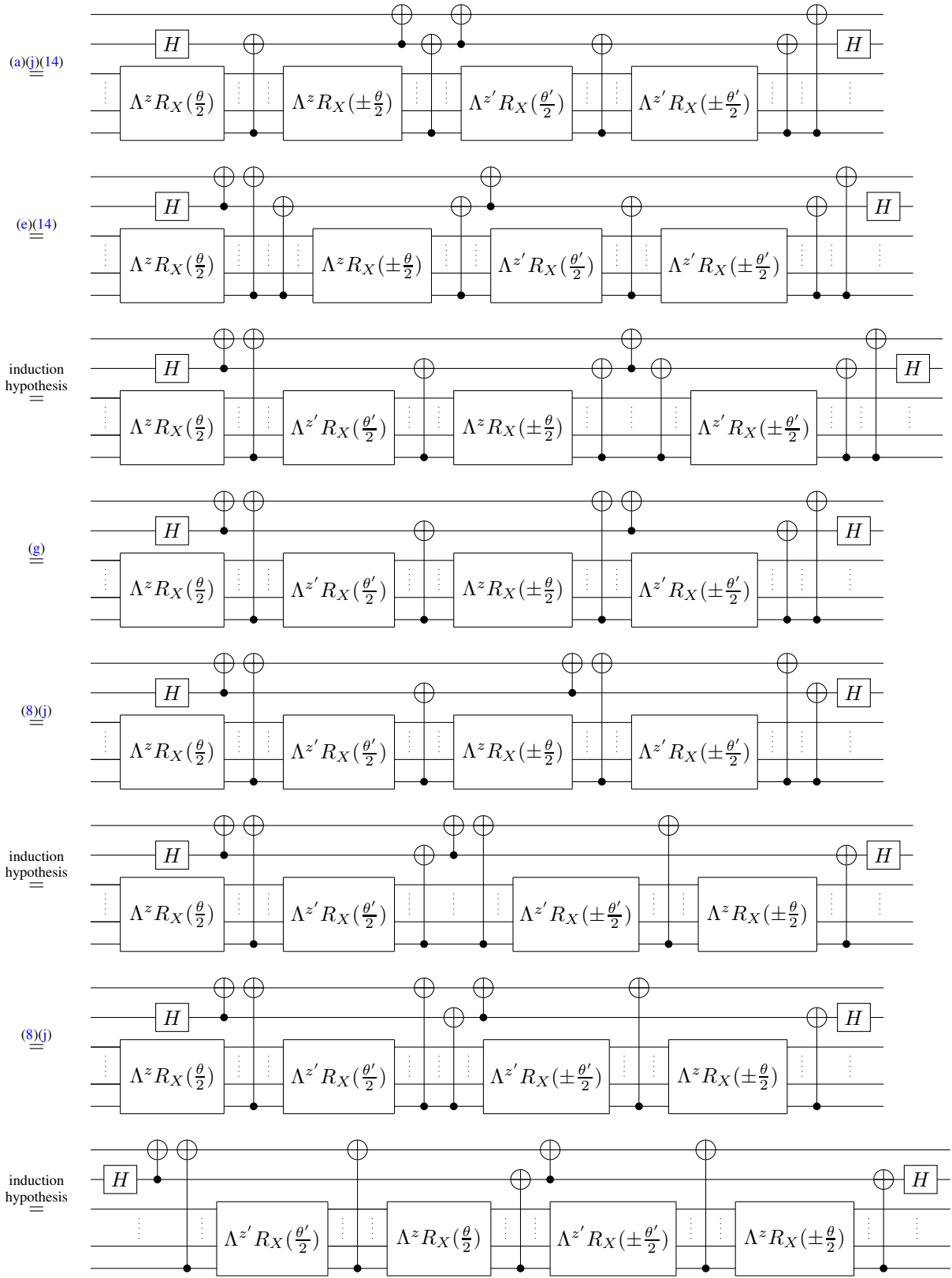
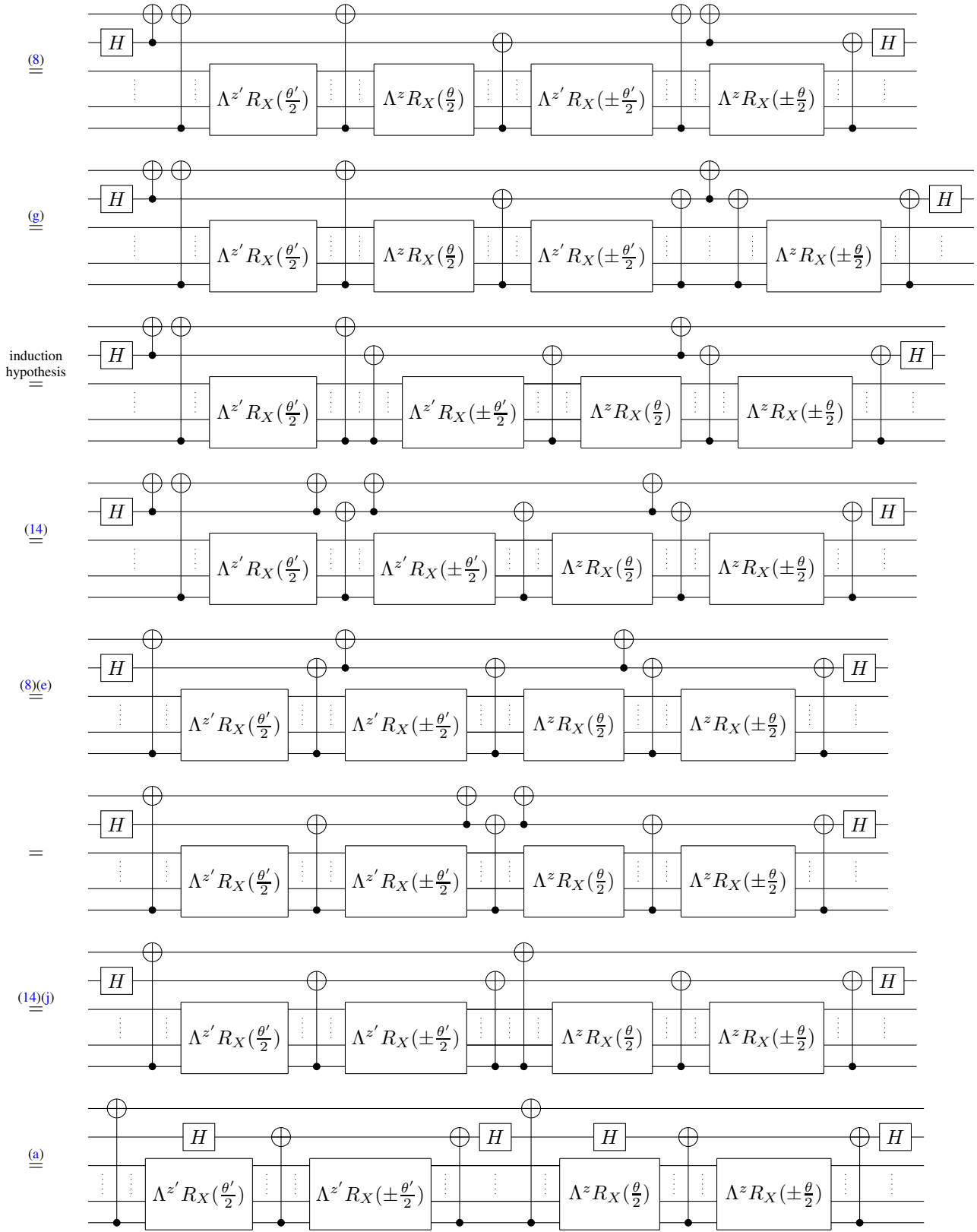
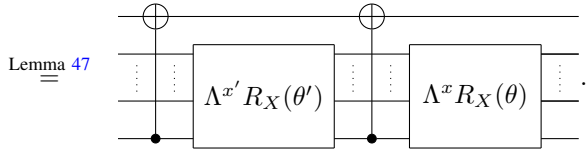
\begin{longtable}{CL}
\input{com-L-prime.tikz}\hspace*{-15em}\\\\
\overset{\text{\cref{decompctrlblancRX}}}{=}&\input{ctrlzRXctrlzprimeRXCNots1.tikz}\\\\
\eqtroiseqref{HH}{commutationCNothaut}{NotClift}&\input{ctrlzRXctrlzprimeRXCNots2.tikz}\\\\
\eqdeuxeqref{CNotCNot}{NotClift}&\input{ctrlzRXctrlzprimeRXCNots3.tikz}\\\\
\overset{\text{induction}}{\overset{\text{hypothesis}}{=}}&\input{ctrlzRXctrlzprimeRXCNots4.tikz}\\\\
\eqeqref{CNotlift}&\input{ctrlzRXctrlzprimeRXCNots5.tikz}\\\\
\eqdeuxeqref{commutationCNotsbas}{commutationCNothaut}&\input{ctrlzRXctrlzprimeRXCNots6.tikz}\\\\
\overset{\text{induction}}{\overset{\text{hypothesis}}{=}}&\input{ctrlzRXctrlzprimeRXCNots7.tikz}\\\\
\eqdeuxeqref{commutationCNotsbas}{commutationCNothaut}&\input{ctrlzRXctrlzprimeRXCNots8.tikz}\\\\
\overset{\text{induction}}{\overset{\text{hypothesis}}{=}}&\input{ctrlzRXctrlzprimeRXCNots9.tikz}\\\\
\eqeqref{commutationCNotsbas}&\input{ctrlzRXctrlzprimeRXCNots10.tikz}\\\\
\eqeqref{CNotlift}&\input{ctrlzRXctrlzprimeRXCNots11.tikz}\\\\
\overset{\text{induction}}{\overset{\text{hypothesis}}{=}}&\input{ctrlzRXctrlzprimeRXCNots12.tikz}\\\\
\eqeqref{NotClift}&\input{ctrlzRXctrlzprimeRXCNots13.tikz}\\\\
\eqdeuxeqref{commutationCNotsbas}{CNotCNot}&\input{ctrlzRXctrlzprimeRXCNots14.tikz}\\\\
=&\input{ctrlzRXctrlzprimeRXCNots15.tikz}\\\\
\eqdeuxeqref{NotClift}{commutationCNothaut}&\input{ctrlzRXctrlzprimeRXCNots16.tikz}\\\\
\eqeqref{HH}&\input{ctrlzRXctrlzprimeRXCNots17.tikz}\\\\
\overset{\text{\cref{decompctrlblancRX}}}{=}&\input{com-R-prime.tikz}.
\end{longtable}
 \end{proof}

\subsection{Proof of Propositions \ref{cor:swap} and \ref{prop:CP}}\label{preuvesswapsmultictrl}

First, we consider the case $G=R_X(\theta)$ of Equations~\eqref{swapCCZ}-\eqref{swapCCZ3}, for which the proof is a direct induction based on \cref{eq:comRX} that is proven in \cref{sec:proofcomRX}.  

Next, we prove \cref{phasemobile} in the case $y=\epsilon$.

We can assume without loss of generality that $x=1^k$. We proceed by induction on $k$.
If $k=0$, then
\begin{longtable}{RCL}
\Lambda^1P(\varphi)&\eqdef&\input{ctrlPdecbrut.tikz}\\\\
&\eqquatreeqref{SS}{S0}{HH}{CNotHH}&\input{ctrlPdecsimple.tikz}\\\\
&\eqeqref{CNotPCNotreversible}&\input{ctrlPdecsimplealenvers.tikz}\\\\
&\eqquatreeqref{SS}{S0}{HH}{CNotHH}&\Lambda^\epsilon_1 P(\varphi).
\end{longtable}

If $k\geq1$, 
then one has
\begin{longtable}{RCL}
\Lambda^{x1}P(\varphi)&\eqdef&\input{LambdaPmove1.tikz}\\\\\\
&\overset{\text{Equations~\eqref{swapCCZ}-\eqref{swapCCZ3}}}{\overset{\text{(case $G=R_X(\theta)$)}}{=}}&\input{LambdaPmove2.tikz}\\\\\\
&\overset{\text{def}}{=}&\input{LambdaPmove3.tikz}\\\\\\
&=&\input{LambdaPmove4.tikz}\\\\\\
&\overset{\text{def}}{=}&\scalebox{0.8}{$\input{LambdaPmove5.tikz}$}\\\\\\
&\overset{\text{\cref{commctrlciblesdistinctes,CZRXCZreversiblecontrole}}}{=}&\scalebox{0.8}{$\input{LambdaPmove6.tikz}$}\\\\\\
&=&\scalebox{0.8}{$\input{LambdaPmove7.tikz}$}\\\\\\
&=&\input{LambdaPmove8.tikz}\\\\\\
&=&\Lambda^x_1 P(\varphi).
\end{longtable}

Now, we can prove Equations~\eqref{swapCCZ}-\eqref{swapCCZ3} in the case $G=s(\psi)$ (the cases $G=P(\varphi)$ and $G=X$ are direct consequences of this case). Without loss of generality we can assume $y=\epsilon$ and consider only \cref{swapCCZ}.

The proof is by induction on the number $r$ of input qubits of $\Lambda^{xabz}G$. If $z=\epsilon$, which is necessarily the case in the base case $r=2$, then the result is a direct consequence of the case $y=\epsilon$ of \cref{phasemobile}. If $z\neq\epsilon$, then using \cref{def:multicontrolled-pos,def:multicontrolled-oriented} (in particular in the case of $\Lambda^{1x}P(\varphi)$), the result is a direct consequence of the induction hypothesis and the case $G=R_X(\theta)$ of Equations~\eqref{swapCCZ}-\eqref{swapCCZ3}.\bigskip

Finally, using the definition of $\Lambda^x_{y1} P(\varphi)$ in terms of $\Lambda^{xy1} P(\varphi)$, the general case of \cref{phasemobile} follows directly from the case $y=\epsilon$ and Equations~\eqref{swapCCZ}-\eqref{swapCCZ3}.

\subsection{Proof of Proposition \ref{prop:sum}}
\label{proofpropsum}

It remains to treat the $\Lambda^xP$ and $\Lambda^xs$ cases of \cref{prop:sum}. 
Those cases are a direct consequence of the following lemma:
\begin{lemma}\label{commctrlphaseenhaut}
For any $x\in\{0,1\}^k$ and $y\in\{0,1\}^\ell$ with $\ell\geq k$,
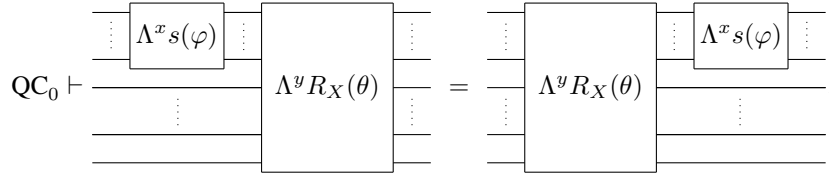
\[\qczero\vdash\input{mctrlxeiphihmctrlyRX.tikz}\ =\ \input{mctrlyRXmctrlxeiphih.tikz}.\]
\end{lemma}

To prove the previous lemma, we do a proof by induction on $k$. 
However, to prove the induction step for $k\geq 2$, we use $\qczero \vdash \Lambda^{1^{k-2}} s(\varphi) \circ \Lambda^{1^{k-2}} s(\varphi')=  \Lambda^{1^{k-2}} s(\varphi+\varphi')$ and 
$\qczero\vdash \Lambda^{1^{k-2}} s(0)= id_{k-1}$, which are the statements of \cref{prop:sum}.

Therefore, we will do a common induction proof for both the other cases of \cref{prop:sum} and for \cref{commctrlphaseenhaut}.
The plan of the proof is the following. First we prove an ancillary equation (\cref{passageCNOtcontroles}) which is derived from previous lemmas. Then we proceed with the induction proof: for $k\geq2$, \cref{commctrlphaseenhaut} is proved with \cref{prop:sum} for $k-2$, while the induction step of \cref{prop:sum} is directly a consequence of \cref{commctrlphaseenhaut} and \cref{prop:sum} for $k-1$, and the $\Lambda^{x}R_X$ case which is already proven.

\begin{proof}
First we prove the following property, which is true for any $a,b\in\{0,1\}$, $z\in\{0,1\}^m$ and $G\in\{s(\varphi),P(\varphi),R_X(\theta),X\}$:
\begin{equation}\label{passageCNOtcontroles}
\qczero\vdash\input{CNothmctrlG.tikz}\ =\ \input{mctrlGCNoth.tikz}\qquad\text{where $c=\begin{cases}b&\text{if $a=0$}\\\bar b&\text{if $a=1$}\end{cases}$}
\end{equation}
To prove \cref{passageCNOtcontroles}, by \cref{CNotX,commutationXCNot,XX} we can assume without loss of generality that $a=b=1$. If $G=R_X(\theta)$, then
\begin{longtable}{CL}
\input{CNothmctrl11zRX.tikz}\hspace*{-7em}\\\\
=&\input{CNotmctrlRX1.tikz}\\\\
\eqdeuxeqref{CNotHH}{HH}&\input{CNotmctrlRX2.tikz}\\\\
\eqdeuxeqref{CNotCNot}{NotClift}&\input{CNotmctrlRX3.tikz}\\\\
\eqdeuxeqref{CNotCNot}{NotClift}&\input{CNotmctrlRX4.tikz}\\\\
\eqtroiseqref{commutationCNotsbas}{commutationCNothaut}{CNotCNot}&\input{CNotmctrlRX5.tikz}\\\\
\eqdeuxeqref{CNotCNot}{NotClift}&\input{CNotmctrlRX6.tikz}\\\\
\eqdeuxeqref{CNotCNot}{NotClift}&\input{CNotmctrlRX7.tikz}\\\\
\eqtroiseqref{commutationCNotsbas}{commutationCNothaut}{CNotCNot}&\input{CNotmctrlRX8.tikz}\\\\
\eqeqref{CNotHH}&\input{CNotmctrlRX9.tikz}\\\\
\eqeqref{CNotCNot}&\input{CNotmctrlRX10.tikz}\\\\
\eqeqref{NotClift}&\input{CNotmctrlRX11.tikz}\\\\
\eqeqref{CNotCNot}&\input{CNotmctrlRX11-1.tikz}\\\\
\eqtroiseqref{NotClift}{commutationCNothaut}{commutationCNotsbas}&\input{CNotmctrlRX12.tikz}\\\\
\eqeqref{eq:comRX}&\input{CNotmctrlRX13.tikz}\\\\
\eqdeuxeqref{commutationCNothaut}{CNotCNot}&\input{CNotmctrlRX14.tikz}\\\\
\eqeqref{eq:comRX}&\input{CNotmctrlRX15.tikz}\\\\
\eqeqref{CNotlift}&\input{CNotmctrlRX16.tikz}\\\\
\eqeqref{NotClift}&\input{CNotmctrlRX17.tikz}\\\\
\eqeqref{eq:comRX}&\input{CNotmctrlRX18.tikz}\\\\
\eqdeuxeqref{commutationCNothaut}{HH}&\input{CNotmctrlRX19.tikz}\\\\
\overset{\text{\cref{decompctrlblancRX},}}{\overset{\text{def}}{=}}&\input{mctrl10zRXCNoth.tikz}
\end{longtable}
Now, to prove \cref{prop:sum} and \cref{commctrlphaseenhaut}, by \cref{XX} we can assume without loss of generality that $x=1^k$. We proceed by induction on $k$. If $k=0$, then \cref{prop:sum} is a consequence of \cref{S0,,SS,,P0,,PP}, and \cref{commctrlphaseenhaut} is a consequence of the topological rules. If $k=1$, then $\Lambda^x s(\varphi)=P(\varphi)$. Let $y=az$ with $a\in\{0,1\}$. By \cref{decompctrlblancRX}, one has
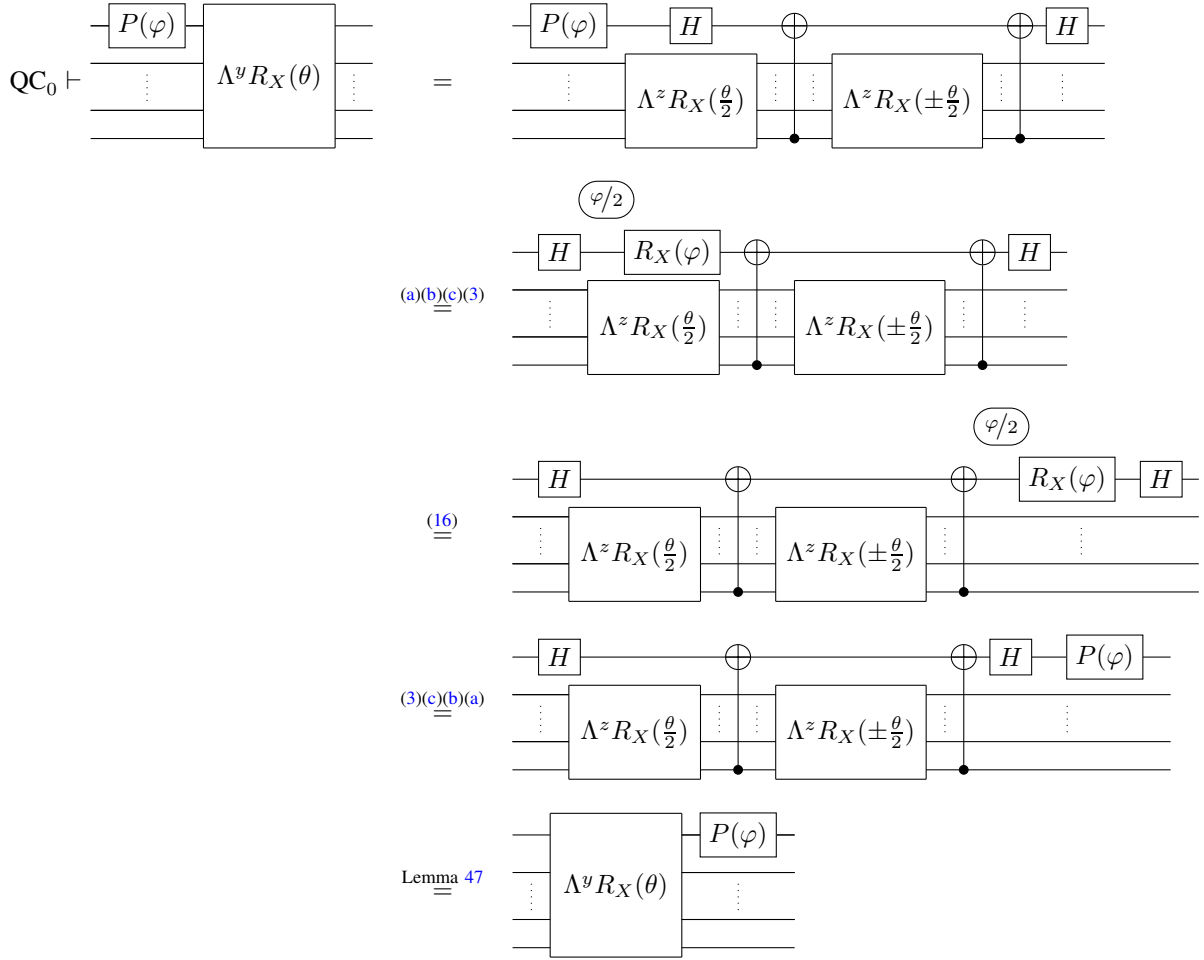
\begin{longtable}{RCL}
\qczero\vdash\input{PphihmctrlyRX.tikz}&=&\input{PphihmctrlyRXdecomp.tikz}\\\\
&\eqquatreeqref{HH}{S0}{SS}{RXgate}&\input{PphihmctrlyRXdecomp1.tikz}\\\\
&\eqeqref{commutationRXCNot}&\input{PphihmctrlyRXdecomp2.tikz}\\\\
&\eqquatreeqref{RXgate}{SS}{S0}{HH}&\input{PphihmctrlyRXdecomp3.tikz}\\\\
&\overset{\text{\cref{decompctrlblancRX}}}=&\input{mctrlyRXPphih.tikz}
\end{longtable}
where the $\pm$ sign is $(-1)^a$.
The case of $k=1$ for \cref{prop:sum} is then a direct consequence of the previous result, the case with $R_X$,  \cref{def:multicontrolled-pos} (case $\lambda^nP(\varphi)$) and \cref{HH,P0,,PP}.

If $k\geq2$, let $z=1^{k-1}$ and $t=1^{k-2}$. To prove \cref{commctrlphaseenhaut}, one has
\begin{longtable}{RCL}
\Lambda^xs(\varphi)&=&\input{mctrlxeiphi1.tikz}\\\\
&\overset{\text{induction hypothesis}}{\overset{\text{of \cref{prop:sum}}}{=}}&\input{mctrlxeiphi2.tikz}\\\\
&\overset{\text{induction hypothesis}}{\overset{\text{of \cref{commctrlphaseenhaut}}}{=}}&\input{mctrlxeiphi3.tikz}\\\\
&\overset{\text{\eqref{HH}, def}}{=}&\input{mctrlxeiphi4.tikz}\\\\
&\eqdeuxeqref{CNotHH}{HH}&\input{mctrlxeiphi5.tikz}\\\\
&\overset{\text{def}}{=}&\input{mctrlxeiphi6.tikz}.
\end{longtable}
\noindent Hence, the commutation with $\Lambda^yR_X(\theta)$ follows by induction hypothesis and \cref{passageCNOtcontroles}, together with \cref{cor:swap}.

Then to prove the $\Lambda^xP$ case of \cref{prop:sum}, one has
\begin{longtable}{RCL}
\Lambda^xP(\varphi')\circ\Lambda^xP(\varphi)&=&\input{propsumcasP1.tikz}\\\\
&\eqeqref{HH}&\input{propsumcasP2.tikz}\\\\
&\overset{\text{induction hypothesis}}{\overset{\text{of \cref{commctrlphaseenhaut}}}{=}}&\input{propsumcasP3.tikz}\\\\
&\overset{\text{$\Lambda^{x}R_X$ case and}}{\overset{\text{induction hypothesis}}{\overset{\text{of \cref{prop:sum}}}{=}}}&\input{propsumcasP4.tikz}\\\\
&=&\lambda^xP(\varphi+\varphi').
\end{longtable}
Finally, the $\Lambda^xs$ case is a direct consequence of the $\Lambda^zP$ case.
\end{proof}

\subsection{Ancillary equations: Equations \eqref{mctrlX} and \eqref{Euler2dmulticontrolled}}
\label{usefulmulctrl}
\begin{lemma}
The following equations can be derived in $\textup{QC}$:
\begin{equation}\label{mctrlX}\begin{array}{rcl}\Lambda^xX&=&\input{mctrlXsimpl.tikz}\end{array}\end{equation}
\begin{equation}\label{Euler2dmulticontrolled}\begin{array}{rcl}\input{Euler2dleft-multicontrolled.tikz}&=&\input{Euler2dright-multicontrolled.tikz}\end{array}\end{equation}
where in \cref{Euler2dmulticontrolled}, the angles are the same as in \cref{Euler2d}.
\end{lemma}

\begin{proof}
If $x=\epsilon$, then \cref{mctrlX} is a direct consequence of \crefnosort{lem:Lambda0,xgate,S0,SS,RXgate}. If $x\neq\epsilon$, then \cref{mctrlX} is a direct consequence of \cref{lem:Lambda0x,lem:Lambda-X,lem:Lambda-P,XX,HH}. 

\noindent Proof of \cref{Euler2dmulticontrolled}:
\begin{longtable}{RCL}
\scalebox{0.8}{$\input{Euler2dleft-multicontrolled.tikz}$}&\equiv&\scalebox{0.8}{$\input{Euler2dleft-multicontrolled1.tikz}$}\\\\
&\overset{\text{\cref{prop:CP,prop:sum}}}{=}&\scalebox{0.8}{$\input{Euler2dleft-multicontrolled2.tikz}$}\\\\
&\eqeqref{Euler3dmulticontrolled}&\scalebox{0.8}{$\input{Euler3dright-multicontrolled-simp-deltas-swaps.tikz}$}
\end{longtable}
By uniqueness of the right-hand side in \cref{Euler2d,Euler3dmulticontrolled}, the $\delta_i$ are such that the last circuit is equal to \vspacebeforeline{0.5em}$\scalebox{0.8}{$\input{Euler3dright-multicontrolled-simp-anglesEuler2d.tikz}$}$\vspace{0.5em}, where the $\beta_j$ are computed in the same way as in \cref{Euler2d}. It follows from \cref{prop:CP,prop:sum} that this is equal modulo $\qczero$ to the right-hand side of \cref{Euler2dmulticontrolled}.
\end{proof}

\section{Proofs of Sections \ref{propertiesmultictrl} and \ref{Eulerandperiod}}
\label{}

\subsection{Proof of Proposition \ref{prop:comb}}\label{preuvepropcomb}
Without loss of generality, we can assume that $y=\epsilon$.

The case where $G=s(\varphi)$ and $x=\epsilon$ follows directly from \crefnosort{XPX,,PP,,P0}. The cases where $G=s(\varphi)$ and $x\neq\epsilon$ follow directly from the case $G=P(\varphi)$, together with \cref{XX}.

By \cref{HH,XX}, the case $G=X$ follows directly from the case $G=P(\pi)$.

The case $G=P(\varphi)$ follows from the case $G=R_X(\theta)$ by a straightforward induction, using \cref{lem:Lambda-P,HH,commctrlphaseenhaut}.

Thus, it suffices to treat the case where $G=R_X(\theta)$. One has
\begin{longtable}{RCL}
\Lambda^{0x}R_X(\theta)\circ\Lambda^{1x}R_X(\theta)&\overset{\text{\cref{lem:Lambda-Rx,decompctrlblancRX}}}{=}&\input{propcomb1.tikz}\\\\
&\eqeqref{HH}&\input{propcomb2.tikz}\\\\
&\eqeqref{eq:comRX}&\input{propcomb3.tikz}\\\\\\
&\overset{\text{\eqref{CNotCNot}, \cref{prop:sum},}}{\overset{\eqref{CNotCNot}\eqref{HH}}{=}}&\input{propcombfinal.tikz}.
\end{longtable}

\subsection{Ancillary lemmas: Lemmas \ref{commctrlPNotCvar} to \ref{commutationctrldotsCNotRX}}
\label{proof:commctrlPNotCvar}

\begin{lemma}\label{commctrlPNotCvar}
For any $x\in\{0,1\}^k$,
\[\input{HmctrlRXHNotC.tikz}\ =\ \input{NotCHmctrlRXH.tikz}\]
\end{lemma}
\begin{proof}
We proceed by induction on $k$. If $k=0$ then the result is a direct consequence of \cref{RXgate,HH,commutationPctrl}. If $k\geq1$, then without loss of generality we can assume that $x=1z$ with $z\in\{0,1\}^{k-1}$. One has
\begin{longtable}{RCL}
\input{HmctrlRXHNotC-grand.tikz}&=&\input{HmctrlRXHNotCdecomp.tikz}\\\\
&\eqeqref{HH}&\input{HmctrlRXHNotCdecomp1.tikz}\\\\
&\eqeqref{CNotHH}&\input{HmctrlRXHNotCdecomp2.tikz}\\\\
&\eqeqref{HH}&\input{HmctrlRXHNotCdecomp3.tikz}\\\\
&\eqdeuxeqref{CNotCNot}{CNotliftvar}&\input{HmctrlRXHNotCdecomp4.tikz}\\\\
&\overset{\text{induction}}{\overset{\text{hypothesis}}{=}}&\input{HmctrlRXHNotCdecomp5.tikz}\\\\
&\eqdeuxeqref{CNotliftvar}{CNotCNot}&\input{HmctrlRXHNotCdecomp7.tikz}\\\\
&\overset{\text{induction}}{\overset{\text{hypothesis}}{=}}&\input{HmctrlRXHNotCdecomp8.tikz}\\\\
&\eqtroiseqref{HH}{CNotHH}{HH}&\input{HmctrlRXHNotCdecomp9.tikz}\\\\
&=&\input{NotCHmctrlRXH-grand.tikz}.
\end{longtable}
\end{proof}

\begin{lemma}\label{commutationctrldotsCNotRX}

For any $x\in\{0,1\}^k$,
\[\input{cLambdaRXCNoth.tikz}\ =\ \input{CNothcLambdaRX.tikz}.\]
\end{lemma}
\begin{proof}~

\begin{longtable}{RCL}
\input{cLambdaRXCNoth.tikz}&=&\input{cLambdaRXCNoth1.tikz}\\\\
&\eqeqref{HH}&\input{cLambdaRXCNoth2.tikz}\\\\
&\eqeqref{CNotHH}&\input{cLambdaRXCNoth3.tikz}\\\\
&\eqeqref{commutationCNotsbas}&\input{cLambdaRXCNoth4.tikz}\\\\
&\eqeqref{CNotHH}&\input{cLambdaRXCNoth5.tikz}\\\\
&\eqeqref{HH}&\input{cLambdaRXCNoth6.tikz}\\\\
&=&\input{CNothcLambdaRX.tikz}
\end{longtable}
\end{proof}

\subsection{Proof of Proposition \ref{commctrl}}\label{preuvecommctrl}

We assume without loss of generality that $y=y'=\epsilon$.\bigskip

First, for the case where $G=R_X(\theta)$ and $G'=R_X(\theta')$, we prove by induction on $k$ 
that for any $x,x'\in\{0,1\}^k$,
\begin{equation}\label{commctrlRX}\qczero\vdash\Lambda^xR_X(\theta)\circ\Lambda^{x'}R_X(\theta')=\Lambda^{x'}R_X(\theta')\circ\Lambda^xR_X(\theta).\end{equation}
The desired result corresponds to \cref{commctrlRX} with $x\neq x'$. Note that when $x=x'$, \cref{commctrlRX} is a consequence of \cref{prop:sum}.

If $k=0$, then 
\cref{commctrlRX} is a direct consequence of \cref{RXRX}. If $k\geq 1$, then we can write $x=az$ and $x'=a'z'$ with $a,a'\in\{0,1\}$. One has (where the $\pm$ signs correspond respectively to $(-1)^a$ and $(-1)^{a'}$):
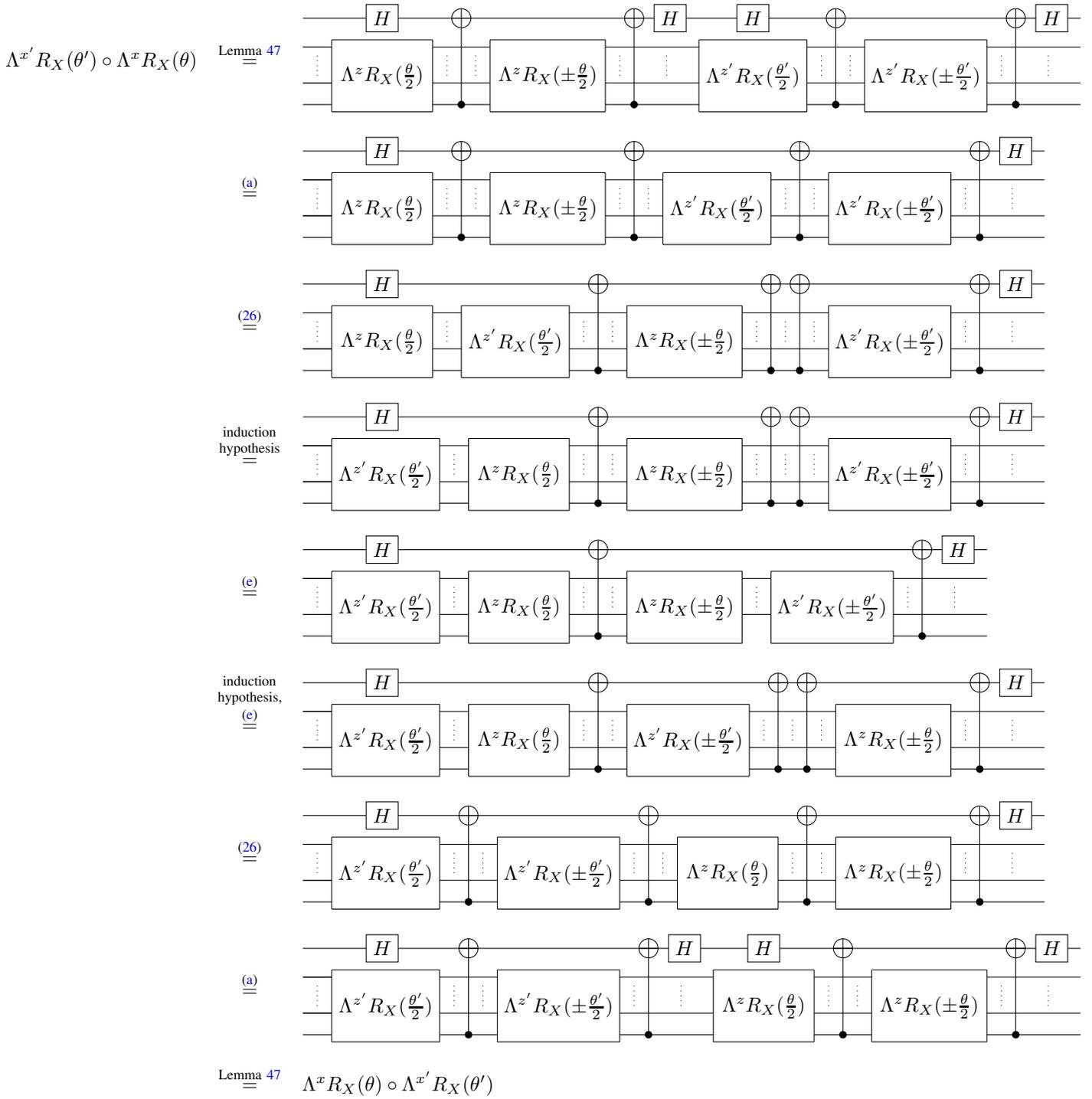
\begin{longtable}{RCL}
\Lambda^{x'}R_X(\theta')\circ\Lambda^xR_X(\theta)&\overset{\text{\cref{decompctrlblancRX}}}{=}&\input{ctrlzRXctrlzprimeRX1.tikz}\\\\
&\eqeqref{HH}&\input{ctrlzRXctrlzprimeRX2.tikz}\\\\
&\eqeqref{commctrlRXconjCNots}&\input{ctrlzRXctrlzprimeRX3.tikz}\\\\
&\overset{\text{induction}}{\overset{\text{hypothesis}}{=}}&\input{ctrlzRXctrlzprimeRX4.tikz}\\\\
&\eqeqref{CNotCNot}&\input{ctrlzRXctrlzprimeRX5.tikz}\\\\
&\overset{\text{induction}}{\overset{\text{hypothesis,}}{\eqeqref{CNotCNot}}}&\input{ctrlzRXctrlzprimeRX6.tikz}\\\\
&\eqeqref{commctrlRXconjCNots}&\input{ctrlzRXctrlzprimeRX7.tikz}\\\\
&\eqeqref{HH}&\input{ctrlzRXctrlzprimeRX8.tikz}\\\\
&\overset{\text{\cref{decompctrlblancRX}}}{=}&\Lambda^xR_X(\theta)\circ\Lambda^{x'}R_X(\theta')
\end{longtable}

If $G=P(\theta)$ and $G'=P(\theta')$, we prove by induction on $k$ that for any $z,z'\in\{0,1\}^k$, %
\begin{equation}\label{commutationLambdaeiphi}
\Lambda^zs(\varphi)\circ\Lambda^{z'}s(\varphi')=\Lambda^{z'}s(\varphi')\circ\Lambda^zs(\varphi).
\end{equation}
The result corresponds to the case where %
$z=x1$ and $z'=x'1$ with $x\neq x'$. Note that the case where $x=x'$ is a consequence of \cref{prop:sum}.

If $k=0$, then %
\cref{commutationLambdaeiphi} is a consequence of the topological rules.

If $k=1$, then 
it is a consequence of \cref{PP,XPX}.

If $k\geq2$, note first that by Equations \eqref{xgate}, \eqref{HH}, \eqref{ZctrlRX}, and \eqref{ZZ} (or \eqref{XPX}, \eqref{HH} and \eqref{RXgate} if $m=0$), for any $x\in\{0,1\}^m$,
\begin{equation}\label{ctrlblancbaseiphi}\qczero\vdash\Lambda^{x0}s(\varphi)\ =\ \input{mctrleiphiblancbas.tikz}.\end{equation}
Let $z=xa$ and $z'=x'a'$ with $a,a'\in\{0,1\}$ and $x,x'\in\{0,1\}^{k-1}$. 
One has (with the $\pm$ signs being $(-1)^{1-a}$ and $(-1)^{1-a'}$ respectively): 
\begin{longtable}{RCL}
\Lambda^{z'}s(\varphi')\circ\Lambda^zs(\varphi)&%
\eqtroiseqref{XX}{HH}{ctrlblancbaseiphi}&\input{mctrlPphicomp-pm.tikz}\\\\
&\overset{\text{\cref{commctrlphaseenhaut}}}{=}&\input{mctrlPphicomp1-pm.tikz}\\\\
&\overset{\text{induction}}{\overset{\text{hypothesis}}{=}}&\input{mctrlPphicomp2-pm.tikz}\\\\
&\eqeqref{commctrlRX}&\input{mctrlPphicomp3-pm.tikz}\\\\
&\overset{\text{\cref{commctrlphaseenhaut}}}{=}&\input{mctrlPphicomp4-pm.tikz}\\\\
&
\eqtroiseqref{XX}{HH}{ctrlblancbaseiphi}&
\Lambda^zs(\varphi)\circ\Lambda^{z'}s(\varphi').
\end{longtable}
For the case where $G=R_X(\theta)$ and $G'=P(\theta')$, we prove by induction on $k\geq1$ that for any $x,x'\in\{0,1\}^k$ with $x\neq x'$,
\begin{equation}\label{commctrlRXHRXH}\qczero\vdash\input{mctrlRXHRXH.tikz}=\input{mctrlHRXHRX.tikz}\end{equation}
Note that by \cref{commctrlphaseenhaut} and the preceding case, \cref{commctrlRXHRXH} is equivalent to the desired result.

If $k=1$, then without loss of generality we can assume that $x=1$ and $x'=0$. One has
\begin{longtable}{RCL}
\input{cRXHcbRXH.tikz}&\overset{\text{\cref{decompctrlblancRX}}}{=}&\input{cRXHcbRXHdecomp.tikz}\\\\
&\eqdeuxeqref{HH}{NotCRXNotCreversible}&\input{cRXHcbRXHdecomp1.tikz}\\\\
&\eqeqref{commutationRXCNot}&\input{cRXHcbRXHdecomp2.tikz}\\\\
&\eqeqref{commutationdecompctrlPRY}&\input{cRXHcbRXHdecomp3.tikz}\\\\
&\eqeqref{commutationRXCNot}&\input{cRXHcbRXHdecomp4.tikz}\\\\
&\eqdeuxeqref{NotCRXNotCreversible}{HH}&\input{cRXHcbRXHdecomp5.tikz}\\\\
&\overset{\text{\cref{decompctrlblancRX}}}{=}&\input{HcbRXHcRX.tikz}
\end{longtable}
If $k\geq2$, then by \cref{cor:swap}, we can assume without loss of generality that we can write $x=az$ and $x'=az'$ with $a,a'\in\{0,1\}$ and $z\neq z'$. One has (where the $\pm$ signs correspond respectively to $(-1)^a$ and $(-1)^{a'}$):
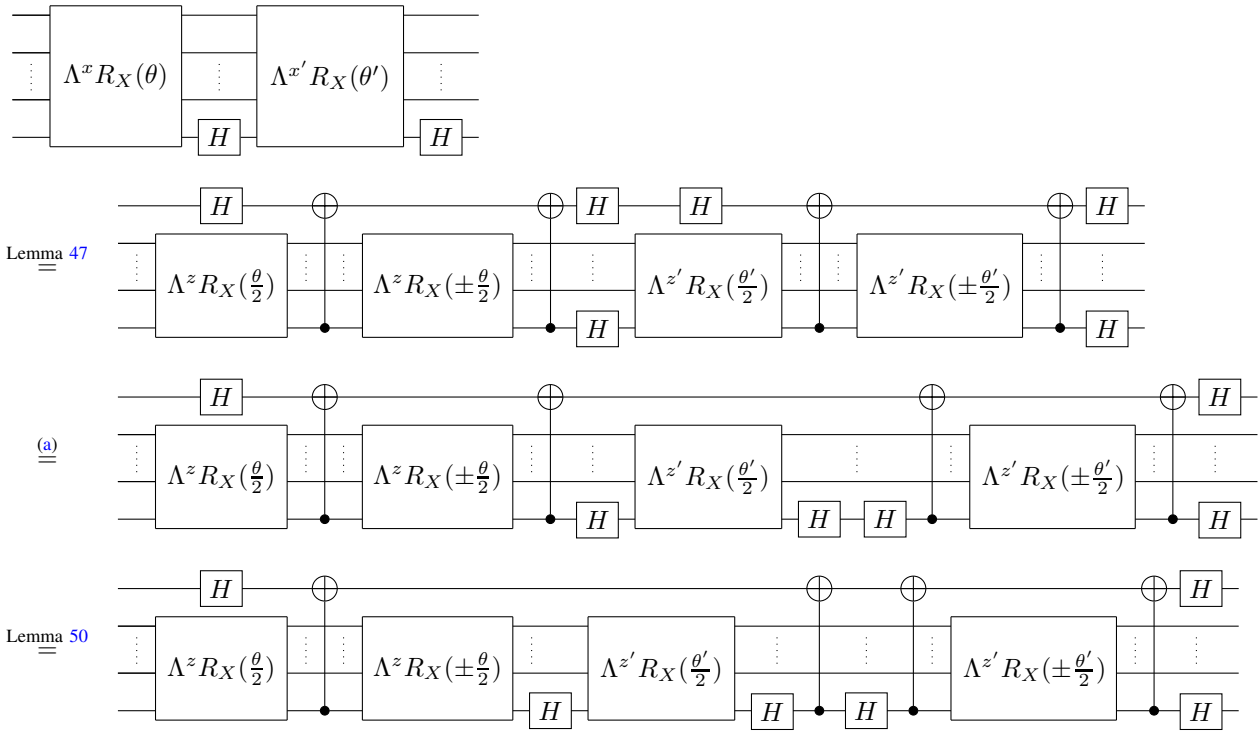
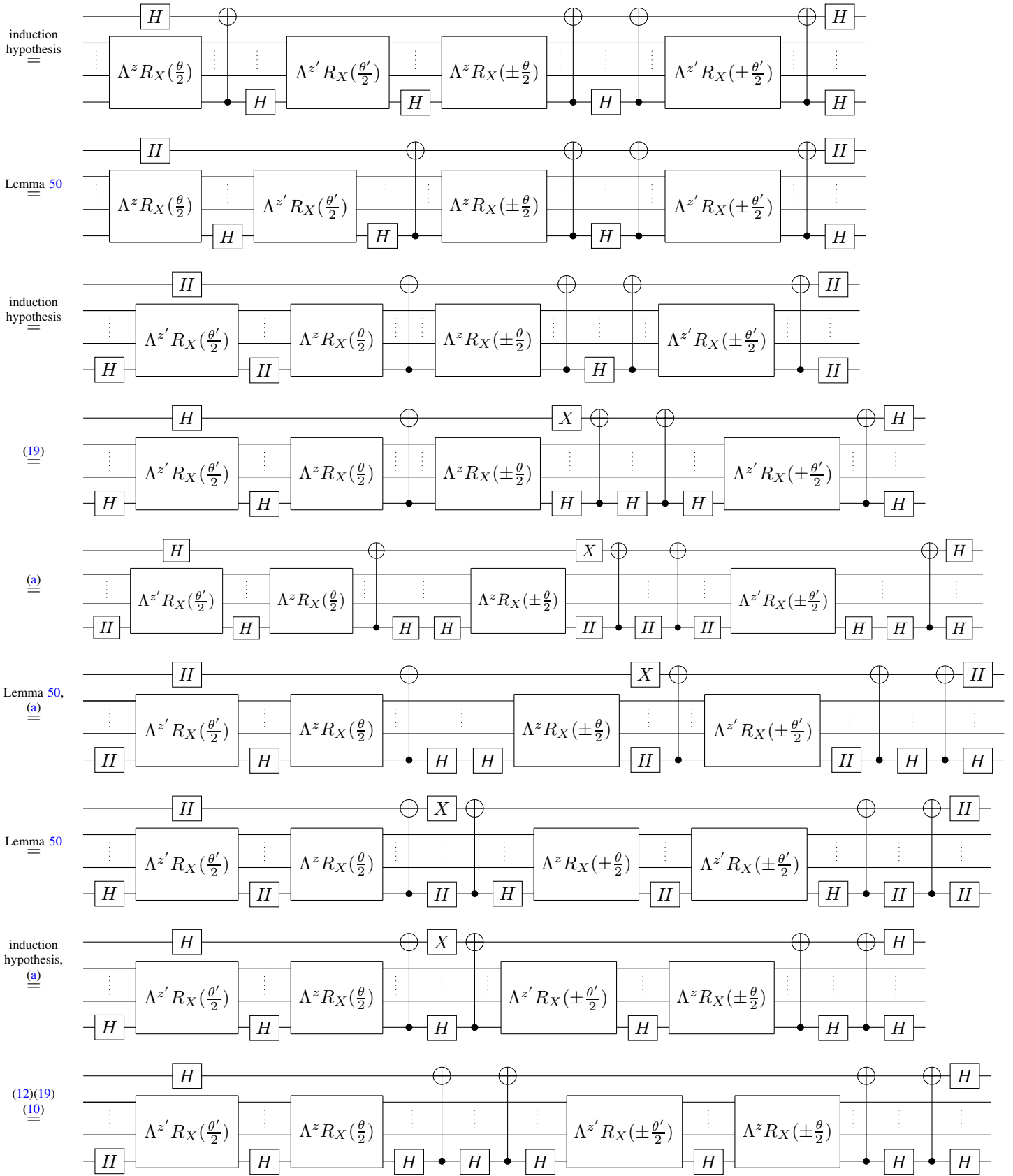
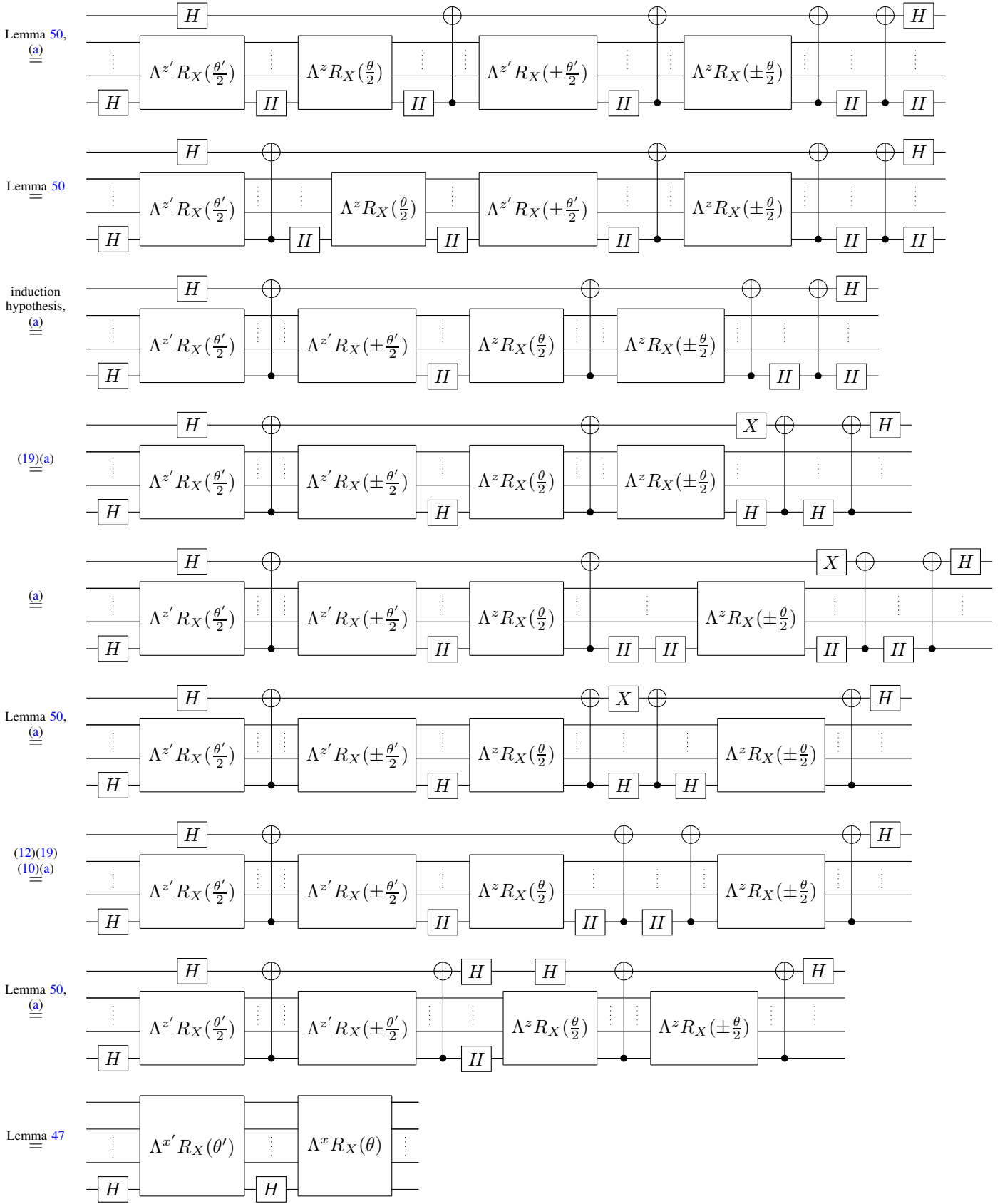
\begin{longtable}{CL}
\input{mctrlRXHRXH-grand.tikz}\hspace*{-15em}\\\\
\overset{\text{\cref{decompctrlblancRX}}}{=}&\input{mctrlRXHRXHdecomp.tikz}\\\\
\eqeqref{HH}&\input{mctrlRXHRXHdecomp1.tikz}\\\\
\overset{\text{\cref{commctrlPNotCvar}}}{=}&\input{mctrlRXHRXHdecomp2.tikz}\\\\
\overset{\text{induction}}{\overset{\text{hypothesis}}{=}}&\input{mctrlRXHRXHdecomp3.tikz}\\\\
\overset{\text{\cref{commctrlPNotCvar}}}{=}&\input{mctrlRXHRXHdecomp4.tikz}\\\\
\overset{\text{induction}}{\overset{\text{hypothesis}}{=}}&\input{mctrlRXHRXHdecomp5.tikz}\\\\
\eqeqref{CNotHCNot}&\input{mctrlRXHRXHdecomp6.tikz}\\\\
\eqeqref{HH}&\scalebox{0.9}{$\input{mctrlRXHRXHdecomp7.tikz}$}\\\\
\overset{\text{\cref{commctrlPNotCvar},}}{\eqeqref{HH}}&\input{mctrlRXHRXHdecomp8.tikz}\\\\
\overset{\text{\cref{commctrlPNotCvar}}}{=}&\input{mctrlRXHRXHdecomp9.tikz}\\\\
\overset{\text{induction}}{\overset{\text{hypothesis,}}{\eqeqref{HH}}}&\input{mctrlRXHRXHdecomp10.tikz}\\\\
\overset{\eqref{commutationXCNot}\eqref{CNotHCNot}}{\eqeqref{XX}}&\input{mctrlRXHRXHdecomp11.tikz}\\\\
\overset{\text{\cref{commctrlPNotCvar},}}{\eqeqref{HH}}&\input{mctrlRXHRXHdecomp12.tikz}\\\\
\overset{\text{\cref{commctrlPNotCvar}}}{=}&\input{mctrlRXHRXHdecomp13.tikz}\\\\
\overset{\text{induction}}{\overset{\text{hypothesis,}}{\eqeqref{HH}}}&\input{mctrlRXHRXHdecomp14.tikz}\\\\
\eqdeuxeqref{CNotHCNot}{HH}&\input{mctrlRXHRXHdecomp15.tikz}\\\\
\eqeqref{HH}&\input{mctrlRXHRXHdecomp16.tikz}\\\\
\overset{\text{\cref{commctrlPNotCvar},}}{\eqeqref{HH}}&\input{mctrlRXHRXHdecomp17.tikz}\\\\
\overset{\eqref{commutationXCNot}\eqref{CNotHCNot}}{\eqdeuxeqref{XX}{HH}}&\input{mctrlRXHRXHdecomp18.tikz}\\\\
\overset{\text{\cref{commctrlPNotCvar},}}{\eqeqref{HH}}&\input{mctrlRXHRXHdecomp19.tikz}\\\\
\overset{\text{\cref{decompctrlblancRX}}}{=}&\input{mctrlHRXHRX-grand.tikz}
\end{longtable}

If $G=X$ or $G'=X$, then by \cref{mctrlX}, the result follows from the preceding cases together with \cref{commctrlphaseenhaut} and \cref{commutationLambdaeiphi}.

\subsection{Ancillary lemmas: Lemmas \ref{symmetriesemicontrolee} to \ref{commctrlphaseenhautP}}\label{commutationsXtarget}

\begin{lemma}\label{symmetriesemicontrolee}
For any $x\in\{0,1\}^k$ and $y\in\{0,1\}^\ell$,
\[\qczero\vdash(id_k\otimes X\otimes id_\ell)\circ\Lambda^x_yX=\Lambda^x_yX\circ(id_k\otimes X\otimes id_\ell)\]
and
\[\qczero\vdash(id_k\otimes X\otimes id_\ell)\circ\Lambda^x_yR_X(\theta)=\Lambda^x_yR_X(\theta)\circ(id_k\otimes X\otimes id_\ell)\]
\end{lemma}
\begin{proof}
The case of $\Lambda^x_yX$ is a direct consequence of \cref{prop:comb,commctrl}. Indeed, using \cref{prop:comb}, $(id_k\otimes X\otimes id_\ell)$ can be decomposed into a product of multi-controlled gates of the form $\Lambda^{x'}_{y'}X$ with $x'\in\{0,1\}^k$ and $y'\in\{0,1\}^\ell$. Then these multi-controlled gates commute with $\Lambda^x_yX$, trivially in the case where $x'y'=xy$, and by \cref{commctrl} in the other cases. %
For the case  of $\Lambda^x_yR_X(\theta)$, note that $\eqref{S0},\eqref{SS}\vdash\begin{tikzpicture}
	\begin{pgfonlayer}{nodelayer}
		\node [style=none] (4) at (-1, 0) {};
		\node [style=none] (5) at (1, 0) {};
		\node [style=boite2] (11) at (0, 0) {$X$}; 
	\end{pgfonlayer}
	\begin{pgfonlayer}{edgelayer}
		\draw (4.center) to (5.center);
	\end{pgfonlayer}
\end{tikzpicture}
=\begin{tikzpicture}
	\begin{pgfonlayer}{nodelayer}
		\node [style=none] (4) at (-1.75, 0) {};
		\node [style=none] (5) at (1.75, 0) {};
		\node [style=boite2] (11) at (0, 0) {$R_X(\pi)$};
		\node [style=cartouche] (12) at (-2.25, 0.75) {$\nicefrac{\pi}{2}$};
	\end{pgfonlayer}
	\begin{pgfonlayer}{edgelayer}
		\draw (4.center) to (5.center);
	\end{pgfonlayer}
\end{tikzpicture}
$. Then $s(\frac\pi2)$ commutes by the topological rules, while the commutation of $(id_k\otimes R_X(\pi)\otimes id_\ell)$ is a direct consequence of \crefnosort{prop:comb,commctrl,prop:sum}: using \cref{prop:comb}, it can be decomposed into a product of multi-controlled gates of the form $\Lambda^{x'}_{y'}R_X(\pi)$ with $x'\in\{0,1\}^k$ and $y'\in\{0,1\}^\ell$. Then these multi-controlled gates commute with $\Lambda^x_yR_X(\theta)$, by \cref{commctrl} in the cases where $x'y'\neq xy$, and by \cref{prop:sum} in the case where $x'y'=xy$.
\end{proof}

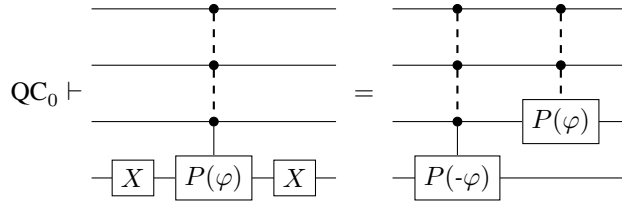
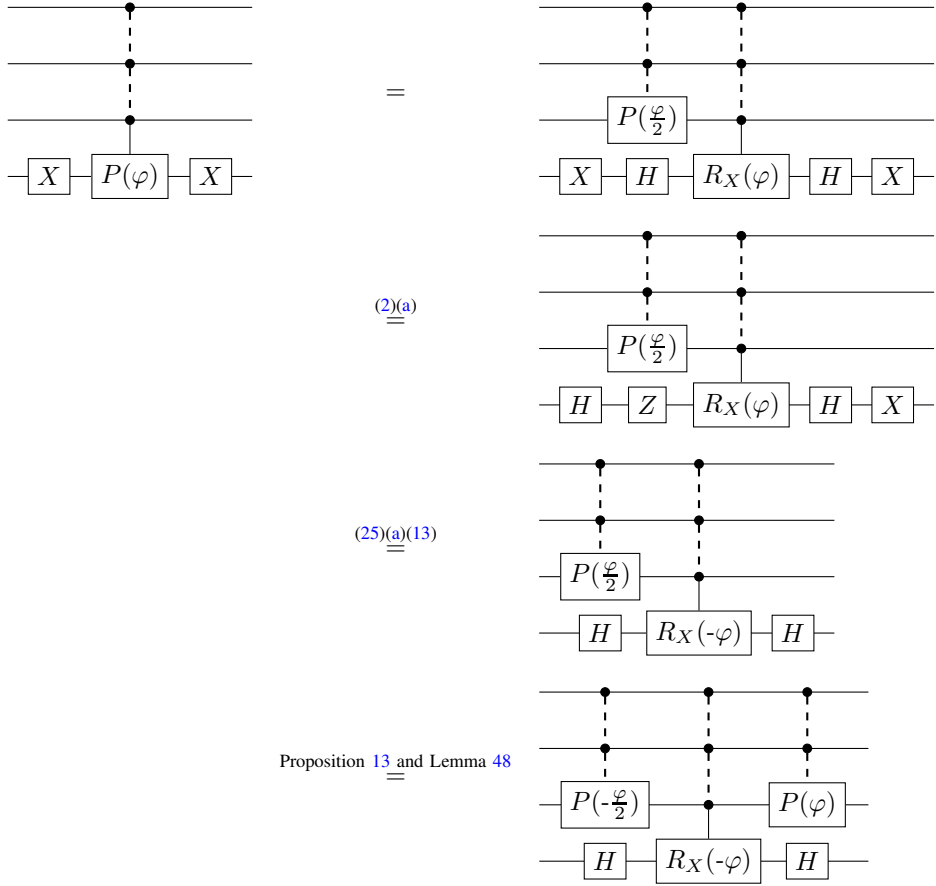
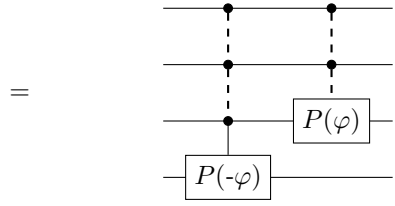
\begin{lemma}\label{antisymmetriesemicontrolee}
\[\qczero\vdash\input{XmctrlPX.tikz}\ =\ \input{mctrlPmoinsphictrlPphihaut.tikz}\]
\end{lemma}
\begin{proof}~
\begin{longtable}{RCL}
\input{XmctrlPX.tikz}&=&\input{XmctrlPX1.tikz}\\\\
&\eqdeuxeqref{xgate}{HH}&\input{XmctrlPX2.tikz}\\\\
&\eqtroiseqref{ZctrlRX}{HH}{ZZ}&\input{XmctrlPX3.tikz}\\\\
&\overset{\text{\crefnosort{prop:sum,commctrlphaseenhaut}}}{=}&\input{XmctrlPX4.tikz}\\\\
&=&\input{mctrlPmoinsphictrlPphihaut.tikz}
\end{longtable}
\end{proof}

\begin{lemma}\label{commctrlphaseenhautP}
  For any $x\in\{0,1\}^k$ and $y\in\{0,1\}^\ell$ with $\ell\geq k$,
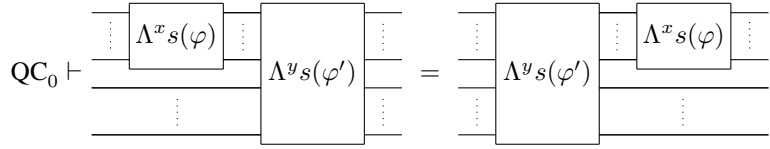
  \[\qczero\vdash\input{mctrlxeiphihmctrlyeiphiprime.tikz}\ =\ \input{mctrlyeiphiprimemctrlxeiphih.tikz}.\]
  \end{lemma}
  \begin{proof}
  We proceed by induction on $\ell-k$. If $\ell=k$ then the result is a consequence of \cref{prop:sum} or \ref{commctrl} (or just of the topological rules if $k=\ell=0$). If $\ell\geq k+1$, then %
  without loss of generality, we can assume that $y=t1$ for some $t\in\{0,1\}^{\ell-1}$. Then by \cref{lem:Lambda-P} (together with \cref{lem:Lambda0x} and \cref{XX}),
  \begin{multline*}
    \qczero\vdash\input{mctrlxeiphihmctrlyeiphiprime.tikz}\\=\input{mctrlxeiphihmctrlyeiphiprime1.tikz}
  \end{multline*}
  so that the commutation follows by induction hypothesis and \cref{commctrlphaseenhaut}.
  \end{proof}

\subsection{Proof of Proposition \ref{commctrlpasenface}}\label{preuvecommctrlpasenface}

First, the cases where $G$ or $G'=X$ follow from the other cases. Indeed, using \cref{mctrlX} and \cref{prop:comb} (together with \cref{cor:swap}), and then \cref{prop:CP}, one gets that for any $t\in\{0,1\}^p$,
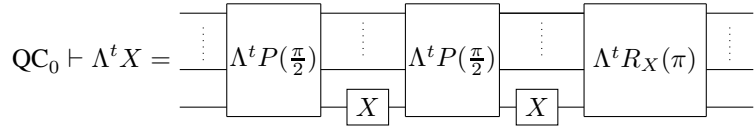
\[\qczero\vdash\Lambda^tX=\input{mctrlXsimplmemetarget.tikz}.\]
Then, if $G$ or $G'=X$, one can use this decomposition and make the multi-controlled parts commute using the other cases. The non-controlled $X$ gates commute with the control dots by changing their colour, with the help of \cref{XX}.  This does not alter the fact that the multi-controlled gates commute, since the $X$ gates are not on the same wire than the control dots of different colours. And since the decomposition produces each time two $X$ gates on the same wire, any control dot gets changed twice, so that it is the same at the end as at the beginning.

Thus, it suffices to treat the cases where $G,G'\in \{R_X(\theta),P(\varphi)\}$.

If $G=R_X(\theta)$ and $G'=P(\varphi)$ (or conversely), then by \cref{prop:CP}, the result is a consequence of \cref{symmetriesemicontrolee,commctrl}.

If $G=P(\varphi)$ and $G'=P(\varphi')$, then by \cref{prop:CP}, the result is a consequence of \cref{antisymmetriesemicontrolee,commctrlphaseenhautP} (together with \cref{XX}) and \cref{commctrl}.

It remains to treat the case where $G=R_X(\theta)$ and $G'=R_X(\theta')$. By \cref{symmetriesemicontrolee}, we can assume without loss of generality that $a=b=1$. By definition of $\Lambda^t_u$, we can also assume without loss of generality that $k=m=0$. Then the hypothesis $xyz\neq x'y'z'$ becomes $y\neq y'$. We proceed by induction on $\ell$. If $\ell=1$, then without loss of generality we can assume that $x=1$ and $x'=0$. One has
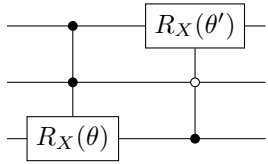
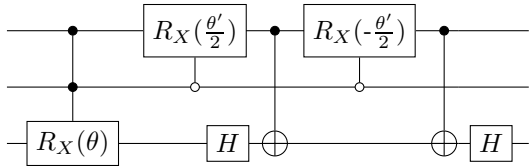
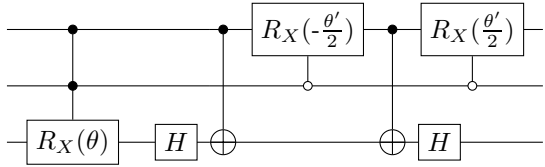
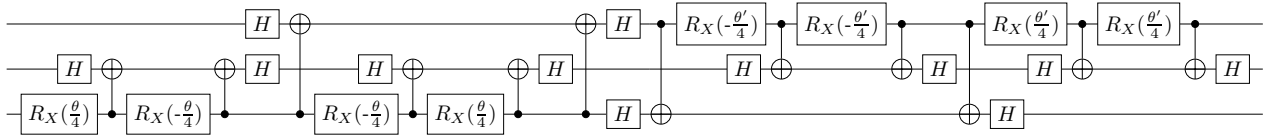
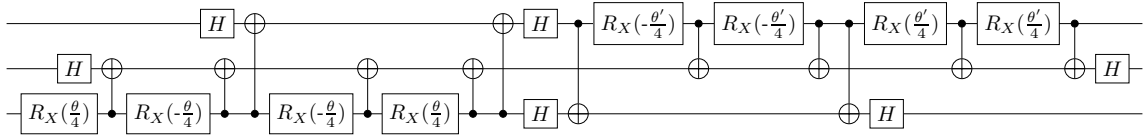
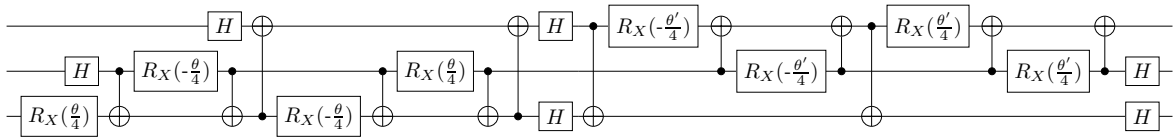
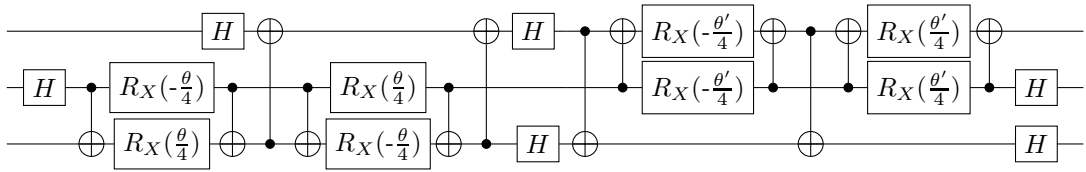
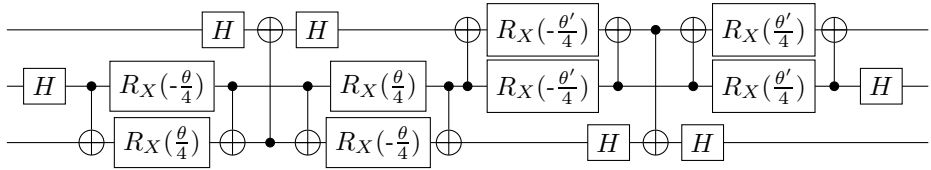
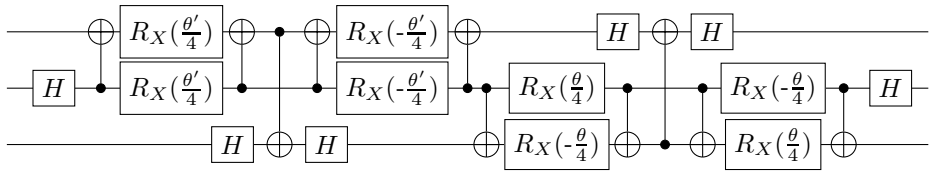
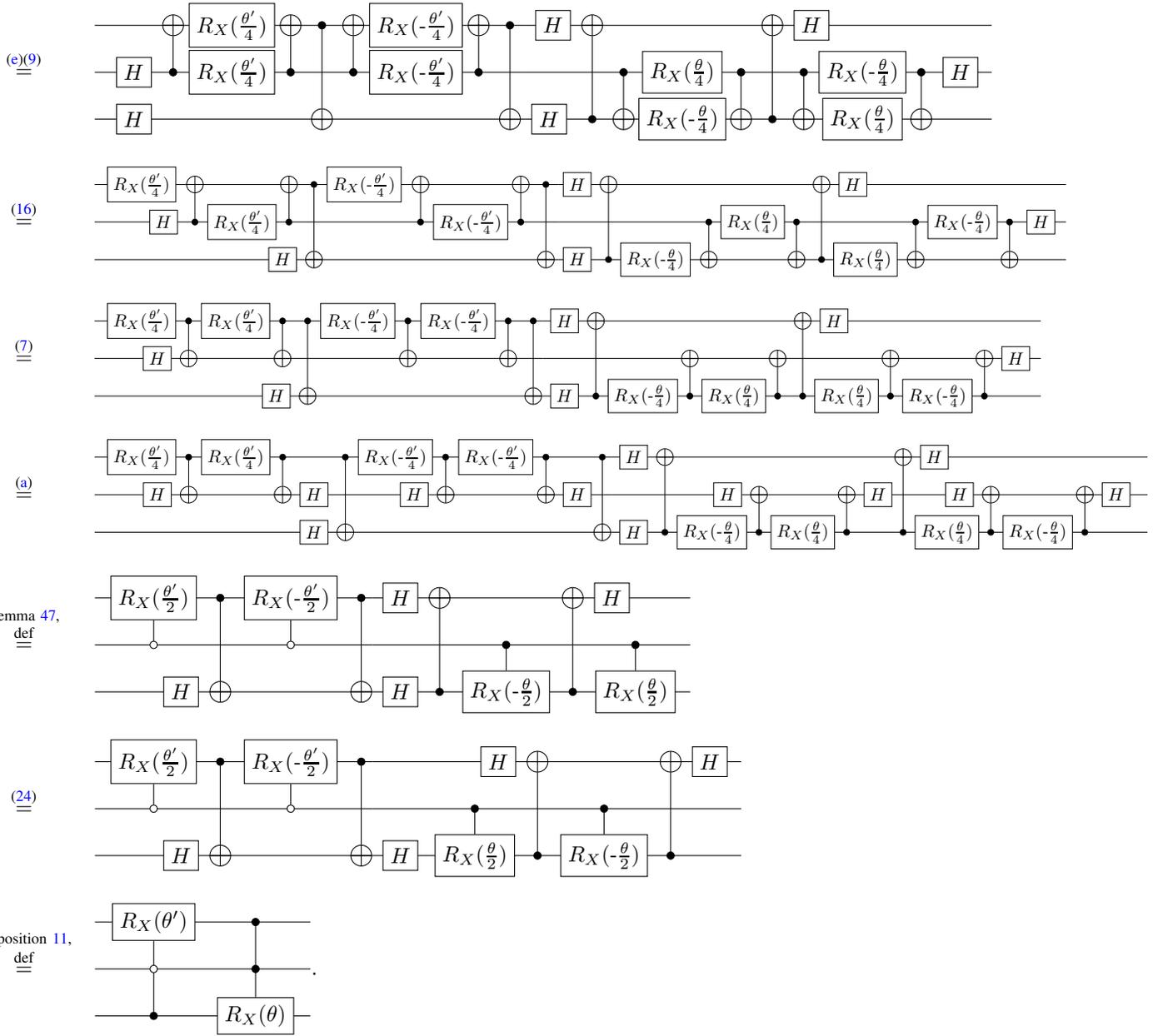
\begin{longtable}{CL}
&\input{ccRXRXcbc.tikz}\\\\
\overset{\text{\cref{cor:swap},}}{\overset{\text{def}}{=}}&\input{ccRXRXcbcsemidec1.tikz}\\\\
\eqeqref{eq:comRX}&\input{ccRXRXcbcsemidec2.tikz}\\\\
\overset{\text{\cref{decompctrlblancRX},}}{\overset{\text{def}}{=}}&\scalebox{0.8}{$\input{ccRXRXcbcdec.tikz}$}\\\\
\eqeqref{HH}&\scalebox{0.8}{$\input{ccRXRXcbcdec1.tikz}$}\\\\
\eqeqref{NotCRXNotCreversible}&\scalebox{0.8}{$\input{ccRXRXcbcdec2.tikz}$}\\\\
\eqeqref{commutationRXCNot}&\input{ccRXRXcbcdec4.tikz}\\\\
\eqdeuxeqref{CNotHH}{CNotCNot}&\input{ccRXRXcbcdec5.tikz}\\\\
\eqeqref{commutationdecompctrlRXpasenface}&\input{ccRXRXcbcdec6.tikz}\\\\
\eqdeuxeqref{CNotCNot}{CNotHH}&\input{ccRXRXcbcdec8.tikz}\\\\
\eqeqref{commutationRXCNot}&\scalebox{0.8}{$\input{ccRXRXcbcdec9.tikz}$}\\\\
\eqeqref{NotCRXNotCreversible}&\scalebox{0.8}{$\input{ccRXRXcbcdec10.tikz}$}\\\\
\eqeqref{HH}&\scalebox{0.8}{$\input{ccRXRXcbcdec11.tikz}$}\\\\
\overset{\text{\cref{decompctrlblancRX},}}{\overset{\text{def}}{=}}&\input{ccRXRXcbcsemidec3.tikz}\\\\
\eqeqref{eq:comRX}&\input{ccRXRXcbcsemidec4.tikz}\\\\
\overset{\text{\cref{cor:swap},}}{\overset{\text{def}}{=}}&\input{RXcbcccRX.tikz}.
\end{longtable}
If $k\geq2$, by \cref{cor:swap} we can assume without loss of generality that $y=at$ and $y'=a't'$ with $a,a'\in\{0,1\}$ and $t\neq t'$. One has (with the $\pm$ signs being $(-1)^a$ and $(-1)^{a'}$ respectively):
\begin{longtable}{RCL}
\input{RXhLambdaccLambdaRxbgrand.tikz}&=&\input{RXhLambdaccLambdaRxbdecomp.tikz}\\\\
&\eqeqref{HH}&\input{RXhLambdaccLambdaRxbdecomp1.tikz}\\\\
&\overset{\text{\cref{commutationctrldotsCNotRX}}}{=}&\input{RXhLambdaccLambdaRxbdecomp2.tikz}\\\\
&\overset{\text{induction}}{\overset{\text{hypothesis}}{=}}&\input{RXhLambdaccLambdaRxbdecomp3.tikz}\\\\
&\overset{\text{\cref{commutationctrldotsCNotRX},}}{\overset{\text{induction}}{\overset{\text{hypothesis}}{=}}}&\input{RXhLambdaccLambdaRxbdecomp4.tikz}\\\\
&\eqeqref{commutationCNotsbas}&\input{RXhLambdaccLambdaRxbdecomp5.tikz}\\\\
&\overset{\text{\cref{cor:swap},}}{\overset{\text{\cref{commutationctrldotsCNotRX}}}{=}}&\input{RXhLambdaccLambdaRxbdecomp6.tikz}\\\\
&\overset{\eqref{commutationCNotsbas},}{\overset{\text{\cref{cor:swap},}}{\overset{\text{\cref{commutationctrldotsCNotRX}}}{=}}}&\input{RXhLambdaccLambdaRxbdecomp7.tikz}\\\\
&\overset{\text{\cref{commutationctrldotsCNotRX},}}{\overset{\text{induction}}{\overset{\text{hypothesis}}{=}}}&\input{RXhLambdaccLambdaRxbdecomp8.tikz}\\\\
&\overset{\eqref{commutationCNotsbas},}{\overset{\text{\cref{cor:swap},}}{\overset{\text{\cref{commutationctrldotsCNotRX}}}{=}}}&\input{RXhLambdaccLambdaRxbdecomp9.tikz}\\\\
&\eqeqref{HH}&\input{RXhLambdaccLambdaRxbdecomp10.tikz}\\\\
&=&\input{cLambdaRxbRXhLambdacgrand.tikz}.
\end{longtable}

\subsection{Proof of Proposition \ref{EulerHmoins}}\label{preuveEulerHmoins}
\begin{longtable}{RCL}
\begin{tikzpicture}
	\begin{pgfonlayer}{nodelayer}
		\node [style=none] (0) at (-1, 0) {};
		\node [style=none] (1) at (1, 0) {};
		\node [style=boite2] (2) at (0, 0) {$H$};
	\end{pgfonlayer}
	\begin{pgfonlayer}{edgelayer}
		\draw (0.center) to (1.center);
	\end{pgfonlayer}
\end{tikzpicture}
&\eqquatreeqref{P0}{PP}{RX0}{RXRX}&\input{EulerHmoins1.tikz}\\\\
&\eqeqref{EulerH}&\input{EulerHmoins2.tikz}\\\\
&\eqeqref{HH}&\input{EulerHmoins.tikz}
\end{longtable}

\subsection{Proof of Proposition \ref{soundnessEuleraxioms}}
\label{appendix:euler3d}

The proof is inspired by the proofs of Lemmas 10 and 11 of \cite{clement2022LOv}.
Given any $n$-qubit quantum circuit $C$, let $\interpg{C}\coloneqq\mathfrak G_n^{-1}\circ\interp{C}\circ \mathfrak G_n$.

\subsubsection{Soundness of Equation \eqref{Euler2d}}

Given any $\alpha_1,\alpha_2,\alpha_3\in\mathbb R$, let $U\coloneqq\interpg{\minitikzfig{RXPRXalphas}}$. We have to prove that there exist unique $\beta_0,\beta_1,\beta_2,\beta_3$ satisfying the conditions of \cref{fig:euler} such that $\interpg{\minitikzfig{RXPRXbetas}}=U$. We are going to first prove that assuming that such $\beta_j$ exist, their values are uniquely determined by $U$. Since we are going do so by giving explicit expressions of the unique possible value of each $\beta_j$ in terms of the entries of $U$, it will then be easy to check that these expressions indeed define angles with the desired properties.

One has
\[U=\interpg{\minitikzfig[0.83]{RXPRXbetas}}=e^{i\beta_0}\begin{pmatrix}\cos\bigl(\frac{\beta_2}2\bigr)&-ie^{i\beta_1}\sin\bigl(\frac{\beta_2}2\bigr)\\-ie^{i\beta_3}\sin\bigl(\frac{\beta_2}2\bigr)&e^{i(\beta_1+\beta_3)}\cos\bigl(\frac{\beta_2}2\bigr)\end{pmatrix}\]

If $U$ has a null entry, then since it is unitary, it is either diagonal or anti-diagonal. If it is diagonal, then $\sin\bigl(\frac{\beta_2}2\bigr)=0$, which, since $\beta_2\in[0,2\pi)$, implies that $\beta_2=0$, which by the constraint on $\beta_1$ and $\beta_2$, implies that $\beta_1=0$. Consequently, $\beta_0=\arg(U_{0,0})$ and $\beta_3=\arg\left(\frac{U_{1,1}}{U_{0,0}}\right)$. If $U$ is anti-diagonal, then $\cos\bigl(\frac{\beta_2}2\bigr)=0$, which, since $\beta_2\in[0,2\pi)$, implies that $\beta_2=\pi$, which by the constraint on $\beta_1$ and $\beta_2$, implies that $\beta_1=0$. Consequently, $\beta_0=\arg\left(\frac{U_{0,1}}{-i}\right)$ and $\beta_3=\arg\left(\frac{U_{1,0}}{U_{0,1}}\right)$.

If $U$ has no null entry, then one has $\beta_2\neq\pi$ and $\dfrac{ie^{-i\beta_1}U_{0,1}}{U_{0,0}}=\tan\bigl(\frac{\beta_2}2\bigr)$. Hence, $\beta_1$ is the unique angle in $[0,\pi)$ such that $\dfrac{ie^{-i\beta_1}U_{0,1}}{U_{0,0}}\in\mathbb R$, namely $\arg\left(\frac{iU_{0,1}}{U_{0,0}}\right)\bmod\pi$. In turn, $\beta_2$ is the unique angle in $[0,2\pi)\setminus\{\pi\}$ such that $\tan\bigl(\frac{\beta_2}2\bigr)=\dfrac{ie^{-i\beta_1}U_{0,1}}{U_{0,0}}$. Finally, one has $e^{i\beta_3}=\frac{\cos(\frac{\beta_2}2)U_{1,0}}{-i\sin(\frac{\beta_2}2)U_{0,0}}$, so that $\beta_3=\arg\left(\frac{\cos(\frac{\beta_2}2)U_{1,0}}{-i\sin(\frac{\beta_2}2)U_{0,0}}\right)$, and $e^{i\beta_0}=\frac{U_{0,0}}{\cos(\frac{\beta_2}2)}$, so that $\beta_0=\arg\left(\frac{U_{0,0}}{\cos(\frac{\beta_2}2)}\right)$.

\subsubsection{Soundness of Equation \eqref{Euler3dmulticontrolled}}

Given any $n$-qubit quantum circuit $C$ such that $\interpg{C}$ is of the form $\left(\begin{array}{c|c}I&0\\\hline 0&U\end{array}\right)$ with $U\in\mathbb C^{3\times3}$, let $\interpt{C}\coloneqq U$.

Given any $\gamma_1,\gamma_2,\gamma_3,\gamma_4\in\mathbb R$, let $U\coloneqq\interpt{\minitikzfig{Euler3dleft-multicontrolled-simp-gammas}}$. We have to prove that there exist unique $\delta_1,\delta_2,\delta_3,\delta_4,\delta_5,\delta_6,\delta_7,\delta_8,\delta_9$ satisfying the conditions of \cref{fig:euler} such that
\[\interpt{\minitikzfig{Euler3dright-multicontrolled-simp-deltas}}=U,\]
or equivalently,
\[\interpt{\minitikzfig{Euler3D-2qubit-R}}=U.\]
We are going to first prove that assuming that such $\delta_j$ exist, their values are uniquely determined by $U$. Since we are going do so by giving explicit expressions of the unique possible value of each $\delta_j$ in terms of the entries of $U$, it will then be easy to check that these expressions indeed define angles with the desired properties.

Let $U_{123}\coloneqq\interpt{\minitikzfig{Euler3dright-123}}=\begin{pmatrix}e^{i\delta_2}&0&0\\0&e^{i(\delta_1+\delta_2)}\cos\bigl(\frac{\delta_3}2\bigr)&-i\sin\bigl(\frac{\delta_3}2\bigr)\\0&-ie^{i(\delta_1+\delta_2)}\sin\bigl(\frac{\delta_3}2\bigr)&\cos\bigl(\frac{\delta_3}2\bigr)\end{pmatrix}$, $U_{4}\coloneqq\interpt{\minitikzfig{Euler3dright-4}}=\begin{pmatrix}\cos\bigl(\frac{\delta_4}2\bigr)&-i\sin\bigl(\frac{\delta_4}2\bigr)&0\\-i\sin\bigl(\frac{\delta_4}2\bigr)&\cos\bigl(\frac{\delta_4}2\bigr)&0\\0&0&1\end{pmatrix}$ and $U_{56}\coloneqq\interpt{\minitikzfig{Euler3dright-56}}=\begin{pmatrix}1&0&0\\0&e^{i\delta_5}\cos\bigl(\frac{\delta_6}2\bigr)&-i\sin\bigl(\frac{\delta_6}2\bigr)\\0&-ie^{i\delta_5}\sin\bigl(\frac{\delta_6}2\bigr)&\cos\bigl(\frac{\delta_6}2\bigr)\end{pmatrix}$. Let also $U_{\mathrm{I}}\coloneqq U_{123}\circ U^\dag$, $U_{\mathrm{II}}\coloneqq U_4\circ U_{\mathrm{I}}$ and $U_{\mathrm{III}}\coloneqq U_{56}\circ U_{\mathrm{II}}$.

By construction, %
\begin{eqnexpr}\label{U3}U_{\mathrm{III}}=\interpt{\minitikzfig{Euler3dright-789}}^\dag=\begin{pmatrix}e^{-i\delta_9}&0&0\\0&e^{-i(\delta_7+\delta_8+\delta_9)}&0\\0&0&e^{-i\delta_8}\end{pmatrix}\end{eqnexpr}
so that
\begin{eqnexpr}\label{U2}U_{\mathrm{II}}=U_{56}^\dag\circ U_{\mathrm{III}}=\begin{pmatrix}e^{-i\delta_9}&0&0\\0&e^{-i(\delta_5+\delta_7+\delta_8+\delta_9)}\cos\bigl(\frac{\delta_6}2\bigr)&ie^{-i(\delta_5+\delta_8)}\sin\bigl(\frac{\delta_6}2\bigr)\\0&ie^{-i(\delta_7+\delta_8+\delta_9)}\sin\bigl(\frac{\delta_6}2\bigr)&e^{-i\delta_8}\cos\bigl(\frac{\delta_6}2\bigr)\end{pmatrix}\end{eqnexpr}
and $U_{\mathrm{I}}=U_4^\dag\circ U_{\mathrm{II}}$. Since $U_4$ acts as the identity on the last entry, this implies that $(U_{\mathrm{I}})_{2,0}=0$.\footnote{Where we denote by $M_{i,j}$ the entry of indices $(i,j)$ of any matrix $M$, the index of the first row and column being $0$.} That is, by definition of $U_{\mathrm{I}}$,
\begin{eqnexpr}\label{premiercoef}\textstyle-ie^{i(\delta_1+\delta_2)}\sin\bigl(\frac{\delta_3}2\bigr)U_{0,1}^\dag+\cos\bigl(\frac{\delta_3}2\bigr)U_{0,2}^\dag=0.\end{eqnexpr}

By direct calculation using the definitions of $U_{\mathrm{I}}$ and $U_{\mathrm{II}}$, one gets $(U_{\mathrm{I}})_{0,0}=e^{i\delta_2}U_{0,0}^\dag$ and $(U_{\mathrm{I}})_{1,0}=e^{i(\delta_1+\delta_2)}\cos\bigl(\frac{\delta_3}2\bigr)U_{0,1}^\dag-i\sin\bigl(\frac{\delta_3}2\bigr)U_{0,2}^\dag$, so that $(U_{\mathrm{II}})_{1,0}=-i\sin\bigl(\frac{\delta_4}2\bigr)(U_{\mathrm{I}})_{0,0}+\cos\bigl(\frac{\delta_4}2\bigr)(U_{\mathrm{I}})_{1,0}=-i\sin\bigl(\frac{\delta_4}2\bigr)e^{i\delta_2}U_{0,0}^\dag+\cos\bigl(\frac{\delta_4}2\bigr)(e^{i(\delta_1+\delta_2)}\cos\bigl(\frac{\delta_3}2\bigr)U_{0,1}^\dag-i\sin\bigl(\frac{\delta_3}2\bigr)U_{0,2}^\dag)$. That is, since by \eqref{U2}, $(U_{\mathrm{II}})_{1,0}=0$:
\begin{eqnexpr}\label{deuxiemecoef}\textstyle-i\sin\bigl(\frac{\delta_4}2\bigr)e^{i\delta_2}U_{0,0}^\dag+\cos\bigl(\frac{\delta_4}2\bigr)\left(e^{i(\delta_1+\delta_2)}\cos\bigl(\frac{\delta_3}2\bigr)U_{0,1}^\dag-i\sin\bigl(\frac{\delta_3}2\bigr)U_{0,2}^\dag\right)=0\end{eqnexpr}
\begin{itemize}
\item If $U_{0,1}=U_{0,2}=0$, then since $U$ is unitary, $U_{0,0}\neq0$ and \eqref{deuxiemecoef} becomes $-i\sin\bigl(\frac{\delta_4}2\bigr)e^{i\delta_2}U_{0,0}^\dag=0$, that is $\sin\bigl(\frac{\delta_4}2\bigr)=0$. Since $\delta_4\in[0,2\pi)$, this implies that $\delta_4=0$, which by the conditions of \cref{fig:euler}, implies that $\delta_1=\delta_2=\delta_3=0$.
\item If $(U_{0,1},U_{0,2})\neq(0,0)$, then $e^{i(\delta_1+\delta_2)}\cos\bigl(\frac{\delta_3}2\bigr)U_{0,1}^\dag-i\sin\bigl(\frac{\delta_3}2\bigr)U_{0,2}^\dag\neq0$. Indeed, if this expression was equal to $0$, by \eqref{premiercoef} this would mean that the non-zero vector $\begin{pmatrix}e^{i(\delta_1+\delta_2)}U_{0,1}^\dag\\U_{0,2}^\dag\end{pmatrix}$ is in the kernel of the matrix $\begin{pmatrix}\cos\bigl(\frac{\delta_3}2\bigr)&-i\sin\bigl(\frac{\delta_3}2\bigr)\\-i\sin\bigl(\frac{\delta_3}2\bigr)&\cos\bigl(\frac{\delta_3}2\bigr)\end{pmatrix}$, whereas this matrix is invertible. Then:
\begin{itemize}
\item If $U_{0,0}=0$, then \eqref{deuxiemecoef} implies that $\cos\bigl(\frac{\delta_4}2\bigr)=0$, which, since $\delta_4\in[0,2\pi)$, implies that $\delta_4=\pi$. By the conditions of \cref{fig:euler}, this implies that $\delta_2=0$. Then:
\begin{itemize}
\item If  $U_{0,2}=0$, then $U_{0,1}\neq0$, and \eqref{premiercoef} implies that $\sin\bigl(\frac{\delta_3}2\bigr)=0$, that is, since $\delta_3\in[0,2\pi)$, that $\delta_3=0$. By the conditions of \cref{fig:euler}, together with the fact that $\delta_4=\pi$, this implies that $\delta_1=0$.
\item If  $U_{0,1}=0$, then $U_{0,2}\neq0$, and \eqref{premiercoef} implies that $\cos\bigl(\frac{\delta_3}2\bigr)=0$, that is, since $\delta_3\in[0,2\pi)$, that $\delta_3=\pi$. By the conditions of \cref{fig:euler}, this implies that $\delta_1=0$.
\item If $U_{0,1},U_{0,2}\neq0$, then \eqref{premiercoef}, on the one hand, implies that $\delta_3\neq\pi$, and on the other hand, is equivalent to
\[\tan\bigl(\tfrac{\delta_3}2\bigr)=\frac{e^{-i\delta_1}U_{0,2}^\dag}{iU_{0,1}^\dag}.\]
Hence, $\delta_1$ is the unique angle in $[0,\pi)$ such that $\frac{e^{-i\delta_1}U_{0,2}^\dag}{iU_{0,1}^\dag}\in\mathbb R$. In turn, $\delta_3$ is the unique angle in $[0,2\pi)$ such that $\tan\bigl(\frac{\delta_3}2\bigr)=\frac{e^{-i\delta_1}U_{0,2}^\dag}{iU_{0,1}^\dag}$.
\end{itemize}
\item If $U_{0,0}\neq0$, then \eqref{deuxiemecoef} can be simplified into
\begin{eqnexpr}\label{deuxiemecoefbis}\textstyle-i\tan\bigl(\frac{\delta_4}2\bigr)e^{i\delta_2}U_{0,0}^\dag+e^{i(\delta_1+\delta_2)}\cos\bigl(\frac{\delta_3}2\bigr)U_{0,1}^\dag-i\sin\bigl(\frac{\delta_3}2\bigr)U_{0,2}^\dag=0.\end{eqnexpr}
\begin{itemize}
\item If  $U_{0,2}=0$, then $U_{0,1}\neq0$, and \eqref{premiercoef} 
implies that $\sin\bigl(\frac{\delta_3}2\bigr)=0$, that is, since $\delta_3\in[0,2\pi)$, that $\delta_3=0$. By the conditions of \cref{fig:euler}, this implies that $\delta_2=0$. Then \eqref{deuxiemecoefbis} becomes
\[\textstyle-i\tan\bigl(\frac{\delta_4}2\bigr)U_{0,0}^\dag+e^{i\delta_1}U_{0,1}^\dag=0\]
that is,
\[\tan\bigl(\tfrac{\delta_4}2\bigr)=\frac{e^{i\delta_1}U_{0,1}^\dag}{iU_{0,0}^\dag}.\]
Hence, $\delta_1$ is the unique angle in $[0,\pi)$ such that $\frac{e^{i\delta_1}U_{0,1}^\dag}{iU_{0,0}^\dag}\in\mathbb R$. In turn, $\delta_4$ is the unique angle in $[0,2\pi)$ such that $\tan\bigl(\frac{\delta_4}2\bigr)=\frac{e^{i\delta_1}U_{0,1}^\dag}{iU_{0,0}^\dag}$.
\item If  $U_{0,1}=0$, then $U_{0,2}\neq0$, and \eqref{premiercoef} 
implies that $\cos\bigl(\frac{\delta_3}2\bigr)=0$, that is, since $\delta_3\in[0,2\pi)$, that $\delta_3=\pi$. By the conditions of \cref{fig:euler}, this implies that $\delta_1=0$. Then \eqref{deuxiemecoefbis} becomes
\[\textstyle-i\tan\bigl(\frac{\delta_4}2\bigr)e^{i\delta_2}U_{0,0}^\dag-iU_{0,2}^\dag=0\]
that is,
\[\tan\bigl(\tfrac{\delta_4}2\bigr)=-\frac{e^{-i\delta_2}U_{0,2}^\dag}{U_{0,0}^\dag}.\]
Hence, $\delta_2$ is the unique angle in $[0,\pi)$ such that $\frac{e^{-i\delta_2}U_{0,2}^\dag}{U_{0,0}^\dag}\in\mathbb R$. In turn, $\delta_4$ is the unique angle in $[0,2\pi)$ such that $\tan\bigl(\frac{\delta_4}2\bigr)=-\frac{e^{-i\delta_2}U_{0,2}^\dag}{U_{0,0}^\dag}$.
\item If $U_{0,1},U_{0,2}\neq0$, then \eqref{premiercoef}, on the one hand, implies that $\delta_3%
\notin\{0,\pi\}$, and on the other hand, is equivalent to
\begin{eqnexpr}\label{eidelta12}e^{i(\delta_1+\delta_2)}=\frac{\cos\bigl(\frac{\delta_3}2\bigr)U_{0,2}^\dag}{i\sin\bigl(\frac{\delta_3}2\bigr)U_{0,1}^\dag}.\end{eqnexpr}
Then by substituting in \eqref{deuxiemecoefbis}, we get
\[\textstyle-i\tan\bigl(\frac{\delta_4}2\bigr)e^{i\delta_2}U_{0,0}^\dag+\dfrac{\cos^2\bigl(\frac{\delta_3}2\bigr)U_{0,2}^\dag}{i\sin\bigl(\frac{\delta_3}2\bigr)}-i\sin\bigl(\frac{\delta_3}2\bigr)U_{0,2}^\dag=0\]
which can be simplified into
\[\textstyle-i\tan\bigl(\frac{\delta_4}2\bigr)e^{i\delta_2}U_{0,0}^\dag+\dfrac{U_{0,2}^\dag}{i\sin\bigl(\frac{\delta_3}2\bigr)}=0\]
which is equivalent to
\begin{eqnexpr}\label{tandelta4}\textstyle\tan\bigl(\frac{\delta_4}2\bigr)=-\dfrac{e^{-i\delta_2}U_{0,2}^\dag}{\sin\bigl(\frac{\delta_3}2\bigr)U_{0,0}^\dag}.\end{eqnexpr}
Hence, $\delta_2$ is the unique angle in $[0,\pi)$ such that $\dfrac{e^{-i\delta_2}U_{0,2}^\dag}{U_{0,0}^\dag}\in\mathbb R$. Then \eqref{eidelta12} can be rephrased into
\[\tan\bigl(\tfrac{\delta_3}2\bigr)=\frac{e^{-i(\delta_1+\delta_2)}U_{0,2}^\dag}{iU_{0,1}^\dag}.\]
Hence, $\delta_1$ is the unique angle in $[0,\pi)$ such that $\frac{e^{-i(\delta_1+\delta_2)}U_{0,2}^\dag}{iU_{0,1}^\dag}\in\mathbb R$. In turn, $\delta_3$ is the unique angle in $[0,2\pi)$ such that $\tan\bigl(\tfrac{\delta_3}2\bigr)=\frac{e^{-i(\delta_1+\delta_2)}U_{0,2}^\dag}{iU_{0,1}^\dag}$. Finally, %
$\delta_4$ is the unique angle in $[0,2\pi)$ satisfying \eqref{tandelta4}.
\end{itemize}
\end{itemize}
\end{itemize}
Thus, assuming that the $\delta_j$ exist, since $U_{\mathrm{I}}$ and $U_{\mathrm{II}}$ only depend on $\delta_1$, $\delta_2$, $\delta_3$, $\delta_4$ and $U$, they are uniquely determined by $U$. Then \eqref{U2} implies that
\begin{itemize}
\item If $(U_{\mathrm{II}})_{1,2}=0$, then $\sin\bigl(\frac{\delta_6}2\bigr)=0$, which means, since $\delta_6\in[0,2\pi)$, that $\delta_6=0$. By the conditions of \cref{fig:euler}, this implies that $\delta_5=0$.
\item If $(U_{\mathrm{II}})_{2,2}=0$, then $\cos\bigl(\frac{\delta_6}2\bigr)=0$, which means, since $\delta_6\in[0,2\pi)$, that $\delta_6=\pi$. By the conditions of \cref{fig:euler}, this implies that $\delta_5=0$.
\item If $(U_{\mathrm{II}})_{1,2}=0,(U_{\mathrm{II}})_{2,2}\neq0$, then
\[\tan\bigl(\tfrac{\delta_6}2\bigr)=\frac{e^{i\delta_5}(U_{\mathrm{II}})_{1,2}}{i(U_{\mathrm{II}})_{2,2}}.\]
Hence, $\delta_5$ is the unique angle in $[0,\pi)$ such that $\frac{e^{i\delta_5}(U_{\mathrm{II}})_{1,2}}{i(U_{\mathrm{II}})_{2,2}}\in\mathbb R$. In turn, $\delta_6$ is the unique angle in $[0,2\pi)$ such that $\tan\bigl(\frac{\delta_6}2\bigr)=\frac{e^{i\delta_5}(U_{\mathrm{II}})_{1,2}}{i(U_{\mathrm{II}})_{2,2}}$.
\end{itemize}
Thus, assuming that the $\delta_j$ exist, since $U_{\mathrm{III}}$ only depends on $\delta_5$, $\delta_6$ and $U_{\mathrm{II}}$, it is uniquely determined by $U$. Then by \eqref{U3}, $\delta_8=\arg((U_{\mathrm{III}})_{2,2}^\dag)$, $\delta_9=\arg((U_{\mathrm{III}})_{0,0}^\dag)$ and $\delta_7=\arg\left(\frac{(U_{\mathrm{III}})_{0,0}(U_{\mathrm{III}})_{2,2}}{(U_{\mathrm{III}})_{1,1}}\right)$.

\subsection{Proof of Proposition \ref{klfollowfromEuler}}\label{preuveklfollowfromEuler}
\noindent Proof of Equation \eqref{PP}:
\begin{longtable}{RCL}
\begin{tikzpicture}
	\begin{pgfonlayer}{nodelayer}
		\node [style=none] (6) at (-3, 0) {};
		\node [style=none] (7) at (3, 0) {};
		\node [style=boite2] (15) at (-1.5, 0) {$P(\varphi_1)$};
		\node [style=boite2] (18) at (1.425, 0) {$P(\varphi_2)$};
	\end{pgfonlayer}
	\begin{pgfonlayer}{edgelayer}
		\draw (6.center) to (7.center);
	\end{pgfonlayer}
\end{tikzpicture}
&\eqeqref{HH}&\input{PP1.tikz}\\\\
&\eqtroiseqref{S0}{SS}{RXgate}&\input{PP2.tikz}\\\\
&\eqeqref{P0}&\input{PP3.tikz}\\\\
&\eqeqref{Euler2d}&\input{PP4.tikz}\\\\
&\eqeqref{Euler2d}&\input{PP5.tikz}\\\\
&\eqdeuxeqref{P0}{RX0}&\input{PP6.tikz}\\\\
&\eqquatreeqref{RXgate}{HH}{SS}{S0}&\begin{tikzpicture}
	\begin{pgfonlayer}{nodelayer}
		\node [style=none] (4) at (-2.25, 0) {};
		\node [style=none] (5) at (2.25, 0) {};
		\node [style=boite2] (11) at (0, 0) {$P(\varphi_1{+}\varphi_2)$};
	\end{pgfonlayer}
	\begin{pgfonlayer}{edgelayer}
		\draw (4.center) to (5.center);
	\end{pgfonlayer}
\end{tikzpicture}

\end{longtable}
The first use of \cref{Euler2d} is valid since \cref{Euler2d} is applied from the left to the right. The second use of \cref{Euler2d} is valid since it 
preserves the semantics. Note that one can show that $\beta_1=\beta_3=0$, $\beta_2=\varphi_1+\varphi_2\bmod 2\pi$ and $\beta_0=\begin{cases}0&\text{if $(\varphi_1+\varphi_2\bmod 4\pi)\in[0,2\pi)$}\\\pi&\text{if $(\varphi_1+\varphi_2\bmod 4\pi)\in[2\pi,4\pi)$}\end{cases}$\bigskip.

\noindent Proof of Equation \eqref{XPX}:
  \begin{longtable}{RCL}
  \input{XPX.tikz}&\eqdeuxeqref{xgate}{zgate}&\input{XPX0.tikz}\\\\
  &\eqtroiseqref{S0}{SS}{RXgate}&\input{XPX1.tikz}\\\\
  &\eqdeuxeqref{Euler2d}{SS}&\input{RXPRXbetas1.tikz}
  \end{longtable}
One has $\beta_1=\beta_2=0$,
  $\beta_3=-\varphi \bmod 2\pi$ and $\beta_0=\varphi-\pi \bmod 2\pi$%
.  Indeed, this choice of angles
  satisfies the conditions of \cref{Euler2d} and is sound with respect to the semantics, and \cref{soundnessEuleraxioms} guarantees that this is the only possible choice.
  Thus, by \crefnosort{P0,RX0}, this implies that one can
  transform $\scalebox{0.8}{\input{XPX.tikz}}$ into
  $\scalebox{0.8}{\input{RXPRXbetas4bis.tikz}}\eqdeuxeqref{S0}{SS}
  \scalebox{0.8}{\begin{tikzpicture}
	\begin{pgfonlayer}{nodelayer}
		\node [style=none] (4) at (0.75, 0) {};
		\node [style=none] (5) at (6.5, 0) {};
		\node [style=boite2] (12) at (3.75, 0) {$P(\text{-}\varphi \bmod 2\pi)$};
		\node [style=cartouche] (14) at (0.25, 0.7) {$\varphi$};
	\end{pgfonlayer}
	\begin{pgfonlayer}{edgelayer}
		\draw (4.center) to (5.center);
	\end{pgfonlayer}
\end{tikzpicture}
}$\smallskip. Finally,
  $\scalebox{0.8}{\begin{tikzpicture}
	\begin{pgfonlayer}{nodelayer}
		\node [style=none] (4) at (-1.75, 0) {};
		\node [style=none] (5) at (1.75, 0) {};
		\node [style=boite2] (12) at (0, 0) {$P(\text{-}\varphi )$};
	\end{pgfonlayer}
	\begin{pgfonlayer}{edgelayer}
		\draw (4.center) to (5.center);
	\end{pgfonlayer}
\end{tikzpicture}
}\eqeqref{RX0}
    \scalebox{0.8}{\input{RXPRXbetas7.tikz}}\eqdeuxeqref{Euler2d}{S0}
  \scalebox{0.8}{\input{RXPRXbetas3.tikz}} \eqdeuxeqref{P0}{RX0} $
  $\scalebox{0.8}{\begin{tikzpicture}
	\begin{pgfonlayer}{nodelayer}
		\node [style=none] (4) at (0.75, 0) {};
		\node [style=none] (5) at (6.5, 0) {};
		\node [style=boite2] (12) at (3.75, 0) {$P(\text{-}\varphi \bmod 2\pi)$};
	\end{pgfonlayer}
	\begin{pgfonlayer}{edgelayer}
		\draw (4.center) to (5.center);
	\end{pgfonlayer}
\end{tikzpicture}
}$, which terminates the proof.

\subsection{Proof of Proposition \ref{prop:CCX}}
\label{proof:CCX}

First, we can show that $\qc \vdash \Lambda^x P(2\pi) =  id_{k+1}$ as follows: 

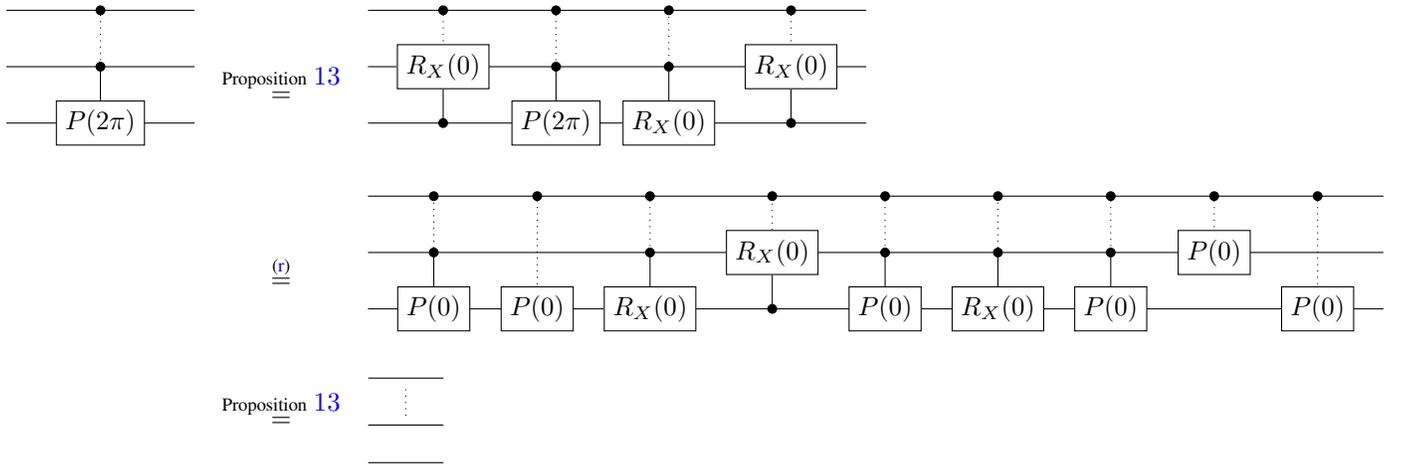
\begin{eqnarray*}\input{CCP2Pi1.tikz}&\overset{\text{Proposition }\ref{prop:sum}}{=}&\input{CCP2Pi2.tikz}\\[0.5cm]
     &\eqeqref{Euler3dmulticontrolled}&\input{CCP2Pi3.tikz}\\[0.5cm]
     &\overset{\text{Proposition }\ref{prop:sum}}{=}&\input{idk1.tikz}\\[0.5cm]
     \end{eqnarray*}

It follows that, for $x\in \{1\}^k$:

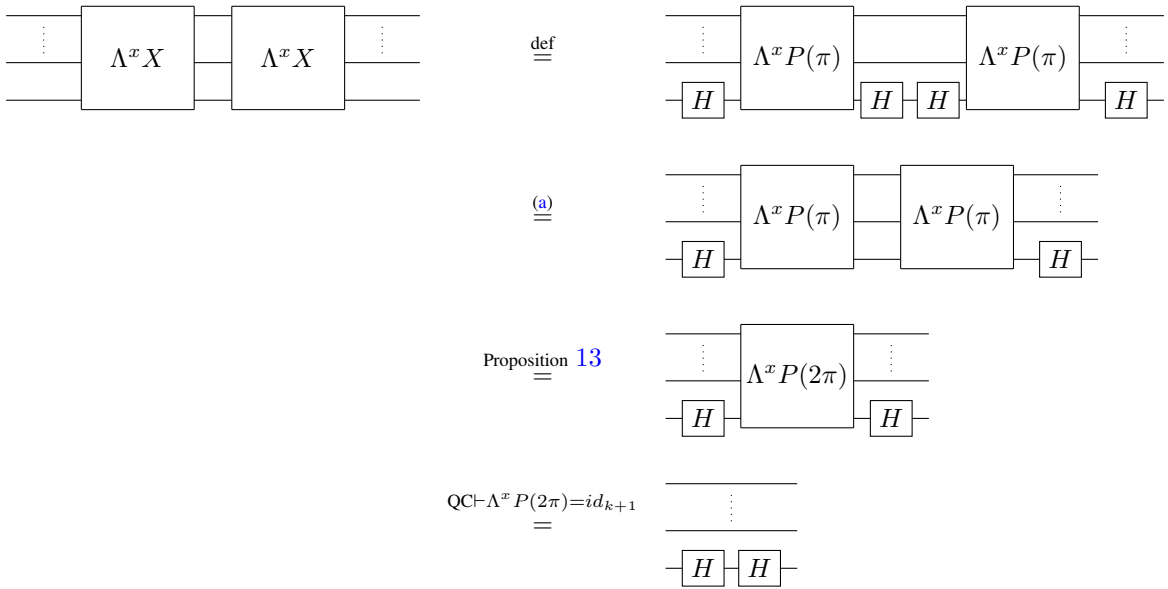
\begin{eqnarray*}\input{lambdaXlambdaX1.tikz}&\overset{\text{def}}{=}&\input{lambdaXlambdaX2.tikz}\\[0.5cm]
     &\eqeqref{HH}&\input{lambdaXlambdaX3.tikz}\\[0.5cm]
     &\overset{\text{Proposition }\ref{prop:sum}}{=}&\input{lambdaXlambdaX4.tikz}\\[0.5cm]
     &\overset{\textup{QC} \vdash \Lambda^{x}P(2\pi)=id_{k+1}}{\overset{}{=}}&\input{lambdaXlambdaX5.tikz}\\[0.5cm]
     &\eqeqref{HH}&\input{idk1.tikz}\\[0.5cm]
     \end{eqnarray*}

\subsection{Proof of Proposition \ref{prop:period}}
\label{proof:propperiod}
First, we prove the case for $x=\epsilon$ : 
\[\qczero \vdash R_X(4\pi) = id_1 \qquad\qczero \vdash  P(2\pi) =  id_1  \qquad\qczero \vdash  s(2\pi) =  id_0\].
  
\begin{eqnarray*}\begin{tikzpicture}
	\begin{pgfonlayer}{nodelayer}
		\node [style=none] (4) at (-1.5, 0) {};
		\node [style=none] (5) at (1.5, 0) {};
		\node [style=boite2] (11) at (0, 0) {$P(2\pi)$};
	\end{pgfonlayer}
	\begin{pgfonlayer}{edgelayer}
		\draw (4.center) to (5.center);
	\end{pgfonlayer}
\end{tikzpicture}
&\eqdeuxeqref{zgate}{PP}&\begin{tikzpicture}
	\begin{pgfonlayer}{nodelayer}
		\node [style=none] (4) at (-2, 0) {};
		\node [style=none] (5) at (1.75, 0) {};
		\node [style=boite2] (12) at (-1, 0) {$Z$};
		\node [style=boite2] (13) at (0.75, 0) {$Z$};
	\end{pgfonlayer}
	\begin{pgfonlayer}{edgelayer}
		\draw (4.center) to (5.center);
	\end{pgfonlayer}
\end{tikzpicture}
\\[0.5cm]
     &\eqeqref{ZZ}&\begin{tikzpicture}
	\begin{pgfonlayer}{nodelayer}
		\node [style=none] (0) at (-0.8, 0) {};
		\node [style=none] (1) at (0.8, 0) {};
	\end{pgfonlayer}
	\begin{pgfonlayer}{edgelayer}
		\draw [style=noire] (0.center) to (1.center);
	\end{pgfonlayer}
\end{tikzpicture}\\[0.5cm]
     \end{eqnarray*}

\begin{eqnarray*}\begin{tikzpicture}
	\begin{pgfonlayer}{nodelayer}
		\node [style=cartouche] (0) at (0, -0.1) {$2\pi$};
	\end{pgfonlayer}
\end{tikzpicture}
&\eqeqref{S0}&\begin{tikzpicture}[scale=0.2]
	\begin{pgfonlayer}{nodelayer}
		\node [style=sdiagrammevide] (0) at (0, 0) {};
	\end{pgfonlayer}
\end{tikzpicture}
\\[0.5cm]
          \end{eqnarray*}

It follows that:
\begin{eqnarray*}\begin{tikzpicture}
	\begin{pgfonlayer}{nodelayer}
		\node [style=none] (0) at (-1.75, 0) {};
		\node [style=none] (1) at (1.75, 0) {};
		\node [style=boite2] (2) at (0, 0) {$R_X(4\pi)$};
	\end{pgfonlayer}
	\begin{pgfonlayer}{edgelayer}
		\draw (0.center) to (1.center);
	\end{pgfonlayer}
\end{tikzpicture}
&\eqeqref{RXgate}&\input{HP4piH.tikz}\\[0.5cm]
          &=&\begin{tikzpicture}
	\begin{pgfonlayer}{nodelayer}
		\node [style=none] (4) at (-2, 0) {};
		\node [style=none] (5) at (1.75, 0) {};
		\node [style=boite2] (12) at (-1, 0) {$H$};
		\node [style=boite2] (13) at (0.75, 0) {$H$};
	\end{pgfonlayer}
	\begin{pgfonlayer}{edgelayer}
		\draw (4.center) to (5.center);
	\end{pgfonlayer}
\end{tikzpicture}
\\[0.5cm]
          &\eqeqref{HH}&\begin{tikzpicture}
	\begin{pgfonlayer}{nodelayer}
		\node [style=none] (0) at (-0.8, 0) {};
		\node [style=none] (1) at (0.8, 0) {};
	\end{pgfonlayer}
	\begin{pgfonlayer}{edgelayer}
		\draw [style=noire] (0.center) to (1.center);
	\end{pgfonlayer}
\end{tikzpicture}\\[0.5cm]
          \end{eqnarray*}

We can now prove the general case, first by noticing that $\qc \vdash \Lambda^x P(2\pi)= id_{k+1}$, as proven in \cref{proof:CCX}.

As $\Lambda^{x1}s(2\pi) = \Lambda^x P(2\pi)$, we have for any $x \in \{1\}^k$, $\qc \vdash \Lambda^{x}s(2\pi)=id_{k}$.

Finally:

\begin{eqnarray*}\input{lambdaxRX4pi.tikz}&=&\input{lambdaxRX4pi1.tikz}\\[0.5cm]
     &\overset{\qc \vdash id_{k+1}=\Lambda^{1x}s(2\pi)}{\overset{}{=}}&\input{lambdaxRX4pi2.tikz}\\[0.5cm]
     &\overset{\text{Proposition \ref{prop:sum}}}{=}&\input{lambdaxRX4pi3.tikz}\\[0.5cm]
     &\overset{\qc \vdash \Lambda^{1x}s(2\pi)=id_{k+1}}{\overset{\qc \vdash \Lambda^{1x}P(2\pi)=id_{k+2}}{\overset{}{=}}}&\input{lambdaxRX4pi4.tikz}\\[0.5cm]
     &\eqeqref{HH}&\input{id1k1.tikz}\\[0.5cm]
     \end{eqnarray*}

\section{Proofs of Section \ref{sec:completeness}}

\subsection{Proof of Theorem \ref{thm:LOPPcompleteness}}
\label{proof:thmLOPPcompleteness}

One can easily show that every equation of \cref{axiomsLOPP} is sound with respect to the semantics. Regarding the completeness proof, we use the rewriting system of \cref{PPRS}  that has been introduced in 
\cite{clement2022LOv}. 
This rewriting system has been proved to be strongly normalising, moreover it has been proved that any two swap-free circuits having the same semantics are reduced to the same normal form \cite{clement2022LOv}.

Using \cref{swapbspisur2} one can transform any circuit into a swap-free circuit. As a consequence, to prove the completeness it only remains to show that every rule of \cref{PPRS} can be derived using the equations of  \cref{axiomsLOPP}. 

\begin{figure}[tb]
\newcommand{\nspazer}{-0.2em}
\centering
\scalebox{0.9}{\begin{minipage}{\textwidth}
\begin{vwcol}[widths={0.5,0.5},
 sep=.1cm, justify=flush,rule=0pt,indent=0em] 
\begin{align}
  \label{phasemod2pi}\input{convtp-phase-shift-.tikz}\ &\to\ \input{convtp-phase-shiftthetamod2pi-.tikz}\\[1.5em]
    \label{bsmod2pi}\input{bs-.tikz}\ &\to\ \input{convtp-bsphimod2pi-.tikz}\\[1.5em]
      \label{fusionphaseshifts}\input{convtp-phase-shifts-12.tikz}\ &\to\ \input{convtp-phase-shift-1plus2.tikz}\\[1.5em]
         \label{zerophaseshifts}\input{convtp-phase-shift-zero.tikz}\ &\to\ \input{filcourt.tikz}\\[1.5em]
         \label{zerobs}\input{bs0.tikz}\ &\to \ \input{filsparalleleslongbs-m-.tikz}
  \end{align}

  \begin{align}
  \label{removebottomphase}\input{convtp-phasebbs.tikz}\ &\to\ \input{convtp-moinsthetahbsthetatheta.tikz}\\[0.6em]
  \label{passagepisur2}\input{convtp-phasehbspisur2.tikz}\ &\to\ 
\input{convtp-bspisur2thetab.tikz}\\[0.6em]
  \label{passagephasepi}\input{convtp-phasehbs.tikz}\ &\to\ \input{convtp-thetamoinspihbspimoinsphipib.tikz}\\[0.6em]
  \label{soustractionpi}
\input{bstheta4.tikz}\ &\to\ \input{bstheta4moinspipis.tikz}
  \end{align}
   \end{vwcol}

   \begin{equation}\label{glissadeEulerscalaires}\begin{array}{rcl}\input{convtp-bsyangbaxterpointeenbas-etoiles.tikz}&\to&\input{bsyangbaxterpointeenhaut.tikz}\end{array}\end{equation}

  \begin{equation}  \label{fusionEulerbsphasebs}\qquad\qquad\begin{array}{rcl}\input{bsphasebsalpha-etoile.tikz}&\to&\input{phasebsphasebeta.tikz}\end{array}\end{equation}
\end{minipage}}
  \caption{Rewriting rules of PPRS. $\protect\minitikzfig[0.75]{convtp-phase-shift-etoile-s}$ denotes either $\protect\minitikzfig[0.75]{convtp-phase-shift-s}$ or $\protect\minitikzfig[0.75]{filcourt-s}$. The conditions on the angles are given in \cite{clement2022LOv}, note that for \cref{glissadeEulerscalaires,fusionEulerbsphasebs} they are the same as in \cref{axiomsLOPP} (with $\varphi_1$, $\varphi_2$ and $\alpha_2$ taken as being $0$ if missing).
   \label{PPRS}
  }
\end{figure}
First we can notice that Rule \eqref{fusionEulerbsphasebs} is exactly the same as \cref{Eulerbsphasebs} (up to \cref{phase0}). 

Rule \eqref{phasemod2pi} is derived from \cref{phase0} and \cref{phaseaddition}.

Rule \eqref{bsmod2pi} is derived from \cref{Eulerbsphasebs} with $\alpha_1=\alpha_2=0$ and $\alpha_3=\psi+2k\pi$.

Rule \eqref{fusionphaseshifts} is derived from \cref{phaseaddition}.

Rule \eqref{zerophaseshifts} is derived from \cref{phase0}.

Rule \eqref{zerobs} is derived from \cref{bs0}.

Rule \eqref{removebottomphase} is derived from \cref{phaseaddition}, \cref{phase0} and \cref{globalphasepropagationbs}.

Rule \eqref{passagepisur2} is derived from \cref{swapbspisur2} and \cref{phaseaddition}.

Rule \eqref{passagephasepi} is derived from \cref{Eulerbsphasebs} with $\alpha_1=0$, $\alpha_2=\varphi_0$ and $\alpha_3=\theta_0$.

Rule \eqref{soustractionpi} is derived from \cref{Eulerbsphasebs} with $\alpha_1=\alpha_2=0$ and $\alpha_3=\theta_4$.

Regarding Rule \eqref{glissadeEulerscalaires}, its LHS can be transformed as follows: 
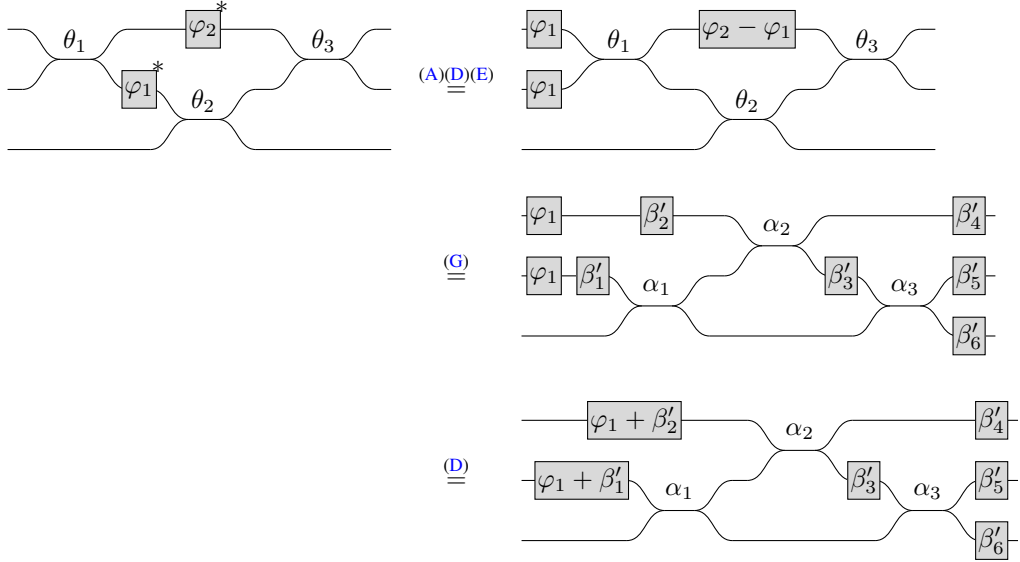
\begin{eqnarray*}\!\!\!\!\!\!\!\!\!\!\!\!\!\!\!\!\input{convtp-bsyangbaxterpointeenbas-etoiles.tikz}&\eqtroiseqref{phase0}{phaseaddition}{globalphasepropagationbs}&\input{convtp-bsyangbaxterpointeenbas-etoiles1.tikz}\\[0.5cm]
     &\eqeqref{Eulerscalaires}&\input{bsyangbaxterpointeenhaut1.tikz}\\[0.5cm]
     &\eqeqref{phaseaddition}&\input{bsyangbaxterpointeenhaut2.tikz}\\[0.5cm]
     \end{eqnarray*}

Note that the angles in the resulting circuit are not necessarily those of the RHS of Rule \eqref{glissadeEulerscalaires}.

However, one can show that 
it can be put in normal 
form using the rules of \cref{PPRS} except 
 
Rule \eqref{glissadeEulerscalaires}. 
As we have seen above that  each of these rules can be derived using equations of \cref{axiomsLOPP}, this shows that  Rule \eqref{glissadeEulerscalaires} can also be derived using the equations of \cref{axiomsLOPP}.

\subsection{Useful Definitions}\label{usefuldef}

\begin{definition}
Given $x\in\{0,1\}^k$, $y\in\{0,1\}^\ell$ and $G\in \{s(\psi),X,R_X(\theta),P(\varphi)\}$, we define
\[\bar\Lambda^x_y G\coloneqq\prod_{\begin{scriptarray}{c}\\[-1.5em]x'\in\{0,1\}^k\\[-0.2em]y'\in\{0,1\}^\ell\\[-0.2em]x'y'\neq xy\end{scriptarray}}\Lambda^{x'}_{y'}G\]
where the product denotes a sequential composition taken in an arbitrary order.

\end{definition}

\begin{definition}
Given $x\in\{0,1\}^k$, $y\in\{0,1\}^\ell$ and $z\in\{0,1\}^m$, we define 
\[\Lambda\tripleindice xyz\gcnot\coloneqq\Lambda^{x1y}_z X,\qquad \Lambda\tripleindice xyz\gnotc\coloneqq\Lambda^x_{y1z}X,\qquad \bar\Lambda\tripleindice xyz\gcnot\coloneqq\prod_{\begin{scriptarray}{c}\\[-1.5em]x'\in\{0,1\}^k\\[-0.2em]y'\in\{0,1\}^\ell\\[-0.2em]z'\in\{0,1\}^m\\[-0.2em]x'y'z'\neq xyz\end{scriptarray}}\Lambda^{x'1y'}_{z'}X\quad\text{ and }\quad\bar\Lambda\tripleindice xyz\gnotc\coloneqq\prod_{\begin{scriptarray}{c}\\[-1.5em]x'\in\{0,1\}^k\\[-0.2em]y'\in\{0,1\}^\ell\\[-0.2em]z'\in\{0,1\}^m\\[-0.2em]x'y'z'\neq xyz\end{scriptarray}}\Lambda^{x'}_{y'1z'}X.\]

\end{definition}

\subsection{Ancillary lemmas: Lemmas \ref{antisymmetriecontrolee} to \ref{decodagefilsdisjoints}}
\label{proof:usefullemmas}

\begin{lemma}\label{antisymmetriecontrolee}
\[\qc\vdash\input{mctrlXmctrlPmctrlX.tikz}\ =\ \input{mctrlPmoinsphictrlPphihaut.tikz}\]

\end{lemma}
\begin{proof}~
\begin{longtable}{RCL}
\input{mctrlXmctrlPmctrlX.tikz}&\overset{\text{\crefnosort{XX,prop:CCX,commctrl,prop:comb}}}{=}&\input{mctrlXmctrlPmctrlX1.tikz}\\\\
&\overset{\text{\crefnosort{commctrl,prop:CCX}}}{=}&\input{XmctrlPX.tikz}\\\\
&\overset{\text{\cref{antisymmetriesemicontrolee}}}{=}&\input{mctrlPmoinsphictrlPphihaut.tikz}
\end{longtable}
where $\vec 1$ denotes a list of appropriate length whose elements are all equal to $1$.
\end{proof}

\begin{lemma}\label{passagephasepicircuits}
\[\qczero\vdash\input{mctrlZmctrlRX.tikz}\ =\ \input{mctrlRXmoinsthetamctrlZ.tikz}\]

\end{lemma}
\begin{proof}~
\begin{longtable}{RCL}
\input{mctrlZmctrlRX.tikz}&\overset{\text{\crefnosort{zgate,prop:sum,commctrl,prop:comb}}}{=}&\input{mctrlZmctrlRX1.tikz}\\\\
&\overset{\text{\crefnosort{commctrl}}}{=}&\input{mctrlZmctrlRX2.tikz}\\\\
&
\eqeqref{ZctrlRX}&\input{mctrlZmctrlRX3.tikz}\\\\
&\overset{\text{\crefnosort{prop:comb,commctrl,prop:sum,zgate}}}{=}&\input{mctrlRXmoinsthetamctrlZ.tikz}
\end{longtable}
\end{proof}

\begin{lemma}\label{CRX2pi}
For any $x\in\{0,1\}^k$,
\[\qc\vdash\Lambda^{x}R_X(2\pi)=\Lambda^xs(\pi)\otimes\gid\]
\end{lemma}
\begin{proof}~
\begin{longtable}{RCL}
\Lambda^{x}R_X(2\pi)&\overset{\text{\eqref{XX}, \eqref{HH}, \crefnosort{prop:sum,lem:Lambda-P}}}{=}&\input{mctrlRX2pi1.tikz}\\\\
&\overset{\text{\crefnosort{prop:period,prop:sum,HH}}}{=}&\Lambda^xs(\pi)\otimes\gid
\end{longtable}

\end{proof}

\begin{lemma}\label{passagephasepihbcircuits}
\[\qc\vdash\input{lemmepassagepihb1.tikz}\ =\ \input{lemmepassagepihb6.tikz}\]
\end{lemma}
\begin{proof}~
\begin{longtable}{RCL}
\input{lemmepassagepihb1.tikz}&\overset{\text{\crefnosort{passagephasepicircuits,prop:period}}}{=}&\input{lemmepassagepihb2.tikz}\\\\
&\overset{\text{\crefnosort{prop:sum,prop:period}}}{=}&\input{lemmepassagepihb3.tikz}\\\\
&\overset{\text{\cref{CRX2pi}}}{=}&\input{lemmepassagepihb4.tikz}\\\\
&\overset{\text{\cref{prop:CP}}}{=}&\input{lemmepassagepihb5.tikz}\\\\
&\overset{\text{\crefnosort{prop:comb,commctrl,prop:sum}}}{=}&\input{lemmepassagepihb6.tikz}
\end{longtable}
\end{proof}

\begin{lemma}\label{decodagefilsdisjoints}
For any raw optical circuits $C_1:\ell_1\to\ell_1$ and $C_2:\ell_2\to\ell_2$, and any $k,\ell,n$ with $\ell\geq\ell_1$ and $k+\ell\leq 2^n$,
\[\qczero\vdash D_{k+\ell,n}(C_2)\circ D_{k,n}(C_1)=D_{k,n}(C_1)\circ D_{k+\ell,n}(C_2).\]
\end{lemma}

\begin{proof}
We proceed by structural induction on $C_1$ and $C_2$.
\begin{itemize}
\item If $C_1=C_1''\circ C_1'$, then
\[D_{k+\ell,n}(C_2)\circ D_{k,n}(C_1)= D_{k+\ell,n}(C_2)\circ (D_{k,n}(C_1'')\circ D_{k,n}(C_1'))\]
while
\[D_{k,n}(C_1)\circ D_{k+\ell,n}(C_2)= (D_{k,n}(C_1'')\circ D_{k,n}(C_1'))\circ D_{k+\ell,n}(C_2)\]
so the result follows by \cref{assoccomp} of quantum circuits and the induction hypothesis.
\item The case $C_2=C_2''\circ C_2'$ is similar to the previous one.
\item If $C_1=C_1'\otimes C_1''$ with $C_1':\ell_1'\to\ell_1'$, then
\[D_{k+\ell,n}(C_2)\circ D_{k,n}(C_1)= D_{k+\ell,n}(C_2)\circ (D_{k+\ell_1',n}(C_1'')\circ D_{k,n}(C_1'))\]
while
\[D_{k,n}(C_1)\circ D_{k+\ell,n}(C_2)= (D_{k+\ell_1',n}(C_1'')\circ D_{k,n}(C_1'))\circ D_{k+\ell,n}(C_2)\]
so the result follows by \cref{assoccomp} of quantum circuits and the induction hypothesis.
\item The case $C_2=C_2'\otimes C_2''$ is similar to the previous one.
\item If $C_1$ or $C_2$ is $\begin{tikzpicture}[scale=0.2]
	\begin{pgfonlayer}{nodelayer}
		\node [style=sdiagrammevide] (0) at (0, 0) {};
	\end{pgfonlayer}
\end{tikzpicture}
$ or $\gid$, then the results follows from \cref{idneutre} of quantum circuits.
\item If $C_1,C_2\in\{\begin{tikzpicture}[scale=0.3]
	\begin{pgfonlayer}{nodelayer}
		\node [style=none] (5) at (-1.6, 0) {};
		\node [style=none] (9) at (1.6, 0) {};
		\node [style=PolRot,scale=0.75] (11) at (0, 0) {\small $\varphi$};
	\end{pgfonlayer}
	\begin{pgfonlayer}{edgelayer}
		\draw (5.center) to (11);
		\draw (11) to (9.center);
	\end{pgfonlayer}
\end{tikzpicture}
,\input{bs-xs.tikz},\input{swap-xs.tikz}\}$, then $D_{k,n}(C_1)=\Lambda^{G_n(k)}s(\varphi)$, $\Lambda^{x_{k,n}}_{y_{k,n}}R_X(-2\theta)$ or $\Lambda^{x_{k,n}}_{y_{k,n}}X$ and $D_{k+\ell,n}(C_2)=\Lambda^{G_n(k+\ell)}s(\varphi)$, $\Lambda^{x_{k+\ell,n}}_{y_{k+\ell,n}}R_X(-2\theta)$ or $\Lambda^{x_{k+\ell,n}}_{y_{k+\ell,n}}X$. 
Using the definitions of $G_n(k)$, $x_{k,n}$ and $y_{k,n}$, it is easy to check that in any case, $D_{k,n}(C_1)$ and $D_{k+\ell,n}(C_2)$ satisfy the premises of either \cref{commctrl} or \ref{commctrlpasenface} and therefore commute.
\end{itemize}
\end{proof}

\subsection{Proof of Lemma \ref{DE}}\label{lemmasforbasecaseDE}

For the sake of clarity, the proofs are given separately in \cref{preuveDupsilonLambda,preuvesigmatoswaps,preuvebasecaseDE}.

\begin{lemma}\label{DupsilonLambda}
For any $N\geq1$, $i\in\{0,...,N-1\}$, $b\in\{0,...,2^i-1\}$ and $a\in\{0,...,2^{N-i-1}-1\}$,
\[\textup{QC}\vdash D(\upsilon_{N,i,b,a})=\Lambda^{G_i(b)}_{G_{N-i-1}(2^{N-i-1}-a-1)}X\]
where $\upsilon_{N,i,b,a}$ is defined in \cref{defsigma}, and given $n\in\mathbb N$ and $k\in\{0,...,2^n-1\}$, $G_n(k)\in\{0,1\}^n$ is the $n$-bit Gray code of $k$, defined in \cref{defgraycode}. Note that $G_{N-i-1}(2^{N-i-1}-a-1)$ differs from $G_{N-i-1}(a)$ by only the first bit.
\end{lemma}

\begin{lemma}\label{sigmatoswaps}
For any $k,\ell,n\in\mathbb N$,
\[\textup{QC}\vdash D(\sigma_{k,n,\ell})=id_k\otimes\sigma_{n,\ell}.\]
where $\sigma_{0,0}\coloneqq\emptyc$ and $\sigma_{n,\ell}\coloneqq\sigma_{n+\ell-1}^\ell$, where $\sigma_{n+\ell-1}$ is defined in \cref{fig:axiom}.
\end{lemma}

\begin{lemma}\label{basecaseDE}
For any $g\in%
\{\emptyc,\gid,s(\varphi),\gh,\gp,\gcnot,\gswap\}$,
\[\textup{QC}\vdash D(E_{k,\ell}(g))=id_k\otimes g\otimes id_\ell.\]
\end{lemma}

\subsubsection{Proof of Lemma \ref{DupsilonLambda}}\label{preuveDupsilonLambda}
We proceed by induction on $a$.

It follows from the definition of $D$ that
\[D(\upsilon_{N,i,b,0})\eqdef D\left(\input{swapib2pNi1-PHOL.tikz}\right)=\Lambda^{G_i(b)}_{G_{N-i-1}(2^{N-i-1}-1)}X.\]

Assuming for some $a\in\{1,...,2^{N-i-1}-1\}$ that $\qczero\vdash D(\upsilon_{N,i,b,a-1})=\Lambda^{G_i(b)}_{G_{N-i-1}(2^{N-i-1}-a)}X$, by definition of $\upsilon_{N,i,b,a}$, one has\footnote{Note that this product of five raw circuits should be written with more parentheses since the composition is not associative. We have omitted these parentheses by abuse of language in order to lighten the notations. In the following, we will similarly omit the associativity parentheses whenever this does not create ambiguity.}
\[\qczero\vdash D(\upsilon_{N,i,b,a})=D(s_{-a})\circ D(s_{+a})\circ\left(\Lambda^{G_i(b)}_{G_{N-i-1}(2^{N-i-1}-a)}X\right)\circ D(s_{+a})\circ D(s_{-a})\medskip\]
where $s_{+a}=\input{swapib2pNi1plusa-PHOL.tikz}$\quad and\quad $s_{-a}=\input{swapib2pNi1moinsa-PHOL.tikz}$\bigskip.

Due to the properties of Gray codes, $G_{N-i-1}(2^{N-i-1}-a-1)$ differs from $G_{N-i-1}(2^{N-i-1}-a)$ by only one bit. That is, there exist $k,\ell\geq0$ with $k+\ell=N-i-2$, $x\in\{0,1\}^k$, $y\in\{0,1\}^\ell$ and $\alpha\in\{0,1\}$, such that
\[G_{N-i-1}(2^{N-i-1}-a-1)=x\alpha y\quad\text{and}\quad G_{N-i-1}(2^{N-i-1}-a)=x\bar\alpha y\]
where $\bar\alpha\coloneqq1-\alpha$.

Additionally, $G_{N-i}(2^{N-i-1}-a-1)$ differs from $G_{N-i}(2^{N-i-1}+a)$ by only the first bit, and $G_{N-i}(2^{N-i-1}-a)$ also differs from $G_{N-i}(2^{N-i-1}+a-1)$ by only the first bit. Therefore, there exists $\beta\in\{0,1\}$ such that
\[\begin{array}{ll}G_{N-i}(2^{N-i-1}-a-1)=\beta x\alpha y,&G_{N-i}(2^{N-i-1}+a)=\bar\beta x\alpha y,\\[0.8em]
G_{N-i}(2^{N-i-1}-a)=\beta x\bar\alpha y\quad\text{and}&G_{N-i}(2^{N-i-1}+a-1)=\bar\beta x \bar\alpha y.\end{array}\]

It follows from the definition of $D$ that $D(s_{-a})=\Lambda^{G_i(b).\beta x}_yX$ and $D(s_{+a})=\Lambda^{G_i(b).\bar\beta x}_yX$. Hence, by \cref{commctrl,cor:swap,prop:comb}, $\qczero\vdash D(s_{-a})\circ D(s_{+a})=D(s_{+a})\circ D(s_{-a})=
\scalebox{0.69}{$\input{LambdaGibxyX.tikz}$}=(\sigma_{1,i}\otimes id_{N-i-1})\circ\left(\gid\otimes\Lambda^{G_i(b) x}_yX\right)\circ(\sigma_{i,1}\otimes id_{N-i-1})$, so that
\[\qc\vdash D(\upsilon_{N,i,b,a})=(\sigma_{1,i}\otimes id_{N-i-1})\circ\left(\gid\otimes\Lambda^{G_i(b) x}_yX\right)\circ\left(\Lambda^{\epsilon}_{G_i(b)x\bar\alpha y}X\right)\circ \left(\gid\otimes\Lambda^{G_i(b) x}_yX\right)\circ(\sigma_{i,1}\otimes id_{N-i-1})\]
with
\begin{longtable}{CL}
\multicolumn{1}{R}{\qc\vdash}&\left(\gid\otimes\Lambda^{G_i(b) x}_yX\right)\circ\left(\Lambda^{\epsilon}_{G_i(b)x\bar\alpha y}X\right)\circ \left(\gid\otimes\Lambda^{G_i(b) x}_yX\right)\\\\
\overset{\text{\crefnosort{prop:CCX,commctrl,prop:comb}}}{=}&(id_{N-1}\otimes X)\circ\left(\gid\otimes\bar\Lambda^{G_i(b) x}_yX\right)\circ\left(\Lambda^{\epsilon}_{G_i(b)x\bar\alpha y}X\right)\circ \left(\gid\otimes\bar\Lambda^{G_i(b) x}_yX\right)\circ(id_{N-1}\otimes X)\\\\
\overset{\text{\cref{commctrl,prop:comb}}}{=}&(id_{N-1}\otimes X)\circ\left(\gid\otimes\bar\Lambda^{G_i(b) x}_yX\right)\circ \left(\gid\otimes\bar\Lambda^{G_i(b) x}_yX\right)\circ\left(\Lambda^{\epsilon}_{G_i(b)x\bar\alpha y}X\right)\circ(id_{N-1}\otimes X)\\\\
\overset{\text{\cref{prop:CCX,commctrl}}}{=}&(id_{N-1}\otimes X)\circ\left(\Lambda^{\epsilon}_{G_i(b)x\bar\alpha y}X\right)\circ(id_{N-1}\otimes X)
\end{longtable}

In other words,
\[\qc\vdash D(\upsilon_{N,i,b,a})=\left(id_{i+k+1}\otimes X\otimes id_{\ell}\right)\circ\left(\Lambda^{G_i(b)}_{x\bar\alpha y}X\right)\circ \left(id_{i+k+1}\otimes X\otimes id_{\ell}\right).\]

By definition of $\Lambda^{G_i(b)}_{x\bar\alpha y}X$ and \cref{XX}, this implies that \[\qc\vdash D(\upsilon_{N,i,b,a})=\Lambda^{G_i(b)}_{x\alpha y}X\]
which, since $x\alpha y=G_{N-i-1}(2^{N-i-1}-a-1)$, is the desired property.

\begin{remark}
By defining $\upsilon_{N,i,b,a}$ in a less natural way using not only $\gid$ and $\gswap$ but also $\begin{tikzpicture}[scale=0.3]
	\begin{pgfonlayer}{nodelayer}
		\node [style=none] (5) at (-1.6, 0) {};
		\node [style=none] (9) at (1.6, 0) {};
		\node [style=PolRot,scale=0.75] (11) at (0, 0) {\small $\varphi$};
	\end{pgfonlayer}
	\begin{pgfonlayer}{edgelayer}
		\draw (5.center) to (11);
		\draw (11) to (9.center);
	\end{pgfonlayer}
\end{tikzpicture}
$ and $\input{bs-xs.tikz}$, one could avoid using \cref{prop:CCX} and get the stronger result that $\qczero\vdash D(\upsilon_{N,i,b,a})=\Lambda^{G_i(b)}_{G_{N-i-1}(2^{N-i-1}-a-1)}X$, which would in turn imply that the equalities of \cref{sigmatoswaps,basecaseDE} can also be taken modulo $\qczero$ instead of $\textup{QC}$.
\end{remark}

\subsubsection{Proof of Lemma \ref{sigmatoswaps}}\label{preuvesigmatoswaps}
First, if $n=1$, by definition (see \cref{defdecoding,defsigma}), one has
\[D(\sigma_{k,1,\ell})=\prod_{j=k+1}^{k+\ell}P_jQ_jP_j\]
where $M\coloneqq k+\ell+1$, $\displaystyle P_j\coloneqq\prod_{\begin{scriptarray}{c}\\[-1.5em]b=0\\[-0.4em]b\bmod 4\in\{1,2\}\end{scriptarray}}^{2^{j}-1}\hspace{-1.5em}\prod_{a=0}^{2^{M-j-1}-1}\hspace{-1em}D(\upsilon_{M,j,b,a})$ and $\displaystyle Q_j\coloneqq\prod_{b=0}^{2^{j-1}-1}\ \ \prod_{a=0}^{2^{M-j-3}-1}\hspace{-1em}D(\upsilon_{M,j-1,b,a})$.

By \cref{DupsilonLambda}, for all $j$,
\[\textup{QC}\vdash P_j=\prod_{\begin{scriptarray}{c}\\[-1.5em]b=0\\[-0.4em]b\bmod 4\in\{1,2\}\end{scriptarray}}^{2^{j}-1}\hspace{-1.5em}\prod_{a=0}^{2^{M-j-1}-1}\hspace{-1em}\Lambda^{G_j(b)}_{G_{M-j-1}(2^{M-j-1}-a-1)}X\]
It is easy to check that when $a$ goes %
from $0$ to $2^{M-j-1}-1$, %
$G_{M-j-1}(2^{M-j-1}-a-1)$ takes all possible values in $\{0,1\}^{M-j-1}$, once each, and that when $b$ takes all possible values between $0$ and $2^j-1$ that are congruent to $1$ or $2$ modulo $4$, $G_j(b)$ takes, once each, all values in $\{0,1\}^j$ in which the last bit has value $1$. Hence, it follows from \crefnosort{prop:comb,commctrl,ctrlXCNot} that
\[\textup{QC}\vdash P_j=id_{j-1}\otimes \gcnot\otimes id_{M-j-1}.\]

Again by \cref{DupsilonLambda}, for all $j$,
\[\textup{QC}\vdash Q_j=\prod_{b=0}^{2^{j-1}-1}\ \ \prod_{a=0}^{2^{M-j-3}-1}\hspace{-1em}\Lambda^{G_{j-1}(b)}_{G_{M-j}(2^{M-j}-a-1)}X\]
Similarly, it is easy to check that when $b$ goes from $0$ to $2^{j-1}-1$, $G_{j-1}(b)$ takes all values in $\{0,1\}^{j-1}$, once each, and that when $a$ goes from $0$ to $2^{M-j-3}$, $G_{M-j}(2^{M-j}-a-1)$ takes, once each, all values in $\{0,1\}^{M-j}$ in which the first bit has value $1$. Hence, it follows from \crefnosort{prop:comb,commctrl,ctrlXCNot} that
\[\textup{QC}\vdash Q_j=id_{j-1}\otimes \gnotc\otimes id_{M-j-1}.\]

Thus,
\[\textup{QC}\vdash D(\sigma_{k,1,\ell})=\prod_{j=k+1}^{k+\ell}id_{j-1}\otimes\scalebox{0.69}{$\input{tripleCNot-s.tikz}$}\otimes id_{M-j-1}.\]
By \cref{tripleCNotswap}, this implies that
\begin{equation}\label{sigmatoswapscase1}\qc\vdash D(\sigma_{k,1,\ell})=\prod_{j=k+1}^{k+\ell}id_{j-1}\otimes \gswap\otimes id_{M-j-1}\equiv id_k\otimes\sigma_{1,\ell}.\end{equation}

Finally, if $n>1$, then
\begin{longtable}{RCL}
D(\sigma_{k,n,\ell})&\eqdef&D(\sigma_{k,1,\ell+n-1}^n)\\[0.5em]
&\eqdef&D(\sigma_{k,1,\ell+n-1})^n\\[0.5em]
&\eqeqref{sigmatoswapscase1}&(id_k\otimes\sigma_{1,\ell+n-1})^n\\[0.5em]
&\equiv&id_k\otimes\sigma_{n,\ell}.
\end{longtable}

\subsubsection{Proof of Lemma \ref{basecaseDE}}\label{preuvebasecaseDE}
If $g=\begin{tikzpicture}[scale=0.2]
	\begin{pgfonlayer}{nodelayer}
		\node [style=sdiagrammevide] (0) at (0, 0) {};
	\end{pgfonlayer}
\end{tikzpicture}
$ or $\begin{tikzpicture}
	\begin{pgfonlayer}{nodelayer}
		\node [style=none] (0) at (-0.8, 0) {};
		\node [style=none] (1) at (0.8, 0) {};
	\end{pgfonlayer}
	\begin{pgfonlayer}{edgelayer}
		\draw [style=noire] (0.center) to (1.center);
	\end{pgfonlayer}
\end{tikzpicture}$ then the result follows directly from the definitions.

If $g=s(\varphi)$, then it follows from the definitions of $E_{k,
\ell}$ and $D$ that
\[D(E_{k,
\ell}(s(\varphi)))=\prod_{x\in\{0,1\}^{k+\ell}}\Lambda^xs(\varphi)\]
where we use the notation $\prod_{x\in\{0,1\}^{k+\ell}}$ to denote the product without specifying the order of the factors. By \cref{prop:comb,commctrl}, this implies that
\[\qc\vdash D(E_{k,
\ell}(s(\varphi)))=
id_{k+\ell}\otimes s(\varphi)\]
which is equal to $id_k\otimes s(\varphi)\otimes id_{\ell}$ by the topological rules of 
quantum circuits\bigskip.

If $g=
\gp$, then it follows from the definitions that if $k=\ell=0$,
\[D(E_{0,0}(\gp))=D(\input{Z-PHOL.tikz})\equiv\Lambda^1s(\varphi)=P(\varphi).\]
and if $(k,\ell)\neq(0,0)$,
\[D(E_{k,\ell}(P(\varphi))=D(\sigma_{k,\ell,1})%
\circ D\left(\left(\input{Z-PHOL-doublesym.tikz}\right)^{\otimes {2^{k+\ell-1}}}\right)%
\circ D(\sigma_{k,1,\ell})\]
with
\[D\left(\left(\input{Z-PHOL-doublesym.tikz}\right)^{\otimes {2^{k+\ell-1}}}\right)=\prod_{x\in\{0,1\}^{k+\ell}}\Lambda^{x1}s(\varphi)=\prod_{x\in\{0,1\}^{k+\ell}}\Lambda^xP(\varphi).\]
By \cref{prop:comb,commctrl}, this product is equal modulo $\qczero$ to $id_{k+\ell}\otimes P(\varphi)$. Then, \cref{sigmatoswaps} together with topological rules 
of quantum circuits gives us the result\bigskip.

If $g=\gh$, then it follows from the definitions that if $k=\ell=0$,
\begin{longtable}{R@{\ }L}D(E_{0,0}(\gh))=D(\input{H-LOPP-xs.tikz})\equiv&\Lambda^1s(-\frac\pi2)\circ\Lambda^\epsilon_\epsilon R_X(-\frac\pi2)\circ\Lambda^1s(-\frac\pi2)\\\\
=&\input{EulerHmoins.tikz}\\\\
\eqeqref{EulerHmoins}&\begin{tikzpicture}
	\begin{pgfonlayer}{nodelayer}
		\node [style=none] (0) at (-1, 0) {};
		\node [style=none] (1) at (1, 0) {};
		\node [style=boite2] (2) at (0, 0) {$H$};
	\end{pgfonlayer}
	\begin{pgfonlayer}{edgelayer}
		\draw (0.center) to (1.center);
	\end{pgfonlayer}
\end{tikzpicture}
\end{longtable}
and if $(k,\ell)\neq(0,0)$,
\[D(E_{k,\ell}(\gh)=D(\sigma_{k,\ell,1})\circ D\left(\left(\input{H-LOPP-doublesym-xs.tikz}\right)^{\otimes {2^{k+\ell-1}}}\right)\circ D(\sigma_{k,1,\ell})\]
with
\[D\left(\left(\input{H-LOPP-doublesym-xs.tikz}\right)^{\otimes {2^{k+\ell-1}}}\right)\equiv\prod_{x\in\{0,1\}^{k+\ell}}\left(\left(\prod_{a\in\{0,1\}}\Lambda^{xa1}s(-\frac\pi2)\right)\circ\left(\prod_{a\in\{0,1\}}\Lambda^{xa}R_X\bigl(-\frac\pi2\bigr)\right)\circ\left(\prod_{a\in\{0,1\}}\Lambda^{xa1}s(-\frac\pi2)\right)\right).\]
By \cref{prop:comb,commctrl}, this product is equal modulo $\qczero$ to $id_{k+\ell}\otimes \scalebox{0.69}{$\input{EulerHmoins.tikz}$}$, which by \cref{EulerHmoins} is equal modulo $\qczero$ to \gh. Then, \cref{sigmatoswaps} together with topological rules of quantum circuits gives us the result\bigskip.

If $g=\gcnot$, then it follows from the definitions  that if $k=\ell=0$,
\[D(E_{0,0}(\gcnot))=D\left(\input{CNot-PHOL.tikz}\right)\equiv\Lambda^1_\epsilon X\]
which is equal to $\gcnot$ modulo $\qczero$ by \cref{ctrlXCNot};

and if $(k,\ell)\neq(0,0)$,
\[D(E_{k,\ell}(\gcnot)=D(\sigma_{k,\ell,2})
\circ D\left(\left(\input{CNot-PHOL-doublesym.tikz}\right)^{\otimes {2^{k+\ell-1}}}\right)
\circ D(\sigma_{k,2,\ell})\]
with
\[D\left(\left(\input{CNot-PHOL-doublesym.tikz}\right)^{\otimes {2^{k+\ell-1}}}\right)\equiv\prod_{x\in\{0,1\}^{k+\ell}}\Lambda^{x1}X.\]
By \cref{prop:comb,commctrl}, this product is equal modulo $\qczero$ to $id_{k+\ell}\otimes \Lambda^1X$, which by \cref{ctrlXCNot} is equal modulo $\qczero$ to $id_{k+\ell}\otimes \gcnot$.
Then, \cref{sigmatoswaps} together with topological rules 
of quantum circuits gives us the result\bigskip.

If $g=\gswap$, then it follows from the definitions that
\[D(E_{k,2,\ell}(\gswap)=D(\sigma_{k,\ell,2})
\circ D(\sigma_{k+\ell,1,1})%
\circ D(\sigma_{k,2,\ell})\]
By \cref{sigmatoswaps}, this is equal modulo $\textup{QC}$ to $(id_k\otimes\sigma_{\ell,2})\circ(id_{k+\ell}\otimes \gswap)\otimes(id_k\otimes\sigma_{\ell,2})$, which by the topological rules of 
quantum circuits, is equal to $id_k\otimes \gswap\otimes id_\ell$.

\subsection{Proof of Lemma \ref{decodingtoporules}}\label{preuvedecodingtoporules}

\begin{definition}[Context]
A $N$-mode raw context $\C[\cdot]_i$ with $i\in\mathbb N$ is inductively defined as follows:
\begin{itemize}
\item $[\cdot]_i$ is a $i$-mode raw context
\item if $\C[\cdot]_i$ is a $N$-mode raw context and $C$ is a $M$-mode raw optical circuit then $\C[\cdot]_i\otimes C$ and $C\otimes \C[\cdot]_i$ are $N{+}M$-mode raw contexts
\item if $\C[\cdot]_i$ is a $N$-mode raw context and $C$ is a $N$-mode raw optical circuit then $\C[\cdot]_i\circ C$ and $C\circ \C[\cdot]_i$ are $N$-mode raw contexts.
\end{itemize}
\end{definition}

\begin{definition}[Substitution]
Given a $N$-mode raw context $\C[\cdot]_i$ and a $i$-mode raw circuit $C$, we define the substituted circuit $\C[C]$ as the $N$-mode raw circuit obtained by replacing the hole $[\cdot]_i$ by $C$ in $\C[\cdot]_i$. 
\end{definition}

To prove \cref{decodingtoporules}, it suffices to prove that for each rule of \cref{fig:axiom}, of the form $C_1=C_2$ with $C_1,C_2\in\LOPPbarebf(i,i)$, and any $2^n$-mode raw context $\C[\cdot]_i$, one has $\qc\vdash D(\C[C_1])=D(\C[C_2])$. For this purpose, we prove a slightly more general result, namely that for any $k,n$ and any $\ell$-mode raw context $\C[\cdot]_i$ with $k+\ell\leq 2^n$, one has $\qc\vdash D_{k,n}(\C[C_1])=D(\C[C_2])$. We proceed by induction on $\C[\cdot]_i$:
\begin{itemize}
\item If $\C[\cdot]_i=C\circ \C'[\cdot]_i$, then $D_{k,n}(\C[C_1])=D_{k,n}(C)\circ D_{k,n}(\C'[C_1])$ and $D_{k,n}(\C[C_2])=D_{k,n}(C)\circ D_{k,n}(\C'[C_2])$, so the result follows by induction hypothesis. The case $\C[\cdot]_i=\C'[\cdot]_i\circ C$ is similar.
\item If $\C[\cdot]_i=C\otimes \C'[\cdot]_i$ with $C:\ell_1\to\ell_1$, then $D_{k,n}(\C[C_1])=D_{k+\ell_1,n}(\C'[C_1])\circ D_{k,n}(C)$ and $D_{k,n}(\C[C_2])=D_{k+\ell_1,n}(\C'[C_2])\circ D_{k,n}(C)$, so the result follows by induction hypothesis. The case $\C[\cdot]_i=\C'[\cdot]_i\otimes C$ is similar.
\end{itemize}
It remains to prove for each rule of \cref{fig:axiom}, of the form $C_1=C_2$ with $C_1,C_2\in\LOPPbarebf(i,i)$, that for any $k,n$ with $k+i\leq 2^n$, one has $\qc\vdash D_{k,n}(C_1)=D_{k,n}(C_2)$.

For \cref{assoccomp}, for any $C_1,C_2,C_3:\ell\to \ell$,
\[D_{k,n}((C_3\circ C_2)\circ C_1)=(D_{k,n}(C_3)\circ D_{k,n}(C_2))\circ D_{k,n}(C_1)\]
and
\[D_{k,n}(C_3\circ (C_2\circ C_1))=D_{k,n}(C_3)\circ (D_{k,n}(C_2)\circ D_{k,n}(C_1)).\]
Both are equal according to \cref{assoccomp} of quantum circuits.

For \cref{assoctens}, for any optical circuits $C_1:\ell_1\to\ell_1$, $C_2:\ell_2\to\ell_2$ and $C_3:\ell_3\to\ell_3$,
\[D_{k,n}((C_1\otimes C_2)\otimes C_3)=D_{k+\ell_1+\ell_2,n}(C_3)\circ (D_{k+\ell_1,n}(C_2)\circ D_{k,n}(C_1))\]
and
\[D_{k,n}(C_1\otimes (C_2\otimes C_3))=(D_{k+\ell_1+\ell_2,n}(C_3)\circ D_{k+\ell_1,n}(C_2))\circ D_{k,n}(C_1).\]
Again, both are equal according to \cref{assoccomp} of quantum circuits.

For \cref{idneutre}, for any $\ell$-mode optical circuit $C$, by definition of $id_\ell$ and $D_{k,n}$,
\[D_{k,n}(id_\ell\circ C)=(id_n\circ(id_n\circ(\cdots\circ(id_n\circ id_n))\cdots))\circ D_{k,n}(C)\]
with $\ell+1$ occurences of $id_n$ in the right-hand side. This is equal to $D_{k,n}(C)$ according to \cref{idneutre} of quantum circuits. Similarly, $D_{k,n}(C\circ id_\ell)\equiv D_{k,n}(C)$.

For \cref{videneutre}, for any $\ell$-mode optical circuit $C$,
\[D_{k,n}(\begin{tikzpicture}[scale=0.2]
	\begin{pgfonlayer}{nodelayer}
		\node [style=sdiagrammevide] (0) at (0, 0) {};
	\end{pgfonlayer}
\end{tikzpicture}
\otimes C)=D_{k,n}(C)\circ id_\ell\]
which is equal to $D_{k,n}(C)$ according to \cref{idneutre} of quantum circuits. Similarly, $D_{k,n}(C\otimes \begin{tikzpicture}[scale=0.2]
	\begin{pgfonlayer}{nodelayer}
		\node [style=sdiagrammevide] (0) at (0, 0) {};
	\end{pgfonlayer}
\end{tikzpicture}
)\equiv D_{k,n}(C)$.

For \cref{mixedprod}, for any optical circuits $C_1,C_2:\ell\to\ell$ and $C_3,C_4:m\to m$,
\[D_{k,n}((C_2\circ C_1)\otimes (C_4\circ C_3))=(D_{k+\ell,n}(C_4)\circ D_{k+\ell,n}(C_3))\circ(D_{k,n}(C_2)\circ D_{k,n}(C_1))\]
and
\[D_{k,n}((C_2\otimes C_4)\circ (C_1\otimes C_3))=(D_{k+\ell,n}(C_4)\circ D_{k,n}(C_2))\circ(D_{k+\ell,n}(C_3)\circ D_{k,n}(C_1)).\]
The result follows from \cref{assoccomp} of quantum circuits and \cref{decodagefilsdisjoints}.

For \cref{doubleswap}, one has
\[D_{k,n}(\gswap\circ\gswap)=\Lambda^{x_{k,n}}_{y_{k,n}}X\circ\Lambda^{x_{k,n}}_{y_{k,n}}X\]
which by \cref{prop:CCX}, implies that%
\[\qc\vdash D_{k,n}(\gswap\circ\gswap)=id_n.\]
On the other hand,
\[D_{k,n}(\begin{tikzpicture}
	\begin{pgfonlayer}{nodelayer}
		\node [style=none] (0) at (-0.8, 0) {};
		\node [style=none] (1) at (0.8, 0) {};
	\end{pgfonlayer}
	\begin{pgfonlayer}{edgelayer}
		\draw [style=noire] (0.center) to (1.center);
	\end{pgfonlayer}
\end{tikzpicture}\otimes\begin{tikzpicture}
	\begin{pgfonlayer}{nodelayer}
		\node [style=none] (0) at (-0.8, 0) {};
		\node [style=none] (1) at (0.8, 0) {};
	\end{pgfonlayer}
	\begin{pgfonlayer}{edgelayer}
		\draw [style=noire] (0.center) to (1.center);
	\end{pgfonlayer}
\end{tikzpicture})=id_n\circ id_n\equiv id_n.\]

For \cref{naturaliteswap}, we proceed by induction on $C$.
\begin{itemize}
\item If $C=C_1\circ C_2$, then $\sigma_k\circ((C_1\circ C_2)\otimes\gid)\equiv(\sigma_k\circ(C_1\otimes\gid))\circ(C_2\otimes\gid)$, and the derivation of the equivalence does not use \cref{naturaliteswap}. Hence it follows from the paragraphs above that
\[\qc\vdash D_{k,n}(\sigma_k\circ((C_1\circ C_2)\otimes\gid))=D_{k,n}((\sigma_k\circ(C_1\otimes\gid))\circ(C_2\otimes\gid)).\]
It follows similarly from those paragraphs that
\[\qc\vdash D_{k,n}((\gid\otimes(C_1\circ C_2))\circ\sigma_k)=D_{k,n}((C_1\otimes\gid)\circ((C_2\otimes\gid)\circ\sigma_k)).\]
The equality modulo $\qc$ of the two right-hand sides follows from the induction hypothesis, together with the compatibility of $D_{k,n}$ with \cref{assoccomp} modulo $\qc$, which is proved above.

\item If $C=C_1\otimes C_2$ with $C_1:\ell_1\to\ell_1$ and $C_2:\ell_2\to\ell_2$, then
\[\sigma_k\circ((C_1\otimes C_2)\otimes\gid)\equiv((\sigma_{\ell_1}\circ(C_1\otimes\gid))\otimes id_{\ell_2})\circ(id_{\ell_1}\otimes(\sigma_{\ell_2}\circ(C_2\otimes\gid)))\]
and the derivation of the equivalence does not use \cref{naturaliteswap}, so that by the paragraphs above (together with \cref{idneutre} of quantum circuits),
\[\qc\vdash D_{k,n}(\sigma_k\circ((C_1\otimes C_2)\otimes\gid))=D_{k,n}(\sigma_{\ell_1}\circ(C_1\otimes\gid))\circ D_{k+\ell_1}(\sigma_{\ell_2}\circ(C_2\otimes\gid)).\]
The result follows by applying a similar transformation to the right-hand side of \cref{naturaliteswap} and applying the induction hypothesis.
\item If $C=\emptyc\text{ or }\gid$, then the result follows from \cref{idneutre,videneutre} of quantum circuits.
\item If $C=\begin{tikzpicture}[scale=0.3]
	\begin{pgfonlayer}{nodelayer}
		\node [style=none] (5) at (-1.6, 0) {};
		\node [style=none] (9) at (1.6, 0) {};
		\node [style=PolRot,scale=0.75] (11) at (0, 0) {\small $\varphi$};
	\end{pgfonlayer}
	\begin{pgfonlayer}{edgelayer}
		\draw (5.center) to (11);
		\draw (11) to (9.center);
	\end{pgfonlayer}
\end{tikzpicture}
$, let us write $G_n(k)$ as $xay$ with $a\in\{0,1\}$ and $y=\epsilon$ if $k$ is even or $y=1.0^q$ for some $q$ if $k$ is odd. Then by definition of $D_{k,n}$ and \cref{phasemobile}, if $a=1$ then
\[\qc\vdash D_{k,n}(\sigma_1\circ(\begin{tikzpicture}[scale=0.3]
	\begin{pgfonlayer}{nodelayer}
		\node [style=none] (5) at (-1.6, 0) {};
		\node [style=none] (9) at (1.6, 0) {};
		\node [style=PolRot,scale=0.75] (11) at (0, 0) {\small $\varphi$};
	\end{pgfonlayer}
	\begin{pgfonlayer}{edgelayer}
		\draw (5.center) to (11);
		\draw (11) to (9.center);
	\end{pgfonlayer}
\end{tikzpicture}
\otimes\gid))=\Lambda^x_y X\circ\Lambda^x_y P(\varphi)\]
and
\[\qc\vdash D_{k,n}((\gid\otimes\begin{tikzpicture}[scale=0.3]
	\begin{pgfonlayer}{nodelayer}
		\node [style=none] (5) at (-1.6, 0) {};
		\node [style=none] (9) at (1.6, 0) {};
		\node [style=PolRot,scale=0.75] (11) at (0, 0) {\small $\varphi$};
	\end{pgfonlayer}
	\begin{pgfonlayer}{edgelayer}
		\draw (5.center) to (11);
		\draw (11) to (9.center);
	\end{pgfonlayer}
\end{tikzpicture}
)\circ\sigma_1)=(id_{|x|}\otimes X\otimes id_{|y|})\circ\Lambda^x_y P(\varphi)(id_{|x|}\otimes X\otimes id_{|y|})\circ\Lambda^x_y X.\]
By \cref{prop:CCX,commctrl,prop:comb}, the following equalities are true modulo $\qc$:
\begin{longtable}{RCL}
\Lambda^x_y X\circ\Lambda^x_y P(\varphi)&=&(id_{|x|}\otimes X\otimes id_{|y|})\circ\bar\Lambda^x_y X\circ\Lambda^x_y P(\varphi)\\
&=&(id_{|x|}\otimes X\otimes id_{|y|})\circ\Lambda^x_y P(\varphi)\circ\bar\Lambda^x_y X\\
&=&(id_{|x|}\otimes X\otimes id_{|y|})\circ\Lambda^x_y P(\varphi)\circ(id_{|x|}\otimes X\otimes id_{|y|})\circ\Lambda^x_y X
\end{longtable}
which gives us the result. The case $a=0$ is similar.
\item If $C=\input{bs-xs.tikz}$, by the properties of the Gray code, exactly one bit differs between $G_n(k)$ and $G_n(k+1)$, as well as between $G_n(k+1)$ and $G_n(k+2)$, and in exactly one of the two cases this is the last bit that differs (namely between $G_n(k)$ and $G_n(k+1)$ if $k$ is even, and between $G_n(k+1)$ and $G_n(k+2)$ if $k$ is odd). Hence we can write $G_n(k)$ as $xayb$ with $a,b\in\{0,1\}$, in such a way that $G_n(k+2)=x\bar ay\bar b$ and $G_n(k+1)=xay\bar b\text{ or }x\bar ayb$ depending on the parity of $k$. We treat the case where $k$ is even, the case with $k$ odd is similar. Then
\[D_{k,n}(\sigma_2\circ(\input{bs-xs.tikz}\otimes\gid))\equiv\Lambda^{xay}X\circ\Lambda^x_{y\bar b}X\circ\Lambda^{xay}R_X(-2\theta)\]
and
\[D_{k,n}((\gid\otimes\input{bs-xs.tikz})\circ\sigma_2)\equiv\Lambda^x_{y\bar b}R_X(-2\theta)\circ\Lambda^{xay}X\circ\Lambda^x_{y\bar b}X\]
so by \cref{symmetriesemicontrolee,XX}, it suffices to prove that for any $\theta$,
\[\qc\vdash\Lambda^{x1y}X\circ\Lambda^x_{y1}X\circ\Lambda^{x1y}R_X(\theta)=\Lambda^x_{y1}R_X(\theta)\circ\Lambda^{x1y}X\circ\Lambda^x_{y1}X.\]
To prove this, one has, modulo $\qc$ (together with the topological rules of quantum circuits):
\begin{longtable}{RCL}
\Lambda^{x1y}X\circ\Lambda^x_{y1}X\circ\Lambda^{x1y}R_X(\theta)&\overset{\text{\crefnosort{prop:CCX,commctrl,cor:swap,prop:comb,ctrlXCNot}}}{=}&\scalebox{0.85}{$\input{decodagenaturalite1.tikz}$}\\\\
&\overset{\text{\cref{commctrl,commctrlpasenface}}}{=}&\scalebox{0.85}{$\input{decodagenaturalite2.tikz}$}\\\\
&\overset{\text{\crefnosort{prop:CCX,commctrl,cor:swap,prop:comb,ctrlXCNot}}}{=}&\scalebox{0.85}{$\input{decodagenaturalite3.tikz}$}\\\\
&\overset{\text{\cref{commctrlpasenface}}}{=}&\scalebox{0.85}{$\input{decodagenaturalite4.tikz}$}\\\\
&\eqeqref{CNotCNot}&\scalebox{0.85}{$\input{decodagenaturalite5.tikz}$}\\\\
&\eqeqref{tripleCNotswap}&\scalebox{0.85}{$\input{decodagenaturalite6.tikz}$}\\\\
&\overset{\text{\cref{cor:swap}}}{=}&\scalebox{0.85}{$\input{decodagenaturalite7.tikz}$}\\\\
&\eqeqref{tripleCNotswap}&\scalebox{0.85}{$\input{decodagenaturalite8.tikz}$}\\\\
&\eqtroiseqref{XX}{CNotX}{commutationXCNot}&\scalebox{0.85}{$\input{decodagenaturalite9.tikz}$}\\\\
&\overset{\text{\cref{ctrlXCNot,prop:comb,commctrl}}}{=}&\scalebox{0.85}{$\input{decodagenaturalite10.tikz}$}\\\\
&\overset{\text{\cref{symmetriesemicontrolee}, \eqref{XX}}}{=}&\scalebox{0.85}{$\input{decodagenaturalite11.tikz}$}\\\\
&\overset{\text{\cref{prop:CCX,commctrl,cor:swap,prop:comb,ctrlXCNot}}}{=}&\scalebox{0.85}{$\input{decodagenaturalite12.tikz}$}\\\\
&\overset{\text{\cref{commctrlpasenface}}}{=}&\scalebox{0.85}{$\input{decodagenaturalite13.tikz}$}\\\\
&\overset{\text{\crefnosort{prop:CCX,commctrl,cor:swap,prop:comb,ctrlXCNot}}}{=}&
\Lambda^x_{y1}R_X(\theta)\circ\Lambda^{x1y}X\circ\Lambda^x_{y1}X.
\end{longtable}
\item The case $C=\input{swap-xs.tikz}$ is similar to the preceding one, with $R_X(\theta)$ replaced by $X$.
\end{itemize}

\subsection{Proof of Lemma \ref{decodingLOPPrules}}\label{preuvedecodingLOPPrules}

By \cref{decodingtoporules}, to prove \cref{decodingLOPPrules}, it suffices to prove that for each rule of \cref{axiomsLOPP}, of the form $C_1=C_2$ with $C_1,C_2\in\LOPPbarebf(i,i)$ (see \cref{rawinfigures}), and any $2^n$-mode raw context $\C[\cdot]_i$, one has $\qc\vdash D(\C[C_1])=D(\C[C_2])$. For this purpose, we prove a slightly more general result, namely that for any $k,n$ and any $\ell$-mode raw context $\C[\cdot]_i$ with $k+\ell\leq 2^n$, one has $\qc\vdash D_{k,n}(\C[C_1])=D(\C[C_2])$. We proceed by induction on $\C[\cdot]_i$:
\begin{itemize}
\item If $\C[\cdot]_i=C\circ \C'[\cdot]_i$, then $D_{k,n}(\C[C_1])=D_{k,n}(C)\circ D_{k,n}(\C'[C_1])$ and $D_{k,n}(\C[C_2])=D_{k,n}(C)\circ D_{k,n}(\C'[C_2])$, so the result follows by induction hypothesis. The case $\C[\cdot]_i=\C'[\cdot]_i\circ C$ is similar.
\item If $\C[\cdot]_i=C\otimes \C'[\cdot]_i$ with $C:\ell_1\to\ell_1$, then $D_{k,n}(\C[C_1])=D_{k+\ell_1,n}(\C'[C_1])\circ D_{k,n}(C)$ and $D_{k,n}(\C[C_2])=D_{k+\ell_1,n}(\C'[C_2])\circ D_{k,n}(C)$, so the result follows by induction hypothesis. The case $\C[\cdot]_i=\C'[\cdot]_i\otimes C$ is similar.
\end{itemize}
It remains to prove for each rule of \cref{axiomsLOPP}, of the form $C_1=C_2$ with $C_1,C_2\in\LOPPbarebf(i,i)$, that for any $k,n$ with $k+i\leq 2^n$, one has $\qc\vdash D_{k,n}(C_1)=D_{k,n}(C_2)$. Again by \cref{decodingtoporules}, it suffices to prove that $\qc\vdash D_{k,n}(C'_1)=D_{k,n}(C'_2)$ for arbitrary $C'_1$ and $C'_2$ such that $C'_1\equiv C_1$ and $C'_2\equiv C_2$. 

For \cref{phase0}, one has $D_{k,n}(\scalebox{0.69}{$\begin{tikzpicture}
	\begin{pgfonlayer}{nodelayer}
		\node [style=none] (5) at (-1.1, 0) {};
		\node [style=none] (9) at (1.1, 0) {};
		\node [style=PolRot] (11) at (0, 0) {$0$};
	\end{pgfonlayer}
	\begin{pgfonlayer}{edgelayer}
		\draw (5.center) to (11);
		\draw (11) to (9.center);
	\end{pgfonlayer}
\end{tikzpicture}
$})=\Lambda^{G_n(k)}s(0)$, $D_{k,n}(\scalebox{0.69}{$\begin{tikzpicture}
	\begin{pgfonlayer}{nodelayer}
		\node [style=none] (5) at (-1.1, 0) {};
		\node [style=none] (9) at (1.1, 0) {};
		\node [style=PolRot] (11) at (0, 0) {$2\pi$};
	\end{pgfonlayer}
	\begin{pgfonlayer}{edgelayer}
		\draw (5.center) to (11);
		\draw (11) to (9.center);
	\end{pgfonlayer}
\end{tikzpicture}
$})=\Lambda^{G_n(k)}s(2\pi)$ and $D_{k,n}(\gid)=id_n$. The three are equal modulo $\textup{QC}$ by \cref{prop:sum,prop:period}.

For \cref{bs0}, one has $D_{k,n}(\scalebox{0.69}{$\input{bs0-s.tikz}$})=\Lambda^{x_{k,n}}_{y_{k,n}}R_X(0)$ (where $x_{k,n}$ and $y_{k,n}$ are defined in \cref{defdecoding}) and $D_{k,n}(\scalebox{0.4}{$\input{filsparalleleslongbs-m.tikz}$})=id_n\circ id_n\equiv id_n$. The two are equal modulo $\textup{QC}$ 
by \cref{prop:sum}.

For \cref{swapbspisur2}, one has $D_{k,n}(\scalebox{0.69}{$\input{swap-s.tikz}$})=\Lambda^{x_{k,n}}_{y_{k,n}}X$, and $D_{k,n}(\scalebox{0.575}{$\input{bspissur2-ms.tikz}$})=%
\left(\displaystyle\prod_{j\in\{k,k+1\}}\Lambda^{G_n(j)}s(-\frac\pi2)\right)\circ\Lambda^{x_{k,n}}_{y_{k,n}}R_X(-\pi)$.
Note that the definitions imply that
\begin{equation}\label{Grayxy}\{G_n(k),G_n(k+1)\}=\{x_{k,n}0y_{k,n},x_{k,n}1y_{k,n}\}.\end{equation}
Therefore,
\begin{longtable}{RCL}D_{k,n}(\scalebox{0.575}{$\input{bspissur2-ms.tikz}$})&=&\sigma_{1,|x_{k,n}|}\circ\left(\displaystyle\prod_{a\in\{0,1\}}\Lambda^{ax_{k,n}y_{k,n}}s(-\frac\pi2)\right)\circ\Lambda^{\epsilon}_{x_{k,n}y_{k,n}}R_X(-\pi)\circ \sigma_{|x_{k,n}|,1}\\\\
&\overset{\text{\cref{prop:comb,commctrl}}}{=}&\sigma_{1,|x_{k,n}|}\circ\left(\gid\otimes\Lambda^{x_{k,n}y_{k,n}}s(-\frac\pi2)\right)\circ\Lambda^{\epsilon}_{x_{k,n}y_{k,n}}R_X(-\pi)\circ \sigma_{|x_{k,n}|,1}
\end{longtable}
which by \cref{commctrlphaseenhaut,def:multicontrolled,prop:period,HH}, is equal modulo $\qc$ to $\Lambda^{x_{k,n}}_{y_{k,n}}X$.

For \cref{phaseaddition}, one has $D_{k,n}(\scalebox{0.69}{$\input{convtp-phase-shifts-12.tikz}$})=\Lambda^{G_n(k)}s(\varphi_2)\circ\Lambda^{G_n(k)}s(\varphi_1)$ and $D_{k,n}(\scalebox{0.69}{$\begin{tikzpicture}
	\begin{pgfonlayer}{nodelayer}
		\node [style=none] (14) at (-1.725, 0) {};
		\node [style=PolRot] (15) at (0, 0) {$\varphi_1{+}\varphi_2$};
		\node [style=none] (16) at (1.675, 0) {};
	\end{pgfonlayer}
	\begin{pgfonlayer}{edgelayer}
		\draw (15) to (14.center);
		\draw (15) to (16.center);
	\end{pgfonlayer}
\end{tikzpicture}
$})=\Lambda^{G_n(k)}s(\varphi_1+\varphi_2)$. Both are equal modulo $\qc$ by \cref{prop:sum}.

For \cref{globalphasepropagationbs}, one has
\begin{longtable}{RCL}D_{k,n}(\scalebox{0.69}{$\input{thetathetabs-ms.tikz}$})&=&\Lambda^{x_{k,n}}_{y_{k,n}}R_X(-2\theta)\circ\left(\displaystyle\prod_{j\in\{k,k+1\}}\Lambda^{G_n(j)}s(\varphi)\right)\\\\
&\eqeqref{Grayxy}&\Lambda^{x_{k,n}}_{y_{k,n}}R_X(-2\theta)\circ\left(\displaystyle\prod_{a\in\{0,1\}}\Lambda^{x_{k,n}ay_{k,n}}s(\varphi)\right)\\\\
&=&\sigma_{1,|x_{k,n}|}\circ\Lambda^{\epsilon}_{x_{k,n}y_{k,n}}R_X(-2\theta)\circ\left(\displaystyle\prod_{a\in\{0,1\}}\Lambda^{ax_{k,n}y_{k,n}}s(\varphi)\right)\circ \sigma_{|x_{k,n}|,1}\\\\
&\overset{\text{\cref{prop:comb,commctrl}}}{=}&\sigma_{1,|x_{k,n}|}\circ\Lambda^{\epsilon}_{x_{k,n}y_{k,n}}R_X(-2\theta)\circ\left(\gid\otimes\Lambda^{x_{k,n}y_{k,n}}s(\varphi)\right)\circ \sigma_{|x_{k,n}|,1}\\\\
&\overset{\text{\cref{commctrlphaseenhaut}}}{=}&\sigma_{1,|x_{k,n}|}\circ\left(\gid\otimes\Lambda^{x_{k,n}y_{k,n}}s(\varphi)\right)\circ\Lambda^{\epsilon}_{x_{k,n}y_{k,n}}R_X(-2\theta)\circ \sigma_{|x_{k,n}|,1}\\\\
&\overset{\text{\cref{prop:comb,commctrl}}}{=}&\sigma_{1,|x_{k,n}|}\circ\left(\displaystyle\prod_{a\in\{0,1\}}\Lambda^{ax_{k,n}y_{k,n}}s(\varphi)\right)\circ\Lambda^{\epsilon}_{x_{k,n}y_{k,n}}R_X(-2\theta)\circ \sigma_{|x_{k,n}|,1}\\\\
&=&\left(\displaystyle\prod_{a\in\{0,1\}}\Lambda^{x_{k,n}ay_{k,n}}s(\varphi)\right)\circ\Lambda^{x_{k,n}}_{y_{k,n}}R_X(-2\theta)\\\\
&\eqeqref{Grayxy}&\left(\displaystyle\prod_{j\in\{k,k+1\}}\Lambda^{G_n(j)}s(\varphi)\right)\circ\Lambda^{x_{k,n}}_{y_{k,n}}R_X(-2\theta)\\\\
&=&D_{k,n}(\scalebox{0.69}{$\input{bsthetatheta-ms.tikz}$}).
\end{longtable}

For \cref{Eulerbsphasebs}, one has

\[D_{k,n}(\scalebox{0.69}{$\input{bsphasebsalpha-ms.tikz}$})\equiv\Lambda^{x_{k,n}}_{y_{k,n}}R_X(-2\alpha_3)\circ\Lambda^{G_n(k)}s(\alpha_2)\circ\Lambda^{x_{k,n}}_{y_{k,n}}R_X(-2\alpha_1)\]
and
\[D_{k,n}(\scalebox{0.69}{$\input{phasebsphasebeta-ms.tikz}$})\equiv\Lambda^{G_n(k+1)}s(\beta_4)\circ\Lambda^{G_n(k)}s(\beta_3)\circ\Lambda^{x_{k,n}}_{y_{k,n}}R_X(-2\beta_2)\circ\Lambda^{G_n(k)}s(\beta_1).\]
Note that for some $a_k\in\{0,1\}$, one has $G_n(k)=x_{k,n}a_ky_{k,n}$ and $G_n(k+1)=x_{k,n}\bar a_ky_{k,n}$. Therefore, by \cref{prop:CP}, for any $\varphi\in\mathbb R$, one has $\qc\vdash\Lambda^{G_n(k)}s(\varphi)=\Lambda^{x_{k,n}}_{y_{k,n}}P(\varphi)$ and $\qc\vdash\Lambda^{G_n(k+1)}s(\varphi)=(id_{|x_{k,n}|}\otimes X\otimes id_{|y_{k,n}|})\circ\Lambda^{x_{k,n}}_{y_{k,n}}P(\varphi)\circ(id_{|x_{k,n}|}\otimes X\otimes id_{|y_{k,n}|})$, or conversely. Thus, up to using \cref{XX} and possibly \cref{symmetriesemicontrolee}, it suffices to prove that $\lambda^{n-1}R_X(-2\alpha_3)\circ\lambda^{n-1}P(\alpha_2)\circ\lambda^{n-1}R_X(-2\alpha_1)=(id_{n-1}\otimes X)\circ\lambda^{n-1}P(\beta_4)\circ(id_{n-1}\otimes X)\circ\lambda^{n-1}P(\beta_3)\circ\lambda^{n-1}R_X(-2\beta_2)\circ\lambda^{n-1}P(\beta_1)$. One has
\begin{longtable}{RCL}
\scalebox{0.8}{$\input{phasebsphasebeta-circuit.tikz}$}&\overset{\text{\cref{antisymmetriesemicontrolee}}}{=}&\scalebox{0.8}{$\input{phasebsphasebeta-circuit1.tikz}$}\\\\
&\overset{\text{\cref{prop:sum}}}{=}&\scalebox{0.8}{$\input{phasebsphasebeta-circuit2.tikz}$}\\\\
&\overset{\text{\cref{prop:period}}}{=}&\scalebox{0.8}{$\input{phasebsphasebeta-circuit3.tikz}$}
\end{longtable}

Because of the conditions on the angles in the right-hand side of \cref{Eulerbsphasebs}, if $\beta_2=0$ then the angles of the last circuit satisfy the conditions so that it matches the right-hand side of \cref{Euler2dmulticontrolled}. Hence, since it has the same semantics as $\lambda^{n-1}R_X(-2\alpha_3)\circ\lambda^{n-1}P(\alpha_2)\circ\lambda^{n-1}R_X(-2\alpha_1)$, both circuits are equal according to \cref{Euler2dmulticontrolled}.

If $\beta_2\neq0$, then
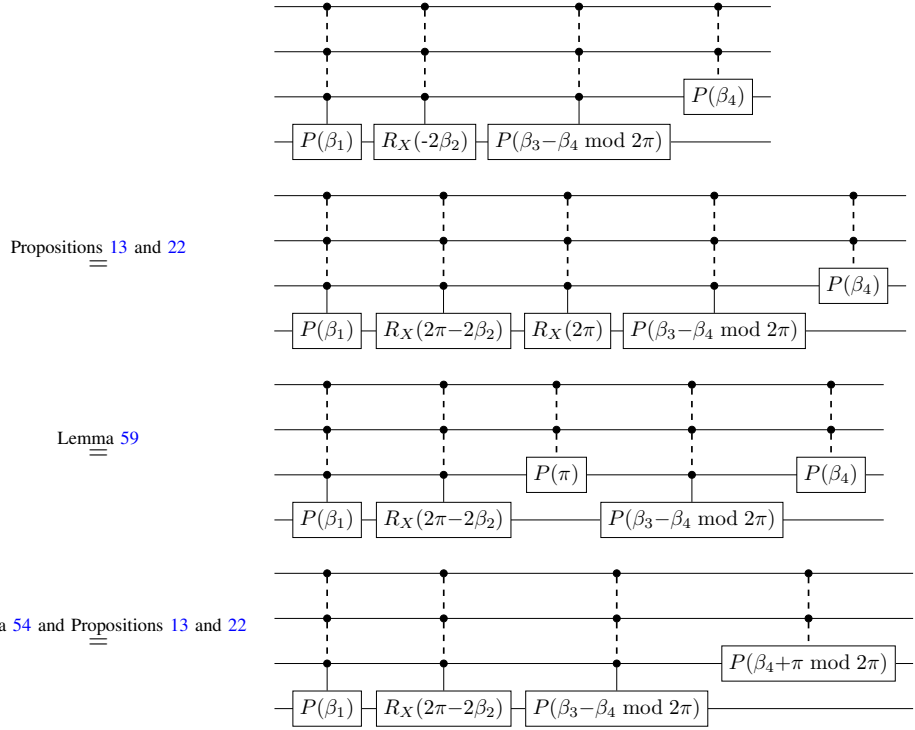
\begin{longtable}{CL}
&\scalebox{0.8}{$\input{phasebsphasebeta-circuit3.tikz}$}\\\\
\overset{\text{\cref{prop:period,prop:sum}}}{=}&\scalebox{0.8}{$\input{phasebsphasebeta-circuit4.tikz}$}\\\\
\overset{\text{\cref{CRX2pi}}}{=}&\scalebox{0.8}{$\input{phasebsphasebeta-circuit5.tikz}$}\\\\
\overset{\text{\cref{commctrlphaseenhautP,prop:sum,prop:period}}}{=}&\scalebox{0.8}{$\input{phasebsphasebeta-circuit6.tikz}$}
\end{longtable}
Because of the conditions on the angles in the right-hand side of \cref{Eulerbsphasebs}, one has $\beta_2\in(0,\pi)$, so that $2\pi-2\beta_2\in(0,2\pi)$, and if $2\pi-2\beta_2=\pi$ then $\beta_2=\frac\pi2$, so that $\beta_1=0$. Hence, the angles of the last circuit satisfy the conditions so that it matches the right-hand side of \cref{Euler2dmulticontrolled}. Again, since it has the same semantics as $\lambda^{n-1}R_X(-2\alpha_3)\circ\lambda^{n-1}P(\alpha_2)\circ\lambda^{n-1}R_X(-2\alpha_1)$, both circuits are equal according to \cref{Euler2dmulticontrolled}.

For \cref{Eulerscalaires}, by the properties of the Gray code, exactly one bit differs between $G_n(k)$ and $G_n(k+1)$, as well as between $G_n(k+1)$ and $G_n(k+2)$, and in exactly one of the two cases this is the last bit that differs (namely between $G_n(k)$ and $G_n(k+1)$ if $k$ is even, and between $G_n(k+1)$ and $G_n(k+2)$ if $k$ is odd). Hence we can write $G_n(k)$ as $xayb$ with $a,b\in\{0,1\}$, in such a way that $G_n(k+2)=x\bar ay\bar b$ and $G_n(k+1)=xay\bar b\text{ or }x\bar ayb$ depending on the parity of $k$. We treat the case where $k$ is even, the case with $k$ odd is similar. One has

\[D_{k,n}\left(\scalebox{0.69}{$\input{bsyangbaxterpointeenbas-simp-gammas-ms.tikz}$}\right)\equiv\Lambda^{xay}R_X(-2\gamma_4)\circ\Lambda^x_{y\bar b}R_X(-2\gamma_3)\circ\Lambda^{xayb}s(\gamma_2)\circ\Lambda^{xay}R_X(-2\gamma_1)\]
and
\[D_{k,n}\left(\scalebox{0.69}{$\input{bsyangbaxterpointeenhaut-deltas-ms.tikz}$}\right)\equiv\begin{array}[t]{l}\Lambda^{x\bar ay\bar b}s(\delta_9)\circ\Lambda^{xay\bar b}s(\delta_8)\circ\Lambda^{xayb}s(\delta_7)\circ\Lambda^x_{y\bar b}R_X(-2\delta_6)\circ\Lambda^{xay\bar b}s(\delta_5)\\
\circ\Lambda^{xay}R_X(-2\delta_4)\circ\Lambda^x_{y\bar b}R_X(-2\delta_3)\circ\Lambda^{xayb}s(\delta_2)\circ\Lambda^{xay\bar b}s(\delta_1).\end{array}\]

Up to using \cref{XX}, we can assume that the components of $x$ and $y$ are all equal to $1$. Up to using additionally \cref{symmetriesemicontrolee}, we can assume that $a=1$ and $b=0$. Finally, up to deforming the circuits, we can assume that $y=\epsilon$. Thus, it suffices to prove that
\[\qc\vdash\Lambda^{x1}R_X(-2\gamma_4)\circ\Lambda^x_{1}R_X(-2\gamma_3)\circ\Lambda^{x10}s(\gamma_2)\circ\Lambda^{x1}R_X(-2\gamma_1)=\begin{array}[t]{l}\Lambda^{x01}s(\delta_9)\circ\Lambda^{x11}s(\delta_8)\circ\Lambda^{x10}s(\delta_7)\circ\Lambda^x_{1}R_X(-2\delta_6)\circ\Lambda^{x11}s(\delta_5)\circ\\
\Lambda^{x1}R_X(-2\delta_4)\circ\Lambda^x_{1}R_X(-2\delta_3)\circ\Lambda^{x10}s(\delta_2)\circ\Lambda^{x11}s(\delta_1)\end{array}\]
where $x=1^{n-2}$.

The left-hand side is equal to
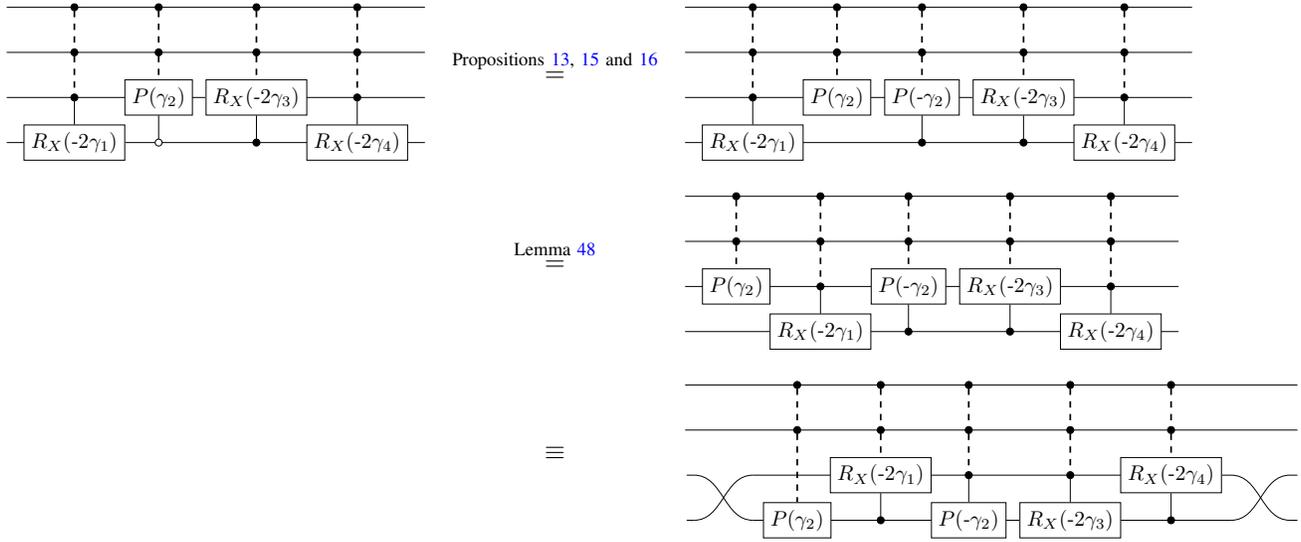
\begin{longtable}{RCL}
\scalebox{0.8}{$\input{Eulerscalairesleft-circuit.tikz}$}&\overset{\text{\cref{prop:sum,commctrl,prop:comb}}}{=}&\scalebox{0.8}{$\input{Eulerscalairesleft-circuit1.tikz}$}\\\\
&\overset{\text{\cref{commctrlphaseenhaut}}}{=}&\scalebox{0.8}{$\input{Eulerscalairesleft-circuit2.tikz}$}\\\\
&\equiv&\scalebox{0.8}{$\input{Eulerscalairesleft-circuit-arrange.tikz}$}
\end{longtable}
while the right-hand side is equal to
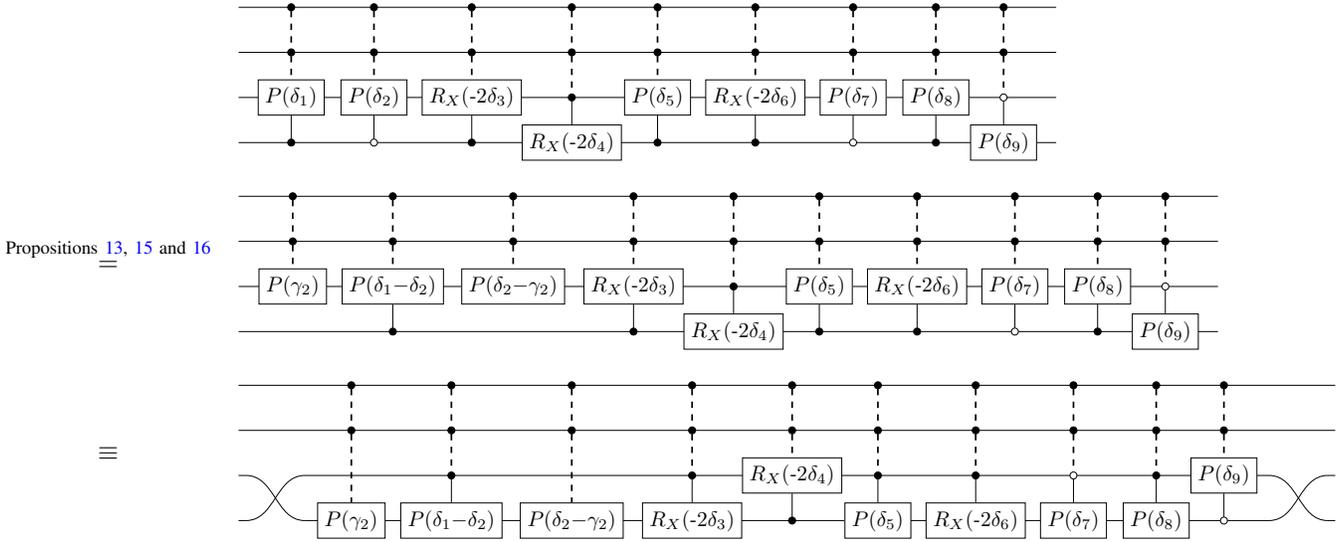
\begin{longtable}{CL}
&\scalebox{0.8}{$\input{Eulerscalairesright-circuit.tikz}$}\\\\
\overset{\text{\cref{prop:sum,commctrl,
prop:comb}}}=&\scalebox{0.8}{$\input{Eulerscalairesright-circuit2.tikz}$}\\\\
\equiv&\scalebox{0.8}{$\input{Eulerscalairesright-circuit3.tikz}$}
\end{longtable}

Hence, it suffices to prove that

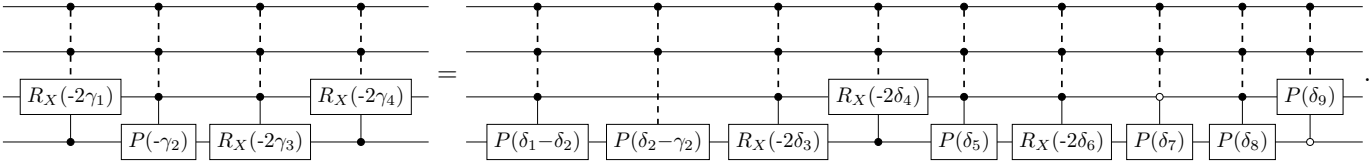
\[\scalebox{0.8}{$\input{Euler3dmulticontrolledleft-preuve.tikz}$}=\scalebox{0.8}{$\input{Euler3dmulticontrolledright-preuve.tikz}$}.\]
The left-hand side matches the left-hand side of \cref{Euler3dmulticontrolled}, hence it suffices to prove that the right-hand side can be put in the form of the right-hand side of \cref{Euler3dmulticontrolled} with the angles satisfying the conditions. One has
\begin{longtable}{CL}
&\scalebox{0.8}{$\input{Euler3dmulticontrolledright-preuve.tikz}$}\\\\
\overset{\text{\cref{prop:sum,prop:comb}}}{=}&\scalebox{0.8}{$\input{Euler3dmulticontrolledright-preuve1.tikz}$}\\\\
\overset{\text{\cref{prop:CP,prop:sum}}}{=}&\scalebox{0.8}{$\input{Euler3dmulticontrolledright-preuve2.tikz}$}\\\\
\overset{\text{\cref{prop:comb,commctrl,prop:sum}}}{=}&\scalebox{0.8}{$\input{Euler3dmulticontrolledright-preuve3.tikz}$}\\\\
\overset{\text{\cref{prop:sum}}}{=}&\scalebox{0.8}{$\input{Euler3dmulticontrolledright-preuve4.tikz}$}\\\\
\overset{\text{\cref{prop:comb,commctrl,prop:CP,prop:sum}}}{=}&\scalebox{0.8}{$\input{Euler3dmulticontrolledright-preuve5.tikz}$}
\end{longtable}

It remains to %
prove that any circuit of the form $\scalebox{0.8}{$\input{Euler3dright-multicontrolled-simp-deltas.tikz}$}$ can be transformed using the axioms of $\qc$ in such a way that the angles satisfy the conditions given in \cref{fig:euler}. We treat the conditions in the following order (note that some of the conditions of \cref{fig:euler} have been split into two parts):
\begin{itemize}
\item \hyperref[delta3in02pi]{$\delta_3\in[0,2\pi)$}
\item \hyperref[delta4in02pi]{$\delta_4\in[0,2\pi)$}
\item \hyperref[delta6in02pi]{$\delta_6\in[0,2\pi)$}
\item \hyperref[delta3zeroimpldelta2zero]{if $\delta_3=0$ then $\delta_2=0$}
\item \hyperref[delta4piimpldelta2zeroifdelta3notzero]{if $\delta_3\neq0$ but $\delta_4=\pi$ then $\delta_2=0$}
\item \hyperref[delta3zerodelta4piimpldelta1zero]{if $\delta_3=0$ and $\delta_4=\pi$ then $\delta_1=0$}
\item \hyperref[delta3piimpldelta1zero]{if $\delta_3=\pi$ then $\delta_1=0$}
\item \hyperref[delta4zeroimpldelta123zero]{if $\delta_4=0$ then $\delta_1=\delta_2=\delta_3=0$}
\item \hyperref[delta1in0piifdelta3notzero]{if $\delta_3\neq0$ then $\delta_1\in[0,\pi)$}
\item \hyperref[delta1in0piifdelta3zero]{if $\delta_3=0$ then $\delta_1\in[0,\pi)$}
\item \hyperref[delta6zeroimpldelta5zero]{if $\delta_6=0$ then $\delta_5=0$}
\item \hyperref[delta6piimpldelta5zero]{if $\delta_6=\pi$ then $\delta_5=0$}
\item \hyperref[delta2in0pi]{$\delta_2\in[0,\pi)$}
\item \hyperref[delta5in0pi]{$\delta_5\in[0,\pi)$}
\item \hyperref[delta789in02pi]{$\delta_7,\delta_8,\delta_9\in[0,2\pi)$}.
\end{itemize}
For each of them, we prove that given a circuit satisfying the previous conditions, we can transform it into a circuit satisfying also the considered condition.\bigskip

\phantomsection\label{delta3in02pi}If $\delta_3\notin[0,2\pi)$, then by \cref{prop:period}, we can assume that it is in $[0,4\pi)$, and then if it is in $[2\pi,4\pi)$, then:
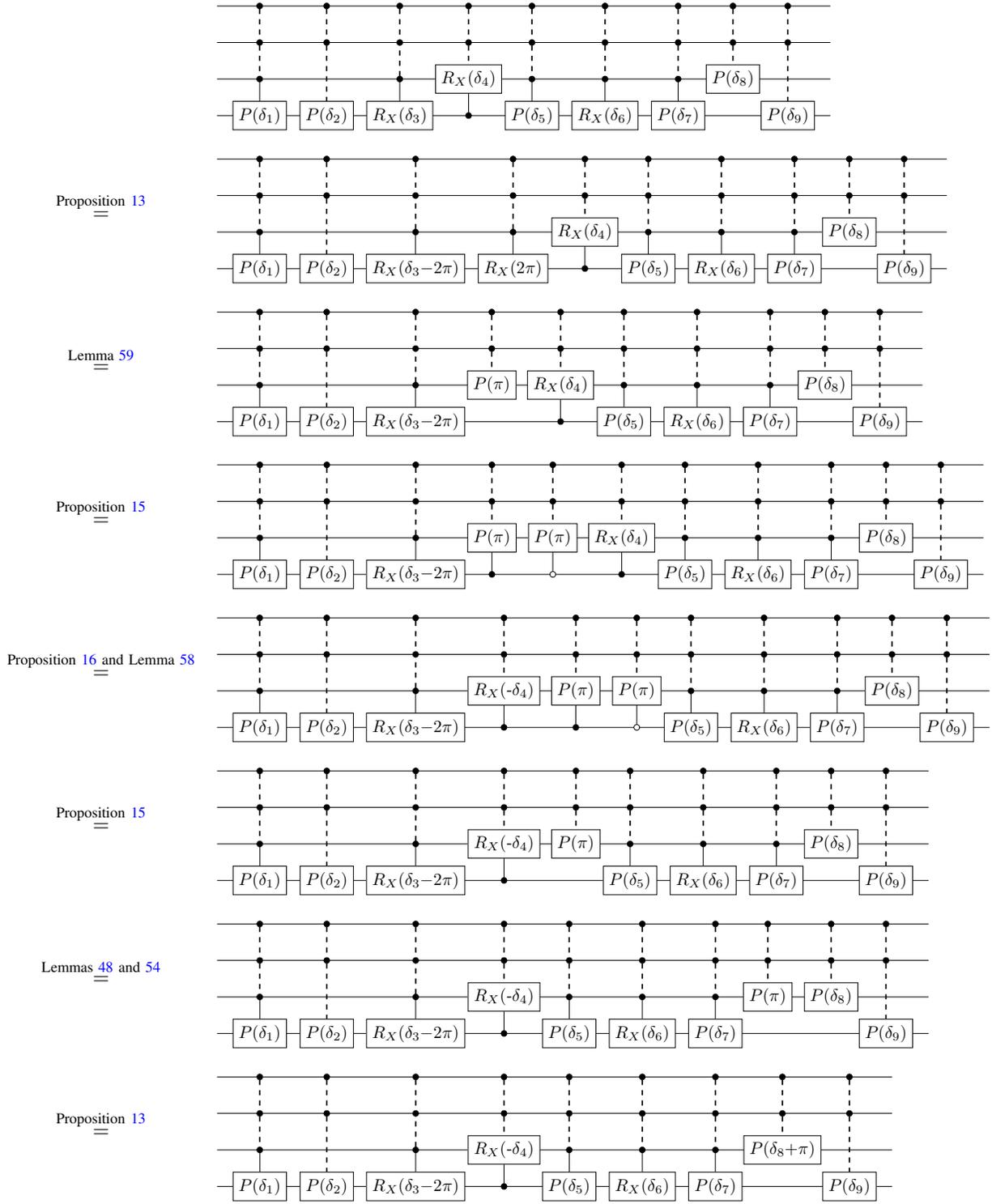
\begin{longtable}{CL}
&\scalebox{0.8}{$\input{Euler3dright-multicontrolled-simp-deltas.tikz}$}\\\\
\overset{\text{\cref{prop:sum}}}{=}&\scalebox{0.8}{$\input{anglesbonsintervalles1bis.tikz}$}\\\\
\overset{\text{\cref{CRX2pi}}}{=}&%

\scalebox{0.8}{$\input{anglesbonsintervalles3.tikz}$}\\\\
\overset{\text{\cref{prop:comb}}}{=}&\scalebox{0.8}{$\input{anglesbonsintervalles4.tikz}$}\\\\
\overset{\text{\cref{commctrl,passagephasepicircuits}}}{=}&\scalebox{0.8}{$\input{anglesbonsintervalles5.tikz}$}\\\\
\overset{\text{\cref{prop:comb}}}{=}&\scalebox{0.8}{$\input{anglesbonsintervalles6.tikz}$}\\\\
\overset{\text{\cref{commctrlphaseenhaut,commctrlphaseenhautP}}}{=}&\scalebox{0.8}{$\input{anglesbonsintervalles7.tikz}$}\\\\
\overset{\text{\cref{prop:sum}}}{=}&\scalebox{0.8}{$\input{anglesbonsintervalles8.tikz}$}
\end{longtable}
with $\delta_3-2\pi\in[0,2\pi)$. Hence, we can assume that $\delta_3\in[0,2\pi)$.\bigskip

\phantomsection\label{delta4in02pi}If $\delta_4\notin[0,2\pi)$, then by \cref{prop:period}, we can ensure that it is in $[0,4\pi)$, and then if it is in $[2\pi,4\pi)$, then:
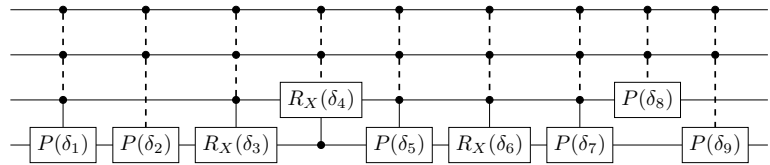
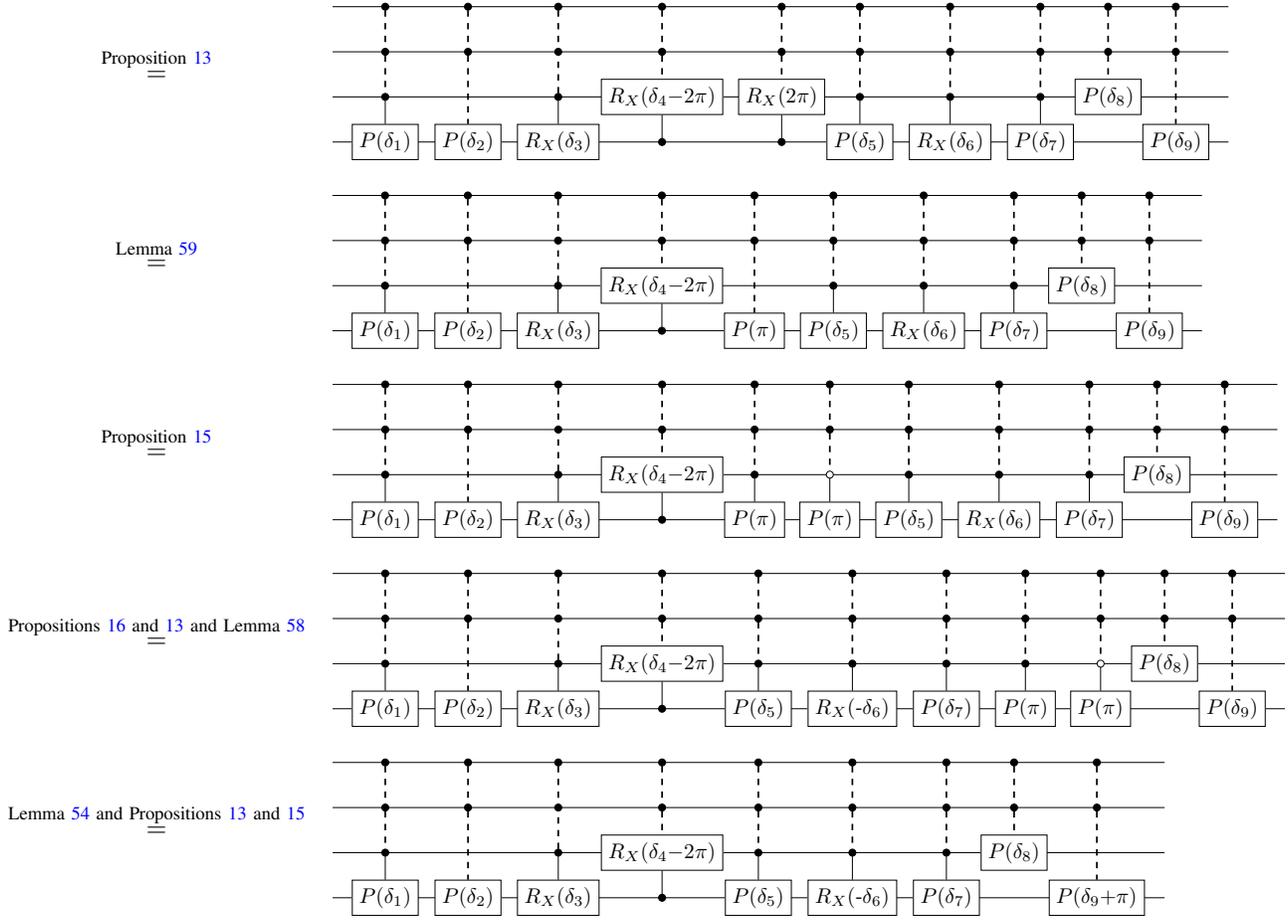
\begin{longtable}{CL}
&\scalebox{0.8}{$\input{Euler3dright-multicontrolled-simp-deltas.tikz}$}\\\\
\overset{\text{\cref{prop:sum}}}{=}&\scalebox{0.8}{$\input{anglesbonsintervalles9.tikz}$}\\\\
\overset{\text{\cref{CRX2pi}}}{=}&\scalebox{0.8}{$\input{anglesbonsintervalles10.tikz}$}\\\\
\overset{\text{\cref{prop:comb}}}{=}&\scalebox{0.8}{$\input{anglesbonsintervalles11.tikz}$}\\\\
\overset{\text{\crefnosort{commctrl,prop:sum,passagephasepicircuits}}}{=}&\scalebox{0.8}{$\input{anglesbonsintervalles12.tikz}$}\\\\
\overset{\text{\cref{commctrlphaseenhautP,prop:comb,prop:sum}}}{=}&\scalebox{0.8}{$\input{anglesbonsintervalles13.tikz}$}
\end{longtable}
with $\delta_4-2\pi\in[0,2\pi)$. Hence, we can assume additionally that $\delta_4\in[0,2\pi)$.\bigskip

\phantomsection\label{delta6in02pi}If $\delta_6\notin[0,2\pi)$, then by \cref{prop:period}, we can ensure that it is in $[0,4\pi)$, and then if it is in $[2\pi,4\pi)$, then:
\begin{longtable}{CL}
&\scalebox{0.8}{$\input{Euler3dright-multicontrolled-simp-deltas.tikz}$}\\\\
\overset{\text{\cref{prop:sum}}}{=}&\scalebox{0.8}{$\input{anglesbonsintervalles14.tikz}$}\\\\
\overset{\text{\cref{CRX2pi}}}{=}&\scalebox{0.8}{$\input{anglesbonsintervalles15.tikz}$}\\\\
\overset{\text{\cref{commctrlphaseenhautP,prop:sum}}}{=}&\scalebox{0.8}{$\input{anglesbonsintervalles16.tikz}$}
\end{longtable}
with $\delta_6-2\pi\in[0,2\pi)$. Hence, we can assume additionally that $\delta_6\in[0,2\pi)$.\bigskip

\phantomsection\label{delta3zeroimpldelta2zero}If $\delta_3=0$ but $\delta_2\neq0$, then:
\begin{longtable}{CL}
&\scalebox{0.8}{$\input{anglesbonsintervalles16etdemi.tikz}$}\\\\
\overset{\text{\cref{prop:sum}}}{=}&\scalebox{0.8}{$\input{anglesbonsintervalles17.tikz}$}\\\\
\overset{\text{\cref{commctrlphaseenhaut}}}{=}&\scalebox{0.8}{$\input{anglesbonsintervalles18.tikz}$}\\\\
\overset{\text{\cref{prop:comb}}}{=}&\scalebox{0.8}{$\input{anglesbonsintervalles19.tikz}$}\\\\
\overset{\text{\crefnosort{commctrl,prop:sum}}}{=}&\scalebox{0.8}{$\input{anglesbonsintervalles20.tikz}$}\\\\
\overset{\text{\cref{prop:sum,prop:comb}}}{=}&\scalebox{0.8}{$\input{anglesbonsintervalles21.tikz}$}\\\\
\overset{\text{\cref{prop:comb,commctrl,prop:CP,prop:sum}}}{=}&\scalebox{0.8}{$\input{anglesbonsintervalles22.tikz}$}\\\\
\overset{\text{\cref{prop:sum}}}{=}&\scalebox{0.8}{$\input{anglesbonsintervalles23.tikz}$}.
\end{longtable}
Hence, we can assume additionally that if $\delta_3=0$ then $\delta_2=0$.\bigskip

\phantomsection\label{delta4piimpldelta2zeroifdelta3notzero}If $\delta_3\neq0$, and $\delta_4=\pi$ but $\delta_2\neq0$, then:
\begin{longtable}{CL}
&\scalebox{0.8}{$\input{anglesbonsintervalles44.tikz}$}\\\\
\overset{\text{\cref{prop:sum,mctrlX}}}{=}&\scalebox{0.8}{$\input{anglesbonsintervalles45.tikz}$}\\\\
\overset{\text{\cref{prop:comb}}}{=}&\scalebox{0.8}{$\input{anglesbonsintervalles46.tikz}$}\\\\
\overset{\text{\cref{prop:sum,commctrl}}}{=}&\scalebox{0.8}{$\input{anglesbonsintervalles47.tikz}$}\\\\
\overset{\text{\cref{prop:CP}}}{=}&\scalebox{0.8}{$\input{anglesbonsintervalles48.tikz}$}\\\\
\overset{\text{\cref{commctrlphaseenhautP}}}{=}&\scalebox{0.8}{$\input{anglesbonsintervalles49.tikz}$}\\\\
\overset{\text{\cref{antisymmetriesemicontrolee,antisymmetriecontrolee}}}{=}&\scalebox{0.8}{$\input{anglesbonsintervalles50.tikz}$}\\\\
\overset{\text{\cref{prop:CCX}}}{=}&\scalebox{0.8}{$\input{anglesbonsintervalles51.tikz}$}\\\\
\overset{\text{\crefnosort{prop:CP,prop:sum}}}{=}&\scalebox{0.8}{$\input{anglesbonsintervalles52.tikz}$}\\\\
\overset{\text{\cref{mctrlX,prop:sum}}}{=}&\scalebox{0.8}{$\input{anglesbonsintervalles53.tikz}$}\\\\
\overset{\text{\cref{prop:sum}}}{=}&\scalebox{0.8}{$\input{anglesbonsintervalles54.tikz}$}
\end{longtable}
Hence, we can assume additionally that if $\delta_4=\pi$ then $\delta_2=0$ (note that by the previous assumption we already had $\delta_2=0$ when $\delta_3=0$).\bigskip

\phantomsection\label{delta3zerodelta4piimpldelta1zero}If $\delta_3=0$ and $\delta_4=\pi$, then by assumption, $\delta_2=0$. If we do not have additionally that $\delta_1=0$, then:
\begin{longtable}{CL}
&\scalebox{0.8}{$\input{anglesbonsintervalles55.tikz}$}\\\\
\overset{\text{\cref{prop:sum}}}{=}&\scalebox{0.8}{$\input{anglesbonsintervalles56.tikz}$}\\\\
\overset{\text{\cref{prop:sum,mctrlX}}}{=}&\scalebox{0.8}{$\input{anglesbonsintervalles57.tikz}$}\\\\
\overset{\text{\cref{prop:CP}}}{=}&\scalebox{0.8}{$\input{anglesbonsintervalles58.tikz}$}\\\\
\overset{\text{\cref{commctrlphaseenhautP}}}{=}&\scalebox{0.8}{$\input{anglesbonsintervalles59.tikz}$}\\\\
\overset{\text{\crefnosort{prop:CCX,antisymmetriecontrolee,antisymmetriesemicontrolee}}}{=}&\scalebox{0.8}{$\input{anglesbonsintervalles60.tikz}$}\\\\
\overset{\text{\cref{prop:CP}}}{=}&\scalebox{0.8}{$\input{anglesbonsintervalles61.tikz}$}\\\\
\overset{\text{\cref{commctrl}}}{=}&\scalebox{0.8}{$\input{anglesbonsintervalles62.tikz}$}\\\\
\overset{\text{\cref{prop:sum,commctrlphaseenhautP,prop:comb}}}{=}&\scalebox{0.8}{$\input{anglesbonsintervalles63.tikz}$}\\\\
\overset{\text{\cref{mctrlX,prop:sum}}}{=}&\scalebox{0.8}{$\input{anglesbonsintervalles64.tikz}$}\\\\
\overset{\text{\cref{prop:sum}}}{=}&\scalebox{0.8}{$\input{anglesbonsintervalles65.tikz}$}
\end{longtable}
Hence, we can assume additionally that if $\delta_3=0$ and $\delta_4=\pi$ then $\delta_1=0$.\bigskip

\phantomsection\label{delta3piimpldelta1zero}If $\delta_3=\pi$ but $\delta_1\neq0$, then:
\begin{longtable}{CL}
&\scalebox{0.8}{$\input{anglesbonsintervalles24.tikz}$}\\\\
\overset{\text{\cref{prop:sum,mctrlX}}}{=}&\scalebox{0.8}{$\input{anglesbonsintervalles25.tikz}$}\\\\
\overset{\text{\cref{prop:comb,prop:sum,commctrl}}}{=}&\scalebox{0.8}{$\input{anglesbonsintervalles26.tikz}$}\\\\
\overset{\text{\cref{commctrlphaseenhautP}}}{=}&\scalebox{0.8}{$\input{anglesbonsintervalles27.tikz}$}\\\\
\overset{\text{\crefnosort{prop:CCX,antisymmetriecontrolee,antisymmetriesemicontrolee}}}{=}&\scalebox{0.8}{$\input{anglesbonsintervalles28.tikz}$}\\\\
\overset{\text{\cref{prop:CP}}}{=}&\scalebox{0.8}{$\input{anglesbonsintervalles29.tikz}$}\\\\
\overset{\text{\cref{commctrl}}}{=}&\scalebox{0.8}{$\input{anglesbonsintervalles30.tikz}$}\\\\
\overset{\text{\cref{prop:CP}}}{=}&\scalebox{0.8}{$\input{anglesbonsintervalles31.tikz}$}\\\\
\overset{\text{\cref{commctrl,prop:sum,prop:comb}}}{=}&\scalebox{0.8}{$\input{anglesbonsintervalles32.tikz}$}\\\\
\overset{\text{\cref{prop:CP}}}{=}&\scalebox{0.8}{$\input{anglesbonsintervalles33.tikz}$}\\\\
\overset{\text{\cref{commctrlphaseenhaut,commctrlphaseenhautP,prop:sum}}}{=}&\scalebox{0.8}{$\input{anglesbonsintervalles34.tikz}$}\\\\
\overset{\text{\cref{mctrlX,prop:sum}}}{=}&\scalebox{0.8}{$\input{anglesbonsintervalles35.tikz}$}
\end{longtable}
Hence, we can assume additionally that if $\delta_3=\pi$ then $\delta_1=0$.\bigskip

\phantomsection\label{delta4zeroimpldelta123zero}If $\delta_4=0$ but $(\delta_1,\delta_2,\delta_3)\neq(0,0,0)$, then:
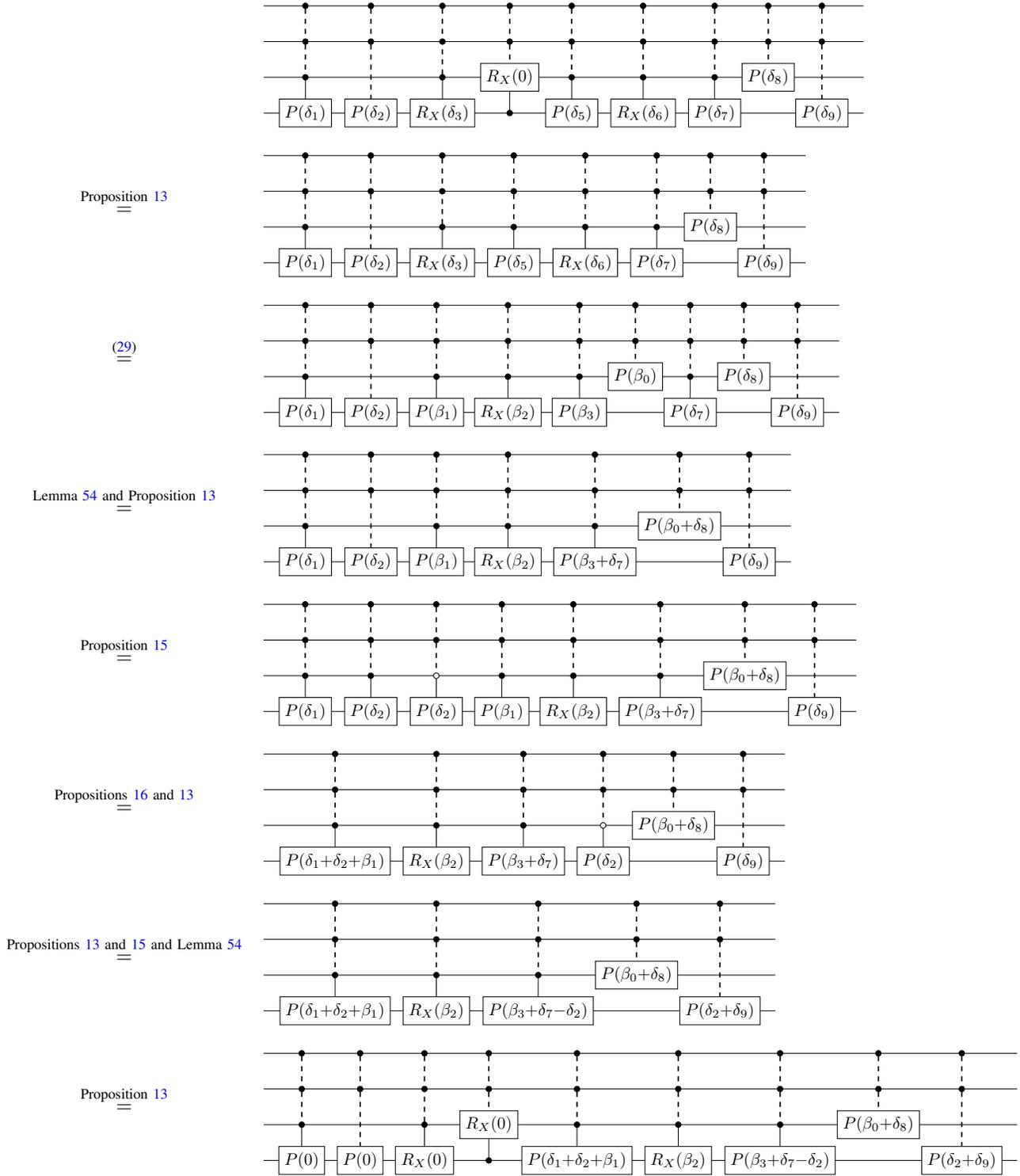
\begin{longtable}{CL}
&\scalebox{0.8}{$\input{anglesbonsintervalles36.tikz}$}\\\\
\overset{\text{\cref{prop:sum}}}{=}&\scalebox{0.8}{$\input{anglesbonsintervalles37.tikz}$}\\\\
\eqeqref{Euler2dmulticontrolled}&\scalebox{0.8}{$\input{anglesbonsintervalles38.tikz}$}\\\\
\overset{\text{\cref{commctrlphaseenhautP,prop:sum}}}{=}&\scalebox{0.8}{$\input{anglesbonsintervalles39.tikz}$}\\\\
\overset{\text{\cref{prop:comb}}}{=}&\scalebox{0.8}{$\input{anglesbonsintervalles40.tikz}$}\\\\
\overset{\text{\crefnosort{commctrl,prop:sum}}}{=}&\scalebox{0.8}{$\input{anglesbonsintervalles41.tikz}$}\\\\
\overset{\text{\cref{prop:sum,commctrlphaseenhautP,prop:comb}}}{=}&\scalebox{0.8}{$\input{anglesbonsintervalles42.tikz}$}\\\\
\overset{\text{\cref{prop:sum}}}{=}&\scalebox{0.8}{$\input{anglesbonsintervalles43.tikz}$}
\end{longtable}
where $\beta_0,\beta_1,\beta_2\text{ and }\beta_3$ satisfy the conditions given in \cref{fig:euler}. In particular, $\beta_2\in[0,2\pi)$, so that the previous assumptions are preserved. This implies that we can assume additionally that if $\delta_4=0$ then $\delta_1=\delta_2=\delta_3=0$.\bigskip

\phantomsection\label{delta1in0piifdelta3notzero}If $\delta_1\notin[0,\pi)$, then by \cref{prop:period}, we can ensure that it is in $[0,2\pi)$, and then if it is in $[\pi,2\pi)$, then, if $\delta_3\neq0$:
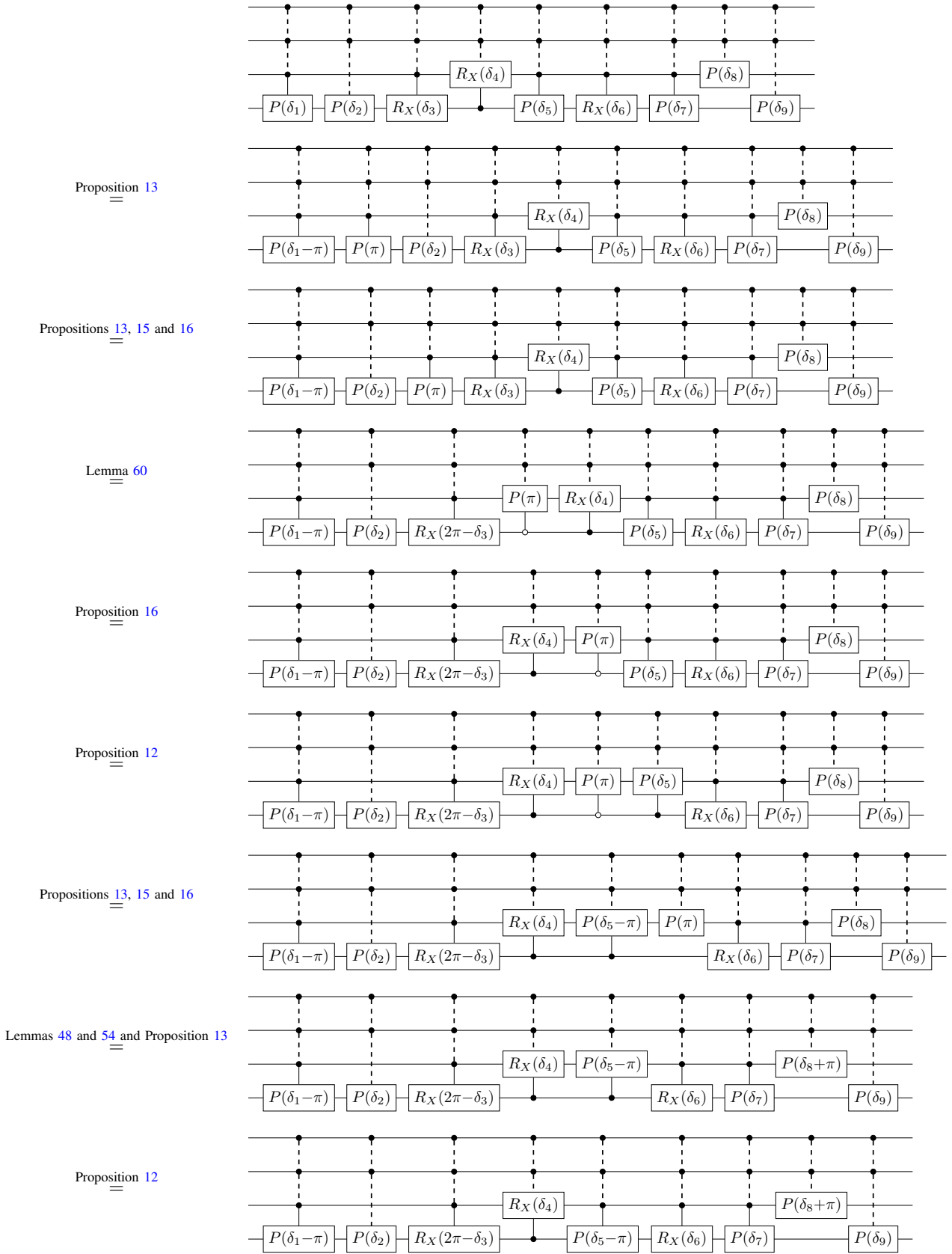
\begin{longtable}{CL}
&\scalebox{0.8}{$\input{Euler3dright-multicontrolled-simp-deltas.tikz}$}\\\\
\overset{\text{\cref{prop:sum}}}{=}&\scalebox{0.8}{$\input{anglesbonsintervalles80.tikz}$}\\\\
\overset{\text{\cref{prop:comb,prop:sum,commctrl}}}{=}&\scalebox{0.8}{$\input{anglesbonsintervalles81.tikz}$}\\\\
\overset{\text{\cref{passagephasepihbcircuits}}}{=}&\scalebox{0.8}{$\input{anglesbonsintervalles86.tikz}$}\\\\
\overset{\text{\cref{commctrl}}}{=}&\scalebox{0.8}{$\input{anglesbonsintervalles87.tikz}$}\\\\
\overset{\text{\cref{prop:CP}}}{=}&\scalebox{0.8}{$\input{anglesbonsintervalles88.tikz}$}\\\\
\overset{\text{\cref{commctrl,prop:sum,prop:comb}}}{=}&\scalebox{0.8}{$\input{anglesbonsintervalles89.tikz}$}\\\\
\overset{\text{\cref{commctrlphaseenhaut,commctrlphaseenhautP,prop:sum}}}{=}&\scalebox{0.8}{$\input{anglesbonsintervalles90.tikz}$}\\\\
\overset{\text{\cref{prop:CP}}}{=}&\scalebox{0.8}{$\input{anglesbonsintervalles91.tikz}$}
\end{longtable}
with $\delta_1-\pi\in[0,\pi)$. Moreover, since $\delta_3\neq0$, one has $2\pi-\delta_3\in[0,2\pi)$, so that the previous assumptions are preserved.

\phantomsection\label{delta1in0piifdelta3zero}And, still in the case where $\delta_1\in[\pi,2\pi)$, if $\delta_3=0$, then by assumption, $\delta_2=0$, and one has:
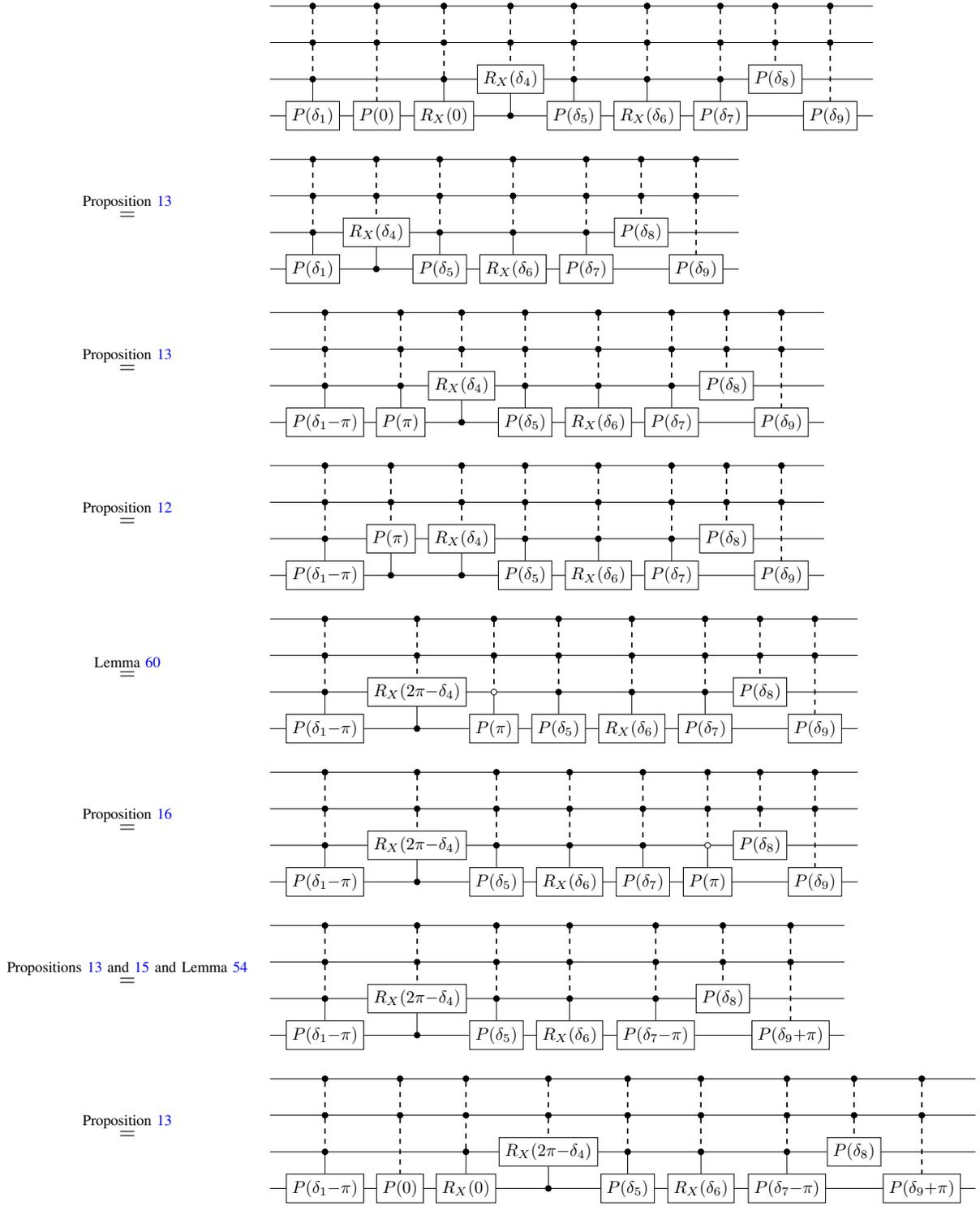
\begin{longtable}{CL}
&\scalebox{0.8}{$\input{anglesbonsintervalles92.tikz}$}\\\\
\overset{\text{\cref{prop:sum}}}{=}&\scalebox{0.8}{$\input{anglesbonsintervalles93.tikz}$}\\\\
\overset{\text{\cref{prop:sum}}}{=}&\scalebox{0.8}{$\input{anglesbonsintervalles93etdemi.tikz}$}\\\\
\overset{\text{\cref{prop:CP}}}{=}&\scalebox{0.8}{$\input{anglesbonsintervalles94.tikz}$}\\\\
\overset{\text{\cref{passagephasepihbcircuits}}}{=}&\scalebox{0.8}{$\input{anglesbonsintervalles95.tikz}$}\\\\
\overset{\text{\cref{commctrl}}}{=}&\scalebox{0.8}{$\input{anglesbonsintervalles96.tikz}$}\\\\
\overset{\text{\cref{prop:sum,commctrlphaseenhautP,prop:comb}}}{=}&\scalebox{0.8}{$\input{anglesbonsintervalles97.tikz}$}\\\\
\overset{\text{\cref{prop:sum}}}{=}&\scalebox{0.8}{$\input{anglesbonsintervalles98.tikz}$}
\end{longtable}
with $\delta_1-\pi\in[0,\pi)$.\bigskip

\phantomsection\label{delta6zeroimpldelta5zero}If $\delta_6=0$ but $\delta_5\neq0$, then:
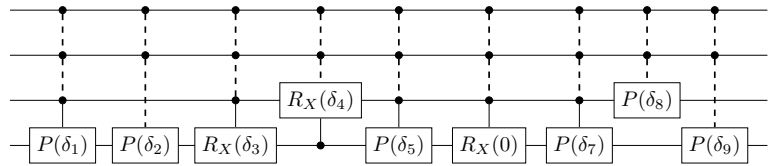
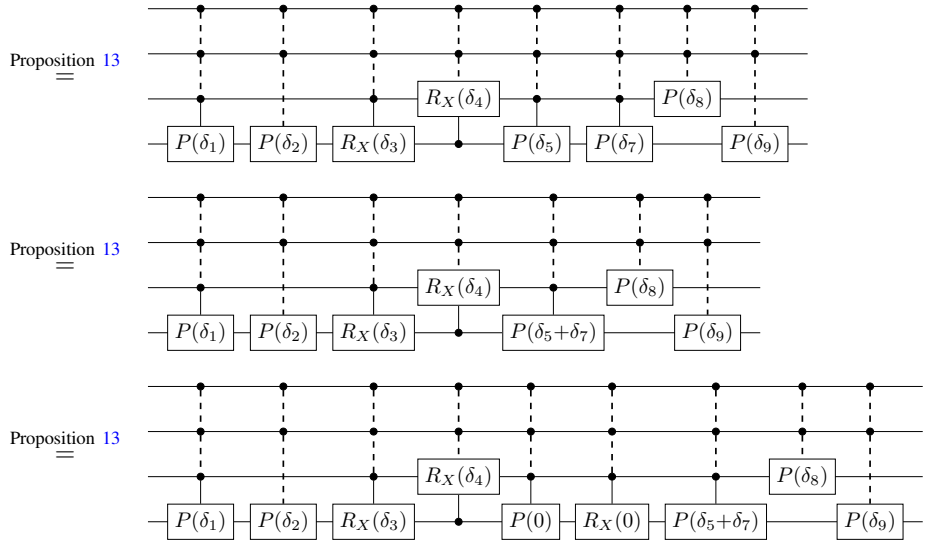
\begin{longtable}{CL}
&\scalebox{0.8}{$\input{anglesbonsintervalles66.tikz}$}\\\\
\overset{\text{\cref{prop:sum}}}{=}&\scalebox{0.8}{$\input{anglesbonsintervalles67.tikz}$}\\\\
\overset{\text{\cref{prop:sum}}}{=}&\scalebox{0.8}{$\input{anglesbonsintervalles68.tikz}$}\\\\
\overset{\text{\cref{prop:sum}}}{=}&\scalebox{0.8}{$\input{anglesbonsintervalles69.tikz}$}
\end{longtable}
Hence, we can assume additionally that if $\delta_6=0$ then $\delta_5=0$.\bigskip

\phantomsection\label{delta6piimpldelta5zero}If $\delta_6=\pi$ but $\delta_5\neq0$, then:
\begin{longtable}{CL}
&\scalebox{0.8}{$\input{anglesbonsintervalles70.tikz}$}\\\\
\overset{\text{\cref{prop:sum,mctrlX}}}{=}&\scalebox{0.8}{$\input{anglesbonsintervalles71.tikz}$}\\\\
\overset{\text{\cref{commctrlphaseenhautP}}}{=}&\scalebox{0.8}{$\input{anglesbonsintervalles72.tikz}$}\\\\
\overset{\text{\crefnosort{prop:CCX,antisymmetriecontrolee,antisymmetriesemicontrolee}}}{=}&\scalebox{0.8}{$\input{anglesbonsintervalles73.tikz}$}\\\\
\overset{\text{\cref{prop:CP}}}{=}&\scalebox{0.8}{$\input{anglesbonsintervalles74.tikz}$}\\\\
\overset{\text{\crefnosort{commctrl,prop:sum,prop:comb}}}{=}&\scalebox{0.8}{$\input{anglesbonsintervalles75.tikz}$}\\\\
\overset{\text{\cref{prop:sum}}}{=}&\scalebox{0.8}{$\input{anglesbonsintervalles76.tikz}$}\\\\
\overset{\text{\cref{prop:CP}}}{=}&\scalebox{0.8}{$\input{anglesbonsintervalles77.tikz}$}\\\\
\overset{\text{\cref{mctrlX,prop:sum}}}{=}&\scalebox{0.8}{$\input{anglesbonsintervalles78.tikz}$}\\\\
\overset{\text{\cref{prop:sum}}}{=}&\scalebox{0.8}{$\input{anglesbonsintervalles79.tikz}$}
\end{longtable}
Hence, we can assume additionally that if $\delta_6=\pi$ then $\delta_5=0$.\bigskip

\phantomsection\label{delta2in0pi}If $\delta_2\notin[0,\pi)$, then by \cref{prop:period}, we can ensure that it is in $[0,2\pi)$, and then if it is in $[\pi,2\pi)$, then:
\begin{longtable}{CL}
&\scalebox{0.8}{$\input{Euler3dright-multicontrolled-simp-deltas.tikz}$}\\\\
\overset{\text{\cref{prop:sum}}}{=}&\scalebox{0.8}{$\input{anglesbonsintervalles99.tikz}$}\\\\
\overset{\text{\cref{prop:comb}}}{=}&\scalebox{0.8}{$\input{anglesbonsintervalles100.tikz}$}\\\\
\overset{\text{\crefnosort{commctrl,prop:CP}}}{=}&\scalebox{0.8}{$\input{anglesbonsintervalles101.tikz}$}\\\\
\overset{\text{\cref{symmetriesemicontrolee}}}{=}&\scalebox{0.8}{$\input{anglesbonsintervalles102.tikz}$}\\\\
\overset{\text{\cref{passagephasepihbcircuits}}}{=}&\scalebox{0.8}{$\input{anglesbonsintervalles103.tikz}$}\\\\
\overset{\text{\cref{symmetriesemicontrolee,XX}}}{=}&\scalebox{0.8}{$\input{anglesbonsintervalles104.tikz}$}\\\\
\overset{\text{\cref{passagephasepihbcircuits}}}{=}&\scalebox{0.8}{$\input{anglesbonsintervalles105.tikz}$}\\\\
\overset{\text{\cref{commctrl}}}{=}&\scalebox{0.8}{$\input{anglesbonsintervalles106.tikz}$}\\\\
\overset{\text{\crefnosort{prop:CP,commctrl,prop:comb}}}{=}&\scalebox{0.8}{$\input{anglesbonsintervalles107.tikz}$}\\\\
\overset{\text{\cref{commctrlphaseenhautP,commctrlphaseenhaut,prop:sum}}}{=}&\scalebox{0.8}{$\input{anglesbonsintervalles108.tikz}$}
\end{longtable}
with $\delta_2-\pi\in[0,\pi)$. Moreover, since $\delta_2\neq0$, by assumption $\delta_3\neq0$ and $\delta_4\neq0$, so that $2\pi-\delta_3$ and $2\pi-\delta_4$ are still in $[0,2\pi)$ and the previous assumptions are preserved.\bigskip

\phantomsection\label{delta5in0pi}If $\delta_5\notin[0,\pi)$, then by \cref{prop:period}, we can ensure that it is in $[0,2\pi)$, and then if it is in $[\pi,2\pi)$, then:
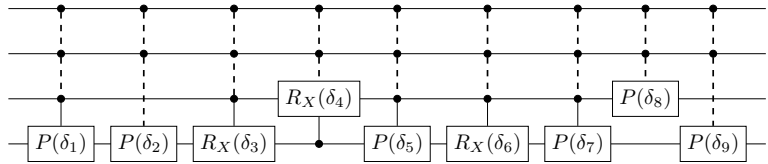
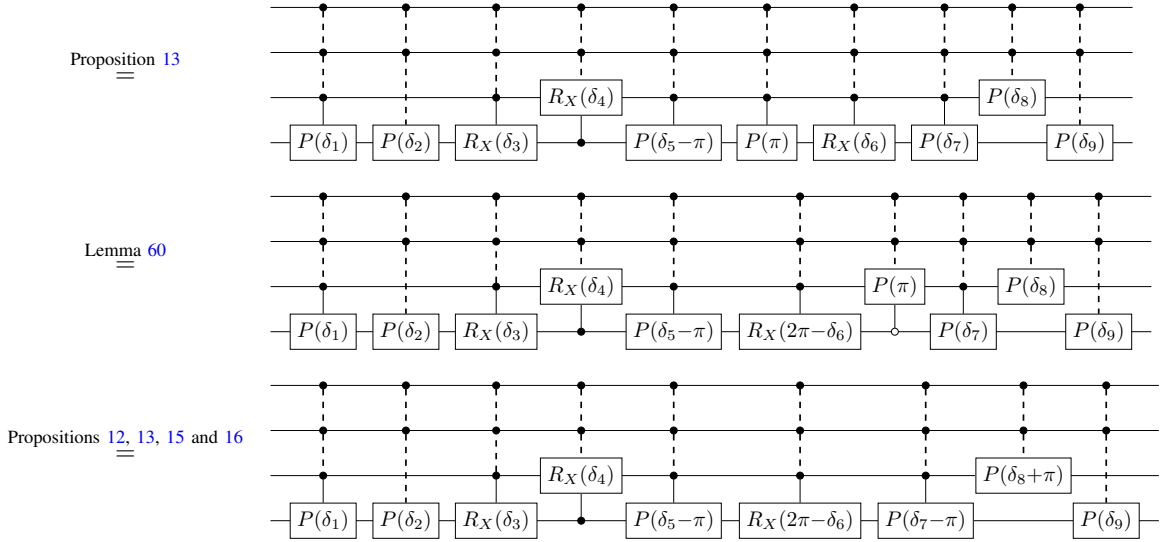
\begin{longtable}{CL}
&\scalebox{0.8}{$\input{Euler3dright-multicontrolled-simp-deltas.tikz}$}\\\\
\overset{\text{\cref{prop:sum}}}{=}&\scalebox{0.8}{$\input{anglesbonsintervalles109.tikz}$}\\\\
\overset{\text{\cref{passagephasepihbcircuits}}}{=}&\scalebox{0.8}{$\input{anglesbonsintervalles110.tikz}$}\\\\
\overset{\text{\cref{prop:CP,commctrl,prop:sum,prop:comb}}}{=}&\scalebox{0.8}{$\input{anglesbonsintervalles111.tikz}$}
\end{longtable}
with $\delta_5-\pi\in[0,\pi)$. Moreover, since $\delta_5\neq0$, by assumption $\delta_6\neq0$, so that $2\pi-\delta_6\in[0,2\pi)$ and the previous assumptions are preserved.\bigskip

\phantomsection\label{delta789in02pi}Finally, by \cref{prop:period} we can put $\delta_7$, $\delta_8$ and $\delta_9$ in $[0,2\pi)$ without modifying the other angles.

\subsection[Definition of sigma\_k,n,l]{Definition of $\sigma_{k,n,\ell}$}\label{defsigma}

$\sigma_{k,0,\ell}\coloneqq(\begin{tikzpicture}
	\begin{pgfonlayer}{nodelayer}
		\node [style=none] (0) at (-0.8, 0) {};
		\node [style=none] (1) at (0.8, 0) {};
	\end{pgfonlayer}
	\begin{pgfonlayer}{edgelayer}
		\draw [style=noire] (0.center) to (1.center);
	\end{pgfonlayer}
\end{tikzpicture})^{\otimes 2^{k+\ell}}$ and $\forall n\geq2,\ \displaystyle\sigma_{k,n,\ell}\coloneqq%
\sigma_{k,1,\ell+n-1}^n$, with 
\[\sigma_{k,1,\ell}=\prod_{j=k+1}^{k+\ell}\mathcal P_j\mathcal Q_j\mathcal P_j\]
where
\begin{itemize}
\item given a family of $N$-mode circuits $C_A,...,C_B$,\quad $\displaystyle\prod_{i=A}^BC_i\coloneqq 
(\ldots((C_B\circ C_{B-1})\circ C_{B-2})\circ \ldots )\circ C_A$,
\item $M\coloneqq k+\ell+1$
\item $\mathcal P_j$ is a raw optical circuit such that $\mathfrak G_n\circ\interp{\mathcal P_j}\circ\mathfrak G_n^{-1}=id_{j-1}\otimes\interp{\gcnot}\otimes id_{M-j-1}$, defined as\newline $\displaystyle \mathcal P_j\coloneqq\hspace{-1.5em}\prod_{\begin{scriptarray}{c}\\[-1.5em]b=0\\[-0.4em]b\bmod 4\in\{1,2\}\end{scriptarray}}^{2^{j}-1}\hspace{-1.5em}\prod_{a=0}^{2^{M-j-1}-1}\hspace{-1em}\upsilon_{M,j,b,a}$
\item $\mathcal Q_j$ is a raw optical circuit such that $\mathfrak G_n\circ\interp{\mathcal Q_j}\circ\mathfrak G_n^{-1}=id_{j-1}\otimes\interp{\gnotc}\otimes id_{M-j-1}$, defined as\newline $\displaystyle \mathcal Q_j\coloneqq\prod_{b=0}^{2^{j-1}-1}\ \ \prod_{a=0}^{2^{M-j-3}-1}\hspace{-1em}\upsilon_{M,j-1,b,a}$
\item $\upsilon_{N,i,b,a}$ is a raw optical circuit such that $\upsilon_{N,i,b,a}\equiv\scalebox{0.69}{\input{swap-ibaN-etroit.tikz}}$\bigskip. It is defined for any $N\geq1$, $i\in\{0,...,N-1\}$, $b\in\{0,...,2^i-1\}$ and $a\in\{0,...,2^{N-i-1}-1\}$, 
by finite induction on $a$ by \[\upsilon_{N,i,b,0}\coloneqq\input{swapib2pNi1-PHOL.tikz},\] and for $a\in\{1,...,2^{N-i-1}-1\}$, 
\[\upsilon_{N,i,b,a}\coloneqq s_{-a}\circ s_{+a}\circ\upsilon_{N,i,b,a-1}\circ s_{+a}\circ s_{-a}\bigskip,\] where $s_{+a}\coloneqq\input{swapib2pNi1plusa-PHOL.tikz}$\quad and\quad $s_{-a}\coloneqq\input{swapib2pNi1moinsa-PHOL.tikz}$\bigskip.
\end{itemize}

\end{document}